%% file: fas-eprint.tex
\newif\ifacm
\acmfalse

\newif\ifhldiff
\hldifffalse

\newif\ifcameraready
\camerareadyfalse

\ifacm 
\documentclass[sigconf]{acmart}
\else

\documentclass[11pt,letterpaper]{article}
\usepackage[margin=1in]{geometry}

\fi

\usepackage{amsthm}

\ifacm\else 
\usepackage{cite}
\fi
\usepackage[utf8]{inputenc}

\ifacm\else 
\usepackage[bookmarks=true,pdfstartview=FitH,colorlinks,linkcolor=darkred,filecolor=darkred,citecolor=darkred,urlcolor=darkred]{hyperref}
\fi

\usepackage{multicol}
\usepackage{colortbl}
\usepackage{mdframed}
\usepackage{color,soul}

\usepackage{paralist}
\usepackage{pdfpages}
\usepackage{enumitem}
\usepackage{float}
\usepackage{booktabs}
\usepackage[override]{cmtt}
\usepackage[dvipsnames]{xcolor}
\colorlet{hldiffcolor}{ForestGreen}
\ifacm\else
\usepackage{authblk}
\fi
\usepackage{graphicx,color,eso-pic}

\usepackage{algorithm}% http://ctan.org/pkg/algorithms
\usepackage{algpseudocode}% http://ctan.org/pkg/algorithmicx

\usepackage{soul}
\usepackage[normalem]{ulem}

\usepackage{tablefootnote}

\usepackage{tabu}

\usepackage[
	lambda,
	advantage,
	operators,
	sets,
	adversary,
	landau,
	probability,
	notions,
	logic,
	ff,
	mm,
	primitives,
	events,
	complexity,
	asymptotics,
	keys]
{cryptocode}
\usepackage{pgfplots}
\usepgfplotslibrary{fillbetween}
\usepackage{balance}

\ifacm
\usepackage{fancyhdr}
\else 
\fi

\input{macros}

\ifacm %%%%%%%%%%%%%%%%%%%%%%%%%%%%%%%%%%%%%%%%%%%%%%%%%%%%%%%%%%%%%%

\copyrightyear{2024}
\acmYear{2024}
\setcopyright{rightsretained}
\acmConference[CCS '24]{Proceedings of the 2024 ACM SIGSAC Conference on Computer and Communications Security}{October 14--18, 2024}{Salt Lake City, UT, USA}
\acmBooktitle{Proceedings of the 2024 ACM SIGSAC Conference on Computer and Communications Security (CCS '24), October 14--18, 2024, Salt Lake City, UT, USA}
\acmDOI{10.1145/3658644.3690240}
\acmISBN{979-8-4007-0636-3/24/10}

\makeatletter
\gdef\@copyrightpermission{
  \begin{minipage}{0.3\columnwidth}
   \href{https://creativecommons.org/licenses/by/4.0/}{\includegraphics[width=0.90\textwidth]{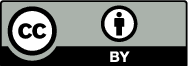}}
  \end{minipage}\hfill
  \begin{minipage}{0.7\columnwidth}
   \href{https://creativecommons.org/licenses/by/4.0/}{This work is licensed under a Creative Commons Attribution International 4.0 License.}
  \end{minipage}
  \vspace{5pt}
}
\makeatother

\begin{document}
% \fancyhead{}
\pagestyle{fancy}
\begin{CCSXML}
<ccs2024>
<concept>
<concept_id>10002978.10002979</concept_id>
<concept_desc>Security and privacy~Cryptography</concept_desc>
<concept_significance>500</concept_significance>
</concept>
</ccs2024>
\end{CCSXML}
\ccsdesc[500]{Security and privacy~Cryptography}

%\title{Functional Adaptor Signatures}
\ifhldiff
\title{Functional Adaptor Signatures: Beyond All-or-Nothing Blockchain-based Payments {\color{hldiffcolor}(Highlight Version)}}
\else
\title{Functional Adaptor Signatures: Beyond All-or-Nothing Blockchain-based Payments}
\fi
%\title{Functional Adaptor Signatures for Fine-grained Sale of Digital Goods}
%\title{Functional Adaptor Signatures: A New Vision for Blockchain-based Payments}

\author{Nikhil Vanjani}
\affiliation{%
 \institution{Carnegie Mellon University}
 \city{Pittsburgh}
 \country{USA}
 }
\email{nvanjani@cmu.edu}
\author{Pratik Soni}
\affiliation{%
 \institution{University of Utah}
 \city{Salt Lake City}
 \country{USA}
 }
\email{psoni@cs.utah.edu}
\author{Sri AravindaKrishnan Thyagarajan}
\affiliation{%
 \institution{University of Sydney}
 \city{Sydney}
 \country{Australia}
 }
\email{t.srikrishnan@gmail.com}

\input{abstract}

\keywords{Adaptor Signatures; Functional Encryption}

% \makealltitles
\maketitle

\else %%%%%%%%%%%%%%%%%%%%%%%%%%%%%%%%%%%%%%%%%%%%%%%%%%%%%%%%%%%%%%

%\title{Functional Adaptor Signatures}
\title{Functional Adaptor Signatures: Beyond All-or-Nothing Blockchain-based Payments}
%\title{Functional Adaptor Signatures for Fine-grained Sale of Digital Goods}
%\title{Functional Adaptor Signatures: A New Vision for Blockchain-based Payments}

\author[1]{Nikhil Vanjani} %\\ {\tt nvanjani@andrew.cmu.edu}}
\author[2]{Pratik Soni} %\\{\tt psoni@andrew.cmu.edu} 
\author[3]{Sri AravindaKrishnan Thyagarajan} %\\{\tt t.srikrishnan@gmail.com} 

\affil[1]{Carnegie Mellon University}
\affil[2]{University of Utah}
\affil[3]{University of Sydney}
\date{}

\renewcommand{\paragraph}[1]{\vspace{10pt}\noindent\textbf{#1}}

% \copyrightyear{2024}
% \acmYear{2024}
% \setcopyright{rightsretained}
% \acmConference[CCS '24]{Proceedings of the 2024 ACM SIGSAC Conference on Computer and Communications Security}{October 14--18, 2024}{Salt Lake City, UT, USA}
% \acmBooktitle{Proceedings of the 2024 ACM SIGSAC Conference on Computer and Communications Security (CCS '24), October 14--18, 2024, Salt Lake City, UT, USA}
% \acmDOI{10.1145/3658644.3690240}
% \acmISBN{979-8-4007-0636-3/24/10}

% \renewcommand{\shortauthors}{Nikhil Vanjani, Pratik Soni, \& Sri AravindaKrishnan Thyagarajan}

\begin{document}

% \begin{titlepage}
\maketitle
%\noindent
%\makebox[\linewidth]{\small \today}
\input{abstract}
% \end{titlepage}

\newpage
{\small\tableofcontents}
\newpage

\fi %%%%%%%%%%%%%%%%%%%%%%%%%%%%%%%%%%%%%%%%%%%%%%%%%%%%%%%%%%%%%%

% \tableofcontents

% \newpage
\pagestyle{plain}
\pnote{we can even move the syntax of different algorithm to the appendix, we have already introduced them earlier in section 1/2; this will create more space for lemmas about hybrid indistinguishability in 5 and robustness proofs in 6.}
% \pnote{I think the small plaintext discussion should be moved to 1 or 2}
\input{intro}

\input{overview}
% \input{prelims}
\input{prelims-eprint}
% \input{definitions}
\input{definitions-eprint}

\input{fas-construction}
\input{fas-from-prime-groups-eprint}

% \input{fas-from-lattices}
\input{fas-from-lattices-eprint}

\input{performance-evaluation}

\section*{Acknowledgements}
This work was supported in part through a gift awarded to Pratik Soni and Sri AravindaKrishnan Thyagarajan under the Stellar Development Foundation Academic Research Grants program.
This work was also sponsored by the Algorand Foundation, the National Science Foundation under award numbers 2044679 and 2212746, and the Office of Naval Research under award number N000142212064.  The views and conclusions contained in this document are those of the author and should not be interpreted as representing the official policies, either expressed or implied, of any sponsoring institution, the U.S. government or any other entity.

\ifacm
\bibliographystyle{ACM-Reference-Format}
\else
\bibliographystyle{alpha}
\fi
\balance
\bibliography{refs}

\ifcameraready\else

\appendix 
\input{fas-proof-appendix}

\input{fas-construction-weak}

\fi

\ifacm
\pagestyle{fancy}
\else
\pagestyle{plain}
\fi
\end{document}

%% file: macros.tex
\sethlcolor{lightgray}
\mathchardef\mhyphen="2D

\newcommand{\msf}[1]{\ensuremath{{\mathsf {#1}}}}
\newcommand{\mcal}[1]{\ensuremath{\mathcal {#1}}}

\newcommand{\algA}{{\ensuremath{\mcal{A}}}\xspace}

\newcommand{\Sim}{\ensuremath{\msf{Sim}}\xspace}

\ifacm 
\newcommand{\A}{\ensuremath{{\mcal A}}\xspace}
\newcommand{\B}{\ensuremath{{\mcal B}}\xspace}
\newcommand{\C}{\ensuremath{{\mcal C}}\xspace}

\renewcommand{\S}{{\mcal S}}

\else
\newcommand{\A}{\ensuremath{{\cal A}}\xspace}
\newcommand{\B}{\ensuremath{{\cal B}}\xspace}
\newcommand{\C}{\ensuremath{{\cal C}}\xspace}

\renewcommand{\S}{{\cal S}}

\fi

\renewcommand{\H}{{\mathsf{H}}}

\newcommand{\ct}{\ensuremath{{\sf ct}}\xspace}

\renewcommand{\ppt}{\ensuremath{{\text{p.p.t.}}}\xspace}

\newcommand{\R}{\ensuremath{\mathbb{R}}}
\newcommand{\Ring}{\ensuremath{\mcal{R}}\xspace}

\renewcommand{\secparam}{\ensuremath{\lambda}\xspace}

\newcommand{\ipe}{\ensuremath{{{\sf IPFE}}}\xspace}

\newcommand{\ipfe}{\ensuremath{{{\sf IPFE}}}\xspace}
\newcommand{\IPFE}{\ensuremath{{{\sf IPFE}}}\xspace}
\newcommand{\ipfeSimReal}{\ensuremath{{\sf IPFE\text{-}SIM\text{-}Real}_{\A, \ipfe}}\xspace}
\newcommand{\ipfeSimIdeal}{\ensuremath{{\sf IPFE\text{-}SIM\text{-}Ideal}_{\A, \ipfe}^\Sim}\xspace}
\newcommand{\ipfeIND}{\ensuremath{{\sf IPFE\text{-}IND\text{-}SEL}_{\A, \ipfe}}\xspace}

\newcommand{\hnfsis}{\ensuremath{{{\sf HNF\text{-}SIS}}}\xspace}
\newcommand{\sis}{\ensuremath{{{\sf SIS}}}\xspace}

\newcommand{\isis}{\ensuremath{{{\sf ISIS}}}\xspace}

\newcommand{\lwe}{\ensuremath{{{\sf LWE}}}\xspace}
\newcommand{\LWE}{\ensuremath{{{\sf LWE}}}\xspace}

\newcommand{\mheLWE}{\ensuremath{{{\sf mheLWE}}}\xspace}

\newcommand{\crs}{\ensuremath{{\sf crs}}\xspace}

\newcommand{\Gen}{\ensuremath{{\sf Gen}}\xspace}
\newcommand{\Setup}{\ensuremath{{\sf Setup}}\xspace}

\newcommand{\Hyb}{\ensuremath{{\sf Hyb}}\xspace}

\newcommand{\Enc}{\ensuremath{{\sf Enc}}\xspace}

\newcommand{\Dec}{\ensuremath{{\sf Dec}}\xspace}

\newcommand{\KGen}{\ensuremath{{\sf KGen}}\xspace}
\newcommand{\PubKGen}{\ensuremath{{\sf PubKGen}}\xspace}

\newcommand{\ret}{\ensuremath{{\bf ret}}\xspace}

\newcommand{\TrapdoorGen}{\ensuremath{{\sf TrapdoorGen}}\xspace}

\newcommand{\GenR}{\ensuremath{\mathsf{GenR}}\xspace}

\newcommand{\AdvertisementGen}{\ensuremath{{\mathsf{AdGen}}}\xspace}
\newcommand{\AdGen}{\ensuremath{{\mathsf{AdGen}}}\xspace}

\newcommand{\AdvertisementVerify}{\ensuremath{{\mathsf{AdVerify}}}\xspace}
\newcommand{\AdVerify}{\ensuremath{{\mathsf{AdVerify}}}\xspace}
\newcommand{\AuxGen}{\ensuremath{{\mathsf{AuxGen}}}\xspace}
\newcommand{\AuxVerify}{\ensuremath{{\mathsf{AuxVerify}}}\xspace}
\newcommand{\PreSign}{\ensuremath{\mathsf{PreSign}}\xspace}
\newcommand{\FPreSign}{\ensuremath{{\mathsf{FPreSign}}}\xspace}
\newcommand{\PreVerify}{\ensuremath{\mathsf{PreVerify}}\xspace}
\newcommand{\FPreVerify}{\ensuremath{{\mathsf{FPreVerify}}}\xspace}
\newcommand{\Adapt}{\ensuremath{\mathsf{Adapt}}\xspace}
\newcommand{\Ext}{\ensuremath{\mathsf{Ext}}\xspace}
\newcommand{\FExt}{\ensuremath{{\mathsf{FExt}}}\xspace}

\newcommand{\aSigForge}{\ensuremath{\mathsf{aSigForge}}\xspace}
\newcommand{\faSigForge}{\ensuremath{\mathsf{faEUF{\mhyphen}CMA}}\xspace}
\newcommand{\aExt}{\ensuremath{\mathsf{aExt}}\xspace}
\newcommand{\aUniqueExt}{\ensuremath{\mathsf{aUniqueExt}}\xspace}
\newcommand{\aWitExt}{\ensuremath{\mathsf{aWitExt}}\xspace}
\newcommand{\aUnlink}{\ensuremath{\mathsf{aUnlink}}\xspace}
\newcommand{\faWitExt}{\ensuremath{\mathsf{faWitExt}}\xspace}
\newcommand{\faWitInd}{\ensuremath{\mathsf{faWI}}\xspace}
\newcommand{\faWI}{\ensuremath{\mathsf{faWI}}\xspace}

\newcommand{\faWH}{\ensuremath{\mathsf{faWH}}\xspace}

\newcommand{\faZKReal}{\ensuremath{\mathsf{faZKReal}}\xspace}

\newcommand{\faZKIdeal}{\ensuremath{\mathsf{faZKIdeal}}\xspace}

\renewcommand{\st}{\ensuremath{\mathsf{st}}\xspace}
\newcommand{\stmt}{\ensuremath{\mathsf{stmt}}\xspace}
\newcommand{\wit}{\ensuremath{\mathsf{wit}}\xspace}

\newcommand{\DS}{\ensuremath{\mathsf{DS}}\xspace}
\newcommand{\ds}{\ensuremath{\mathsf{DS}}\xspace}
\newcommand{\AS}{\ensuremath{\mathsf{AS}}\xspace}
\newcommand{\as}{\ensuremath{\mathsf{AS}}\xspace}

\newcommand{\FAS}{\ensuremath{\mathsf{FAS}}\xspace}
\newcommand{\fas}{\ensuremath{\mathsf{FAS}}\xspace}
\newcommand{\Sch}{\ensuremath{\mathsf{Sch}}\xspace}
\newcommand{\Lyu}{\ensuremath{\mathsf{Lyu'}}\xspace}

\newcommand{\Sign}{{\ensuremath{{\sf Sign}}\xspace}}

\newcommand{\Prove}{{\ensuremath{{\sf Prove}}\xspace}}
\newcommand{\Verify}{{\ensuremath{{\sf Vf}}\xspace}}

\newcommand{\NIZK}{{\ensuremath{\mathsf{NIZK}}\xspace}}
\renewcommand{\nizk}{{\ensuremath{\mathsf{NIZK}}\xspace}}

\ifacm
\newcommand{\D}{\ensuremath{{\mcal D}}\xspace}
\else
\newcommand{\D}{\ensuremath{{\cal D}}\xspace}
\fi

\definecolor{darkgreen}{rgb}{0,0.5,0}
\definecolor{lightblue}{RGB}{0,176,240}
\definecolor{darkblue}{RGB}{0,112,192}
\definecolor{lightpurple}{RGB}{124, 66, 168}
\definecolor{grey}{RGB}{139, 137, 137}
\definecolor{maroon}{RGB}{178, 34, 34}
\definecolor{green}{RGB}{34, 139, 34}
\definecolor{types}{RGB}{72, 61, 139}
\definecolor{gold}{rgb}{0.8, 0.33, 0.0}
\definecolor{mygreen}{HTML}{588D6A}
\definecolor{myred}{HTML}{C86733}
\definecolor{myblue}{HTML}{5B68FF}
\definecolor{myshadow}{HTML}{E6C5B4}

\definecolor{darkgray}{gray}{0.3}

\newcommand{\skiptext}[1]{}

\newcommand{\getr}{\ensuremath{{\overset{\$}{\leftarrow}}}\xspace}

\newcommand{\mpk}{\ensuremath{{\sf mpk}}\xspace}
\newcommand{\msk}{\ensuremath{{\sf msk}}\xspace}

\usepackage{boxedminipage}
\usepackage{url}
\usepackage{graphicx}
\ifacm
\usepackage[justification=centering]{caption}
\else
\usepackage[bf,justification=centering]{caption}
\fi
\usepackage[justification=centering]{subcaption}
\usepackage{color}
\usepackage{microtype}
\usepackage{xspace}
\usepackage{multirow}
\usepackage{amsthm,amstext}
\ifacm\else
\usepackage{amsmath,amstext,amsfonts,amssymb,latexsym}
\fi
\usepackage{amsbsy}
\usepackage{stmaryrd}
\usepackage{pifont}

\usepackage[capitalize,noabbrev]{cleveref} %% for referencing use \cref

\usepackage{wrapfig}
\definecolor{darkred}{rgb}{0.5, 0, 0}
\definecolor{darkgreen}{rgb}{0, 0.5, 0}
\definecolor{darkblue}{rgb}{0,0,0.5}

\newcommand\markx[2]{}

\renewcommand{\sk}{\ensuremath{\mathsf{sk}}\xspace}

\newcommand{\aux}{{\sf aux}}

\renewcommand{\negl}{{\sf negl}}

\newcommand{\N}{\ensuremath{\mathbb{N}}\xspace}
\newcommand{\G}{\ensuremath{\mathbb{G}}\xspace}
\newcommand{\Z}{\ensuremath{\mathbb{Z}}\xspace}

\newcommand{\ignore}[1]{}

\renewcommand{\part}{\ensuremath{\mcal{S}}\xspace}

\newcounter{task}

\newtheorem{thm}{Theorem}[section]      % A counter for Algorithms etc

\newtheorem{theorem}[thm]{Theorem}

\newtheorem{conjecture}[thm]{Conjecture}
\newtheorem{lemma}[thm]{Lemma}
\newtheorem{claim}[thm]{Claim}
\newtheorem{corollary}[thm]{Corollary}
\newtheorem{fact}[thm]{Fact}
\newtheorem{proposition}[thm]{Proposition}
\newtheorem{example}[thm]{Example}

\newtheorem{remark}[thm]{Remark}
\Crefname{claim}{Claim}{Claims}
\crefname{claim}{claim}{claims}

\newtheorem{definition}[thm]{Definition}

\newtheoremstyle{boxes}% name  
{2pt}% ?Space above ? 
{0pt}% ?Space below ?
{}% ?Body font ?
{}% ?Indent amount ?
{\bfseries}% ?Theorem head font?
{}% ?Punctuation after theorem head ?
{\newline}% ?Space after theorem head ?
{\thmname{#1}\thmnumber{ #2}:  
\thmnote{#3}}

\theoremstyle{boxes}
\newcommand{\elaine}[1]{{\footnotesize\color{magenta}[Elaine: #1]}}
\newcommand{\nikhil}[1]{{\footnotesize\color{red}[Nikhil: #1]}}
\newcommand{\pnote}[1]{{\footnotesize\color{green}[Pratik: #1]}}
\newcommand{\anote}[1]{{\footnotesize\color{orange}[Aravind: #1]}}
\newcommand{\rex}[1]{{\footnotesize\color{teal}[Rex: #1]}}

\renewcommand{\elaine}[1]{}
\renewcommand{\nikhil}[1]{}
\renewcommand{\pnote}[1]{}
\renewcommand{\anote}[1]{}
\renewcommand{\rex}[1]{}

\newcounter{cnt:challenge}

\newcommand{\bfa}{{\bf a}}
\newcommand{\bfb}{{\bf b}}
\newcommand{\bfc}{{\bf c}}

\newcommand{\bfe}{{\bf e}}

\newcommand{\bfr}{{\bf r}}
\newcommand{\bfs}{{\bf s}}
\newcommand{\bft}{{\bf t}}
\newcommand{\bfu}{{\bf u}}
\newcommand{\bfv}{{\bf v}}
\newcommand{\bfw}{{\bf w}}
\newcommand{\bfx}{{\bf x}}
\newcommand{\bfy}{{\bf y}}
\newcommand{\bfz}{{\bf z}}
\newcommand{\bfA}{{\bf A}}

\newcommand{\bfE}{{\bf E}}

\newcommand{\bfR}{{\bf R}}
\newcommand{\bfS}{{\bf S}}
\newcommand{\bfT}{{\bf T}}
\newcommand{\bfU}{{\bf U}}

\newcommand{\bfX}{{\bf X}}

\newcommand{\bfZ}{{\bf Z}}
\newcommand{\bfnum}[1]{{\bf #1}}

\renewcommand{\pp}{\mathsf{pp}}

\newcommand{\advt}{\mathsf{advt}}

\renewcommand{\state}{\ensuremath{\mathsf{st}}}
\newcommand{\GroupGen}{\ensuremath{\mathsf{GGen}}\xspace}
\newcommand{\DL}{\ensuremath{\mathsf{DL}}\xspace}
\newcommand{\DLog}{\ensuremath{\mathsf{DLog}}\xspace}

\newcommand{\fpS}{\ensuremath{\mathsf{fpS}}\xspace}
\newcommand{\td}{\ensuremath{\mathsf{td}}\xspace}

%% file: abstract.tex
\begin{abstract}
%Attempt 3

In scenarios where a seller holds sensitive data $x$, like employee / patient records or ecological data, and a buyer seeks to obtain an evaluation of specific function $f$ on this data, solutions in trustless digital environments like blockchain-based Web3 systems typically fall into two categories: (1) Smart contract-powered solutions and (2) cryptographic solutions leveraging tools such as adaptor signatures.
The former approach offers atomic transactions where the buyer learns the function evaluation $f(x)$ (and not $x$ entirely) upon payment.
However, this approach is often inefficient, costly, lacks privacy for the seller's data, and is incompatible with systems that do not support smart contracts with required functionalities.
In contrast, the adaptor signature-based approach addresses all of the above issues but comes with an "all-or-nothing" guarantee, where the buyer fully extracts $x$ and does not support functional extraction of the sensitive data.
In this work, we aim to bridge the gap between these approaches, developing a solution that enables fair functional sales of information while offering improved efficiency, privacy, and compatibility similar to adaptor signatures.

Towards this, we propose \emph{functional adaptor signatures (FAS)} a novel cryptographic primitive that achieves all the desired properties as listed above.
Using FAS, the seller can publish an advertisement committing to $x$. 
The buyer can pre-sign the payment transaction w.r.t.\ a function $f$, and send it along with the transaction to the seller.
The seller adapts the pre-signature into a valid (buyer's) signature and posts the payment and the adapted signature on the blockchain to get paid.
Finally, using the pre-signature and the posted signature, the buyer efficiently extracts $f(x)$, and completes the sale.
We formalize the security properties of FAS, among which is a new notion called \emph{witness privacy} to capture seller's privacy, which ensures the buyer does not learn anything beyond $f(x)$.
We present multiple variants of witness privacy, namely, \emph{witness hiding, witness indistinguishability}, and \emph{zero-knowledge}, to capture varying levels of leakage about $x$ beyond $f(x)$ to a malicious buyer.

We introduce two efficient constructions of FAS supporting linear functions (like statistics/aggregates, kernels in machine learning, etc.), that satisfy the strongest notion of witness privacy.
 One construction is based on prime-order groups and compatible with Schnorr signatures for payments, and the other is based on lattices and compatible with a variant of Lyubashevsky's signature scheme.
A central conceptual contribution of our work lies in revealing a surprising connection between functional encryption, a well-explored concept over the past decade, and adaptor signatures, a relatively new primitive in the cryptographic landscape.
On a technical level, we avoid heavy cryptographic machinery and achieve improved efficiency, by making black-box use of building blocks like \emph{inner product functional encryption (IPFE)} while relying on certain security-enhancing techniques for the IPFE in a non-black-box manner.
We implement our FAS construction for Schnorr signatures and show that for reasonably sized seller witnesses, the different operations are quite efficient even for commodity hardware.

\end{abstract}

%% file: intro.tex
\section{Introduction}

The shift from centralized Web2 to decentralized blockchain-based Web3 solutions has transformed digital goods trading.
Smart contracts have been pivotal in this transition, and they facilitate seamless transactions between sellers and buyers.
Imagine a seller offering a solution to a computational puzzle or a problem for sale.
Typical smart contracts let the seller specify the terms of their digital offering, referred to as an \emph{advertisement}, creating a transparent and accessible mechanism for buyers.
The buyers lock coins to the contract as an expression of interest to buy the solution.
Now the seller publishes the solution in the form of a transaction on the blockchain that invokes the contract.
The contract executes a validation procedure to verify the solution's correctness, and if successful, transfers the locked coins to the seller.
If the seller fails to respond with a correct solution, the contract refunds the buyer after the expiry of a timeout.
This template underpins various digital trading scenarios including Hash Timelock Contracts~\cite{lightningnetwork}, Bridges~\cite{bridge}, NFT sales~\cite{nft}, smart contract-based offline services~\cite{chainlink}, and e-commerce~\cite{khan2022revolutionizing}.

Despite the advantages of smart contracts, their current implementation faces significant challenges:
\begin{itemize}[leftmargin=*]
    \item \emph{Cost and Efficiency}: Sellers often incur hefty fees for posting advertisements on smart contracts and verifying secret solutions, resulting in higher transaction costs. Ethereum measures these costs in gas, and many popular smart contract applications have significantly higher gas costs compared to regular user-to-user payment transactions~\cite{livegascost}.
    \item \emph{Privacy}: Openly disclosing the secret solution compromises privacy. Encrypting values on the contract to address this necessitates more expensive and intricate cryptographic tools,  further exacerbating costs. Further, the approach affects the 
    \ifcameraready
    fungibility
    \else
    fungibility\footnote{A blockchain system is
    said to be fungible if all units/coins in the system have the same value, independent of their history.} 
    \fi
    of the system as pointed out in many previous works~\cite{erwig2021two,vts,thyagarajan2021lockable,thyagarajan2022universal,hanzlik2022sweep}.
    \item \emph{Compatibility}: Finally, the method relies on \emph{complex} smart contracts, rendering itself incompatible with prominent blockchains like Bitcoin. This poor adaptability hampers scalability in the broader Web3 ecosystem.
\end{itemize}

\smallskip\noindent\textbf{Adaptor Signatures.} As the community delves deeper into enhancing the efficiency and simplicity of smart contracts, cryptographic tools like adaptor signatures~\cite{malavolta2019anonymous,pq-as,asig,erwig2021two,thyagarajan2022universal,hanzlik2022sweep,qin2023blindhub} have emerged as promising solutions. 
Adaptor signatures help model a blockchain-based fair exchange between a buyer for its signature $\sigma$ (on a transaction) and an NP witness $x$ that is known to the seller.
\ifhldiff
{\color{hldiffcolor}
\fi
More formally, the seller first samples an NP statement $X$ along with its witness $x$. The statement $X$ binds the seller to its witness which it publishes to advertise its intent to sell $x$.\footnote{Typically, this step is not explicitly stated in the formalization of adaptor signatures. We adopt this extension to aid our functional adaptor signature formalization.} 
Now, adaptor signatures offers four algorithms $\PreSign, \PreVerify, \Adapt, \Ext$, that work as follows.
The buyer using the pre-sign algorithm $\PreSign$ first generates a pre-signature $\widetilde{\sigma}$ on the payment transaction $m$ w.r.t.\ the signing key $\sk$ and seller's statement $X$.
The seller verifies the validity of $\widetilde{\sigma}$ using $\PreVerify$ algorithm. Then, using the $\Adapt$ algorithm, the seller adapts $\widetilde{\sigma}$ into a full signature $\sigma$ on $m$ using its witness $x$. To redeem the transaction $m$, the seller posts $\sigma$ on the blockchain. 
\ifhldiff
}
\fi
% Adaptor signatures enable a signer to pre-sign a message $m$ with respect to a statement $X$ of an NP language $L$ and send it along with the pre-signature $\widetilde\sigma$ to the designated recipient.
% The recipient can adapt the pre-signature into a valid signature $\sigma$ using the secret witness $x$ corresponding to the statement $X$. 
% In the case of digital trading, the buyer generates the pre-signature on the message which is a payment transaction.
% The seller receives the pre-signature and the transaction, and adapts using the secret solution as the witness $x$.
Central to the efficacy of adaptor signatures is their unique property that allows the buyer, with access to the pre-signature  $\widetilde\sigma$ and the posted signature $\sigma$, to efficiently \emph{extract} the secret solution $x$ using the algorithm $\Ext$, thus completing the fair exchange.
This innovative approach addresses the drawbacks of traditional smart contracts:
\begin{itemize}[leftmargin=*]
	\item \emph{Improved efficiency and privacy:} Adaptor signatures require only a single transaction and signature to be posted on the blockchain, reducing transaction costs and enhancing efficiency. Moreover, with buyers locally extracting secrets, transactions resemble regular payments, thereby improving the privacy and fungibility of coins in the system.
	\item \emph{Enhanced compatibility:} Signature verification, a universally supported operation, ensures compatibility across all blockchain systems, making adaptor signatures-based digital sales compatible with the broader Web3 landscape.
\end{itemize}

\smallskip\noindent\textbf{Looking ahead: Granular and functional sale of digital information.} However, while adaptor signatures offer a promising solution for digital transactions, they come with certain limitations. 
They operate on an all-or-nothing basis, revealing the entire secret upon successful transactions or learning nothing in case the seller aborts. This lack of granularity contrasts with the concept of \emph{functional encryption (FE)}~\cite{fedefn}, where recipients can decrypt and learn specific functions applied to the encrypted message instead of the complete message.
Exploring granularity not only enables more nuanced and functional sales of digital information but also ensures the privacy of the seller's data.
For instance, the seller holding medical records might desire to safeguard specific functions of the records, adhering to medical data privacy regulations like the Health Insurance Portability and Accountability Act (HIPAA) in the United States.
The privacy requirement here is that the buyer learns only the specific function for which the payment is made, ensuring compliance with legal frameworks.
Below we expand on this intuition and give two illustrative examples of real-world scenarios where functional sales of digital information can be critical.

\smallskip\noindent\textbf{Application 1: Medical information.} In this setting the seller holds a medical database and is authorized to sell functions of the database by the patients or appropriate authorities. 
For instance, the buyer who might be an insurance firm, can query some aggregate/statistics of the demography in the database.
If it's a research organization, more complex functions like correlations of pre-existing medical conditions and cancer can be queried. The seller has the financial incentive to provide the information, which, in the long run, can contribute greatly to medical advancements.

\smallskip\noindent\textbf{Application 2: ML model training information.} In this case, the seller owns a dataset, and the buyer, with an ML model, seeks to train the model on the seller's dataset. The buyer presents the model to the seller, who then trains the model over the dataset and returns the result for a price. Training may involve simple similarity measures like computing the kernel of two vectors or more complex functions with large datasets (e.g., regression analysis, clustering). 

Therefore, we ask the following question,
\begin{quote}
\centering
	\emph{Can we facilitate the functional sale of digital information using adaptor signatures?}
\end{quote}

More concretely, let the seller hold secret data $x$, and the buyer wants access to a function $f$. 
\emph{Can we expand the functionality of adaptor signatures such that, at the end of the exchange, the seller obtains a payment transaction and a signature from the buyer, and the buyer learns $f(x)$, and not the entirety of $x$?}

While general functions are the ultimate goal, this work focuses on the restricted function class of linear functions (where $f$ is a linear function), arguing that they are already widely applicable. We show that practical constructions using lightweight cryptographic tools can be achieved for linear functions, with some insights into the challenges of achieving a more expressive class of functions.

%A FAS scheme has the following interfaces:
%\begin{itemize}[leftmargin=*]
%	\item $\Setup(\secparam)$: Output the common reference string $\crs$ that will be input to all the interfaces.
%	\item $\AdvertisementGen(\crs,X,x)$: The seller executes to generate a public advertisement $\advt$ that embeds a statement $X$ of the language $L_R$ and the corresponding witness $x$. 
%	\item $\AdvertisementVerify(\crs, X, \advt)$: Interested buyers can publicly verify the well-formedness of the advertisement $\advt$.
%	\item $\FPreSign(\advt, \sk, m, X, f)$: A buyer executes to generate pre-signature $\widetilde{\sigma}$ using the signing key $\sk$ on the message $m$ for a function $f$.
%	\item $  \FPreVerify(\advt, \vk, m, X, f, \widetilde{\sigma})$: The seller can verify if the pre-signature is valid or not under the verification key $\vk$ (of the buyer).
%	\item $\Adapt(\advt, \state, \vk, m, X, x, f, \widetilde{\sigma})$: The seller can adapt the pre-signature into a valid signature $\sigma$ using the witness $x$, and return $\sigma$.
%	\item $\FExt(\advt, \widetilde{\sigma}, \sigma, X, f)$: The buyer can efficiently extract a value $z= f(z)$.
%\end{itemize}

\begin{figure}[t]
    \centering
    \captionsetup{justification=centering}
    \includegraphics[width=1.0\linewidth]{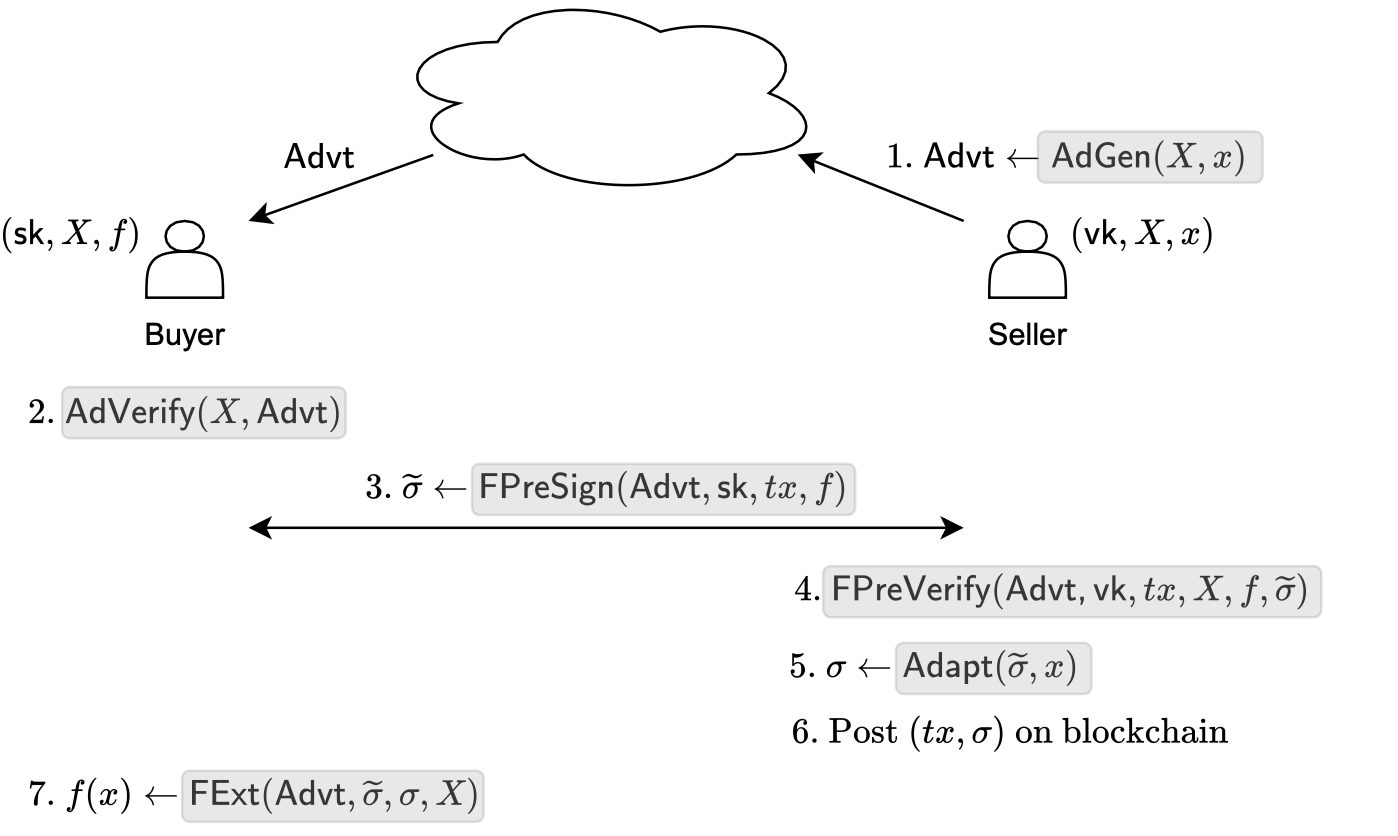}
    \caption{Functional payments via simplified FAS interface.}
    \label{fig:fas-flow}
    \label{fig:fas-protocol}
\end{figure}

\subsection{Our Contributions}
\label{sec:intro-contributions}
Below we summarize the contributions of this work.

\smallskip\noindent\textbf{New cryptographic primitive - Functional Adaptor Signatures.} 
To answer the above question affirmatively, we formally introduce a novel cryptographic primitive called \emph{functional adaptor signatures (FAS)} (in \cref{sec:FAS}). It is similar in functionality to standard adaptor signatures, except that it is additionally parameterized by a family of functions $\mathcal{F}$. This addition introduces new interfaces to FAS; we describe these below in the context of a digital trade for ease of understanding. A pictorial description is given in~\cref{fig:fas-flow}.

The seller of the trade generates an advertisement $\advt$ using the FAS interface $\AdvertisementGen$, that embeds the statement $X$ and the corresponding secret witness $x$, where $(X,x) \in R$ for some NP relation $R$. 
Anyone can publicly verify the well-formedness of $\advt$ using the $\AdvertisementVerify$ algorithm.
The buyer executes the functional pre-sign algorithm $\FPreSign$ to generate a pre-signature $\widetilde{\sigma}$ on the message $m$ (the payment transaction, denoted as $\mathit{tx}$ in~\cref{fig:fas-flow}) for the statement $X$ and a function $f$.
$\FPreVerify$ algorithm helps the seller to verify if $\widetilde{\sigma}$ is valid with respect to the message $m$, statement $X$ and function $f$.
If the pre-signature is valid, the seller using the $\Adapt$ algorithm and the witness $x$ can adapt $\widetilde{\sigma}$ into a signature $\sigma$ on the message $m$.
The seller may now post the transaction and the signature on the blockchain to get paid.
Finally, the buyer can efficiently extract the value $y = f(x)$ using the functional extract algorithm $\FExt$ given access to $\advt$, $\widetilde{\sigma}$, and $\sigma$.
With the above outline of FAS, we can see that we gain the same advantages of adaptor signatures over smart contracts for digital trading, while also achieving the desired granularity.

\smallskip\noindent\textbf{Definitional framework.} Recent works~\cite{as-stronger,as-foundations} have highlighted the inherent complexities in designing the security framework for even standard adaptor signatures, with many subtle nuances that require a lot of care.
In this work, we overcome these challenges and present a comprehensive definitional framework for the security of FAS (in~\cref{sec:FAS}).
We establish the formal foundations of FAS, constituting a substantial contribution of this work.

Broadly speaking, for a secure FAS, we are required to upgrade the security definitions of standard adaptor signatures to handle additional constraints posed by the functional aspects of the primitive.
Moreover, we introduce additional novel security properties for a secure FAS to satisfy, among which \emph{witness privacy} is a prominent one.
This property captures the leakage a malicious buyer learns about $x$ beyond $f(x)$.
To formalize this, we present three different notions of witness privacy for FAS, akin to the different privacy guarantees of cryptographic proofs.
To be more specific, we present (1) \emph{witness hiding}, where the malicious buyer after learning $f(x)$ cannot output the secret witness $x$, (2) \emph{witness indistinguishable}, where the malicious buyer learns $f(x)$ and cannot distinguish between two witnesses $x = x_0$ and $x = x_1$ (for the statement $X$) as long as $f(x_0) = f(x_1)$, and (3) \emph{zero-knowledge}, where the malicious buyer learns no other information about $x$, other than $f(x)$.

\smallskip\noindent\textbf{Efficient constructions.} From a practical perspective, we give several efficient constructions of FAS for the class of linear functions (like aggregates, statistics, kernel computation in ML, etc.) that are compatible with standard signatures like Schnorr and ECDSA.
Notably, these signature schemes are widely used in blockchain systems for transaction verification, thus making our constructions ready for deployment in current systems.
Towards a post-quantum instantiation of FAS, when the buyer has linearly independent function queries, we also present a FAS construction compatible with a variant of the lattice-based signature scheme of Lyubashevsky~\cite{lyu12lattice-sig}.
% {\color{red}
Later in~\cref{sec:TOinstantiations}, we discuss how the linear-independency requirement can be relaxed with minor trade-offs in efficiency.
% }

Most importantly, the key conceptual contribution of our work is the surprising connection between two seemingly unrelated cryptographic tools, namely, functional encryption~\cite{fedefn} that has been extensively studied for over a decade and a relatively new primitive, adaptor signatures.
Concretely, we construct FAS for a family of inner-product functions $\mcal{F}_{\sf IP}$ using three lightweight building blocks:
\begin{enumerate}[leftmargin=*]
	\item[(i)] a functional encryption scheme $\ipfe$ for $\mcal{F}_{\sf IP}$,
	\item[(ii)] an adaptor signature scheme $\as$ w.r.t.\ 
    % {\color{red}
    the digital signature scheme $\ds$
    % } 
    and some hard relation $R_\ipfe$ (related to $\ipfe$), 
	\item[(iii)] NIZK argument system for NP.
\end{enumerate}
A nice feature of our construction is that we show it satisfies the zero-knowledge style witness privacy assuming 
% {\color{red} 
only 
% } 
selective, IND-security of $\ipfe$ and 
% {\color{red}
% adaptive 
% } 
zero-knowledge of $\nizk$.

We show two ways to instantiate the construction: 
\begin{itemize}[leftmargin=*]
\item 
{\bf Prime-order groups (\Cref{sec:fas-from-prime-groups}).}
Using a varaint of the DDH-based IPFE scheme by Abdalla et al.~\cite{ipe}, we get $R_\ipfe$ is the discrete log relation. So, we use  Schnorr adaptor signature~\cite{asig} for $\as$. 
% Further, we can use NIZKs for NP from DDH~\cite{nizk-from-ddh}.
\item 
{\bf Lattices (\Cref{sec:fas-from-lattices}).}
Using a variant of the LWE-based IPFE scheme by Agrawal et al.~\cite{ipe01}, we get $R_\ipfe$ is the short integer solution (SIS) relation. So, we use a variant of the post-quantum adaptor signature~\cite{pq-as} for $\as$. 
\end{itemize}

Throughout the paper, for conceptual clarity, we abstract out the relation $R$ to be any NP relation, and make black-box use of a NIZK~\cite{peikert2019noninteractive} for relation $R$.
However, we emphasize that in both cases, the NIZK can be efficiently instantiated depending on the exact application (i.e., for a concrete relation $R$) and plugged in without affecting  other components. 
Lastly, for a weaker witness privacy notion of witness-indistinguishability, we present a round-optimal construction 
\ifcameraready
(details in the full version of the paper).
\else 
(\Cref{sec:overview_noninteractive,sec:fas-construction-weak}).
\fi
%This construction uses the same instantiations of building blocks as before, but uses them in a slightly different way. 

\smallskip\noindent\textbf{Performance evaluation.} 
To assess the practicality of our FAS, we provide
an open-source implementation~\cite{fas-impl} of our prime-order group-based instantiation satisfying zero-knowledge style witness privacy in Python. The details of our performance evaluation are given in~\cref{sec:implementation}.
We also perform a series of benchmarks on our FAS scheme for a wide range of parameter settings.
Our results show that our scheme is practical for a wide variety of real-world scenarios.
For instance, on a MacBook Pro, for 300KB size witness, 
\ifhldiff
{\color{hldiffcolor}
\fi
the time taken for pre-signing is 0.344 seconds, pre-verification is 0.424 seconds, adapt is 0.035 seconds, functional-extract is 1.025 seconds. 
\ifhldiff
}
\fi
% \anote{Please give one sentence.}
We share more details of our findings in~\Cref{sec:implementation}.
% \nikhil{comment on insights.}

\smallskip\noindent\textbf{Applicability of Linear Functions.} Linear functions although restricted, allow us to capture many scenarios including learning statistical information (like mean, average, weighted linear functions, select entries) of records in a medical database, employee records of an organization, weather data, or ecological records of endangered species. Linear functions also allow a buyer to measure the proximity of a vector he holds with the secret vector of the seller. Such proximity measures are quite fundamental in Machine Learning where they are referred to as computing the kernel of two vectors or graphs in the case of~\cite{avrachenkov2017kernels}.

%% file: overview.tex
\section{Technical Overview}

%Throughtout the rest of this paper, we represent such linear functions by vectors $\bfy = (y_1, \ldots, y_\ell) \in \mathbb{Z}_p^\ell$ and the corresponding evaluation on $\bfx$ by the inner-product (mod $p$) of the vectors $\bfx$ and $\bfy$. That is, $\bfy(\bfx) = \langle \bfx, \bfy \rangle\mod p = \sum_{i=1}^{\ell} x_i \cdot y_i \mod p$.

In this overview, we present our approach to fair payments for learning linear functions of a seller's witness. 
%{\color{red}We first introduce functional adaptor signatures and its security, and outline the concepts behind our construction of functional adaptor signatures for linear functions. Finally, we discuss our implementation, and efficiency metrics, and state questions left open by our work.} \pnote{this can be removed}

To describe our techniques, we use the following setup: let $p$ and $\ell$ be integers where $p$ is prime and $L$ be an NP language containing statements $X$ with witness $\bfx = (x_1, \ldots, x_\ell)\in \mathbb{Z}_p^\ell$. Our goal is to support linear functions over $\bfx$. 
Throughout this paper, we represent such linear functions by vectors $\bfy = (y_1, \ldots, y_\ell) \in \mathbb{Z}_p^\ell$ and the corresponding evaluation on $\bfx$ by the inner-product of the vectors $\bfx$ and $\bfy$ modulo $p$.
%(the lattice based construction evaluates the function $\bmod{p}$). 
That is, $f_\bfy(\bfx) = \langle \bfx, \bfy \rangle = \sum_{i=1}^{\ell} x_i \cdot y_i \mod p$.
\footnote{
While our lattice-based instantiation is for computing inner products$\mod p$, the group-based instantiation is for computing integer inner products with bounded outputs, i.e., $f_\bfy(\bfx) \in \{0, \ldots, B\}$ for some apriori fixed bound $B \ll p$.	
} 

\newcommand{\PVer}{\mathsf{PreVerify}}

\subsection{Functional Adaptor Signatures for Fair Functional Payments}
%We have a fair exchange scenario where the seller has a witness $\bfx$ which is a vector in $\mathbb{Z}^\ell_p$, and only desires to sell some linear functions $\bfy$ on input $\bfx$. That is, we want the buyer to learn no more information about $\bfx$ beyond $f_\bfy(\bfx)$.
%Herein lies the premise of \emph{functional adaptor signatures (FAS)}. 
%Like adaptor signatures, FAS involves algorithms like $\FPreSign$ and $\FPreVerify$, $\Adapt$, and $\FExt$. However, the key distinction lies in the fact that $\FPreSign$ and $\FPreVerify$ are defined w.r.t.\ a function $\bfy$ from the buyer, while the functional extraction algorithm $\FExt$, returns $f_\bfy(\bfx)$ instead of $\bfx$ itself, ensuring the buyer learns only the function's result.

We have a fair exchange scenario where the seller has a witness $\bfx$ which is a vector in $\mathbb{Z}^\ell_p$, and only desires to sell some linear functions $\bfy$ on input $\bfx$. 
To enable such an exchange FAS, like adaptor signatures, involves algorithms like $\FPreSign$ and $\FPreVerify$, $\Adapt$, and $\FExt$. However, the key distinction lies in the fact that $\FPreSign$ and $\FPreVerify$ are defined w.r.t.\ a function $\bfy$ from the buyer, while the functional extraction algorithm $\FExt$, returns $f_\bfy(\bfx)$ instead of $\bfx$ itself, ensuring the buyer learns only the function's result.

\ifacm
\subsubsection{\bf Security of FAS} 
\else
\subsubsection{Security of FAS} 
\fi
\label{sec:overview-fas-security}
We consider two cases, where the adversary can corrupt either the seller or the buyer. 
Similar to standard adaptor signatures, when the seller is corrupt, FAS must mainly guarantee (1) \emph{unforgeability}, where an adversary who doesn't know the witness cannot forge honest buyers' signatures, and (2) \emph{witness extractability}, where the adversary cannot publish a buyer's signature such that when the buyer tries to extract using $\FExt$, it gets $z' \ne f_\bfy(\bfx)$. \footnote{We also require \emph{advertisement soundness} and \emph{pre-signature validity} properties against a corrupt seller, and defer their discussion 
\ifcameraready
to the full version.
\else
to~\Cref{sec:defs-appendix}.
\fi
}
Crucially, in the case of FAS, the adversary is allowed to get many functional pre-signatures for any choice of function $\bfy$.

On the other hand, when the buyer is corrupt, similar to standard adaptor signatures, FAS must guarantee \emph{pre-signature adaptability}, where if the adversary outputs a valid functional pre-signature, then adapting it using a valid witness $\bfx$ should result in a valid signature under the buyer's key.
Most importantly, FAS must guarantee protection of the seller's witness $\bfx$ which we capture in the notion called \emph{witness privacy}.

%We define three different variants of witness privacy.
%We explain the strongest variant called \emph{zero-knowledge}.
%\footnote{
%We refer the reader to~\Cref{def:fas-wh,def:fas-wi} for the two weak variants: \emph{witness hiding, witness indistinguishability}.}
% The first variant being \emph{witness hiding} is arguably the most intuitive, saying that the adversary after having learned $f_\bfy(\bfx)$ cannot output $\bfx$.
% The second variant called \emph{witness indistinguishability}, says that a malicious buyer's interaction remains indistinguishable across any two choices of seller's witnesses $\bfx_0 \neq \bfx_1$ for the NP statement $X$ as long as $f_\bfy(\bfx_0) = f_\bfy(\bfx_1)$. 
% This guarantee while weak, still captures many real-life scenarios.
% For instance, the seller could hold one of the two employee records (or medical databases) 
% $\bfx_0$ or $\bfx_1$ such that $f_\bfy(\bfx_0) = f_\bfy(\bfx_1)$.
% Here, the malicious buyer is guaranteed to learn only $f_\bfy(\bfx_0)$ (or $f_\bfy(\bfx_1)$) but not other sensitive details like gender, or social security numbers of entries in the database.

In more detail, FAS satisfies \emph{zero-knowledge} witness privacy if a malicious buyer learns no more information about $\bfx$ than $f_\bfy(\bfx)$ after the interaction. To formally define this in~\cref{sec:FAS}, we carefully borrow formalism from the related privacy notion of zero-knowledge for cryptographic proof systems~\cite{GMR89}. Informally, for every malicious buyer (that can sample its signing key), we require the existence of an efficient \ppt simulator $\Sim$ that can simulate the buyer's view only using $f_\bfy(\bfx)$. Philosophically, the existence of such an efficient simulator ensures that whatever the buyer learned from interacting with the seller, it could have on its own from $f_\bfy(\bfx)$ by running in polynomial time. While our notion shares similarities and is inspired by cryptographic proof systems, we need additional care to correctly model the fact that the buyer is allowed to ask for multiple functions which, moreover, can be adaptively chosen. 
To gain a comprehensive understanding of witness privacy, we also introduce two natural relaxed variants referred to as \emph{witness hiding} and \emph{witness indistinguishability} 
\ifcameraready
in the full version, 
\else
in~\Cref{def:fas-wh,def:fas-wi}, 
\fi
and study their relation with zero-knowledge witness privacy \nikhil{this study is not yet in the paper, should add}. Interestingly, in~\cref{sec:overview_noninteractive}, we show that the relaxation to witness indistinguishability enables a round optimal construction.

%\anote{Anything more to add? Would be nice to say something nontrivial we faced.}
%Defining the above properties of FAS is highly non-trivial and we model them in the form of carefully designed security games in~\cref{sec:FAS}.

For the rest of this overview, we will focus on building an efficient FAS centered around the witness privacy guarantee given it is the unique property of FAS compared to standard adaptor signatures. Moreover, our focus will be on the strongest notion of zero-knowledge witness privacy.

\ifacm
\subsubsection{\bf Strawman Construction of FAS} 
\else
\subsubsection{Strawman Construction of FAS} 
\fi
To build FAS, one can combine Schnorr adaptor signature (discussed above) along with a general-purpose non-interactive zero-knowledge (NIZK) proof. We present this strawman solution below and highlight a key efficiency challenge that serves as the starting point of our construction.

Recall the protocol flow in~\Cref{fig:fas-flow}. We describe how $\advt, \widetilde{\sigma}, \sigma$ are computed and the functional extraction process. 
The seller publishes $\advt = (\ct, \pi)$, where $\ct$ is a semantically-secure public-key encryption of the witness $\bfx$ and NIZK proof $\pi$ that certifies $\ct$ encrypts a witness for the statement $X$. The pre-signature $\widetilde{\sigma}$ is computed interactively among the buyer and the seller as follows. (i) The buyer sends $f$ to the seller. (ii) The seller computes a secret-key encryption $\ct_z$ of the requested evaluation $z = f(\bfx)$ using a fresh secret-key $k$. Then, the seller encodes $k$ in $\mathbb{Z}_p$ and \emph{engages with the buyer in an adaptor execution to sell the witness for the discrete logarithm of group element $K$, where $K = g^k$ for some generator $g$ of the group $\G$}. That is, the seller sends $(\ct_z, K, \pi_z)$ to the buyer, where NIZK proof $\pi_z$ certifies $\ct_z$ encrypts $f(\bfx)$ using the key $k$ where $\bfx$ is encrypted in $\ct$. The buyer treats $K$ as the Schnorr adaptor statement and computes a pre-signature $\widetilde{\sigma}$ on the transaction message $m$ if and only if the NIZK proof $\pi_z$ verifies. On obtaining $\widetilde{\sigma}$, the seller adapts it into a valid signature $\sigma$ on $m$ using the key $k$ as the witness. For functional extraction, the buyer runs $\Ext$ algorithm of the Schnorr adaptor signature to extract $k$ from $(\widetilde{\sigma}, \sigma)$ and then decrypts $\ct_z$ using it to recover the function evaluation $z$.

The desired privacy of FAS is attained by relying on the zero-knowledge property of both NIZK proofs and the semantic security of $\ct$. 
That is, the seller side of interaction can be efficiently simulated just from the evaluation $z = f(\bfx)$. The security against a malicious seller relies on the security of the underlying adaptor signature scheme and the soundness of the NIZK protocol. 

While this approach enables fair functional payments, it requires that the seller compute NIZK proofs for every $\FPreSign$ interaction with a buyer. The seller's computation in computing these proofs (and their size) grows \emph{polynomially} in the witness size $\bfx$ and is proportional to the complexity of the function $f$. While the proof sizes can be reduced by relying on general-purpose ZK-SNARKs~\cite{BCCT13}, the proving cost growing super-linearly in $|\bfx|$ will become prohibitively expensive even for moderately sized databases $\bfx$. 
This cost also affects the scalability of the solution at the application layer.
For instance, if the seller is dealing with several thousand buyers each with a different function, the seller will be computationally strained, leading to delays and a potential Denial of Service (DoS) attack vector via malicious buyers.
 A more efficient solution would allow (a) the seller's computation to grow linearly in the size of $\bfx$ and (b) the seller's communication to be proportional to the size of the evaluation $z$ which would be significantly smaller than $\bfx$ itself. As we explain below, we achieve these efficiency targets without relying on NIZKs, which ensures we use underlying cryptographic primitives in a \emph{black-box} manner.

%\pnote{current motivation sounds weak, justify better. some rough numbers for reasonable size of $\bfx$.}

\subsection{Our Techniques: Functional Encryption + Adaptor Signatures}
\label{sec:TOtechniques}
The blueprint of our FAS construction is to combine adaptor signatures with another cryptographic tool called functional encryption (FE).
FE was introduced as an enrichment of any (public-key) encryption scheme that allows for fine-grained decryption.
 In particular, in addition to algorithms $(\Setup, \Enc, \Dec)$ which are typical to any encryption scheme, a FE scheme for a function class $\mathcal{F}$ provides an additional functional key-generation algorithm $\KGen$.
 The algorithm takes as input the master secret-key $\msk$ and function $f \in \mathcal{F}$, and outputs a functional secret-key $\sk_f$ such that decrypting any encryption of $\bfx$ (w.r.t.\ $\msk$) with $\sk_f$ reveals $f(\bfx)$ and no more. As discussed earlier, this notion of strong privacy where no information about $\bfx$ is leaked other than $f(\bfx)$ is formalized using the \emph{simulation} paradigm. For simplicity, we will henceforth refer to FE schemes satisfying this privacy property as \emph{simulation-secure}. Numerous works have built FE schemes for rich classes of functions from a variety of hardness assumptions. More specifically, for the case of linear functions, we know of group-based schemes for prime-order groups~\cite{ipe,ipe01,ipe02,ad-sim-ipfe} as well as for unknown-order groups~\cite{ipe01,ad-sim-ipfe}; and post-quantum constructions~\cite{ipe,ipe01,ad-sim-ipfe}  are known from the hardness of the LWE problem.

\smallskip\noindent\textbf{Our construction template in more detail.} In the quest of removing NIZKs from $\FPreSign$, we modify the above steps of the fair exchange as follows: (a) The $\AdvertisementGen$ remains identical except that $\ct$ is computed using the FE scheme's encryption algorithm, (b) In the $\FPreSign$ interaction, the seller instead computes the functional secret-key $\sk_f$ and engages with the buyer in an adaptor execution to sell $\sk_f$, (c) the $\Adapt$ algorithm remains the same, and (d) the $\FExt$ algorithm extracts the functional secret-key $\sk_f$ and decrypts $\ct$ directly to return $f(\bfx)$. If the FE scheme satisfies simulation-security and the NIZK (in the advertisement) is zero-knowledge, then seller's interaction can be simulated only via $f(\bfx)$, which ensures the desired seller privacy.

{An astute reader might ask why to use an encryption of $\bfx$ in the advertisement and not a one-way function of $\bfx$, similar to how statement $X$ and witness $x$ have a discrete log relation in standard adaptor signature constructions. A crucial advantage of our construction template is that since $\bfx$ is encrypted under semantically secure encryption, we can let witness $\bfx$ be of low-entropy like user databases, which are the main applications we are targeting for FAS. A similar setting would be insecure when using one-way functions instead of encryption since $\bfx$ is not a high-entropy witness. Moreover, achieving zero-knowledge witness privacy seems hopeless when using one-way functions for setting up $X$.}

%Given a functional encryption scheme $\FEnc = (\Key Generation, \Enc, \Dec, \mathsf{FGen})$ for linear functions over $\mathbb{Z}+p^\ell$

\smallskip\noindent\textbf{Challenges.} Unfortunately the above approach doesn't quite work as-is and runs into two main challenges:
\begin{enumerate}[leftmargin=*]
	\item \textit{Correctness of the Adaptor Statement}: Suppose that one could encode $\sk_f$ as an integer $x_f$ in $\mathbb{Z}_p$ (for an appropriate $p$), and then use the Schnorr adaptor signature for the statement $X_f = g^{x_f}$. Still, for fairness, the buyer needs to gain confidence in the validity of $X_f$ (e.g., learning the incorrect functional key $\tilde{x}$ will disallow the buyer from learning the correct evaluation) during $\FPreSign$. Augmenting the seller's message with a NIZK proof that attests to the correctness of $X_f$ (i.e., discrete-logarithm of $X_f$ is indeed the correct $\sk_f$) would be sufficient. But the use of NIZKs was exactly what we were trying to avoid, and hence it seems we are back to square one.

	\item \textit{Compatibility with Adaptor Signatures}: Recall that known adaptor signature schemes only facilitate the sale of witnesses for special algebraic languages. For the above approach to work, we need a \emph{simulation-secure} FE scheme where the functional key-generation algorithm exhibits the required structural compatibility. For e.g., to rely on Schnorr adaptor signatures that allow selling discrete-logarithms $x \in \mathbb{Z}_p$ of prime-order group elements $X = g^x$, we need a FE scheme where the functional secret-key $\sk_f$ w.r.t.\ any linear function $f$ is such that $\sk_f \in \mathbb{Z}_p$. In fact,~\cite{ipe} presents such an IPFE scheme, but the scheme satisfies only a weaker notion of IND-security (and not simulation security) which is insufficient for the template to work.

\end{enumerate} 

\smallskip\noindent\textbf{Our Technique 1: Augmenting IPFE with $\PubKGen$ algorithm.}
To avoid using NIZKs during $\FPreSign$, we observe that the IPFE schemes for computing inner product of vectors of length $\ell$ are obtained by starting with $\ell$ instances of some PKE scheme and a functional secret key is essentially a linear combination of the secret keys of the $\ell$ PKE instances (as pointed out in~\cite{ipe}). Combining this observation with the fact that the public key and secret key of the underlying PKE satisfy some algebraic relation, one can envision that a similar linear combination of the public keys of the PKE schemes might give us a commitment to the functional secret key that can be directly used as the statement $X_f$. Crucially, if $X_f$ can be deterministically computed using $\ipfe$'s master public key and the function description $f$, then such an algorithm would be public and the buyer can compute $X_f$ without any interaction with the seller, thus avoiding NIZKs during $\FPreSign$ altogether. We formalize this vision by augmenting $\ipfe$ with a public deterministic algorithm $\PubKGen$. The requirement that the output of $\PubKGen$, i.e., $X_f$ must be a commitment to the output of $\KGen$, i.e., $\sk_f$, is formalized via a \emph{compliance property} (\Cref{def:ipfe-compliant}). Informally, $R_\ipfe$-compliance says that for a given relation $R_\ipfe$, it must be the case that $(X_f, \sk_f) \in R_\ipfe$. 

For instance, if we set the relation $R_\ipfe$ to be the discrete-log relation, then the key structure becomes compatible with the hard relation of Schnorr adaptor signatures. As an example, the IND-secure IPFE scheme of~\cite{ipe} can be shown to have such a key structure. In the IND-secure IPFE scheme of~\cite{ipe}, the master secret key is $\bfs = (s_1, \ldots, s_\ell) \in \Z_p^\ell$, the master public key is $g^\bfs = (g^{s_1}, \ldots, g^{s_\ell})$, the functional secret key for function $\bfy = (y_1, \ldots, y_\ell)$ is $\sk_\bfy = \sum_{i \in [\ell]} s_i y_i \bmod{p}$. Then, $\PubKGen$ can be defined to output $\pk_\bfy = \prod_{i \in [\ell]} (g^{s_i})^{y_i}$. Consequently, we can observe that $\pk_\bfy = g^{\sk_\bfy}$, thus $(\pk_\bfy, \sk_\bfy)$ satisfy the discrete log relation. Defining $\PubKGen$ in this way makes it compliant with the Schnorr adaptor signatures: if the buyer wants to learn function $\bfy$ evaluated on the seller's witness, he can locally compute $\pk_\bfy$ and use it as the adaptor statement. Upon receiving a pre-signature from the buyer, the seller can locally compute $\sk_\bfy$ and use it to adapt the pre-signature. Eventually, the buyer can extract $\sk_\bfy$, thus enabling functional decryption. Crucially notice that because of deterministic nature of $\PubKGen$, the pre-signing becomes non-interactive.
But as noted earlier, IND-secure IPFE isn’t sufficient for the FAS template to work, so we are not done yet.

\smallskip\noindent\textbf{Our Technique 2: Simulation-Secure Compatible IPFE.} To resolve the compatibility with known adaptor signatures, we rely on the IND-secure IPFE to simulation-secure IPFE compiler of~\cite{ad-sim-ipfe}. More specifically, their compiler lifts any IND-secure IPFE to achieve simulation-security without changing the structure of the functional secret keys. 
This is achieved by increasing the length of function vectors from $\ell$ to $2\ell$, where the first $\ell$ slots encode the function as before, and the extra $\ell$ slots encode some random coins. These random coins are used by the IPFE simulator to argue simulation security. To preserve correctness, the length of the message vector to be encrypted is also increased from $\ell$ to $2\ell$, where the first $\ell$ slots encode the message as before, and the extra $\ell$ slots are set to zero.
A small caveat is that the compiler results in a stateful functional secret key generation algorithm. 
We make it stateless which allows us to make an optimization where we increase the vector lengths from $\ell$ to only $\ell+1$. 
% \pnote{nikhil: add a sentence or so about how the compiler does this. E.g., This is achieved by adding redundency in the master secret-key that introduces additional randomness which is used by the simulator.} 
Instantiating their compiler with the DDH-based IND-secure IPFE scheme of~\cite{ipe} results in a simulation-secure FE that is compatible with Schnorr adaptor signatures. 
We also extend this to the post-quantum setting by using a simplified variant of the LWE-based simulation-secure IPFE scheme of~\cite{ad-sim-ipfe} 
\footnote{Making the functional secret key generation algorithm stateless restricts the functionality in the LWE setting. In particular, such IPFE can only handle linearly independent function queries. We show how to slightly modify our $\fas$ construction to ensure that linearly dependent function queries from different buyers can be augmented with some extra slots that always guarantee linear independence of function queries to the underlying $\ipfe$. 
\ifcameraready
See the full version
\else 
See~\Cref{sec:fas-from-lattices-modified} 
\fi
for more details.}
that is compatible with the Dilithium adaptor signature~\cite{PQadaptor1} instantiated on unstructured lattices.

% \pnote{From the description of~\cite{ipe01} in the tech part, it seems they do sel-to-adaptive IND-CPA. I wrote the overview thinking they do ind-to-sim. @nikhil, please check this part.}

While the above technique addresses the compatibility w.r.t.\ known adaptor signatures, 
proving \emph{zero-knowledge} witness privacy requires more care and we will introduce another technique for it. Let us first explain the challenge with proving \emph{zero-knowledge}.
Even when starting with a base $\ipfe$ scheme (that is IND-secure) with a deterministic functional key-generation algorithm, the IPFE compiler introduces randomness in the functional key-generation algorithm.
\footnote{DDH-based IPFE scheme of~\cite{ipe01} has been shown to be simulation-secure in~\cite{ad-sim-ipfe}. Interestingly, it does not introduce randomness in functional key-generation and for proving simulation-security, the simulator relies on the fact that there can be many master secret keys corresponding to a master public key. But unfortunately, it does not satisfy the IPFE compliance property as its master keys do not satisfy the discrete log relation. Hence, it is incompatible with Schnorr adaptor signatures. Despite many attempts, we could not make the two compatible.} 
This randomization is crucial for the IPFE simulator's correctness and essential to work with the base FE scheme that is only IND-secure. In the context of FAS, during $\FPreSign$, if the seller could share the randomness used as part of the generating $\sk_f$, then the buyer can mimic the same computation using the seller's randomness and determine the validity of the $\sk_f$ embedded in $X_f$. While this approach makes pre-signing interactive yet it would certainly avoid the use of general-purpose NIZKs as seen in the template above because the check is \emph{canonical} now. Unfortunately, exposing the seller's randomness would compromise simulation security of the compiled IPFE scheme as the IPFE simulator would then get stuck in the analysis. Consequently, we can't show \emph{zero-knowledge} of FAS. Note that an apporach to determine validity of $X_f$ without having the seller share the random coins could be to use NIZKs, but as said before, we want to avoid it.

\smallskip\noindent\textbf{Our Technique 3: Interactive pre-signing without NIZKs.} 
To address the \emph{zero-knowledge} of FAS, we open up the IPFE compiler of~\cite{ad-sim-ipfe}, that is, make non-black-box use of the IPFE compiler and reveal only a part of the randomness of the functional key-generation algorithm. By carefully choosing the randomness revealed by the seller, we can guarantee the validity of the adaptor statement $X_f$ to the buyer (without expensive NIZKs) and also allow flexibility to the IPFE simulator to successfully finish the simulation (See ~\Cref{remark:removing-nizk,remark:ipfe-simulation} for more details).

We emphasize that the non-black-box use of the techniques from~\cite{ad-sim-ipfe} is a significant technical part of this work. We elaborate on these aspects of our construction in more detail in~\Cref{sec:overview-construction}. 
% \nikhil{change this ref to section 2.3 or remarks 2.1 and 2.2?}
% \anote{Is this correct citation?}

% \begin{remark}
% Note that typically standard adaptor signatures are defined w.r.t.\ some hard relation $R_0$. But, in case of lattice-based instantiations~\cite{pq-as}, the definition needs to be generalized to a pair of hard relations $R_0, R'_0$, where $R'_0$ is a extended relation of $R_0$, i.e., $R_0 \subset R'_0$. In such a case, even though witness $x$ satisfying $(X, x) \in R_0$ is used to adapt a pre-signature, the resulting extracted witness $x'$ satisfies $(X, x') \in R'_0$. Since in our FAS construction the witness is $\sk_f$, the extracted witness will be some value $\sk'_f$ that does not guarantee decryption correctness. To ensure that $\sk'_f$ also decrypts $\ipfe$ ciphertext correctly, we will additionally require a robustness property from our IPFE scheme (\Cref{def:ipfe-robust}).
% \end{remark}

\subsection{Our Construction}
\label{sec:overview-construction}
Suppose that the relation $R$ is such that for any statement $X$ and witness $\bfx$ satisfying $(X, \bfx) \in R$, it is the case that $\bfx \in \Z_p^\ell$. 
Further, 
% \smallskip\noindent 
suppose that inner-product functions are of the form $\bfy \in \Z_p^\ell$. Then, our $\fas$ construction is obtained by employing the abovementioned techniques.
We remind the reader of the protocol flow in~\Cref{fig:fas-protocol}.
% \pnote{good to point to the figure here already}
\begin{itemize}[leftmargin=*]
\item $\Setup$:
Run by a trusted third party, it samples common reference string $\crs$ for NIZK and public parameters $\pp'$ for IPFE. 

% \begin{figure}[H]
% \begin{pchstack}[boxed]
% \procedure[linenumbering, mode=text]{$\Setup(1^\secparam)$}{
% Sample $\crs \gets \nizk.\Setup(1^\secparam)$
% \\ 
% % \nikhil{does the nizk setup need to take $1^\ell$ as input too?}
% Sample $\pp' \gets \ipfe.\Gen(1^\secparam)$
% \\ 
% Return $\pp := (\crs, \pp')$
% }
% \end{pchstack}
% \end{figure}

\item $\AdGen$:
The seller samples $(\mpk, \msk)$ corresponding to IND-secure IPFE and random coins $\bft$ used by the non-black-box IPFE compiler to upgrade from IND-secure IPFE to simulation-secure IPFE. It computes the elongated vector $\widetilde{\bfx}$ used by the IPFE compiler, encrypts $\widetilde{\bfx}$ to obtain $\ct$ and computes a NIZK proof $\pi$ certifying that $\ct$ encrypts a witness corresponding to $X$. 

\item $\AdVerify$:
The buyer verifies the NIZK proof $\pi$.

\item Interactive pre-signing:
For the sake of notational simplicity, we denote the 3-round interactive functional pre-signing via three algorithms: $\AuxGen$, $\AuxVerify, \FPreSign$. 
\begin{itemize}[leftmargin=*]
\item 
First round: the buyer sends the function $\bfy$ to the seller. 
\item 
Second round: the seller runs the $\AuxGen$ algorithm and sends auxiliary value $\aux_\bfy$ along with a proof $\pi_\bfy$ validating authenticity of $\aux_\bfy$ to the buyer.
Specifically, the seller uses the random coins $\bft$ of the IPFE compiler to compute $f_\bfy(\bft)$. Then, it sets the elongated vector $\widetilde{\bfy} = (\bfy^T, f_\bfy(\bft))^T$ and generates $\pk_\bfy$ for it. It sends $(\aux_\bfy, \pi_\bfy) := (\pk_\bfy, f_\bfy(\bft))$.
\item 
Third round: the buyer verifies the auxiliary value via $\AuxVerify$ algorithm, i.e., it creates $\widetilde{\bfy} = (\bfy^T, \pi_\bfy)^T$ and checks if $\aux_\bfy$ matches the output of $\ipfe.\PubKGen(\mpk, \widetilde{\bfy})$. If it verifies, then, it computes pre-signature $\widetilde{\sigma}$ for the adaptor statement $\aux_\bfy$ and sends it to the seller.
\end{itemize}
\item $\FPreVerify$: The seller verfies that $\widetilde{\sigma}$ corresponds to $\bfy$, i.e., it verifies that $\widetilde{\sigma}$ corresponds to $\aux_\bfy$ and $\aux_\bfy$ corresponds to $\bfy$.
\item $\Adapt$: The seller computes IPFE functional key $\sk_\bfy$ for  function $\widetilde{\bfy}$. Since $\sk_\bfy$ is a witness to $\pk_\bfy$, it uses $\sk_\bfy$ to adapt $\widetilde{\sigma}$ into $\sigma$.
\item $\FExt$: The buyer extracts the IPFE functional key $\sk_\bfy$ from $(\widetilde{\sigma}, \sigma)$ and uses it to decrypt $\ct$ and recover $f_{\bfy}(\bfx)$.
\end{itemize}

\input{fig-fas-construcition}

Our formal construction is as in~\Cref{fig:fas-construction}.
We defer the security theorem to~\Cref{sec:fas-construction}. 
We alluded to in `Our Technique 3' on how to avoid NIZKs in interactive pre-signing by making non-blackbox use of the IPFE compiler. Having described the full details of FAS construction, we now elucidate on it in~\Cref{remark:removing-nizk,remark:ipfe-simulation}.

% \pnote{do you need the paranthesis around aux.verify and as.preverify in fpreverify?}

\begin{remark}[Simulatability with leakage $f_\bfy(\bft)$ on $\bft$]
\label{remark:ipfe-simulation}
Note that while proving zero-knowledge of FAS, the FAS simulator would still have to reveal $f_\bfy(\bft)$.
This helps us avoid NIZKs in the interactive pre-signing as discussed before.
Also note that FAS simulator would use the IPFE simulator internally, and $f_\bfy(\bft)$ is a leakage on 
the random coins $\bft$ of the IPFE (compiler's) simulator. 
Crucially, this does not break simulation security of IPFE.
This is because the IPFE compiler of~\cite{ad-sim-ipfe} --- and hence the IPFE simulator too -- inherently leaks $f_\bfy(\bft)$.
The latter is because $\ipfe$ decryption implicitly takes $\widetilde{\bfy}$ as input and hence, the decrypter needs to know $f_\bfy(\bft)$ to ensure decryption correctness. 
Hence, from an honest seller's perspective, the $\fas$ zero-knowledge simulation won't be affected by giving away $f_\bfy(\bft)$ as part of $\pi_\bfy$.
\end{remark}

\begin{remark}[Purpose of $\pi_\bfy$]
\label{remark:removing-nizk}
Note that $\pi_\bfy$ enables verifying that $\aux_\bfy$ is functional public key corresponding to $\widetilde{\bfy} = (\bfy^T, \pi_\bfy)^T$. Regarding $\widetilde{\bfy}$, the buyer knows $\bfy$ but it does not know whether the last slot of $\widetilde{\bfy}$ chosen by the seller is well-formed, i.e., is $\pi_\bfy = f_\bfy(\bft)$? 
% In security against a malicious buyer, the seller is honest and thus $\pi_\bfy = f_\bfy(\bft)$ holds.
In security against a malicious seller, we argue that we do not need to guarantee this.
In particular, it does not affect an honest buyer whether $\pi_\bfy = f_\bfy(\bft)$ holds or not. Suppose that the malicous seller chooses arbitrary value ${\sf val}$ and sends $(\aux_\bfy, \pi_\bfy) = (\pk_\bfy, {\sf val})$, where $\pk_\bfy = \ipfe.\PubKGen(\mpk, \widetilde{\bfy})$ and $\widetilde{\bfy} = (\bfy^T, {\sf val})^T$.
This would certainly ensure that $\AuxVerify$ passes and the buyer continues the protocol with the seller. Crucially though, we argue that if the honest buyer does end up making the payment to the seller, then, the honest buyer indeed learns $f_\bfy(\bfx)$ at the end of the protocol disregard of what ${\sf val}$ was chosen by the malicious seller. 
This is because \emph{advertisement soundness} would guarantee the honest buyer that $\ct$ encrypts a vector $\widetilde{\bfx}$ of the form $\widetilde{\bfx} = (\bfx^T, 0)^T$, where $(X, \bfx) \in R$.
Further, witness extractability would guarantee that the honest buyer extracts a functional secret key $\sk_\bfy$ corresponding to $\widetilde{\bfy} = (\bfy^T, {\sf val})^T$. Thus, IPFE correctness would guarantee that the honest buyer recovers $f_{\widetilde{\bfy}}(\widetilde{\bfx})$ which is same as $f_\bfy(\bfx)$ since the last slot of $\widetilde{\bfx}$ is zero.
% One might ask if the seller needs to give a NIZK proof that $\pi_\bfy = f_\bfy(\bft)$? 
% And the answer is no, because as far as an honest buyer is concerned, the last slot of $\widetilde{\bfy}$ can be set to any arbitrary value in $\Z_p$ and it still won't affect the decryption outcome since the last slot of $\widetilde{\bfx}$ is $0$. 
\end{remark}

% \fi

\subsection{Instantiations}
\label{sec:TOinstantiations}

% \anote{I took this paragraph from the technical sections. Please integrate it here with the text appropriately.}
In~\Cref{sec:fas-from-prime-groups}, we provide an instantiation from prime order groups for the  function class of 
\ifhldiff
{\color{hldiffcolor}
\fi
inner products over integers with output values polynomially bounded by $B \ll p$. 
For example, if each entry of the witness/function vector is bounded by $10^3$, and the vectors have $10^6$ entries, then, the inner product value would be bounded by $10^{12}$, so we can set $B=10^{12}$. We explain below application scenarios where such restrictions are reasonable.
\ifhldiff
}
\fi
We also present an instantiation from lattices in~\Cref{sec:fas-from-lattices} that removes such restrictions. 
\ifhldiff
{\color{hldiffcolor}
\fi
Specifically, it works for the function class of inner products modulo prime integer $p$. 
Further, it is post-quantum secure.
\ifhldiff
}
\fi
% \pnote{this doesn't address one of reviewer's suggestion wherein they wanted us to say 1-2 sentences on the function classes. For prime order groups, do you want to give an example of the function class and maybe tie it to what FE can handle?} 
% \anote{I moved the plaintext discussion here.}
% \pnote{What is restricted about the lattice class, it sounds identical to the class mentioned in 2.0}
% In both the instantiations, we will have $\mcal{M} = \Z_p^\ell$, where $p$ is a prime and $\ell$ is an integer.

\ifacm
\subsubsection{\bf Prime-Order Groups based Instantiation.} 
\else
\subsubsection{Prime-Order Groups based Instantiation.} 
\fi
% bounded linear functions; 
To instantiate our FAS construction in~\Cref{sec:construction} from prime order groups, 
% it suffices to instantiate the building blocks $\IPFE$ and $\AS$ from prime order groups, while ensuring the two are compatible with each other. 
% We can instantiate the $\nizk$ for the language $L_\nizk$ depending on the concrete relation $R$. 
% That is, given $R$ and the above instantiation of $\IPFE$, we can pick the most efficient $\nizk$ for $L_\nizk$.
% Since this is not the main contribution of this work, we will assume $R$ to be a general NP relation and rely on the existence of $\nizk$ for general NP~\cite{nizk-from-ddh,peikert2019noninteractive} and not delve into its details here.
%The asterisk $^*$ in the section title is to note that $\nizk$ for NP are not known from prime order groups, so we do not specify its details here. As the lattice-based $\nizk$ systems support all NP langauges \pnote{which papers to cite?}, we can use those even for the language $L_\nizk$ related to the prime-order groups that the $\ipfe$ instantiation induces here.
we instantiate the $\as$ as Schnorr adaptor signature scheme~\cite{erwig2021two} w.r.t. the hard relation $R_\DL$, where $R_\DL$ is the discrete-log relation in prime order 
\ifcameraready
groups. 
\else 
groups (see~\Cref{def:dl-relation}). 
\fi
We then set $R_\ipfe = R_\DL$ and show that a varaiant of the selective, IND-secure IPFE scheme by Abdalla et al.~\cite{ipe} satisfies $R_\DL$-compliance. 
% when appropriately augmented with a $\PubKGen$ algorithm. 
% We describe the instantiations of these two building blocks in~\cref{sec:detailed_instantiation} 
% and the resulting $\fas$ scheme in~\Cref{thm:fas-strongly-secure-prime-order-groups}. 
% % With the above instantiations, we obtain the following theorem.

\ifhldiff
{\color{hldiffcolor}
\fi
\smallskip\noindent\textbf{Small plaintext space: reason and application scope.}
This limitation is because the function output here is encoded in the exponent and functional extraction of FAS instantiation involves a discrete log computation step. For values in the exponent bounded by $B$, this incurs a running time of $O(\sqrt{B})$. This is efficient only if bound $B$ is a polynomial. 
% For the sake of efficient decryption, our implementation is limited to small plaintext/function space.
We argue that this is sufficient for several realistic application scenarios. We benchmarked computation for datasets having a maximum of $10^6$ entries, resulting in a 32 MB dataset. For example, the breast cancer dataset on Kaggle~\cite{kaggle-cancer} has $20000$ entries, total size $50$ KB, sum of entries $\approx 10^6$. Assuming each entry of a function is at most $10^4$, we get $B=10^{10}$. Thus, the benchmarks in~\Cref{fig:fas-perf} for $(\ell=10^4, B=10^{10})$ fit closest for this dataset.
% \textcolor{red}{
Other scenarios occur when the secret database is a company's employee record with their age, years of service with the employer, retirement contribution, etc. Such values while sensitive, are typically bounded (e.g., by $B=10^{10}$) and the buyer may wish to learn statistical/aggregate information, like sum, weighted mean or average, etc. whose result is also in the bounded plaintext space. The buyer could be a research organization that is studying the workforce in companies or factories, or it could be a recruitment recommendation agency that tries to match clients with potential employers with appropriate aggregate requirements of the client.
% }
\ifhldiff
}
\fi

\ifacm
\subsubsection{\bf Lattice based Instantiation.} 
\else
\subsubsection{Lattice based Instantiation.} 
\fi
% \anote{Talk about the model relaxation as well for lattice-based construction.}
% To instantiate our FAS construction in~\Cref{sec:construction} from lattices, 
% it suffices to instantiate the building blocks IPFE and AS from lattices, while ensuring the two are compatible with each other. We describe these next.
We instantiate the $\as$ as the lattice-based adaptor signature scheme by Esgin et al.~\cite{pq-as} w.r.t. the inhomogenous short integer solution relation  
\ifcameraready
$R_\isis$. 
\else
$R_\isis$  (see~\Cref{def:isis}). 
\fi
Our instantiation is over the ring of integers $\mcal{R} = \Z$. Consequently, its security follows from plain SIS and LWE assumptions. 
% One can look at the resulting $\as$ to be w.r.t.\ a digital scheme $\Lyu$ that is somewhere in between Lyubashevsky's signature scheme~\cite{lyu12lattice-sig} and Dilithium~\cite{dilithium} instantiated with unstructured lattices. 
% \nikhil{is there a closer match? Lyu12 is from unstructured lattices but uses gaussian sampling. If there is some work that does Lyu12 but does not do gaussian sampling, that would be the closest to our digital signature.}
% Consequently, we have $R_\ipfe = R_\isis$ 
% % (and $R'_\ipfe = R'_\isis$ such that $R_\isis \subset R'_\isis$) 
% as defined in~\Cref{sec:as-lattices}.
Further, we show that a variant of the IND-secure IPFE scheme by Agrawal et al.~\cite[Section 4.2]{ipe01} satisfies $R_\isis$-compliance. 
% when appropriately augmented with a $\PubKGen$ algorithm. 
\ifcameraready\else
\footnote{Note that in lattice-based adaptor signature scheme by Esgin et al.~\cite{pq-as}, the extracted witness does not satisfy $R_\isis$ but it satisfies an extended relation $R'_\isis$ such that $R_\isis \subset R'_\isis$. This is due to the rejection sampling step involved in these signatures. Using such signatures along with $\ipfe$ thus becomes technically more challenging as the extracted witness acts as the $\ipfe$ decryption key. So, we additionally require a decryption robustness property from $\ipfe$ (\Cref{def:ipfe-robust}).
}
\fi
% We describe the instantiations of these two building blocks in~\Cref{sec:fas-from-lattices} and the resulting $\fas$ scheme in~\Cref{thm:fas-strongly-secure-lattices}.

For technical reasons related to IPFE, the resulting FAS construction can only support linearly independent function request across all buyers (See~\Cref{remark:linear-independence} for details). This could be quite restrictive since the scheme cannot support same function requests from different buyers. 
\ifcameraready
In the full version,
\else 
In~\Cref{sec:fas-from-lattices-modified}, 
\fi
we show how to modify the FAS construction to remove this restriction.

% We note that in~\cite[Section 4.2]{ipe01}, $\KGen$ is stateful to support key generation queries that are linearly dependent modulo $p$. This is problematic for us as the corresponding $\PubKGen$ will also have to be stateful then, but that defeats the purpose of this algorithm being public. So, to remedy the situation, in the lattice construction we assume that all $\KGen$ queries to $\ipfe$ are going to be linearly independent. Naively instantiating our $\fas$ construction will result in $\fas$ only supporting linearly independent functions, which is a very restrictive assumption. To avoid such an assumption at the application layer, we show in~\Cref{sec:fas-from-lattices-modified} how to assign unique IDs to each buyer that guarantees linearly dependent queries by different buyers will map to linearly independent key generation queries to the underlying $\ipfe$. The latter construction needs to fix an upper bound on the number of buyers (a mild assumption), as this number influences the construction details. 

\subsection{Towards Non-Interactive Pre-Signing}\label{sec:overview_noninteractive} One of the main features of adaptor signatures that enable scalability is the non-interactive nature of the pre-signature generation algorithm. More specifically, the buyer in an adaptor signature execution can generate a pre-signature just from the seller's advertisement and, more importantly, without further interaction with the seller. Since the exchange of pre-signatures is done off-chain and hence the latency doesn't impact the blockchain eco-system, a non-interactive pre-signing phase is more preferable.

We relaxed pre-signing to be interactive to achieve strong witness privacy for FAS, namely, the malicious buyer learns no information about $\bfx$ other than $f(\bfx)$.
This relaxation is not new to our work and was already introduced in~\cite{qin2023blindhub} to achieve the advanced feature of \emph{blindness} for adaptor signatures.
Despite numerous attempts, the interactive nature of pre-signing seems necessary for this notion of strong privacy.
An immediate challenge towards non-interactive pre-signing would be to build a simulation-secure IPFE scheme where the buyer can itself generate $\pk_f \gets \ipfe.\PubKGen(\mpk, f)$ such that the seller can generate the corresponding functional secret-key $\sk_f \gets \ipfe.\KGen(\msk, f)$.
Unfortunately, the randomized nature of the key-generation algorithm of simulation-secure scheme is a major hurdle to enabling this.

%To enable the advanced feature of functional adaptor signatures, we relax the pre-signing phase to be interactive.  Despite numerous attempts, the interactive nature of pre-signing seems necessary for achieving simulation-security for functional adaptor signatures. An immediate challenge towards non-interactive pre-signing would be to build a simulation-secure functional encryption scheme where the buyer can itself generate a ``public" functional key w.r.t.\ function $f$ such that the seller can generate the corresponding functional secret-key $\sk_f$. Unfortunately, the randomized nature of the key-generation algorithm of simulation-secure scheme is a major hurdle towards enabling this.

We explore whether a non-interactive pre-signing can be achieved for FAS for relaxed privacy definitions. Towards this, we study FAS that only satisfies \emph{witness indistinguishability}.
The construction of such a FAS follows from our construction template by replacing the simulation-secure IPFE with an appropriate IND-secure IPFE. 
To understand the intuition, we refer the reader to the IPFE compliance example discussed in our technique 1 in~\Cref{sec:TOtechniques}.
We provide more details on this construction 
\ifcameraready
in the full version.
\else
in~\Cref{sec:fas-construction-weak}.
\fi
% \nikhil{fix this paragraph}

% \subsection{Implementation and Efficiency Metrics} \anote{Not sure if we want talk about it here. Talking a bit about this in contribution is better}

%\subsection{Open Questions}

% For linear functions: can we develop new FE schemes for linear functions that would support non-interactive pre-signing. Beyond linear functions to cover several machine learning style applications?

%% file: fig-fas-construcition.tex
\ifacm 

\begin{figure*}[t]
\ifhldiff
{\color{hldiffcolor}
\fi
\centering
\captionsetup{justification=centering}
\begin{pchstack}[boxed, space=0.5em]

\begin{pcvstack}
\procedure[linenumbering, mode=text]{$\Setup(1^\secparam)$}{
Sample $\crs \gets \nizk.\Setup(1^\secparam)$
\\ 
% \nikhil{does the nizk setup need to take $1^\ell$ as input too?}
Sample $\pp' \gets \ipfe.\Gen(1^\secparam)$
\\ 
\ret $\pp := (\crs, \pp')$
}

\procedure[linenumbering, mode=text]{$\AdvertisementGen(\pp, X, \bfx)$:}{
Sample random coins $r_0$, $r_1$ 
\\ 
Let $(\mpk, \msk) := \ipe.\Setup(\pp', 1^{\ell+1}; r_0)$
\\ 
Sample $\bft \getr \Z_p^\ell$,
let $\widetilde{\bfx} := (\bfx^T, 0)^T \in \Z_p^{\ell+1}$ 
\\ 
Let $\ct := \ipe.\Enc(\mpk, \widetilde{\bfx}; r_1)$
\\ 
Let $ {\sf stmt} := (X, \pp', \mpk, \ct), {\sf wit} := (r_0, r_1, \bfx)$
\\
Let $\pi \gets \nizk.\Prove(\crs, {\sf stmt}, {\sf wit})$
\\ 
\ret $\advt := (\mpk, \ct, \pi)$, $\state := (\msk, \bft)$
}

\end{pcvstack}
\begin{pcvstack}

\procedure[linenumbering, mode=text]{$\AdvertisementVerify(\pp, X, \advt)$}{
\ret $\nizk.\Verify(\crs, (X, \pp', \mpk, \ct), \pi)$
}

\procedure[linenumbering, mode=text]{$\AuxGen(\advt, \state, \bfy)$}{
Parse $\advt = (\mpk, \ct, \pi)$, $\st = (\msk, \bft)$
\\ 
Let $\widetilde{\bfy} := (\bfy^T, f_\bfy(\bft))^T \in \Z_p^{\ell+1}$
\\ 
Let $\pk_\bfy := \ipfe.\PubKGen(\mpk, \widetilde{\bfy})$ 
\\ 
\ret $\aux_\bfy := \pk_\bfy$, $\pi_\bfy := f_\bfy(\bft)$
}

\procedure[linenumbering, mode=text]{$\AuxVerify(\advt, \bfy, \aux_\bfy, \pi_\bfy)$}{
Parse $\advt = (\mpk, \ct, \pi)$,
let $\widetilde{\bfy} := (\bfy^T, \pi_\bfy)^T$ 
\\ 
\ret $1$ iff $\aux_\bfy = \ipfe.\PubKGen(\mpk, \widetilde{\bfy})$
}

\procedure[linenumbering, mode=text]{$\FPreSign(\advt, \sk, m, X, \bfy, \aux_\bfy)$}{
% Parse $\advt = (\mpk, \ct, \pi)$
% \\ 
\ret $\widetilde{\sigma} \gets \as.\PreSign(\sk, m, \aux_\bfy)$
}

\end{pcvstack}
\begin{pcvstack}

\procedure[linenumbering, mode=text]{$\FPreVerify(\advt, \vk, m, X, \bfy, \aux_\bfy, \pi_\bfy, \widetilde{\sigma})$}{
% Parse $\advt = (\mpk, \ct, \pi)$
% \\ 
\ret $\AuxVerify(\advt, \bfy, \aux_\bfy, \pi_\bfy) \wedge$ 
\pcskipln \\ 
$\quad \qquad \as.\PreVerify(\vk, m, \aux_\bfy, \widetilde{\sigma}) $
}

\procedure[linenumbering, mode=text]{$\Adapt(\advt, \state, \vk, m, X, \bfx, \bfy, \aux_\bfy, \widetilde{\sigma})$}{
Parse $\advt = (\mpk, \ct, \pi)$, $\state = (\msk, \bft)$
\\ 
Let $\widetilde{\bfy} := (\bfy^T, f_\bfy(\bft))^T$ 
% Let $\pk_\bfy := \ipe.\PubKGen(\mpk, \bfy)$
\\ 
Let $\sk_\bfy := \ipe.\KGen(\msk, \widetilde{\bfy})$
\\ 
\ret $\sigma := \as.\Adapt(\vk, m, \aux_\bfy, \sk_\bfy, \widetilde{\sigma})$
}

\procedure[linenumbering, mode=text]{$\FExt(\advt, \widetilde{\sigma}, \sigma, X, \bfy, \aux_\bfy)$}{ 
Parse $\advt = (\mpk, \ct, \pi)$.
\\ 
% Let $\pk_\bfy := \ipe.\PubKGen(\mpk, \bfy)$
% \\ 
Let $z := \as.\Ext(\widetilde{\sigma}, \sigma, \aux_\bfy)$
\\ 
\ret $v := \ipe.\Dec(z, \ct)$
}

\end{pcvstack}
\end{pchstack}
\caption{
\ifhldiff
{\color{hldiffcolor}
\fi
Construction: Functional Adaptor Signatures
\ifhldiff
}
\fi
}
\label{fig:fas-construction}
\label{sec:construction}
\ifhldiff
}
\fi
\end{figure*}

\else 

\begin{figure*}[t]
\ifhldiff
{\color{hldiffcolor}
\fi
\centering
\captionsetup{justification=centering}
\begin{pchstack}[boxed, space=0.5em]

\begin{pcvstack}
\procedure[linenumbering, mode=text]{$\Setup(1^\secparam)$}{
Sample $\crs \gets \nizk.\Setup(1^\secparam)$
\\ 
% \nikhil{does the nizk setup need to take $1^\ell$ as input too?}
Sample $\pp' \gets \ipfe.\Gen(1^\secparam)$
\\ 
\ret $\pp := (\crs, \pp')$
}

\procedure[linenumbering, mode=text]{$\AdvertisementGen(\pp, X, \bfx)$:}{
Sample random coins $r_0$, $r_1$ 
\\ 
Let $(\mpk, \msk) := \ipe.\Setup(\pp', 1^{\ell+1}; r_0)$
\\ 
Sample $\bft \getr \Z_p^\ell$,
let $\widetilde{\bfx} := (\bfx^T, 0)^T \in \Z_p^{\ell+1}$ 
\\ 
Let $\ct := \ipe.\Enc(\mpk, \widetilde{\bfx}; r_1)$
\\ 
Let $ {\sf stmt} := (X, \pp', \mpk, \ct), {\sf wit} := (r_0, r_1, \bfx)$
\\
Let $\pi \gets \nizk.\Prove(\crs, {\sf stmt}, {\sf wit})$
\\ 
\ret $\advt := (\mpk, \ct, \pi)$, $\state := (\msk, \bft)$
}

\procedure[linenumbering, mode=text]{$\AdvertisementVerify(\pp, X, \advt)$}{
\ret $\nizk.\Verify(\crs, (X, \pp', \mpk, \ct), \pi)$
}

\procedure[linenumbering, mode=text]{$\AuxGen(\advt, \state, \bfy)$}{
Parse $\advt = (\mpk, \ct, \pi)$, $\st = (\msk, \bft)$
\\ 
Let $\widetilde{\bfy} := (\bfy^T, f_\bfy(\bft))^T \in \Z_p^{\ell+1}$
\\ 
Let $\pk_\bfy := \ipfe.\PubKGen(\mpk, \widetilde{\bfy})$ 
\\ 
\ret $\aux_\bfy := \pk_\bfy$, $\pi_\bfy := f_\bfy(\bft)$
}

\end{pcvstack}
\begin{pcvstack}

\procedure[linenumbering, mode=text]{$\AuxVerify(\advt, \bfy, \aux_\bfy, \pi_\bfy)$}{
Parse $\advt = (\mpk, \ct, \pi)$,
let $\widetilde{\bfy} := (\bfy^T, \pi_\bfy)^T$ 
\\ 
\ret $1$ iff $\aux_\bfy = \ipfe.\PubKGen(\mpk, \widetilde{\bfy})$
}

\procedure[linenumbering, mode=text]{$\FPreSign(\advt, \sk, m, X, \bfy, \aux_\bfy)$}{
% Parse $\advt = (\mpk, \ct, \pi)$
% \\ 
\ret $\widetilde{\sigma} \gets \as.\PreSign(\sk, m, \aux_\bfy)$
}

\procedure[linenumbering, mode=text]{$\FPreVerify(\advt, \vk, m, X, \bfy, \aux_\bfy, \pi_\bfy, \widetilde{\sigma})$}{
% Parse $\advt = (\mpk, \ct, \pi)$
% \\ 
\ret $\AuxVerify(\advt, \bfy, \aux_\bfy, \pi_\bfy) \wedge$ 
\pcskipln \\ 
$\quad \qquad \as.\PreVerify(\vk, m, \aux_\bfy, \widetilde{\sigma}) $
}

\procedure[linenumbering, mode=text]{$\Adapt(\advt, \state, \vk, m, X, \bfx, \bfy, \aux_\bfy, \widetilde{\sigma})$}{
Parse $\advt = (\mpk, \ct, \pi)$, $\state = (\msk, \bft)$
\\ 
Let $\widetilde{\bfy} := (\bfy^T, f_\bfy(\bft))^T$ 
% Let $\pk_\bfy := \ipe.\PubKGen(\mpk, \bfy)$
\\ 
Let $\sk_\bfy := \ipe.\KGen(\msk, \widetilde{\bfy})$
\\ 
\ret $\sigma := \as.\Adapt(\vk, m, \aux_\bfy, \sk_\bfy, \widetilde{\sigma})$
}

\procedure[linenumbering, mode=text]{$\FExt(\advt, \widetilde{\sigma}, \sigma, X, \bfy, \aux_\bfy)$}{ 
Parse $\advt = (\mpk, \ct, \pi)$.
\\ 
% Let $\pk_\bfy := \ipe.\PubKGen(\mpk, \bfy)$
% \\ 
Let $z := \as.\Ext(\widetilde{\sigma}, \sigma, \aux_\bfy)$
\\ 
\ret $v := \ipe.\Dec(z, \ct)$
}

\end{pcvstack}
\end{pchstack}
\caption{
\ifhldiff
{\color{hldiffcolor}
\fi
Construction: Functional Adaptor Signatures
\ifhldiff
}
\fi
}
\label{fig:fas-construction}
\label{sec:construction}
\ifhldiff
}
\fi
\end{figure*}

\fi

%% file: prelims-eprint.tex
\section{Preliminaries}
\label{sec:prelims}\label{sec:notations}

We recall the cryptographic tools needed in this work.

% \anote{Keep only the notations we need for group-based construction. The lattice one, push to the lattice section in the appendix.}

\smallskip\noindent\textbf{Notations.}
We denote by $x \gets S$ the experiment of sampling $x$ from a probability distribution $S$. 
We denote by $[A(\cdot)]$ the range of an algorithm $A(\cdot)$. 
If $p(\cdot, \cdot)$ denotes a predicate, then $\Pr[p(y, z): x \gets S, (y, z) \gets A(x)]$ is the probability that the predicate $p(y, z)$ is true after the ordered sequence of events $x \gets S$ followed by $(y, z) \gets A(x)$.
 We denote scalars by lower-case alphabets such as $x$, 
vectors by bold-face lower-case alphabets such as $\bfx$, 
matrices by bold-face upper-case alphabets such as $\bfX$. 
$\bfx = (x_1, x_2)^T$ denotes a column-vector with elements $x_1$ and $x_2$.
$\bfx^T$ denotes a row-vector.
Suppose $\G$ is a cyclic group of prime-order $p$ and with generator $g$. 
For a scalar $x \in \Z_p$,
$g^x$ denotes its group encoding. 
For a vector $\bfx = (x_1, \ldots, x_n)^T \in \Z_p^n$, 
$g^\bfx$ denotes its group encoding $(g^{x_1}, \ldots, g^{x_n})^T$.
% \pnote{why is there a transpose operation here?}
% \nikhil{transpose operation because it is written like a row vector.}
Similarly, for a matrix $\bfX$, $g^\bfX$ denotes its group encoding.
Given group encoding of a vector $g^\bfx$ and a vector $\bfy$, we can compute $g^{\bfx^T\bfy}$ efficiently as $g^{\bfx^T\bfy} = \Pi_{i \in [n]} (g^{x_i})^{y_i}$. 

\noindent\textbf{Linear Functions.}
We denote linear functions as inner product computation. 
Let $\mcal{F}_{{\sf IP}, \ell}$ denote the family of inner products of vectors of length $\ell$.
For a prime $p$, in this work we study two restrictions of inner product computations as follows:
\begin{itemize}[leftmargin=*]
\item \emph{Inner products modulo $p$.} 
The function class $\mcal{F}_{{\sf IP}, \ell, p} = \{ f_\bfy : \bfy \in \Z_p^\ell\}$, where $f_\bfy: \Z_p^\ell \to \Z_p$ is defined as $f_\bfy(\bfx) = \bfx^T \bfy \bmod{p}$.

\item \emph{Inner products with output values polynomially bounded by $B \ll p$.}
The function class $\mcal{F}_{{\sf IP}, \ell, p, B} = \{ f_\bfy : \bfy \in \Z_p^\ell\}$, where $f_\bfy: \Z_p^\ell \to \{0, \ldots, B\}$ is defined as $f_\bfy(\bfx) = \bfx^T \bfy \in \{0, \ldots, B\}$.

\end{itemize}
Often we will use $\bfy$ instead of $f_\bfy$ to denote the function. 

\smallskip\noindent\textbf{Digital Signature.}
A digital signature scheme 
$\DS := (\KGen, \Sign, \allowbreak \Verify)$ 
has a key generation algorithm $(\vk, \sk) \gets \KGen(1^\secparam)$ 
that outputs a verification-signing key pair. 
Using signing key $\sk$ we can compute signatures on a message $m$ 
by running $\sigma \gets \Sign(\sk, m)$, 
which can be publicly verified using the corresponding verification key $\vk$ 
by running $\Verify(\vk, m, \sigma)$. 
We require the digital signature scheme to satisfy strong existential unforgeability~\cite{sig}. 

% \smallskip\noindent\textbf{Hard Relations.}
% We recall the notion of a hard relation $R$ with statement/witness pairs 
% $(X, x)$. We denote by $L_R$ the associated language defined as 
% $L_R := \{ \ X \ | \ \exists \ x, (X , x) \in R \ \}$. 
% We want the following properties from the hard relation: (1) Statement-witness pairs $(X , x) \in R$ can be efficiently sampled using a  sampling algorithm $\mathsf{GenR}(1^\secpar)$, (2) for all PPT adversaries  $\adv$ the
% probability of $\adv$ on input $X$ outputting a witness $x'$ such that $(X, x') \in R$ is negligible, and 
% (3) for a given function class $\mcal{F}$, for any PPT adversaries $\adv$ that is given input $X$, and returns $(f,z)$, it must be the case that $(f \in \mcal{F}) \wedge \left (z \in \{f(x') : \exists \ x' \text{ such that } (X,x') \in R \} \right )$ only with negligible probability.
% If a relation satisfies only the first two properties, we call it \emph{hard}, and if it satisfies all three, it we call it \emph{$\mcal{F}$-hard}.
% The formal definitions are described 
% \ifcameraready
% in the full version.
% \else 
% in~\cref{sec:more-prelims}. 
% \fi
% In this work, we use the discrete log language $L_\DL$ defined with respect to a group $\GG$ with generator $g$ and order $p$. The language is defined as $L_\DL := \{ Y \mid \exists y \in \ZZ_p,\ Y = g^y \}$ with corresponding hard relation $R_\DL = \{ (X,x)| X = g^x\}$.

\subsection{Hard Relation}
We recall the notion of a hard relation $R$ with statement/witness pairs 
$(X, x)$. We denote by $L_R$ the associated language defined as 
$L_R := \{ \ X \ | \ \exists \ x, (X , x) \in R \ \}$. 

\begin{definition}[Hardness of $R$]
We say that a relation $R$ is hard, if 
(i) there exists a \ppt sampling algorithm $\GenR(1^\secparam)$ that 
outputs a statement/witness pair $(X, x) \in R$,
(ii) the relation is poly-time decidable,
(iii) for every non-uniform PPT adversary $\A$, there exists a negligible function $\negl$ such that for all $\secparam \in \N$,
\begin{equation}
\Pr[\mathsf{G}_{\A,R}( 1^\secparam) = 1] \le \negl(\secparam) \ ,
\end{equation}
where the game $\mathsf{G}_{\A,R}( 1^\secparam)$ is defined as follows:

\begin{pchstack}[boxed, center, space=1em]
\procedure[linenumbering, mode=text]{Game $\mathsf{G}_{\A,R}( 1^\secparam)$}{
$(X,x) \gets \GenR(1^\secparam)$
\\ 
$(x') \gets \A(1^\secparam, X)$
\\ 
If $(X, x') \in R$: \ret $1$
\\ 
\ret $0$	
}
\end{pchstack}

\end{definition}

\begin{definition}[Hardness of $R$ w.r.t.\ function class $\mcal{F}$]
\label{def:f-hard-relation}
We say that a relation $R$ is $\mcal{F}$-hard, if 
(i) there exists a \ppt sampling algorithm $\GenR(1^\secparam)$ that 
outputs a statement/witness pair $(X, x) \in R$,
(ii) for every non-uniform PPT adversary $\A$ there exists a negligible function $\negl$ such that for all $\secparam \in \N$,
\begin{equation}
\Pr[\mathsf{G}_{\A,R, \mcal{F}}( 1^\secparam) = 1] \le \negl(\secparam) \ ,
\end{equation}
where the game $\mathsf{G}_{\A,R, \mcal{F}}( 1^\secparam)$ is defined as follows:
% \pnote{If simpler, we should define hardness w.r.t.\ on linear functions}

\begin{pchstack}[boxed, center, space=1em]
\procedure[linenumbering, mode=text]{Game $\mathsf{G}_{\A,R, \mcal{F}}( 1^\lambda)$}{
$(X,x) \gets \GenR(1^\lambda)$ 
% \pcskipln \\ \nikhil{can we assume that $\GenR$ outputs all witnesses $x$ for $X$?}
\\ 
$(f,z) \gets \A(1^\lambda, X)$
\\ 
If $(f \in \mcal{F}) \wedge \left (z \in \{f(x') : \exists \ x' \text{ such that } (X,x') \in R \} \right )$: \ret $1$
% \pcskipln \\ \pnote{$A$ wins even if $A$ can learn $f(x')$ for another witness $x'$}
\\ 
\ret $0$	
}
\end{pchstack}
\end{definition}

We note that in the above definition we allow the adversary \A to choose $f$ adaptively as this hardness assumption will later be used in proving unforgeability of our functional adaptor signatures construction where the adversary is allowed to choose the challenge function adaptively. This is necessary as functional adaptor signatures are non-trivial to build only in such a setting and become trivial if the challenge function is selectively chosen upfront by the adversary. Further note that we let the adversary \A win even if \A outputs a function evaluation of a witness $x'$ of $X$ that is different from $x$.
% \pnote{Why we allow $A$ to choose $f$ adaptively, and that we let $A$ win even if $A$ learns a function of another witness.}

\begin{definition}[Hard Relation $R_\DL$]
\label{def:dl-relation}
The discrete log language $L_\DL$ is defined with respect 
to a group $G$ with generator $g$ and order $p$ as  
$L_\DL := \{ \ X \ | \ \exists \ x \in  \Z_p, X = g^x \ \}$ 
with the corresponding hard relation $R_\DL = \{ \ (X, x) \ | \ X = g^{x}\ \}$.
\end{definition}

\subsection{Adaptor Signatures}
Adaptor Signatures were first formally defined in~\cite{asig}.
% and were used to enable one-time fair exchange of a coin for a witness. 
%Two different lines of work have improved the syntax and security subsequently. 
~\cite{pq-as} generalized the syntax in~\cite{asig} to enable constructing adaptor signatures from lattices. 
More concretely, ~\cite{asig} defined adaptor signatures w.r.t.\ a digital signature scheme and a hard relation $R$.
~\cite{pq-as} generalized it to be w.r.t.\  two hard relations $R$ and $R'$ such that such that $R \subseteq R'$. For the group-based constructions, typically we have $R=R'$, but for the lattice based constructions, typically we have $R \neq R'$.
In this work, one of our constructions will be from lattices, hence, we follow this generalized syntax.

In a different vein, ~\cite{as-stronger,as-foundations} have strengthened the security properties of adaptor signatures, with notions like \emph{unforgeability}, \emph{pre-signature extractability}, and \emph{witness extractability}.
~\cite{as-stronger} further added \emph{unique extractability} and \emph{unlinkability} as security properties and unified \emph{unforgeability} and \emph{witness extractability} under a single security property called \emph{extractability}. ~\cite{as-foundations} further added \emph{pre-verify soundness} as a security property. We note that not all security properties are needed always. Within the scope of this work, we will use adaptor signatures as a building block for the construction of FAS and we only need \emph{pre-signature adaptability} and \emph{witness extractability} security properties of adaptor signatures. 
Intuitively, pre-signature adaptability ensures that 
given a valid pre-signature and a witness for the statement, one can always adapt the pre-signature into a valid signature.
Witness extractability requires that given an honestly generated pre-signature and a valid signature, one should be able to extract a witness.
We define these formally below.

\begin{definition}
An adaptor signature scheme 
$\AS_{\DS, R, R'} := (\PreSign, $ $\PreVerify, \Adapt, \Ext)$ 
is defined with respect to a signature scheme $\DS = (\KGen, \Sign, \Verify)$ 
and hard relations $R$ and $R'$ such that $R \subseteq R'$.
Here, $R$ constitutes the relation for the statement-witness pairs generated by $\GenR$ and $R'$ is an extended relation that defines the relation for {\em extracted} witnesses. 
The interfaces are described below.
% \anote{Can we put all interface definitions under definition environment?}

\begin{itemize}[leftmargin=*]
% \item 
% $(\sk, \vk) \gets \KGen(1^\secparam)$:
% The key generation algorithm takes as input the security parameter $1^\secparam$ 
% and outputs a signing key $\sk$ and a verification key $\vk$. 
% This algorithm is the same as in $\DS$. 

\item 
$\widetilde{\sigma} \gets \PreSign(\sk,m,X)$:
The pre-signing algorithm takes as input
a signing key $\sk$,
a message $m$, 
and a statement $X$ for the language $L_R$, 
and outputs a pre-signature $\widetilde{\sigma}$
(we sometimes also refer to this as a partial signature). 

\item 
$0/1 \gets \PreVerify(\vk, m, X, \widetilde{\sigma})$: 
The pre-signature verification algorithm takes as input 
a verification key $\vk$, 
a message $m$, 
a statement $X$ for the language $L_R$, 
and a pre-signature $\widetilde{\sigma}$,
and outputs $0/1$ signifying whether $\widetilde{\sigma}$ 
is correctly generated. 

\item 
$\sigma := \Adapt(\vk, m, X, x, \widetilde{\sigma})$: 
The adapt algorithm transforms 
a pre-signature $\widetilde{\sigma}$ into a valid signature $\sigma$
given the witness $x$ for the instance $X$ of the language $L_R$. 

% \item 
% $\sigma \gets \Sign(\sk, m)$: 
% The signing algorithm takes as input 
% a signing key $\sk$, 
% and a message $m$, 
% and outputs a signature $\sigma$. 
% This algorithm is the same as in $\DS$.

% \item 
% $0/1 \gets \Verify(\vk, m, \sigma)$: 
% The verification algorithm takes as input 
% a verification key $\vk$, 
% a message $m$, 
% and a signature $\sigma$, 
% and outputs $0/1$ signifying whether $\sigma$ 
% is correctly generated. 
% This algorithm is the same as in $\DS$. 

\item 
$x := \Ext(\widetilde{\sigma}, \sigma, X)$: 
The extract algorithm takes as input 
a pre-signature $\widetilde{\sigma}$, 
a signature $\sigma$, 
and an instance $X$, 
and outputs a witness $x'$ such that $(X, x') \in R'$, or $\bot$. 
% This can be formalized as correctness.

\end{itemize}

\end{definition}

Note that an adaptor signature scheme $\AS_{\DS, R, R'}$ also inherits $\KGen, \Sign, \Verify$ algorithms from the signature scheme $\DS$.

Typically, an adaptor signature is run between a buyer and a seller. 
The buyer runs $\KGen, \Sign, \allowbreak \PreSign, \Ext$ algorithms, 
the seller runs $\Adapt$ algorithm, and anyone can run $\PreVerify, \Verify$ algorithms.

\begin{remark}
Looking ahead, in the discrete-log based instantiation, we will have $R'=R$, but in the lattice-based instantiation, we will have $R' \neq R$ and the reason for this extension is the {\em knowledge/soundness gap} inherent in {\em efficient} lattice-based zero-knowledge proofs. 
% \nikhil{explain this more}
\end{remark}

\begin{definition}[Correctness]
\label{def:as-correctness}
% \nikhil{rename to correctness}
An adaptor signature scheme 
$\AS$ satisfies correctness if for every 
$n \in \N$, every message $m \in \{0,1\}^*$, 
and every statement/witness pair $(X, x) \in R$, the following holds:

\[
\Pr \left[
\begin{array}{c}
\PreVerify(\vk, m, X, \widetilde{\sigma}) = 1 \\ 
\wedge \ 
\Verify(\vk, m, \sigma) = 1 \\
\wedge \ 
(X, x') \in R'
\end{array}
:
\begin{array}{l}
(\sk, \vk) \gets \KGen(1^\secparam) \\
\widetilde{\sigma} \gets \PreSign(\sk, m, X) \\ 
\sigma := \Adapt(\vk, m, X, x, \widetilde{\sigma}) \\ 
x' := \Ext(\widetilde{\sigma}, \sigma, X)
\end{array}
\right] = 1
\]

\end{definition}

In terms of security, we want witness extractability against a malicious seller and weak pre-signature adaptability against a malicious seller. We explain these next.

\ignore{

		\begin{definition}[Unforgeability]
		An adaptor signature scheme $\AS$ is aEUF-CMA secure if 
		for every \ppt adversary \A there exists a negligible function $\negl$ such that
		for all $\secparam \in \N$,
		\[
		\Pr[ \aSigForge_{\A,\AS}(1^\secparam) = 1] \leq \negl(\secparam),
		\]
		where the experiment $\aSigForge_{\A,\AS}$ 
		is defined as in~\Cref{fig:asig-unf-exp}.

		\end{definition}

		\begin{figure}[H]
		\centering
		\captionsetup{justification=centering}
		\begin{pchstack}[boxed, space=1em]
		\begin{pcvstack}[space=1em]
		\procedure[linenumbering]{Experiment $\aSigForge_{\A,\AS}$.}{
		\mcal{Q} := \emptyset \\
		(\sk, \vk) \gets \KGen(1^\secparam) \\ 
		m \gets \A^{\mcal{O}_S(\cdot), \mcal{O}_{pS}(\cdot, \cdot)}(\vk) \\ 
		(X, x) \gets \GenR(1^\secparam) \\ 
		\widetilde{\sigma} \gets \PreSign(\sk, m, X) \\ 
		\sigma \gets \A^{\mcal{O}_S(\cdot), \mcal{O}_{pS}(\cdot, \cdot)}(\widetilde{\sigma}, X) \\ 
		\text{\ret \ } ((m \notin \mcal{Q}) \wedge \Verify(\vk, m, \sigma))
		} 

		\end{pcvstack}
		\begin{pcvstack}[space=1em]

		\procedure[linenumbering]{Oracle $\mcal{O}_S(m)$}{
		\sigma \gets \Sign(\sk, m) \\ 
		\mcal{Q} := \mcal{Q} \vee \{m\} \\ 
		\text{\ret \ } \sigma
		}

		\procedure[linenumbering]{Oracle $\mcal{O}_{pS}(m, X)$}{
		\widetilde{\sigma} \gets \PreSign(\sk, m, X) \\ 
		\mcal{Q} := \mcal{Q} \vee \{m\} \\ 
		\text{\ret \ } \widetilde{\sigma}
		}
		\end{pcvstack}

		\end{pchstack}
		\caption{Unforgeability experiment of adaptor signatures}
		\label{fig:asig-unf-exp}
		\end{figure}
}

Witness extractability requires that for an honestly generated pre-signature and a valid signature, one should be able to extract a witness. Intuitively, it tries to capture that an interaction between an honest buyer (i.e., generates valid pre-signature) and a malicious seller trying to get paid (i.e., adapt pre-signature into a valid signature) without revealing a witness (i.e., extraction failure) should be unlikely to occur.

\begin{definition}[Witness Extractability]
\label{def:as-wit-ext}
An adaptor signature scheme $\AS$ is witness extractable if for every PPT adversary \A, there exists a negligible function $\negl$ such that
for all $\secparam \in \N$,
\[
\Pr[ \aWitExt_{\A, \AS}(1^\secparam) = 1] \leq \negl(\secparam),
\]
where the experiment $\aWitExt_{\A, \AS}$ is defined as in~\Cref{fig:asig-wit-ext-exp}.
\end{definition}

\ifacm
	\begin{figure}[H]
	\centering
	\captionsetup{justification=centering}
	% \fbox{
	\begin{pcvstack}[boxed, space=1em]
	\begin{pchstack}[space=1em]
	\procedure[linenumbering]{Experiment $\aWitExt_{\A,\AS}$.}{
	\mcal{Q} := \emptyset \\
	(\sk, \vk) \gets \KGen(1^\secparam) \\ 
	(m, X) \gets \A^{\mcal{O}_S(\cdot), \mcal{O}_{pS}(\cdot, \cdot)}(\vk) \\ 
	\widetilde{\sigma} \gets \PreSign(\sk, m, X) \\ 
	\sigma \gets \A^{\mcal{O}_S(\cdot), \mcal{O}_{pS}(\cdot, \cdot)}(\widetilde{\sigma}) \\ 
	x' := \Ext(\widetilde{\sigma}, \sigma, X) \\ 
	\text{\ret \ } ((m \notin \mcal{Q}) \wedge ((X, x') \notin R') \wedge \Verify(\vk, m, \sigma))
	}
	\end{pchstack}
	\begin{pchstack}[space=1em]
	\procedure[linenumbering]{Oracle $\mcal{O}_S(m)$}{
	\sigma \gets \Sign(\sk, m) \\ 
	\mcal{Q} := \mcal{Q} \vee \{m\} \\ 
	\text{\ret \ } \sigma
	}

	\procedure[linenumbering]{Oracle $\mcal{O}_{pS}(m, X)$}{
	\widetilde{\sigma} \gets \PreSign(\sk, m, X) \\ 
	\mcal{Q} := \mcal{Q} \vee \{m\} \\ 
	\text{\ret \ } \widetilde{\sigma}
	}
	\end{pchstack}
	\end{pcvstack}
	% }
	\caption{Witness Extractability experiment of adaptor signatures}
	\label{fig:asig-wit-ext-exp}
	\end{figure}
\else
	\begin{figure}[H]
	\centering
	\captionsetup{justification=centering}
	% \fbox{
	\begin{pchstack}[boxed, space=1em]
	\begin{pcvstack}[space=1em]
	\procedure[linenumbering]{Experiment $\aWitExt_{\A,\AS}$.}{
	\mcal{Q} := \emptyset \\
	(\sk, \vk) \gets \KGen(1^\secparam) \\ 
	(m, X) \gets \A^{\mcal{O}_S(\cdot), \mcal{O}_{pS}(\cdot, \cdot)}(\vk) \\ 
	\widetilde{\sigma} \gets \PreSign(\sk, m, X) \\ 
	\sigma \gets \A^{\mcal{O}_S(\cdot), \mcal{O}_{pS}(\cdot, \cdot)}(\widetilde{\sigma}) \\ 
	x' := \Ext(\widetilde{\sigma}, \sigma, X) \\ 
	\text{\ret \ } ((m \notin \mcal{Q}) \wedge ((X, x') \notin R') \wedge \Verify(\vk, m, \sigma))
	}
	\end{pcvstack}
	\begin{pcvstack}[space=1em]
	\procedure[linenumbering]{Oracle $\mcal{O}_S(m)$}{
	\sigma \gets \Sign(\sk, m) \\ 
	\mcal{Q} := \mcal{Q} \vee \{m\} \\ 
	\text{\ret \ } \sigma
	}

	\procedure[linenumbering]{Oracle $\mcal{O}_{pS}(m, X)$}{
	\widetilde{\sigma} \gets \PreSign(\sk, m, X) \\ 
	\mcal{Q} := \mcal{Q} \vee \{m\} \\ 
	\text{\ret \ } \widetilde{\sigma}
	}
	\end{pcvstack}
	\end{pchstack}
	% }
	\caption{Witness Extractability experiment of adaptor signatures}
	\label{fig:asig-wit-ext-exp}
	\end{figure}
\fi
Note that, in the above witness extractability definition, the adversary's winning condition is restricted to the extracted witness not being in $R'$.
Since $R \subseteq R'$, $(X, x') \notin R'$ implies $(X, x') \notin R$. Therefore, it is sufficient to ensure that $R'$ is a hard relation, which itself implies that $R$ is also a hard relation. As a result, in our security assumptions, we make sure that $R'$ is a hard relation. 

Weak pre-signature adaptability requires that given a valid pre-signature and a witness for the instance, 
one can always adapt the pre-signature into a valid signature.
Intuitively, this tries to capture that it should be impossible for a malicious buyer to learn the witness without doing the payment. 

\begin{definition}[Weak Pre-signature Adaptability]
\label{def:as-pre-sig-adaptability}
An adaptor signature scheme $\AS$ satisfies weak pre-signature adaptability 
if for any $\secparam \in \N$, 
any message $m \in \{0,1\}^*$, 
any key pair $(\sk, \vk) \gets \KGen(1^\secparam)$, 
any statement/witness pair $(X, x) \in R$, 
and any pre-signature $\widetilde{\sigma} \in \{0,1\}^*$ 
with $\PreVerify(\vk, \allowbreak m, X, \widetilde{\sigma}) = 1$, we have:
\[
\Pr[\Verify(\vk, m, \Adapt(\vk, m, X, x, \widetilde{\sigma})) = 1] = 1.
\]
\end{definition}

Note that the above pre-signature adaptability is called {\em weak} because only statement-witness pairs satisfying $R$ are guaranteed to be adaptable, and not those satisfying $R'$. Therefore, weak pre-signature adaptability does not guarantee, for example, that an {\em extracted} witness can be used to adapt a pre-signature successfully. This is however guaranteed whenever $R=R'$.

We have several constructions of adaptor signatures compatible with ECDSA~\cite{asig}, lattice-based signatures~\cite{pq-as,albrecht2022lattice}, and dichotomic signature schemes~\cite{as-foundations} which is a recently introduced abstraction that captures many popular schemes like Schnorr, CL, BBS, etc. The group-based schemes are constructed w.r.t.\ the discrete logarithm NP language in the respective groups, while the lattice-based schemes are constructed w.r.t.\ the short integer solution NP language.

\subsection{Non-Interactive Zero-Knowledge Arguments}\label{sec:nizk}

\begin{definition}[Secure NIZK Argument System]
\label{def:nizk}
A $\NIZK$ for language $L_R$ in the {\em common reference string (CRS)} model consists of the following possibly randomized algorithms.

\begin{itemize}[itemsep=1pt,leftmargin=*]
\item 
$\crs \gets \Setup(1^\secparam)$:
it takes in a security parameter $\secparam$ 
and outputs the common reference string $\crs$ which is publicly known to everyone.

\item 
$\pi \gets \Prove(\crs, \stmt, \wit)$:
it takes in the common reference string $\crs$, 
a statement $\stmt$, 
and a witness $\wit$, 
outputs a proof $\pi$.

\item 
$0/1 \gets \Verify(\crs, \stmt, \pi)$:
it takes in the common reference string $\crs$, 
a statement $\stmt$, 
and a proof $\pi$, 
and either accepts (with output $1$) the proof, or rejects (with output $0$) it.
\end{itemize}
A secure $\nizk$ argument system must satisfy the following properties:
\begin{itemize}[itemsep=1pt,leftmargin=*]
\item 
{\bf Completeness:}
For all $\stmt, \wit$ where $R(\stmt, \wit) = 1$, 
if $\crs \gets \Setup(1^\secparam)$ and 
$\pi \gets \Prove(\crs, \allowbreak \stmt, \wit)$, then, 
\[
\Pr[\Verify(\crs, \stmt, \pi) = 1] = 1.
\]
\item 
% \pnote{why are we defining statistical soundness?}
{\bf Adaptive Soundness:}
For all non-uniform \ppt provers $\mcal{P}^*$, 
there exists a negligible function $\negl$ such that 
for all $\secparam \in \N$,
if $\crs \gets \Setup(1^\secparam)$ and 
$(\stmt, \pi) \gets \mcal{P}^*(\crs)$, then,
\[
\Pr[\stmt \notin L_R \wedge \Verify(\crs, \stmt, \pi) = 1] \leq \negl(\secparam).
\]
% \pnote{$\cap$ is set intersection, you probably want $\wedge$ which is the logical-and operation}
\item 
{\bf Zero-Knowledge:}
% For all \ppt verifiers $\mcal{V}^*$,
There exists a 
% \pnote{NIZK simulators are strict polytime} 
\ppt simulator $\Sim = (\Setup^*, \allowbreak \Prove^*)$ 
and there exists a negligible function $\negl$ such that 
for all \ppt distinguishers \D, 
for all $\secparam \in \N$,
for all $(\stmt, \wit) \in R$, 
% \pnote{verifier shouldn't be choosing the witness}
% \pnote{what is adaptive about this definition?}
% \pnote{better to define two games: REAL and IDEAL played with the verifier where games output views, and then ask views to be indistinguishbable}
% \pnote{It is better to avoid the notation $\approx_c$ and spell it out completely using a distinguisher.}
\[
| 
\Pr[\D(\crs, \pi) = 1] 
-
\Pr[\D(\crs^*, \pi^*) = 1]
| \leq \negl(\secparam),
\]
where
$\crs \gets \Setup(1^\secparam)$,
% $\crs^* \gets \Setup^*(1^\secparam)$,
$\pi \gets \Prove(\crs, \stmt,$ $ \wit)$, 
$(\crs^*, \td) \allowbreak \gets \Setup^*(1^\secparam)$,
$\pi^* \allowbreak \gets \Prove^*(\crs^*, \td, \allowbreak \stmt)$ and 
$\td$ is a trapdoor used by $\Sim$ to come up with proof $\pi^*$ for a statement $\stmt$ without knowing its witness.
% $\pi^* \gets \Prove^*(\crs^*, \stmt)$.
% \nikhil{check with Pratik if these definitions are alright.}
\end{itemize}

\end{definition}

NIZK arguments are known from 
factoring~\cite{feige1990multiple}, 
and standard assumptions on 
pairings~\cite{canetti2003forward,groth2006non} and 
lattices~\cite{peikert2019noninteractive,canetti2019fiat}.

\subsection{Inner Product Functional Encryption}
\label{sec:ipeprelim}
% \anote{Compress this into short paragraph}
We will need a inner-product functional encryption scheme (IPFE). 
Let $\ell$ be an integer. 
We define IPFE for computing inner products of vectors of length $\ell$. Let the message space be $\mcal{M} \subseteq \Z^\ell$ and function space $\mcal{F}_{{\sf IP}, \ell}$. Messages are denoted by $\bfx \in \mcal{M}$, functions are denoted by $\bfy \in \mcal{F}_{{\sf IP}, \ell}$ and the function evaluation is denoted by $f_\bfy(\bfx)$ (all three of them are parameterized by the prime $p \in \pp$ as output by $\Gen$ below).
% \pnote{maybe emph that we define IPE scheme for message space $\mathbb{Z}_p^n$ and hence our functions are described by a $\mathbb{Z}_p^n$ element.}
% However, we will need our underlying IPE to satisfy a few nice properties, including
% the fact that $\Enc$ and $\KGen$ 
% should still work when taking in the group encoding
% of the plaintext or key vector; moreover, we want that the scheme
% computes the ``inner-product in the exponent''.
Formally, IPFE consists of the 
following algorithms:

\begin{itemize}[itemsep=1pt,leftmargin=*]
\item 
$\pp \leftarrow \Gen(1^\secparam)$: 
it is a {\em randomized} algorithm that 
takes in a security parameter $\secparam$ 
and samples public parameters $\pp$.
We assume that $\pp$ contains a prime number $p$.
% We will assume that $\pp$ contains
% the description of a cyclic group $\G$ 
% of prime order $p$, and
% a random generator $g \in \G$.

\item 
$(\mpk, \msk) \leftarrow \Setup(\pp, \ell)$:
it is a {\em randomized} algorithm that 
takes in the public parameters $\pp$ and the vector length $\ell$,
outputs a master public key $\mpk$ and 
a master secret key $\msk$. 

\item 
$\sk_\bfy := \KGen(\msk, \bfy \in \mcal{F}_{{\sf IP}, \ell})$:
it is a {\em deterministic} algorithm that 
takes in the master secret key $\msk$, and 
a vector $\bfy \in \mcal{F}_{{\sf IP}, \ell}$,
outputs a functional secret key $\sk_\bfy$.

\item 
$\ct \leftarrow \Enc(\mpk, \bfx \in \mcal{M})$:
it is a {\em randomized} algorithm that 
takes in the master public key $\mpk$, 
a plaintext vector $\bfx \in \mcal{M}$, 
and outputs a ciphertext $\ct$.

\item 
$v := \Dec(\sk_\bfy, \ct)$:
it is a {\em deterministic} algorithm that 
takes in the functional secret key $\sk_\bfy$ 
%the group encoding $\grp\bfy$ of $\bfy$, 
and a ciphertext $\ct$, and outputs a decrypted outcome $v$.
\end{itemize}

Further, we augment the IPFE scheme with an additional algorithm $\PubKGen$ defined as follows. 
\begin{itemize}[itemsep=1pt,leftmargin=*]
\item $\pk_\bfy := \PubKGen(\mpk, \bfy \in \mcal{F}_{{\sf IP}, \ell})$:
it is a {\em deterministic} algorithm that 
takes in the master public key $\mpk$, and 
a vector $\bfy \in \mcal{F}_{{\sf IP}, \ell}$,
outputs a functional public key $\pk_\bfy$. 
\end{itemize}

Next, we describe the correctness of IPFE.
\begin{definition}[IPFE Correctness]
\label{def:ipfe-correctness}
For any $\secparam \in \N$,
let $\pp \gets \Gen(1^\secparam)$, 
then for any 
$ \ell \in \N, \bfx \in \mcal{M}, \bfy \in \mcal{F}_{{\sf IP}, \ell}$, 
the following holds:
\[
\Pr\left[
v = f_\bfy(\bfx): 
\begin{array}{l}
(\mpk, \msk) \leftarrow  \Setup(\pp, \ell),\\ 
\sk_\bfy = \KGen(\msk, \bfy),\\ 
\ct \leftarrow \Enc(\mpk, \bfx),\\ 
v = \Dec(\sk_\bfy, \ct)
\end{array}
\right] = 1.
\]
\end{definition}

% \anote{Introduce $\PubKGen$ here.}

% \pnote{we now want to modify this definition be selective in the challenge message.}
Next, we describe the security of IPFE.
\begin{definition}[Selective, IND-security]
\label{def:ipfe-sel-ind-sec}
We say that the \ipe scheme satisfies selective,  IND-security, 
iff for any non-uniform \ppt admissible adversary \A, 
there exists a negligible function $\negl$ 
such that for all $\secparam \in \N, \ell \in \N$,
\ifacm
	\begin{align*}
	& |
	\Pr[\ipfeIND^0(1^\secparam, \ell) = 1] \\
	& \quad - \Pr[\ipfeIND^1(1^\secparam, \ell) = 1]
	| \leq \negl(\secparam),
	\end{align*}
\else 
\[
| \Pr[\ipfeIND^0(1^\secparam, \ell) = 1] - \Pr[\ipfeIND^1(1^\secparam, \ell) = 1] | \leq \negl(\secparam),
\]
\fi
where experiments $\ipfeIND^b(1^\secparam)$ for $b \in \{0,1\}$ are described in~\Cref{fig:ipfe-sel-ind-sec}.
Here, an adversary \algA is said to be {\it admissible} iff the following
holds with probability $1$:
for any $\bfy$
submitted in a $\mcal{O}_\KGen$ oracle query, 
it must be that 
$f_\bfy(\bfx^*_0) = f_\bfy(\bfx^*_1)$.
\end{definition}

\ifacm
	\begin{figure}[H]
	\centering
	\captionsetup{justification=centering}
	% \fbox{
	\begin{pcvstack}[boxed, space=1em]
	\procedure[linenumbering, mode=text]{Experiment $\ipfeIND^b(1^\secparam, \ell)$}{
	$\pp \gets \Gen(1^\secparam)$
	\\ 
	$(\bfx^*_0, \bfx^*_1) \gets \A(\pp)$
	\\ 
	$(\mpk, \msk) \gets \Setup(\pp, \ell)$
	\\ 
	$\ct^* \gets \Enc(\mpk, \bfx^*_b)$
	\\ 
	$b' \gets \A^{\mcal{O}_\KGen(\cdot)}(\mpk, \ct^*)$
	\\ 
	\ret $b'$
	}
	\procedure[linenumbering, mode=text]{Oracle $\mcal{O}_\KGen(\bfy)$}{
	\ret $\sk_\bfy = \KGen(\msk, \bfy)$
	}
	\end{pcvstack}
	% }
	\caption{Selective, IND-security experiment of Inner Product Functional Encryption
	% \pnote{To be consistent with the literature, I would just call it WI and WH instead of WitInd and WitHid.}
	% \nikhil{done}
	}
	\label{fig:ipfe-sel-ind-sec}
	\end{figure}
\else 
	\begin{figure}[H]
	\centering
	\captionsetup{justification=centering}
	% \fbox{
	\begin{pchstack}[boxed, space=1em]
	\procedure[linenumbering, mode=text]{Experiment $\ipfeIND^b(1^\secparam, \ell)$}{
	$\pp \gets \Gen(1^\secparam)$
	\\ 
	$(\bfx^*_0, \bfx^*_1) \gets \A(\pp)$
	\\ 
	$(\mpk, \msk) \gets \Setup(\pp, \ell)$,
	$\ct^* \gets \Enc(\mpk, \bfx^*_b)$
	\\ 
	$b' \gets \A^{\mcal{O}_\KGen(\cdot)}(\mpk, \ct^*)$
	\\ 
	\ret $b'$
	}
	\procedure[linenumbering, mode=text]{Oracle $\mcal{O}_\KGen(\bfy)$}{
	\ret $\sk_\bfy = \KGen(\msk, \bfy)$
	}
	\end{pchstack}
	% }
	\caption{Selective, IND-security experiment of Inner Product Functional Encryption
	% \pnote{To be consistent with the literature, I would just call it WI and WH instead of WitInd and WitHid.}
	% \nikhil{done}
	}
	\label{fig:ipfe-sel-ind-sec}
	\end{figure}
\fi
The above IND-security level is called selective because it requires the adversary to commit to the challenge message pair even before it sees $\mpk$. Stronger security models exist such as (i) semi-adaptive, IND-security where the adversary can commit to the challenge message pair after seeing $\mpk$ but before gaining access to the oracle $\mcal{O}_\KGen$ and (ii) adaptive, IND-security where the adversary can commit to the challenge message pair after seeing $\mpk$ and the power to make queries to the oracle $\mcal{O}_\KGen$ before and after committing to the challenge message pair. But, we only consider selective, IND-security as it suffices for the scope of this work.
We also note that any adaptive, IND-secure IPFE is also selective, IND-secure.

\begin{remark}
We note that $\pk_\bfy$ does not enable decryption by any means, so it will not change the correctness requirement. It will also not affect security of IPFE in any way. To see this, note that $\PubKGen$ is deterministic. So, an adversary can always compute $\pk_\bfy$ on its own. Given that $\KGen$ is also deterministic, one way to think of $\pk_\bfy$ is as a statistically binding and computationally hiding commitment to $\sk_\bfy$. Typically, $\mpk$ also has this property w.r.t.\ $\msk$ and it is always implicit in the security definitions of IPFE. We make this explicit below for $\pk_\bfy$ and $\sk_\bfy$ because our functional adaptor signatures will rely on it. 
We formalize this compliance requirement below.
\end{remark}

\begin{definition}[$R_\ipfe$-Compliant IPFE]
\label{def:ipfe-compliant-appendix}
\label{def:ipfe-compliant}
For a hard relation $R_\ipfe$, we say that an IPFE scheme is $R_\ipfe$ compliant if
for any $\secparam \in \N$,
for any $\pp \gets \ipe.\Gen(1^\secparam)$,
for any $(\mpk, \msk) \gets \ipe.\Setup(\pp, 1^\ell)$, 
for any $\bfy \in \mcal{F}_{{\sf IP}, \ell}$,
let $\pk_\bfy := \PubKGen(\mpk, \bfy)$, and 
$\sk_\bfy := \KGen(\msk, \bfy)$.
Then, it must be the case that $(\pk_\bfy, \sk_\bfy) \in R_\ipe$.
\end{definition}

Next, we describe an additional correctness property of IPFE that is tailored specifically to make IPFE compatible with adaptor signatures when the two are used in tandem in our functional adaptor signature construction. Looking ahead, we need this property because in our functional adaptor signature construction, we will use a standard adaptor signature to sell functional secret keys. Recall that adaptor signatures are defined w.r.t.\ two relations $R$ and $R'$. In this case, $R$ will be $R_\ipfe$ as defined in~\Cref{def:ipfe-compliant-appendix} and let us denote the corresponding extended relation $R'$ by $R'_\ipfe$. Then, the $\as.\Ext$ algorithm of adaptor signature scheme $\as$ will extract a functional secret key $\sk'_\bfy$ as the witness for the adaptor statement $\pk_\bfy$. While $\sk'_\bfy$ will be a witness of $R'_\ipfe$, it may not be a witness of the base relation $R_\ipfe$ since  $R_\ipfe \subseteq R'_\ipfe$, and we still want that $\sk'_\bfy$ should enable meaningful decryption. 
We formalize this robustness requirement below.

\begin{definition}[$R'_\ipfe$-Robust IPFE]
\label{def:ipfe-robust-appendix}
\label{def:ipfe-robust}
Let $R_\ipfe$ be the hard relation as defined in~\Cref{def:ipfe-compliant-appendix}. 
Let $R'_\ipfe$ such that $R_\ipfe \subseteq R'_\ipfe$ be an 
an extended relation of $R_\ipfe$. 
We say that an IPFE scheme is $R'_\ipfe$ robust if 
IPFE correctness holds even w.r.t.\ a functional secret key satisfying $R'_\ipfe$. More formally, 
for any $\secparam \in \N$,
let $\pp \gets \ipfe.\Gen(1^\secparam)$,
then for any $\ell \in \N$, 
$\bfx \in \mcal{M}, \bfy \in \mcal{F}_{{\sf IP}, \ell}$, the following holds: 
\[
\Pr\left[
v = f_\bfy(\bfx): 
\begin{array}{l}
(\mpk, \msk) \leftarrow  \Setup(\pp, \ell),\\ 
\pk_\bfy = \PubKGen(\msk, \bfy),\\ 
\sk'_\bfy \ s.t.\ (\pk_\bfy, \sk'_\bfy) \in R'_\ipfe,\\
\ct \leftarrow \Enc(\mpk, \bfx),\\ 
v = \Dec(\sk'_\bfy, \ct)
\end{array}
\right] = 1.
\]
\end{definition}

\begin{remark}
One can think of decryption correctness (\Cref{def:ipfe-correctness}) as being implicitly defined with respect to $\sk_\bfy$ that is witness to $\pk_\bfy$ for the relation $R_\ipfe$. And one can think of the $R'_\ipfe$-robustness as decryption correctness with respect to $\sk'_\bfy$ that is witness to $\pk_\bfy$ for the relation $R'_\ipfe$. 
% Note that here $R_\ipfe \subseteq R'_\ipfe$. 
% Thus, if $(\pk_\bfy, \sk_\bfy) \in R_\ipfe$, then, $(\pk_\bfy, \sk_\bfy) \in R'_\ipfe$. But, $(\pk_\bfy, \sk'_\bfy) \in R'_\ipfe$ does not necessarily imply $(\pk_\bfy, \sk'_\bfy) \in R_\ipfe$. It is for this reason that we need IPFE special property 2 beyond just IPFE correctness.
\end{remark}

% Lastly, we require a compatability property from our IPFE as follows. This will be useful for proving zero-knowledge witness privacy of our functional adaptor signature scheme.

% \begin{definition}[$(\mcal{M},\mcal{F}_{{\sf IP}, \ell})$- Compatible IPFE]
% \label{def:ipfe-compatible}
% We say that an IPFE scheme for message space $\mcal{M}$ and function space $\mcal{F}_{{\sf IP}, \ell}$ is $(\mcal{M},\mcal{F}_{{\sf IP}, \ell})$-compatible if for any $\bfx \in \mcal{M}$ and any $\bfy \in \mcal{F}_{{\sf IP}, \ell}$, it is the case that $f_\bfy(\bfx) \in \mcal{F}_{{\sf IP}, 1}$.
% \end{definition}

Abdalla et al.~\cite{ipe} constructed selective, IND-secure IPFE schemes from DDH and LWE. Subsequently, Agrawal et al.~\cite{ipe01} constructed adaptive, IND-secure IPFE schemes from DDH, DCR, and LWE. 
Within the scope of this work, we will use the following two IPFE schemes:
\begin{itemize}[leftmargin=*]
\item 
The DDH based IPFE by Abdalla et al.~\cite{ipe}. 
This scheme computes inner products with output values polynomially bounded by $B \ll p$. 
Specifically, the message space is $\mcal{M} = \Z_p^\ell$ and the function class is $\mcal{F}_{{\sf IP}, \ell, p, B} = \{ f_\bfy : \bfy \in \Z_p^\ell\} \subseteq \mcal{F}_{{\sf IP}, \ell}$, where $f_\bfy: \Z_p^\ell \to \{0, \ldots, B\}$ is defined as $f_\bfy(\bfx) = \bfx^T \bfy \in \{0, \ldots, B\}$. 
% \nikhil{add reference to where it is described}
\item 
The LWE based IPFE by Agrawal et al.~\cite[Section 4.2]{ipe01}.
This scheme computes inner products modulo $p$. 
Specifically, the message space is $\mcal{M} = \Z_p^\ell$ and the function class is $\mcal{F}_{{\sf IP}, \ell, p} = \{ f_\bfy : \bfy \in \Z_p^\ell\} \subseteq \mcal{F}_{{\sf IP}, \ell}$, where $f_\bfy: \Z_p^\ell \to \Z_p$ is defined as $f_\bfy(\bfx) = \bfx^T \bfy \bmod{p} \in \Z_p$.
% \nikhil{add reference to where it is described}
\end{itemize}

\subsection{Lattice Preliminaries}

We denote by $\S_c = \{ \bfx \in \Z_q^n : ||\bfx||_\infty \leq c\}$ the set of vectors in $\Z_q^n$ whose maximum absolute coefficient is at most $c \in \Z^+$. 

\begin{definition}[$\sis_{n, m,q,\beta}$]
Let $\bfA \getr \Z_q^{n \times m}$.
Given $\bfA$, \sis problem with parameters $m>1$ and $0 < \beta < q$ asks to find a short non-zero $\bfv \in \Z^m$ such that $\bfA \bfv = 0 \bmod{q}$ and $||\bfv||_\infty \leq \beta$.
\end{definition}

Next, we define $\sis$ in Hermite Normal Form (HNF).
\begin{definition}[$\hnfsis_{n, m,q,\beta}$]
Let $\bfA' \getr \Z_q^{n \times (m-1)}$ and $\bfA = (\bfnum{1} || {\bfA'}) \in \Z_q^{n \times m}$.
Given $\bfA$, \hnfsis problem with parameters $m>1$ and $0 < \beta < q$ asks to find a short non-zero $\bfv \in \Z^m$ such that $\bfA \bfv = 0 \bmod{q}$ and $||\bfv||_\infty \leq \beta$.
\end{definition}

\begin{definition}[$\isis_{n, m, q, \beta}$]
\label{def:isis}
It is the inhomogeneous version of \hnfsis. 
Let $\bfA' \getr \Z_q^{n \times (m-1)}$ and $\bfA = (\bfnum{1} || {\bfA'}) \in \Z_q^{n \times m}$ and
let $\bfu \getr \Z^m$ s.t. $||\bfu||_\infty \leq \beta$. 
Let $\bfb = \bfA \bfu \in \Z_q^n$. 
Given $(\bfA, \bfb)$, the problem with parameters $m>1$ and $0 < \beta < q$ asks to find a short non-zero $\bfv \in \Z^m$ such that $\bfA \bfv = \bfb \bmod{q}$ and $||\bfv||_\infty \leq \beta$. 
\end{definition}

\begin{definition}[$\lwe_{n, m, q, \Psi}$]
\lwe problem with parameters $n, m>0$ and distribution $\Psi$ over $\Z_q$ asks to distinguish between the following two cases with non-negligible advantage: 1) $(\bfA, \bfb) \getr \Z_q^{m \times n} \times \Z_q^m$, and 2) $(\bfA, \bfA \bfs + \bfe)$ for $\bfA \getr \Z_q^{m \times n}$, a secret $\bfs \getr \Z_q^n$ and an error vector $\bfe \gets \Psi^m$.
\end{definition}

\begin{definition}[Multi-hint extended-LWE]
\label{def:mhe-lwe}
Let $q, m, t$ be integers, $\alpha$ be a real and $\tau$ be a distribution over $\Z^{t \times m}$, all of them functions of a parameter $n$. The multi-hint extended-LWE problem $\mheLWE_{n, m, q, \alpha, t, \tau}$ is to distinguish between the distributions of the tuples 
\[
(\bfA, \bfb, \bfZ, \bfZ \bfe) \text{ and }
(\bfA, \bfA \bfs + \bfe, \bfZ, \bfZ \bfe), 
\]
where $\bfA \getr \Z_q^{m \times n}$, $\bfs \getr \Z_q^n$, $\bfb \getr \Z_q^m$, $\bfe \gets D_{\Z, \alpha q}^m$, $\bfZ \gets \tau$.
\end{definition}

\begin{lemma}[\cite{ipe01}, Theorem 4]
Let $n \geq 100$, $q \geq 2$, $t < n$ and $m$ with $m = \Omega(n \log n)$ and $m \leq n^{O(1)}$. 
There exists $\xi \leq O(n^4 m^2 \log^{5/2} n)$ and a distribution $\tau$ over $\Z^{t \times m}$ such that the following statements hold:
\begin{itemize}[leftmargin=*]
\item 
There is a reduction from $\LWE_{n-t, m, q, D_{\Z, \alpha q}}$ to $\mheLWE_{n, m, q, \alpha, t, \tau}$.
\item 
It is possible to sample from $\tau$ in polynomial time in $n$.
\item 
Each entry of matrix $\tau$ is an independent discrete Gaussian $\tau_{i,j} = D_{\Z,\sigma_{i,j},c_{i,j}}$ for some $c_{i,j} \in \{0,1\}$ and $\sigma_{i,j} \geq \Omega(\sqrt{mn \log m})$.
\item 
With probability at least  $1 - n^{-\omega(1)}$, all rows from a sample from $\tau$ have norms at most $\xi$.

\end{itemize}
\end{lemma}

\begin{definition}[Discrete Gaussian distribution]
A discrete gaussian distribution $D_{\Lambda, \sqrt{\Sigma}, \bfc}$, for $\bfc \in \R^n$, $\Sigma$ a positive semi-definite matrix in $\R^{n \times n}$, and a lattice $\Lambda \subset \Z^n$, is a distribution with values in $\Lambda$ and probabilities
\[
\Pr[X = \bfx] \propto \exp{\left(-\frac12 (\bfx - \bfc)^T\Sigma^+(\bfx-\bfc)\right)}.
\]
\end{definition}
Note that $\Sigma^+$ denotes the pseudoinverse of a matrix. If $\Lambda = \Z^n$, we shall just write $D_{\sqrt{\Sigma}, \bfc}$. Further, if $\bfc = 0$, then we shall just write $D_{\sqrt{\Sigma}}$, and if $\sqrt{\Sigma} = \sigma I_n$ for $\sigma \in \R^+$ and $I_n$ an identity matrix, we shall just write $D_\sigma$.

% \nikhil{describe discrete gaussian over $\Z_q$: $\D_{\Z_q, \sigma, 0}$}

%% file: definitions-eprint.tex
\section{FAS Definition}
\label{sec:FAS}
\label{sec:defs-appendix}
% \nikhil{Ensure that figures are latest versions. the ones in this section might be old.}
A functional adaptor signature scheme shares similar syntax to an adaptor signature scheme with few additional interfaces and parameters. 
Specifically, 
the interfaces are $\FAS := (\Setup, \AdvertisementGen, \AdvertisementVerify, \allowbreak \AuxGen, \AuxVerify, \FPreSign, \FPreVerify, \allowbreak \Adapt, \FExt)$. 
% is defined with respect to a signature scheme $\DS = (\KGen, \Sign, \Verify)$, a hard relation $R$, and a family of functions $\mcal{F}$. 
We recall that $(\AuxGen, \AuxVerify, \allowbreak \FPreSign)$ essentially denote the 3-round interactive pre-signing process. We describe the formal definition below.

\begin{definition}[Functional adaptor signature]
\label{def:fas-syntax}
A functional adaptor signature scheme
$\FAS := (\Setup, \AdvertisementGen, \AdvertisementVerify, \AuxGen, $ $\AuxVerify, \FPreSign, \FPreVerify, \Adapt, \FExt)$ 
is defined with respect to a signature scheme $\DS = (\KGen, \Sign, \Verify)$, a hard relation $R$, and a family of functions $\mcal{F}$. The interfaces are described below.

\begin{itemize}[leftmargin=*]
\item 
$\pp \gets \Setup(1^\secparam)$: The setup algorithm takes as input the security parameter $1^\secparam$,
% the function class $\mcal{F}$, 
and outputs public parameters $\pp$.
% It is run by a trusted third party.
% \pnote{I think we should remove $\mcal{F}$ as inputs to the algorithm. Since we say that this scheme is defined wrt some particular $R$ and $\mcal{F}$, the Setup can depend on these.}
% \nikhil{say run by ttp}

\item 
$(\advt, \state) \gets \AdvertisementGen(\pp, X, x)$:
% \nikhil{change $\state$ to $\state$}
The advertisement generation algorithm takes as input 
% the function class $\mcal{F}$, 
a public parameter $\pp$,
a statement $X$ for the language $L_R$, and 
a witness $x$, and 
outputs a public advertisement $\advt$, and 
a secret state $\state$.
% \nikhil{add run by seller}

\item 
$0/1 \gets \AdvertisementVerify(\pp, X, \advt)$: 
The advertisement verify algorithm takes as input 
a public parameter $\pp$,
a statement $X$ for the language $L_R$, and 
a public advertisement $\advt$, and 
outputs $0/1$ signifying whether $\advt$ is generated correctly.

% \pnote{Do we really need to Setup, AdvertisementGen separately? I suspect you are using Setup to get the NIZK crs which neither of the parties should control?}
% \nikhil{yes, Setup is run by a trusted third party, and AdvertisementGen is run by the seller.}

% \item 
% $(\sk, \vk) \gets \KGen(1^\secparam)$:
% The key generation algorithm takes as input the security parameter $1^\secparam$ 
% and outputs a signing key $\sk$ and a verification key $\vk$.
% This algorithm is the same as in $\DS$. 

\item 
$(\aux_f, \pi_f) \gets \AuxGen(\advt, \state, f)$:
% \nikhil{change $\state$ to $\state$}
The auxiliary value generation algorithm takes as input 
% the function class $\mcal{F}$, 
% a public parameter $\pp$,
a public advertisement $\advt$,  
a {\em secret} state $\state$, 
a function $f$ belonging to the family $\mcal{F}$, and 
deterministically outputs an auxiliary value $\aux_f$, and 
a proof $\pi_f$.
% \nikhil{add run by seller}

\item 
$0/1 \gets \AuxVerify(\advt, f, \aux_f, \pi_f)$: 
The auxiliary value verify algorithm takes as input 
% a public parameter $\pp$,
a public advertisement $\advt$,  
a function $f$ belonging to the family $\mcal{F}$, 
an auxiliary value $\aux_f$, and 
a proof $\pi_f$, and
outputs $0/1$ signifying whether $\aux_f$ is generated correctly.

\item 
$\widetilde{\sigma} \gets \FPreSign(\advt, \sk, m, X, f, \aux_f)$:
The functional pre-signing algorithm takes as input
a public advertisement $\advt$,
a signing key $\sk$,
a message $m$, 
a statement $X$ for the language $L_R$,
a function $f$ belonging to the family $\mcal{F}$, and
an auxiliary value $\aux_f$, and 
outputs a pre-signature $\widetilde{\sigma}$.
% (we sometimes also refer to this as a partial signature). \pnote{why introduce another term for this?}

\item 
$0/1 \gets \FPreVerify(\advt, \vk, m, X, f, \aux_f, \pi_f, \widetilde{\sigma})$: 
The pre-signature verification algorithm takes as input 
a public advertisement $\advt$,
a {\em secret} state $\state$, 
a verification key $\vk$, 
a message $m$, 
a statement $X$ for the language $L_R$, 
a function $f$ belonging to the family $\mcal{F}$,
an auxiliary value $\aux_f$, 
a proof $\pi_f$, 
and a pre-signature $\widetilde{\sigma}$,
and outputs $0/1$ signifying whether $\widetilde{\sigma}$ 
is correctly generated. 

\item 
$\sigma := \Adapt(\advt, \state, \vk, m, X, x, f, \aux_f, \widetilde{\sigma})$: 
The adapt algorithm transforms 
a pre-signature $\widetilde{\sigma}$ into a valid signature $\sigma$
given the witness $x$ for the instance $X$ of the language $L_R$, 
a function $f$ belonging to the family $\mcal{F}$,
an auxiliary value $\aux_f$, 
a public advertisement $\advt$,
and a {\em secret} state $\state$. 

% \pnote{Should we emph that Adapt receives a secret-state as input and we envision Adapt and AdvertisementGen to be run by the same party in applications?}
% \nikhil{done}

% \item 
% $\sigma \gets \Sign(\sk, m)$: 
% The signing algorithm takes as input 
% a signing key $\sk$, 
% and a message $m$, 
% and outputs a signature $\sigma$. 
% This algorithm is the same as in $\DS$. 

% \item 
% $0/1 \gets \Verify(\vk, m, \sigma)$: 
% The verification algorithm takes as input 
% a verification key $\vk$, 
% a message $m$, 
% and a signature $\sigma$, 
% and outputs $0/1$ signifying whether $\sigma$ 
% is correctly generated. 
% This algorithm is the same as in $\DS$. 

\item 
$z := \FExt(\advt, \widetilde{\sigma}, \sigma, X, f, \aux_f)$: 
% \nikhil{rename output from $f(x)$ to $z$\textbf{}}
The functional extract algorithm takes as input 
a public advertisement $\advt$,
a pre-signature $\widetilde{\sigma}$, 
a signature $\sigma$, 
an instance $X$ for the language $L_R$, 
and a function $f \in \mcal{F}$,
an auxiliary value $\aux_f$, 
and outputs a value $z$ in the range of $f$. 
% \textbf{}\nikhil{is the definition assuming unique-witness languages?}
% This can be formalized as correctness.

\end{itemize}
\end{definition}
% \anote{Can we use public parameters instead of $\pp$?}
% \nikhil{continue here}

In a typical functional payment application, we will have a seller who holds a witness to the statement and wants to sell a function evaluation on it, and a buyer who wants to buy a function evaluation of the witness. We refer the reader to~\Cref{fig:fas-flow} for the protocol flow.

\begin{definition}[Correctness]
\label{def:fas-correctness}
% \nikhil{rename to correctness}
A functional adaptor signature scheme 
$\FAS$ satisfies correctness if for every 
$\secparam \in \N$, every message $m \in \{0,1\}^*$, 
every statement/witness pair $(X, x) \in R$, 
every function $f \in \mcal{F}$, suppose
$\pp \gets \Setup(1^\secparam)$,
$(\advt, \state) \gets \AdvertisementGen(\pp, X, x)$,
$(\sk, \vk) \gets \KGen(1^\secparam)$,
$(\aux_f, \pi_f) \gets \AuxGen(\advt, \state, f)$,
$\widetilde{\sigma} \gets \FPreSign$ $(\advt, \sk, m, X, f, \aux_f)$,
$\sigma := \Adapt($ $\advt, \state, \vk, m, X, x, f,\allowbreak \aux_f, \widetilde{\sigma})$, 
and $z := \FExt(\advt, \widetilde{\sigma}, \sigma, X, f, \aux_f)$. 
Then the following holds:
\[
\Pr \left[
\begin{array}{c}
\AdvertisementVerify(\pp, X, \advt) = 1 \\ 
% \pnote{best to use logical and operator $\wedge$} \\ 
\wedge \ \AuxVerify(\advt, f, \aux_f, \pi_f) = 1 \\
\wedge \ \FPreVerify(\advt, \vk, m, X, f, \pi_f, \widetilde{\sigma}) = 1 \\ 
% \text{and} \\ 
\wedge \ \Verify(\vk, m, \sigma) = 1 \ \wedge \ z = f(x)
% \text{and} \\ 
 
\end{array}
% : 
% \begin{array}{l}
% \pp \gets \Setup(1^\secparam),\\
% (\advt, \state) \gets \AdvertisementGen(\pp, X, x),\\
% (\sk, \vk) \gets \KGen(1^\secparam),\\
% \widetilde{\sigma} \gets \FPreSign(\advt, \sk, m, X, f),\\
% \sigma := \Adapt(\advt, \state, \vk, m, X, x, f, \widetilde{\sigma}), \\
% z := \FExt(\advt, \widetilde{\sigma}, \sigma, X, f)
% \end{array}
\right] = 1.
\]
% where $\pp \gets \Setup(1^\secparam)$,
% $(\advt, \state) \gets \AdvertisementGen(\pp, X, x)$,
% $(\sk, \vk) \gets \KGen(1^\secparam)$,
% $\widetilde{\sigma} \gets \FPreSign(\advt, \sk, m, X, f)$,
% $\sigma := \Adapt(\advt, \state, \vk, m, X, x, f, \widetilde{\sigma})$, 
% $w := \FExt(\advt, \sigma, \widetilde{\sigma}, X, f)$,
% and probability is taken over the random choices of $\Setup, \AdvertisementGen, \KGen, \FPreSign$ algorithms.

% \nikhil{challenger can't efficiently compute $f^{-1}$. In fact $f^{-1}$ may not even exist for many-to-one functions, for example, inner products.}
% \pnote{this equation is going to become unreadable in the submission format. How about not specifying the random choices in the equation, but adding a sentence after the equation -- "where the probability is taken over the choice of blah blah blha"}

\end{definition}

In terms of security, we want the following properties against malicous seller and buyer:
\begin{itemize}
\item {\bf Malicious seller:}
advertisement soundness, 
pre-signature validity, 
unforgeability, 
witness extractability. 
\item {\bf Malicious buyer:}
pre-signature adaptability, 
witness privacy. 
\end{itemize}
Looking ahead, we will consider three notions of witness privacy: zero-knowledge, witness indistinguishability, and witness hiding. Among these, zero-knowledge is the strongest.
Based on this, we define two overall security levels of functional adaptor signatures below.

\begin{definition}[Strongly-secure Functional Adaptor Signature Scheme]
A $\FAS$ scheme is \emph{strongly-secure} if it is 
advertisement sound, pre-signature valid, unforgeable, witness extractable, pre-signature adaptable and zero-knowledge. 
% \anote{Let's use mathsf for experiment names.}
\label{def:fas-strongly-secure}
\end{definition}

\begin{definition}[Weakly-secure Functional Adaptor Signature Scheme]
A $\FAS$ scheme is \emph{weakly-secure} if it is 
advertisement sound, pre-signature valid, unforgeable, witness extractable, pre-signature adaptable and witness indistinguishable.
\label{def:fas-weakly-secure}
\end{definition}
We describe all the security properties next.

\smallskip\noindent\textbf{Advertisement Soundness.} The property requires that given $\pp$, it should be infeasible for a \ppt malicious seller to find statement $X \notin L_R$ and advertisement $\advt$ such that $\advt$ is accepted by $\AdvertisementVerify$.
% \pnote{is this required? or can we live with a "computational" version?}
% \nikhil{yes, TODO: change}

\begin{definition}[Advertisement Soundness]
\label{def:fas-ad-sound}
A functional adaptor signature scheme $\FAS$ satisfies advertisement soundness if 
for every \ppt adversary \A, there exists a negligible function $\negl$ such that for all 
$\secparam \in \N$, for all NP languages $L_R$, $\pp \gets \Setup(1^\secparam)$, $(X, \advt) \gets \A(\pp)$,
% \nikhil{X is fixed before $crs$ making it non-adaptive soundness. Change quantifiers to get adaptive soundness. Also change to computational version.}
\[
\Pr[X \notin L_R \ \wedge \ \AdvertisementVerify(\pp, X, \advt) = 1 ] \leq \negl(\secparam).
\]
\end{definition}

\smallskip\noindent\textbf{Pre-signature validity.}
This property says that given a valid auxiliary value, one can always compute a valid pre-signature using it. Intuitively, this tries to capture that it should be impossible for a malicious seller to give out a malformed auxiliary value for the function $f$ and yet learn an honestly generated pre-signature for $f$ that is valid.

\begin{definition}[Pre-signature Validity]
A functional adaptor signature scheme $\fas$ satisfies pre-signature validity if for 
any $\secparam \in \N$, 
any $\pp \gets \Setup(1^\secparam)$, 
any $X, \advt$ such that $\AdvertisementVerify(\pp, \advt, $ $X)=1$,
any message $m\in \{0,1\}^*$,
any function class $\mcal{F}$,
any function $f \in \mcal{F}$,
any $(\aux_f, \pi_f)$, 
any key pair $(\sk, \vk) \gets \KGen(1^\secparam)$,
any pre-signature $\widetilde{\sigma} \gets \FPreSign(\advt, \sk, m, X, f, \aux_f)$, 
we have that if $\AuxVerify(\advt, f, \aux_f, \pi_f)=1$, then, 
\[
\Pr[\FPreVerify(\advt, \vk, m, X, f, \aux_f, \pi_f, \widetilde{\sigma})=1]=1.
\]
\end{definition}

\smallskip\noindent\textbf{Unforgeability.} 
This is similar to the unforgeability security property typically defined for adaptor signatures~\cite{asig}.
The property requires that even when the adversary (malicious seller) is given access to pre-signatures with respect to functions $f$ of its choice, it should not be able to forge signatures if it does not know the witness $x^*$ for the challenge statement $X^*$ in the first place. To capture this essence, the experiment (See~\Cref{fig:fasig-unf-exp}) samples $(X^*, x^*)$ at random. Beyond this, the adversary has all the powers of a (malicious) seller starting with generating $\advt$. A curious reader might ask if the adversary doesn't know $x^*$, then, wouldn't it typically fail to pass the $\AdvertisementVerify$ check in step~\ref{step:adverify} in~\Cref{fig:fasig-unf-exp} and thus rendering the experiment after that pointless? The answer is no, because \emph{advertisement soundness} is for statements not in the language $L_R$, but here $X^* \in L_R$ as the experiment chooses it.

% \anote{We do not seem to explain the subtlety about the $\advt$ verifying. }

\begin{definition}[Unforgeability]
A $\FAS$ scheme is $\faSigForge$-secure or simply unforgeable, if 
for every stateful \ppt adversary \A there exists a negligible function $\negl$ such that
for all $\secparam \in \N$,
$\Pr[ \faSigForge_{\A,\FAS}(1^\secparam) = 1] \leq \negl(\secparam)$,
where the experiment $\faSigForge_{\A,\FAS}$ 
is defined as in~\Cref{fig:fasig-unf-exp}.
% \nikhil{change $\faSigForge$ to something more readable}
% \nikhil{discuss why we use the weaker pre-sign oracle and not the stronger one proposed in the indocrypt'22 paper. }
% \anote{Something is wrong in the experiment. The oracle $\fpS$ will return $\bot$ always because $\advt = \bot$. It only becomes non-bot after step 6. That means the oracle access in step 6 is useless. Also, the adversary is required to output $\advt$ without knowing $x^*$. Isn't this already some form of a security violation? }
\end{definition}
\ifacm
	\begin{figure}[t]
	\ifhldiff
	{\color{hldiffcolor}
	\fi
	\centering
	\captionsetup{justification=centering}
	% \fbox{
	\begin{pcvstack}[boxed,space=1em]
	\procedure[linenumbering, mode=text]{Experiment $\faSigForge_{\A,\FAS}(1^\secparam)$}{
	$\mcal{Q} := \emptyset$, 
	% \\ 
	$\pp \gets \Setup(1^\secparam)$, 
	% \\ 
	$(\sk, \vk) \gets \KGen(1^\secparam)$ 
	\\ 
	$(X^*, x^*) \gets \GenR(1^\secparam)$
	\\ 
	% $\advt = \bot$ 
	% \\ 
	% $(\advt, m^*, f^*) \gets \A^{\mcal{O}_S(\cdot), \mcal{O}_{\fpS}(\cdot, \cdot, \cdot)}(\pp, \vk, X^*)$ 
	$(\advt, m^*, f^*, \aux_f^*, \pi_f^*) \gets \A^{\mcal{O}_S(\cdot)}(\pp, \vk, X^*) $ \\
	% $(\advt, m^*, f^*, \aux_f^*, \pi_f^*)  $ \pcskipln \\
	% $\qquad \qquad  \gets \A^{\mcal{O}_S(\cdot)}(\pp, \vk, X^*)$ 
	% \\ 
	\label{step:adverify}
	If $\AdvertisementVerify(\pp, X^*, \advt) =0 \ \vee \ f^* \notin \mcal{F} $ 
	\pcskipln \\
	$\vee \ \AuxVerify(\advt, f^*, \aux_f^*, \pi_f^*)=0$ : 
	\ret $0$ 
	\\ 
	$\widetilde{\sigma}^* \gets \FPreSign(\advt, \sk, m^*, X^*, f^*, \aux_f^*)$ 
	\\ 
	$\sigma^* \gets \A^{\mcal{O}_S(\cdot), \mcal{O}_{\fpS}(\cdot, \cdot, \cdot, \cdot, \cdot)}(\widetilde{\sigma}^*)$ \\
	\ret $((m^* \notin \mcal{Q}) \wedge \Verify(\vk, m^*, \sigma^*))$
	}

	\procedure[linenumbering, mode=text]{Oracle $\mcal{O}_S(m)$}{
	$\sigma \gets \Sign(\sk, m)$,
	% \\
	$\mcal{Q} := \mcal{Q} \vee \{m\}$, 
	% \\ 
	\ret $ \sigma$
	}

	\procedure[linenumbering, mode=text]{Oracle $\mcal{O}_{\fpS}(m, X, f, \aux_f, \pi_f)$}{
	If $\AuxVerify(\advt, f, \aux_f, \pi_f) = 0$: \ret $\bot$ \\
	$\widetilde{\sigma} \gets \FPreSign(\advt, \sk, m, X, f, \aux_f)$, 
	% \\ 
	$\mcal{Q} := \mcal{Q} \vee \{m\}$, 
	% \\ 
	\ret $\widetilde{\sigma}$
	}

	\end{pcvstack}
	% }
	\caption{
	\ifhldiff
	{\color{hldiffcolor}
	\fi
	Unforgeability of functional adaptor signatures 
	\ifhldiff
	}
	\fi
	% \pnote{please see my comment on slack on the kind of access A is given to Pre-sign.}
	% \nikhil{done}
	}
	\label{fig:fasig-unf-exp}
	\ifhldiff
	}
	\fi
	\end{figure}
\else
	\begin{figure}[H]
	\ifhldiff
	{\color{hldiffcolor}
	\fi
	\centering
	\captionsetup{justification=centering}
	% \fbox{
	\begin{pchstack}[boxed,space=1em]
	\begin{pcvstack}
	\procedure[linenumbering, mode=text]{Experiment $\faSigForge_{\A,\FAS}(1^\secparam)$}{
	$\mcal{Q} := \emptyset$, 
	% \\ 
	$\pp \gets \Setup(1^\secparam)$, 
	% \\ 
	$(\sk, \vk) \gets \KGen(1^\secparam)$ 
	\\ 
	$(X^*, x^*) \gets \GenR(1^\secparam)$
	\\ 
	% $\advt = \bot$ 
	% \\ 
	% $(\advt, m^*, f^*) \gets \A^{\mcal{O}_S(\cdot), \mcal{O}_{\fpS}(\cdot, \cdot, \cdot)}(\pp, \vk, X^*)$ 
	$(\advt, m^*, f^*, \aux_f^*, \pi_f^*) \gets \A^{\mcal{O}_S(\cdot)}(\pp, \vk, X^*) $ \\
	% $(\advt, m^*, f^*, \aux_f^*, \pi_f^*)  $ \pcskipln \\
	% $\qquad \qquad  \gets \A^{\mcal{O}_S(\cdot)}(\pp, \vk, X^*)$ 
	% \\ 
	\label{step:adverify}
	If $\AdvertisementVerify(\pp, X^*, \advt) =0 \ \vee \ f^* \notin \mcal{F} $ 
	\pcskipln \\
	$\vee \ \AuxVerify(\advt, f^*, \aux_f^*, \pi_f^*)=0$ : 
	\ret $0$ 
	\\ 
	$\widetilde{\sigma}^* \gets \FPreSign(\advt, \sk, m^*, X^*, f^*, \aux_f^*)$ 
	\\ 
	$\sigma^* \gets \A^{\mcal{O}_S(\cdot), \mcal{O}_{\fpS}(\cdot, \cdot, \cdot, \cdot, \cdot)}(\widetilde{\sigma}^*)$ \\
	\ret $((m^* \notin \mcal{Q}) \wedge \Verify(\vk, m^*, \sigma^*))$
	}
	\end{pcvstack}
	\begin{pcvstack}
	\procedure[linenumbering, mode=text]{Oracle $\mcal{O}_S(m)$}{
	$\sigma \gets \Sign(\sk, m)$
	\\
	$\mcal{Q} := \mcal{Q} \vee \{m\}$
	\\ 
	\ret $ \sigma$
	}

	\procedure[linenumbering, mode=text]{Oracle $\mcal{O}_{\fpS}(m, X, f, \aux_f, \pi_f)$}{
	If $\AuxVerify(\advt, f, \aux_f, \pi_f) = 0$: \ret $\bot$ \\
	$\widetilde{\sigma} \gets \FPreSign(\advt, \sk, m, X, f, \aux_f)$, 
	\\ 
	$\mcal{Q} := \mcal{Q} \vee \{m\}$
	\\ 
	\ret $\widetilde{\sigma}$
	}

	\end{pcvstack}
	\end{pchstack}
	% }
	\caption{
	\ifhldiff
	{\color{hldiffcolor}
	\fi
	Unforgeability of functional adaptor signatures 
	\ifhldiff
	}
	\fi
	% \pnote{please see my comment on slack on the kind of access A is given to Pre-sign.}
	% \nikhil{done}
	}
	\label{fig:fasig-unf-exp}
	\ifhldiff
	}
	\fi
	\end{figure}
\fi
% We also require that, given a pre-signature and a witness for the instance, 
% one can always adapt the pre-signature into a valid signature (pre-signature adaptability).
% \nikhil{comment on it being private and need for $\state$}
% Finally, we require that, given a valid pre-signature and a signature with respect to the same instance, one can efficiently {\color{blue}extract the function evaluation on the corresponding witness} (witness extractability).
% \nikhil{is the definition assuming unique-witness languages?}
\smallskip\noindent\textbf{Witness extractability.}
Here we require that for any function $f\in\mcal{F}$, for any honestly generated pre-signature w.r.t.\ $f$ and a valid signature, one should be able to extract the function evaluation $f$ on a witness. Intuitively, it tries to capture the interaction between an honest buyer (i.e., generates valid pre-signature) and a malicious seller trying to get paid (i.e., adapt pre-signature into a valid signature) without revealing function evaluation $f$ on a witness (i.e., extraction failure) should be unlikely to occur.
This is similar to the witness extractability security property of adaptor signatures (See~\Cref{def:as-wit-ext}) except that in the formal witness extractability experiment~\Cref{fig:fasig-wit-ext-exp}, in step~\ref{step:z-inefficient}, computing the set $Z$ is inefficient.

% \nikhil{is the definition assuming unique-witness languages?}

\begin{definition}[Witness Extractability]
A functional adaptor signature scheme $\FAS$ is $\faWitExt$-secure or simply witness extractable, if 
for every stateful \ppt adversary \A, there exists a negligible function $\negl$ such that 
for all $\secparam \in \N$,
\[
\Pr[ \faWitExt_{\A, \FAS}(1^\secparam) = 1] \leq \negl(\secparam),
\]
where the experiment $\faWitExt_{\A, \FAS}$ is defined as in~\Cref{fig:fasig-wit-ext-exp}.
\end{definition}

% \anote{The check with $\advt = \bot$. The adversary has to somehow set it before the oracle query.}
\ifacm
	\begin{figure}[t]
	\centering
	\captionsetup{justification=centering}
	% \fbox{
	\begin{pcvstack}[boxed, space=1em]
	\begin{pchstack}
	\procedure[linenumbering, mode=text]{Experiment $\faWitExt_{\A,\FAS}(1^\secparam)$}{
	$\mcal{Q} := \emptyset$ 
	\\ 
	$\pp \gets \Setup(1^\secparam)$ 
	\\ 
	$(\sk, \vk) \gets \KGen(1^\secparam)$ 
	\\ 
	% $\advt = \bot$ 
	% \\ 
	$(X^*, \advt, m^*, f^*, \aux_f^*, \pi_f^*)$ \pcskipln \\
	$\qquad \qquad \gets \A^{\mcal{O}_S(\cdot)}(\pp, \vk)$ 
	\\ 
	If $\AdvertisementVerify(\pp, X^*, \advt) =0 \ \vee \ f^* \notin \mcal{F}$
	\pcskipln \\
	$\vee \ \AuxVerify(\advt, f^*, \aux_f^*, \pi_f^*)=0$: 
	\pcskipln \\
	$\pcind$ \ret $0$ 
	\\ 
	$\widetilde{\sigma}^* \gets \FPreSign(\advt, \sk, m^*, X^*, f^*, \aux_f^*)$ 
	\\ 
	$\sigma^* \gets \A^{\mcal{O}_S(\cdot), \mcal{O}_{\fpS}(\cdot, \cdot, \cdot, \cdot, \cdot)}(\widetilde{\sigma}^*)$ 
	\\
	$z := \FExt(\advt, \widetilde{\sigma}^*, \sigma^*, X^*, f^*, \aux_f^*)$ \\
	\label{step:z-inefficient}
	Let $Z = \{ f^*(x) : \exists \ x\ s.t.\ (X^*, x) \in R\}$ \\
	% \nikhil{computing set $W$ requires to find all witnesses $x$ for the statement $X$ which is going to be inefficient.}
	$b := ((m^* \notin \mcal{Q}) \wedge \Verify(\vk, m^*, \sigma^*) \wedge (z \notin Z))$ \\
	\ret $b$
	}

	\procedure[linenumbering, mode=text]{Oracle $\mcal{O}_S(m)$}{
	$\sigma \gets \Sign(\sk, m)$
	\\
	$\mcal{Q} := \mcal{Q} \vee \{m\}$ 
	\\ 
	\ret $\sigma$
	}

	\end{pchstack}

	\procedure[linenumbering, mode=text]{Oracle $\mcal{O}_{\fpS}(m, X, f, \aux_f, \pi_f)$}{
	% If $\advt = \bot$: return $\bot$ 
	% \\ 
	If $\AuxVerify(\advt, f, \aux_f, \pi_f) = 0$: \ret $\bot$
	\\
	$\widetilde{\sigma} \gets \FPreSign(\advt, \sk, m, X, f)$ 
	\\ 
	$\mcal{Q} := \mcal{Q} \vee \{m\}$ 
	\\ 
	\ret $\widetilde{\sigma}$
	}
	\end{pcvstack}
	% }
	\caption{Witness Extractability of functional adaptor signatures}
	\label{fig:fasig-wit-ext-exp}

	\end{figure}
\else
	\begin{figure}[H]
	\centering
	\captionsetup{justification=centering}
	% \fbox{
	\begin{pchstack}[boxed, space=1em]
	\begin{pcvstack}
	\procedure[linenumbering, mode=text]{Experiment $\faWitExt_{\A,\FAS}(1^\secparam)$}{
	$\mcal{Q} := \emptyset$, 
	% \\ 
	$\pp \gets \Setup(1^\secparam)$, 
	% \\ 
	$(\sk, \vk) \gets \KGen(1^\secparam)$ 
	\\ 
	% $\advt = \bot$ 
	% \\ 
	$(X^*, \advt, m^*, f^*, \aux_f^*, \pi_f^*) \gets \A^{\mcal{O}_S(\cdot)}(\pp, \vk)$ 
	\\ 
	If $\AdvertisementVerify(\pp, X^*, \advt) =0 \ \vee \ f^* \notin \mcal{F}$
	\pcskipln \\
	$\vee \ \AuxVerify(\advt, f^*, \aux_f^*, \pi_f^*)=0$: \ret $0$ 
	\\ 
	$\widetilde{\sigma}^* \gets \FPreSign(\advt, \sk, m^*, X^*, f^*, \aux_f^*)$ 
	\\ 
	$\sigma^* \gets \A^{\mcal{O}_S(\cdot), \mcal{O}_{\fpS}(\cdot, \cdot, \cdot, \cdot, \cdot)}(\widetilde{\sigma}^*)$ 
	\\
	$z := \FExt(\advt, \widetilde{\sigma}^*, \sigma^*, X^*, f^*, \aux_f^*)$ \\
	\label{step:z-inefficient}
	Let $Z = \{ f^*(x) : \exists \ x\ s.t.\ (X^*, x) \in R\}$ \\
	% \nikhil{computing set $W$ requires to find all witnesses $x$ for the statement $X$ which is going to be inefficient.}
	\ret $((m^* \notin \mcal{Q}) \wedge \Verify(\vk, m^*, \sigma^*) \wedge (z \notin Z))$
	}

	\end{pcvstack}
	\begin{pcvstack}

	\procedure[linenumbering, mode=text]{Oracle $\mcal{O}_S(m)$}{
	$\sigma \gets \Sign(\sk, m)$
	\\
	$\mcal{Q} := \mcal{Q} \vee \{m\}$ 
	\\ 
	\ret $\sigma$
	}

	\procedure[linenumbering, mode=text]{Oracle $\mcal{O}_{\fpS}(m, X, f, \aux_f, \pi_f)$}{
	% If $\advt = \bot$: return $\bot$ 
	% \\ 
	If $\AuxVerify(\advt, f, \aux_f, \pi_f) = 0$: \ret $\bot$
	\\
	$\widetilde{\sigma} \gets \FPreSign(\advt, \sk, m, X, f)$ 
	\\ 
	$\mcal{Q} := \mcal{Q} \vee \{m\}$ 
	\\ 
	\ret $\widetilde{\sigma}$
	}
	\end{pcvstack}
	\end{pchstack}
	% }
	\caption{Witness Extractability of functional adaptor signatures}
	\label{fig:fasig-wit-ext-exp}

	\end{figure}
\fi
\smallskip\noindent\textbf{Pre-signature adaptability.} 
This property says that given a valid pre-signature and a witness for the instance, 
one can always adapt the pre-signature into a valid signature.
Intuitively, this tries to capture that it should be impossible for a malicious buyer to learn a function evaluation on a witness without making the payment. 
This is similar to the pre-signature adaptability security property of adaptor signatures (See~\Cref{def:as-pre-sig-adaptability}).
% \nikhil{comment on it being private and need for $\state$}

\begin{definition}[Pre-signature Adaptability]
\label{def:fas-pre-sig-adaptability}
A functional adaptor signature scheme $\FAS$ satisfies pre-signature adaptability 
if for any $\secparam \in \N$,
any $\pp \gets \Setup(1^\secparam)$, 
any statement/witness pair $(X, x) \in R$, 
any $(\advt, \state) \gets \AdvertisementGen(\pp, X, x)$,
any message $m \in \{0,1\}^*$, 
any function class $\mcal{F}$, 
any function $f \in \mcal{F}$,
any $(\aux_f, \pi_f) $ $\gets \AuxGen(\advt, \state, f)$,
any key pair $(\sk, \vk) \gets \KGen(1^\secparam)$, 
and any pre-signature $\widetilde{\sigma} \in \{0,1\}^*$ 
with $\FPreVerify(\advt, $ $\vk, m, X, f, \aux_f, \pi_f,\allowbreak \widetilde{\sigma}) = 1$, we have:
\[
\Pr[\Verify(\vk, m, \Adapt(\advt, \state, \vk, m, X, x, f, \aux_f, \widetilde{\sigma})) = 1] = 1.
\]
\end{definition}

\smallskip\noindent\textbf{Witness privacy.} 
Here we want that in an interaction between a malicious buyer wanting to learn function evaluation $f$ on the witness and an honest seller using witness $x$ to adapt the pre-signature into a valid signature, at the end the buyer should not be able to extract any information about $x$ beyond $f(x)$. There are different ways to capture this formally. 
We consider various notions of witness privacy akin to those in functional encryption literature and zero-knowledge proofs literature as follows.
\begin{itemize}[leftmargin=*]
\item {\bf Zero-Knowledge:}
for every $(X, x) \in R$, there exists a \ppt simulator $\Sim$ such that the adversary should not be able to guess whether it is interacting with a real-world challenger that is allowed to use the witness $x$ or the simulator $\Sim$ that is only allowed access to $f(x)$. 
\item {\bf Witness Indistinguishability:}
given two valid witnesses $x_0$ and $x_1$ for a statement $X \in L_R$, the adversary should not be able to guess which witness was used to adapt a pre-signature as long as $f(x_0) = f(x_1)$. This notion is meaningful only for languages that have at least two witnesses. 
\item {\bf Witness Hiding:}
for every $X \in L_R$, given valid pre-signature w.r.t.\ function $f$ and valid signature adapted from it, the adversary should not be able to guess a witness $x$ s.t. $(X, x) \in R$.
\end{itemize}

Zero-Knowledge is the strongest among all in the sense that it implies the other two. All three are meaningful for non-unique witness languages. But for unique witness languages, witness indistinguishability is not meaningful. 
Further, witness-indistinguishability and witness-hiding are incomparable.
Next, we describe all three formally.

\begin{definition}[Zero-Knowledge]
\label{def:fas-zk}
% \pnote{this terminology makes no sense to me, perhaps just use ZK.}\nikhil{done}
A $\FAS$ scheme satisfies zero-knowledge if 
for every \ppt adversary \A, 
there exists a stateful \ppt simulator $\Sim = (\Setup^*, \AdGen^*, \AuxGen^*, \Adapt^*)$,
there exists a negligible function $\negl$ s.t.
for all \ppt distinguishers \D, 
for all $\secparam \in \N$, 
for all $(X, x) \in R$,
$| \Pr[\D(\faZKReal_{\A, \FAS}(1^\secparam, X, x)) = 1] - \Pr[ \D(\faZKIdeal_{\A, \FAS}^\Sim(1^\secparam, X, x)) = 1] | \leq  \negl(\secparam)$,
% \begin{align*}
% & | \Pr[\D(\faZKReal_{\A, \FAS}(1^\secparam, X, x)) = 1] \\ 
% & \quad - \Pr[ \D(\faZKIdeal_{\A, \FAS}^\Sim(1^\secparam, X, x)) = 1] | \leq  \negl(\secparam),
% \end{align*}
% \[
% | \Pr[ \faZKReal_{\A, \FAS}(1^\secparam) = 1] - \Pr[ \faZKIdeal_{\A, \FAS}^\Sim(1^\secparam) = 1] | \leq  \negl(\secparam),
% \]
where experiments $\faZKReal_{\A, \FAS}$ and $\faZKIdeal^\Sim_{\A, \FAS}$ are defined in~\Cref{fig:fas-zk-real,fig:fas-zk-ideal}.
% \pnote{I think we should standardize how we write advantages for decisional games: here it written as a difference of probabilities of $A$ whereas previously it the $1/2$ - negl version.} 
% \nikhil{done}
% \pnote{I think this definition needs to be refined to make it look similar to ZK -- e.g., for ZK there are two adversaries: the verifier who participates in the experiment, and the distinguisher who is given the \emph{view} (all messages including the output) of the verifier from the experiment.}
% \nikhil{done}

\end{definition}
\ifacm
	\begin{figure}[t]
	\centering
	\captionsetup{justification=centering}
	\begin{pcvstack}[boxed, space=1em]
	\procedure[linenumbering, mode=text]{Experiment $\faZKReal_{\A,\FAS}(1^\secparam, X, x)$; $\pcbox{\faZKIdeal^\Sim_{\A,\FAS}(1^\secparam, X, x)}$}{
	$\pp \gets \Setup(1^\secparam)$ $\pcbox{\Setup^*(1^\secparam)}$ 
	\\ 
	$(\advt, \state) \gets \AdvertisementGen(\pp, X, x)$ $\pcbox{\AdGen^*(\pp, X)}$ 
	\\
	% ;\   \pcbox{\AdvertisementGen^*(\pp, X)}$ 
	% \\ 
	$\vk := \bot$
	\\ 
	$\vk \gets \A^{\mcal{O}_{\AuxGen}(\cdot), \mcal{O}_{\Adapt}(\cdot, \cdot, \cdot)}(\pp, \advt, X)$;  \pcskipln \\
	$\pcbox{\vk \gets \A^{\mcal{O}^*_{\AuxGen}(\cdot, \cdot), \mcal{O}^*_{\Adapt}(\cdot, \cdot, \cdot, \cdot)}(\pp, \advt, X)}$
	\\ 
	\ret view of $\A$
	}

	\procedure[linenumbering, mode=text]{Oracle $\mcal{O}_{\AuxGen}(f);\ \pcbox{\mcal{O}^*_{\AuxGen}(f, f(x))}$.}{
	\ret $(\aux_f, \pi_f) \gets \AuxGen(\advt, \state, f);\ \pcbox{\AuxGen^*(\advt, f, f(x))}$  
	}

	\procedure[linenumbering, mode=text]{Oracle $\mcal{O}_{\Adapt}(m, f, \widetilde{\sigma});\ \pcbox{\mcal{O}^*_{\Adapt}(m, f, \widetilde{\sigma}, f(x))}$.}{
	If $\vk = \bot$: \ret $\bot$
	\\ 
	$(\aux_f, \pi_f) \gets \AuxGen(\advt, \state, f);\ \pcbox{\AuxGen^*(\advt, f, f(x))}$ 
	\\ 
	If $\FPreVerify(\advt, \vk, m, X, f, \aux_f, \pi_f, \widetilde{\sigma}) = 0$: \ret $\bot$
	\\ 
	\ret $\sigma := \Adapt(\advt, \state, \vk, m, X, x, f, \aux_f, \widetilde{\sigma})$  \pcskipln \\
	$\pcbox{\ret \ \sigma :=  \Adapt^*(\advt, \vk, m, X, f, \aux_f, \widetilde{\sigma}, f(x))}$
	}

	\end{pcvstack}

	\caption{Zero-Knowledge experiments of FAS
	% , with real-world experiment in plain while the changes for the ideal-world experiment are in a $\pcbox{\text{box}}$.
	}
	\label{fig:fas-zk-real}
	\label{fig:fas-zk-ideal}
	\end{figure}
\else
	\begin{figure}[t]
	\centering
	\captionsetup{justification=centering}
	\begin{pchstack}[boxed, space=0.1em]
	\begin{pcvstack}
	\procedure[linenumbering, mode=text]{Experiment $\faZKReal_{\A,\FAS}(1^\secparam, X, x)$}{
	$\pp \gets \Setup(1^\secparam)$
	\\ 
	$(\advt, \state) \gets \AdvertisementGen(\pp, X, x)$  
	\\
	% ;\   \pcbox{\AdvertisementGen^*(\pp, X)}$ 
	% \\ 
	$\vk := \bot$
	\\ 
	$\vk \gets \A^{\mcal{O}_{\AuxGen}(\cdot), \mcal{O}_{\Adapt}(\cdot, \cdot, \cdot)}(\pp, \advt, X)$
	\\ 
	\ret view of $\A$
	}

	\procedure[linenumbering, mode=text]{Oracle $\mcal{O}_{\AuxGen}(f)$}{
	\ret $(\aux_f, \pi_f) \gets \AuxGen(\advt, \state, f)$  
	}

	\procedure[linenumbering, mode=text]{Oracle $\mcal{O}_{\Adapt}(m, f, \widetilde{\sigma})$}{
	If $\vk = \bot$: \ret $\bot$
	\\ 
	$(\aux_f, \pi_f) \gets \AuxGen(\advt, \state, f)$ 
	\\ 
	If $\FPreVerify(\advt, \vk, m, X, f, \aux_f, \pi_f, \widetilde{\sigma}) = 0$: \pcskipln \\ 
	$\pcind$ \ret $\bot$
	\\ 
	\ret $\sigma := \Adapt(\advt, \state, \vk, m, X, x, f, \aux_f, \widetilde{\sigma})$
	}

	\end{pcvstack}
	\begin{pcvstack}
	\procedure[linenumbering, mode=text]{Experiment $\faZKIdeal^\Sim_{\A,\FAS}(1^\secparam, X, x)$}{
	$\pp \gets \Setup^*(1^\secparam)$
	\\ 
	$\advt \gets \AdvertisementGen^*(\pp, X)$ 
	\\
	% ;\   \pcbox{\AdvertisementGen^*(\pp, X)}$ 
	% \\ 
	$\vk := \bot$
	\\ 
	$\vk \gets \A^{\mcal{O}^*_{\AuxGen}(\cdot, \cdot), \mcal{O}^*_{\Adapt}(\cdot, \cdot, \cdot, \cdot)}(\pp, \advt, X)$
	\\ 
	\ret view of $\A$
	}

	\procedure[linenumbering, mode=text]{Oracle $\mcal{O}^*_{\AuxGen}(f, f(x))$}{
	\ret $(\aux_f, \pi_f) \gets \AuxGen^*(\advt, f, f(x))$  
	}

	\procedure[linenumbering, mode=text]{Oracle $\mcal{O}^*_{\Adapt}(m, f, \widetilde{\sigma}, f(x))$}{
	If $\vk = \bot$: \ret $\bot$
	\\ 
	$(\aux_f, \pi_f) \gets \AuxGen^*(\advt, f, f(x))$ 
	\\ 
	If $\FPreVerify(\advt, \vk, m, X, f, \aux_f, \pi_f, \widetilde{\sigma}) = 0$: \pcskipln \\ 
	$\pcind$ \ret $\bot$
	\\ 
	\ret $\sigma :=  \Adapt^*(\advt, \vk, m, X, f, \aux_f, \widetilde{\sigma}, f(x))$
	}

	\end{pcvstack}
	\end{pchstack}

	\caption{Zero-Knowledge experiments of FAS
	% , with real-world experiment in plain while the changes for the ideal-world experiment are in a $\pcbox{\text{box}}$.
	}
	\label{fig:fas-zk-real}
	\label{fig:fas-zk-ideal}
	\end{figure}
\fi

\begin{definition}[Witness Indistinguishability]
A functional adaptor signature scheme $\FAS$ is $\faWI$-secure or simply witness indistinguishable, if 
for every stateful \ppt adversary \A,
% \pnote{since we are not passing state to $\A$'s in different lines of the game, we should say that $\A$ we consider are \emph{stateful}} 
% \nikhil{done}
there exists a negligible function $\negl$ such that for all $\secparam \in \N$,
\[
|
\Pr[ \faWI^0_{\A, \FAS}(1^\secparam) = 1] 
-
\Pr[ \faWI^1_{\A, \FAS}(1^\secparam) = 1] 
| \leq \negl(\secparam),
\]
where the experiments $\faWI^b_{\A, \FAS}$ for $b \in \{0,1\}$ are defined as in~\Cref{fig:fasig-wit-ind-exp}.
\label{def:fas-wi}
\end{definition}

\ifacm
	\begin{figure}[t]
	\centering
	\captionsetup{justification=centering}
	\begin{pcvstack}[boxed, space=1em]
	\procedure[linenumbering, mode=text]{Experiment $\faWI^b_{\A,\FAS}$}{
	$\crs \gets \Setup(1^\secparam)$ 
	\\ 
	$(\vk, X, x_0, x_1) \gets \A(\crs)$
	% \pnote{why does $\A$ choose $\vk$ here? Is this necessary for application? If not, then perhaps we can have a simpler definition where all the randomness (including pp, state) are sampled before we send things to the adversary.}
	% \nikhil{Intuitively, in witness-privacy we are trying to protect against adversarial buyer. The buyer controls $(\sk, \vk)$, hence, I allowed \A to choose it. In WI, we can additionally allow \A to choose statement $X$ along with two witnesses $x_0$ and $x_1$ and then we want to argue that \A can't learn which one of the two witness is used. Regarding $\advt$, it is output by the $\AdvertisementGen$ algorithm which depends on statement-witness pair, hence, the challenger must sample it after seeing $X, x_0, x_1$.}
	\\ 
	$\text{if } (((X, x_0) \notin R) \vee ((X, x_1) \notin R)): \text{\ret } 0$  
	\\ 
	$(\advt, \state) \gets \AdvertisementGen(\crs, X, x_b)$ 
	% \pnote{I don't think the winning condition follows the right convention:
	% That is, if the winning condition is $b = b'$, then typically we just define a single experiement and smaple $b$ inside it. 
	% If we define two experiments (like you do
	% ), we care about the probability that A outputs $1$. 
	% Maybe what you wrote is equivalent, but since its not the traditional way to write it, I may appear confusing.}
	% \nikhil{resolved}
	\\ 
	$(m, f, \aux_f, \pi_f, \widetilde{\sigma}) \gets \A^{\mcal{O}_\AuxGen (\cdot)}(\advt)$ 
	\\ 
	If $\FPreVerify(\advt, \vk, m, X, f, \aux_f, \pi_f, \widetilde{\sigma}) = 0$: \ret $0$
	\\ 
	If $f(x_0) \neq f(x_1)$: \ret $0$
	\\ 
	$\sigma := \Adapt(\advt, \state, \vk, m, X, x_b, f, \aux_f, \widetilde{\sigma})$ 
	\\ 
	$b' \gets \A(\sigma)$ 
	\\ 
	\ret $b'$ 
	% \nikhil{change to \ret $b'$}
	} 
	\procedure[linenumbering, mode=text]{Oracle $\mcal{O}_{\AuxGen}(f)$}{
	\ret $\AuxGen(\advt, \state, f)$
	}
	\end{pcvstack}
	\caption{Witness Indistinguishability experiment of adaptor signatures
	% \pnote{To be consistent with the literature, I would just call it WI and WH instead of WitInd and WitHid.}
	% \nikhil{done}
	}
	\label{fig:fasig-wit-ind-exp}
	\end{figure}
\else 
	\begin{figure}[H]
	\centering
	\captionsetup{justification=centering}
	\begin{pchstack}[boxed, space=1em]
	\procedure[linenumbering, mode=text]{Experiment $\faWI^b_{\A,\FAS}$}{
	$\crs \gets \Setup(1^\secparam)$ 
	\\ 
	$(\vk, X, x_0, x_1) \gets \A(\crs)$
	% \pnote{why does $\A$ choose $\vk$ here? Is this necessary for application? If not, then perhaps we can have a simpler definition where all the randomness (including pp, state) are sampled before we send things to the adversary.}
	% \nikhil{Intuitively, in witness-privacy we are trying to protect against adversarial buyer. The buyer controls $(\sk, \vk)$, hence, I allowed \A to choose it. In WI, we can additionally allow \A to choose statement $X$ along with two witnesses $x_0$ and $x_1$ and then we want to argue that \A can't learn which one of the two witness is used. Regarding $\advt$, it is output by the $\AdvertisementGen$ algorithm which depends on statement-witness pair, hence, the challenger must sample it after seeing $X, x_0, x_1$.}
	\\ 
	If $(((X, x_0) \notin R) \vee ((X, x_1) \notin R))$: \ret $0$  
	\\ 
	$(\advt, \state) \gets \AdvertisementGen(\crs, X, x_b)$ 
	% \pnote{I don't think the winning condition follows the right convention:
	% That is, if the winning condition is $b = b'$, then typically we just define a single experiement and smaple $b$ inside it. 
	% If we define two experiments (like you do
	% ), we care about the probability that A outputs $1$. 
	% Maybe what you wrote is equivalent, but since its not the traditional way to write it, I may appear confusing.}
	% \nikhil{resolved}
	\\ 
	$(m, f, \aux_f, \pi_f, \widetilde{\sigma}) \gets \A^{\mcal{O}_\AuxGen (\cdot)}(\advt)$ 
	\\ 
	If $\FPreVerify(\advt, \vk, m, X, f, \aux_f, \pi_f, \widetilde{\sigma}) = 0$: \ret $0$
	\\ 
	If $f(x_0) \neq f(x_1)$: \ret $0$
	\\ 
	$\sigma := \Adapt(\advt, \state, \vk, m, X, x_b, f, \aux_f, \widetilde{\sigma})$ 
	\\ 
	$b' \gets \A(\sigma)$ 
	\\ 
	\ret $b'$ 
	% \nikhil{change to \ret $b'$}
	} 
	\procedure[linenumbering, mode=text]{Oracle $\mcal{O}_{\AuxGen}(f)$}{
	\ret $\AuxGen(\advt, \state, f)$
	}
	\end{pchstack}
	\caption{Witness Indistinguishability experiment of adaptor signatures
	% \pnote{To be consistent with the literature, I would just call it WI and WH instead of WitInd and WitHid.}
	% \nikhil{done}
	}
	\label{fig:fasig-wit-ind-exp}
	\end{figure}
\fi
% Combining witness indistinguishability with other security properties, we can define a weakly-secure functional adaptor signature scheme as follows.

% \begin{definition}[Weakly-secure Functional Adaptor Signature Scheme]
% A functional adaptor signature scheme $\FAS$ is \emph{weakly-secure} if it is 
% advertisement sound, $\faSigForge$-secure, pre-signature adaptable, $\faWitExt$-secure and $\faWI$-secure.
% \label{def:fas-weakly-secure}
% \end{definition}

\begin{definition}[Witness Hiding]
A functional adaptor signature scheme $\FAS$ is $\faWH$-secure or simply witness hiding, if 
for every stateful \ppt adversary \A, there exists a negligible function $\negl$ such that 
for all $\secparam \in \N$,
\[
\Pr[ \faWH_{\A, \FAS}(1^\secparam) = 1] \leq \negl(\secparam),
\]
where 
% $|W_R|$ is the size of the witness space $W_R$ for the relation $R$ and, 
the experiment $\faWH_{\A, \FAS}$ is defined as in~\Cref{fig:fasig-wit-hid-exp}.
% \pnote{I don't understand the relevance of the 1/W term. We should define WH for hard-relations -- as done here: https://eprint.iacr.org/2017/330.pdf (def 6,7)}
% \nikhil{resolved}
\label{def:fas-wh}
\end{definition}
\ifacm
	\begin{figure}[t]
	\centering
	\captionsetup{justification=centering}
	\begin{pcvstack}[boxed, space=1em]
	\procedure[linenumbering, mode=text]{Experiment $\faWH_{\A,\FAS}$}{	
	$\crs \gets \Setup(1^\secparam)$ 
	\\ 
	$(X, x) \gets \GenR(1^\secparam)$
	\\ 
	$(\advt, \state) \gets \AdvertisementGen(\crs, X, x)$ 
	\\ 
	$(\vk, m, f, \aux_f, \pi_f, \widetilde{\sigma}) \gets \A^{\mcal{O}_\AuxGen(\cdot)}(\crs, \advt, X)$ 
	\\ 
	If $\FPreVerify(\advt, \vk, m, X, f, \aux_f, \pi_f, \widetilde{\sigma}) = 0$: \ret $0$
	\\ 
	$\sigma := \Adapt(\advt, \state, \vk, m, X, x, f, \aux_f, \widetilde{\sigma})$  
	\\ 
	$x' \gets \A(\sigma)$ 
	\\ 
	\ret $1$ iff $(X, x') \in R$ 
	% \pnote{This winning condition doesn't protect against A's that may guess some other witness for $X$. I think the winning condition should let $A$ win if it outputs \emph{a} witness, instead of the \emph{the} witness. }
	% \nikhil{done}
	} 
	\procedure[linenumbering, mode=text]{Oracle $\mcal{O}_{\AuxGen}(f)$}{
	\ret $\AuxGen(\advt, \state, f)$
	}
	\end{pcvstack}
	\caption{Witness Hiding experiment of adaptor signatures
	% \nikhil{WI may not trivially imply WH as in the WH game, the challenger chooses a random statement-witness pair (supposedly hard instance), but in the WI game, the adversary is allowed to choose statement-witness tuple (could be easy instance).}\pnote{WI and WH are incomparable notions.}
	}
	\label{fig:fasig-wit-hid-exp}
	\end{figure}
\else
	\begin{figure}[H]
	\centering
	\captionsetup{justification=centering}
	\begin{pchstack}[boxed, space=1em]
	\procedure[linenumbering, mode=text]{Experiment $\faWH_{\A,\FAS}$}{	
	$\crs \gets \Setup(1^\secparam)$, 
	% \\ 
	$(X, x) \gets \GenR(1^\secparam)$,
	% \\ 
	$(\advt, \state) \gets \AdvertisementGen(\crs, X, x)$ 
	\\ 
	$(\vk, m, f, \aux_f, \pi_f, \widetilde{\sigma}) \gets \A^{\mcal{O}_\AuxGen(\cdot)}(\crs, \advt, X)$ 
	\\ 
	If $\FPreVerify(\advt, \vk, m, X, f, \aux_f, \pi_f, \widetilde{\sigma}) = 0$: \ret $0$
	\\ 
	$\sigma := \Adapt(\advt, \state, \vk, m, X, x, f, \aux_f, \widetilde{\sigma})$  
	\\ 
	$x' \gets \A(\sigma)$ 
	\\ 
	\ret $1$ iff $(X, x') \in R$ 
	% \pnote{This winning condition doesn't protect against A's that may guess some other witness for $X$. I think the winning condition should let $A$ win if it outputs \emph{a} witness, instead of the \emph{the} witness. }
	% \nikhil{done}
	} 
	\procedure[linenumbering, mode=text]{Oracle $\mcal{O}_{\AuxGen}(f)$}{
	\ret $\AuxGen(\advt, \state, f)$
	}
	\end{pchstack}
	\caption{Witness Hiding experiment of adaptor signatures
	% \nikhil{WI may not trivially imply WH as in the WH game, the challenger chooses a random statement-witness pair (supposedly hard instance), but in the WI game, the adversary is allowed to choose statement-witness tuple (could be easy instance).}\pnote{WI and WH are incomparable notions.}
	}
	\label{fig:fasig-wit-hid-exp}
	\end{figure}
\fi
% \anote{Why not give oracle access to the adversary for different functions? In practice this could be he case right?}

%% file: fas-construction.tex
\section{Strongly-Secure FAS Construction}
\label{sec:fas-construction}
% \nikhil{change title}
% \nikhil{put $\AuxGen$, $\AuxVerify$ and update $\FPreSign$ and $\FPreVerify$ to take $\aux_f$ as inputs.}

In this section, we build a generic construction of FAS w.r.t.\ digital signature scheme $\DS$, any NP relation $R$ with statement/witness pairs $(X, \bfx) \in R$ such that $\bfx \in \mcal{M}$ for some set $\mcal{M} \subseteq \Z^\ell$, and function family $\mcal{F}_{{\sf IP}, \ell}$ for computing inner products of vectors of length $\ell$.
% \nikhil{the DL based construction is indeed for this function class. But, the lattice based construction is for $\mcal{F}_{IP, \Z^\ell, B_y} = \{\bfy \in \Z^\ell\  s.t. ||\bfy||_\infty \leq B_y\ \}$. make this clear.}
Our generic construction of FAS is presented in ~\Cref{fig:fas-construction}
\footnote{~\Cref{fig:fas-construction} is presented with $\mcal{M} = \Z_p^\ell$, but the scheme generalizes to $\mcal{M} \subseteq \Z^\ell$ as well.} 
and it uses the following building blocks:
\begin{itemize}[leftmargin=*]
\item 
A selective, IND-secure $\ipe$ satisfying 
$R_\ipfe$-compliance and 
\ifcameraready
$R'_\ipfe$-robustness (\Cref{def:ipfe-compliant,def:ipfe-robust})  
\else
$R'_\ipfe$-robustness (\Cref{def:ipfe-sel-ind-sec,def:ipfe-compliant,def:ipfe-robust})  
\fi
w.r.t.\ hard relations $R_\ipfe$ and $R'_\ipfe$ such that $R_\ipfe \subseteq R'_\ipfe$. The message space of $\ipfe$ is $\mcal{M}' \subseteq \Z^{\ell+1}$ and the function family is $\mcal{F}_{{\sf IP}, \ell+1}$.

\item 
An adaptor signature scheme $\as$ w.r.t.\ a digital signature scheme $\DS$, and hard relations $R_\ipfe$ and $R'_\ipfe$ that satisfies witness 
\ifcameraready
extractability 
\else
extractability (\Cref{def:as-wit-ext}) 
\fi
and weak pre-signature 
\ifcameraready
adaptability 
\else
adaptability (\Cref{def:as-pre-sig-adaptability}) 
\fi
security properties. 
\end{itemize}

\begin{itemize}[leftmargin=*]
\item 
A non-interactive zero-knowledge argument system $\nizk$ in the common reference string model for the NP language $L_\nizk$:
\[
L_\nizk := \left\{ 
\begin{matrix}
(X, \pp', \mpk, \ct) : \\
\exists (r_0, r_1, \bfx)\ \text{such that }
\pp' \in [\ipe.\Gen(1^\secparam)],\\
(\mpk, \msk) = \ipe.\Setup(\pp', 1^{\ell+1}; r_0),\\
(X, \bfx) \in R,
\ct = \ipe.\Enc(\mpk, (\bfx^T, 0)^T; r_1)
\end{matrix}
\right\}.
\]\
\end{itemize}
% Our construction is as in~\Cref{fig:fas-construction}. 

% \subsection{Construction}
% \label{sec:construction}

% \nikhil{continue here. add auxgen, change $L_\nizk$ witness in other files (done here)}
% \nikhil{remove GroupGen specification. keep it generic for the generic construction.}
\ifhldiff
{\color{hldiffcolor}
\fi
We sketch the security proof of unforgeability and zero-knowledge witness privacy below.
Correctness and formal proofs are deferred 
\ifcameraready
to the full version.
\else
to~\cref{sec:fas-proof-appendix}.
\fi
\ifhldiff
}
\fi

\begin{lemma}\label{lemma:correctness-fas}
Suppose NIZK satisfies correctness, $\AS$ satisfies correctness, and $\ipfe$ satisfies $R_\ipfe$-compliance and  $R'_\ipfe$-robustness. Then, the $\fas$ construction in~\Cref{sec:construction} satisfies correctness.
\end{lemma}

\input{fas-proof}

%% file: fas-proof.tex
% The security theorem is formally stated below. 

\begin{theorem}\label{thm:fas-strongly-secure}
Let $\mcal{F}_{{\sf IP}, \ell}$ be the family of inner products functions of vectors of length $\ell$. 
Let $R$ be any NP relation 
with statement/witness pairs $(X, \bfx)$ such that $\bfx \in \mcal{M}$ for some set $\mcal{M} \subseteq \Z^\ell$.
Suppose that 
% \anote{Refer to definitions}
\begin{itemize}[leftmargin=*]
\item 
\mcal{M} is an additive group,
$R$ is 
\ifcameraready
$\mcal{F}_{{\sf IP}, \ell}$-hard,
\else
$\mcal{F}_{{\sf IP}, \ell}$-hard (\Cref{def:f-hard-relation}),
\fi
\item
$\nizk$ is a secure NIZK argument 
\ifcameraready
system,
\else
system (\Cref{def:nizk}),
\fi
\item
$\as$ is an adaptor signature scheme w.r.t.\ digital signature scheme $\ds$ and hard relations $R_\ipfe, R'_\ipfe$ satisfying weak pre-signature 
\ifcameraready
adaptability, 
\else
adaptability (\Cref{def:as-pre-sig-adaptability}), 
\fi
witness 
\ifcameraready
extractability,
\else 
extractability (\Cref{def:as-wit-ext}),
\fi
% \anote{What is this notion? where are we defining this?} 
% \anote{refer to the definitions here.}
\item
$\ipfe$ is a selective, IND-secure IPFE 
\ifcameraready
scheme 
\else
scheme (\Cref{def:ipfe-sel-ind-sec}) 
\fi
for function family $\mcal{F}_{{\sf IP}, \ell+1}$ satisfying $R_\ipe$-compliance (\Cref{def:ipfe-compliant}), $R'_\ipfe$-robustness (\Cref{def:ipfe-robust}). 
% \anote{We better formalize these properties into definitions}
\end{itemize} 
Then, the construction in~\Cref{sec:construction} is a strongly-secure (\Cref{def:fas-strongly-secure}) functional adaptor signature scheme 
w.r.t.\ digital signature scheme $\DS$, NP relation $R$, and family of inner product functions $\mcal{F}_{{\sf IP}, \ell}$.
\end{theorem}

% Since witness privacy is the new additional privacy requirement of functional adaptor signatures compared to adaptor signatures, 
% We sketch the security proof of unforgeability and zero-knowledge witness privacy. We defer the full formal proof of unforgeability, zero-knowledge and all other security properties to~\Cref{sec:fas-proof-appendix}. 
% Specifically, the proof of the theorem follows from~\Cref{lemma:fas-ad-sound,lemma:fas-pre-sig-validity,lemma:fas-unf,lemma:fas-wit-ext,lemma:fas-pre-sig-adaptability,lemma:fas-zk}.

\ifhldiff
{\color{hldiffcolor}
\fi
\smallskip\noindent\textbf{Proof sketch of Unforgeability.}
A natural idea for proving unforgeability of $\fas$ would be to somehow reduce it to unforgeability of $\as$.
But, this requires non-blackbox use of the underlying $\as$, and thus results in a complex proof. 
\nikhil{why?}
We avoid that altogether and present a counter-intuitive yet elegant way of proving unforgeability of $\fas$ by reducing it to witness extractability of $\as$. Our proof technique does not need to open up the blackbox of $\as$ as evident in the sequence of games below. 

In a little more detail, to show that $\Pr[G_0(1^\secparam) = 1] \leq \negl(\secparam)$, we consider a sequence of games $G_0, G_1, G_2$ as decribed in~\Cref{fig:fasig-unf-games-short}.

\begin{figure}[t]
\ifhldiff
{\color{hldiffcolor}
\fi
\centering
\captionsetup{justification=centering}
\begin{gameproof}
\begin{pcvstack}[boxed, space=1em]
\begin{pcvstack}[space=1em]
\procedure[linenumbering, mode=text]{Games $G_0$, {\color{black} \pcbox{G_1}}, {\color{blue} \pcbox{ \color{black} G_2}}}{
$\mcal{Q} := \emptyset$,  
% \\ 
$\crs \gets \nizk.\Setup(1^\secparam)$, 
% \\ 
$\pp' \gets \ipfe.\Gen(1^\secparam)$
\\ 
% \nikhil{remove GroupGen specification.}
% \pcskipln \\ \nikhil{keep it generic for the generic construction.}
% \\ 
$\pp := (\crs, \pp')$,
% \\ 
$(\sk, \vk) \gets \KGen(\pp')$, 
% \\ 
$(X^*, \bfx^*) \gets \GenR(1^\secparam)$ 
\\ 
$(\advt, m^*, \bfy^*, \aux_\bfy^*, \pi^*_\bfy) \gets \A^{\mcal{O}_S(\cdot)}(\pp, \vk, X^*)$ 
\\ 
Parse $\advt = (\mpk, \ct, \pi)$, let $\stmt := (X^*, \pp', \mpk, \ct)$
\\ 
$\widetilde{\bfy^*} := ({\bfy^*}^T, \pi^*_\bfy)^T$, 
% \\
$\pk_{\bfy^*} := \ipe.\PubKGen(\mpk, \widetilde{\bfy^*})$ 
\\ 
${\sf Bad}_1 := \false, {\sf Bad}_2 := \false$
\\ 
If $\nizk.\Verify(\crs, \stmt, \pi) =0 \ \vee\ \bfy^* \notin \mcal{F}_{{\sf IP}, \ell} \ \vee \ \pk_{\bfy^*} \neq \aux_\bfy^*$: 
\pcskipln \\ 
$\pcind$ \ret $0$
\\
{\color{blue} \pcbox{
\color{black} \pcbox{\text{If $\stmt \notin L_{NIZK}$: ${\sf Bad}_1 := \true$, \ret $0$}}
}}
\\ 
$\widetilde{\sigma}^* \gets \as.\PreSign(\sk, m^*, \aux_\bfy^*)$,  
% \\ 
$\sigma^* \gets \A^{\mcal{O}_S(\cdot), \mcal{O}_{\fpS}(\cdot, \cdot, \cdot, \cdot, \cdot)}(\widetilde{\sigma}^*)$ 
\\
{ \color{blue} \pcbox{ \color{black}
z = \as.\Ext(\widetilde{\sigma}^*, \sigma^*, \aux_\bfy^*)
} 
} \\
{ \color{blue} \pcbox{ \color{black}
\text{If $(m^* \notin \mcal{Q}) \wedge \Verify(\vk, m^*, \sigma^*) \wedge ((\aux_\bfy^*, z) \notin R'_\ipfe)$:}	
} 
} \\
$\pcind$ { \color{blue} \pcbox{ \color{black}
\text{${\sf Bad}_2 := \true$, \ret $0$}	
} 
} \\
\ret $((m^* \notin \mcal{Q}) \wedge \Verify(\vk, m^*, \sigma^*))$
}
\end{pcvstack}
\begin{pcvstack}[space=1em]
\procedure[linenumbering, mode=text]{Oracle $\mcal{O}_S(m)$}{
$\sigma \gets \Sign(\sk, m)$,
% \\
$\mcal{Q} := \mcal{Q} \vee \{m\}$, 
% \\ 
\ret $\sigma$
}

\procedure[linenumbering, mode=text]{Oracle $\mcal{O}_{\fpS}(m, X, \bfy, \aux_\bfy, \pi_\bfy)$}{
If $\AuxVerify(\advt, \bfy, \aux_\bfy, \pi_\bfy) = 0$: \ret $\bot$
\\
$\widetilde{\sigma} \gets \FPreSign(\advt, \sk, m, X, \bfy, \aux_\bfy)$, 
% \\ 
$\mcal{Q} := \mcal{Q} \vee \{m\}$, 
% \\ 
\ret $\widetilde{\sigma}$
}
\end{pcvstack}
\end{pcvstack}
\end{gameproof}
\caption{
\ifhldiff
{\color{hldiffcolor}
\fi
Unforgeability Proof: Games $G_0, G_1, G_2$
\ifhldiff
}
\fi
% \nikhil{continue here. simplify the games.}
}
\label{fig:fasig-unf-games-short}
\ifhldiff
}
\fi
\end{figure}

\begin{itemize}[leftmargin=*]
\item 
Game $G_0$ is the original game $\faSigForge_{\A, \FAS}$, where the adversary \A has to come up with a valid forgery on a message $m^*$ of his choice, while having access to functional pre-sign oracle $\mcal{O}_{\fpS}$ and sign oracle $\mcal{O}_S$. Here, variables ${\sf Bad}_1, {\sf Bad}_2$ do not affect the game and are used only to aide the analysis below.

\item 
Game $G_1$ is same as $G_0$, except that 
before computing the pre-signature, 
the game checks if the NIZK statement $(X^*, \pp', \mpk, \ct)$ is in the language $L_{NIZK}$. 
If no, the game sets the flag ${\sf Bad}_1 = \true$ and returns $0$.
Observe that $G_0$ and $G_1$ are identical until $G_1$ checks the condition for setting ${\sf Bad}_1$. Hence,
\[
| \Pr[G_0 (1^\secparam) = 1] - \Pr[G_1 (1^\secparam) = 1] | \leq \Pr[{\sf Bad}_1 \text{ in $G_1$}]. 
\]
Thus, for proving $| \Pr[G_0 (1^\secparam) = 1] - \Pr[G_1 (1^\secparam) = 1] | \leq \negl(\secparam)$, it suffices to show that in game $G_1$, $\Pr[{\sf Bad}_1] \leq \negl(\secparam)$. 
${\sf Bad}_1 = \true$ implies $\stmt \notin L_{NIZK}$ and $\nizk.\Verify(\crs, \stmt, \pi) = 1$. Thus, the probability bound follows from the adaptive soundness of NIZK argument system.
\item 
Game $G_2$ is same as $G_1$, except that 
when \A outputs the forgery $\sigma^*$, 
the game extracts the witness $z$ of the underlying adaptor signature scheme for the statement $\pk_{\bfy^*}$ and checks if the NIZK statement $(\aux_\bfy^*, z)$ satisfies $R'_\ipfe$. 
If no, the game sets the flag ${\sf Bad}_2 = \true$ and returns $0$.
Observe that $G_0$ and $G_1$ are identical until $G_1$ checks the condition for setting ${\sf Bad}_1$. Hence,
\[
| \Pr[G_1 (1^\secparam) = 1] - \Pr[G_2 (1^\secparam) = 1] | \leq \Pr[{\sf Bad_2} \text{ in $G_2$}].
\]
% Let ${\sf Event}_2 = (m^* \notin \mcal{Q}) \wedge \Verify(\vk, m^*, \sigma^*) \wedge \neg {\sf Bad}_0 \wedge \neg {\sf Bad}_1 \wedge {\sf Bad}_2$.
Thus, for proving $| \Pr[G_1 (1^\secparam) = 1] - \Pr[G_2 (1^\secparam) = 1] | \leq \negl(\secparam)$, it suffices to show that in game $G_2$, $\Pr[{\sf Bad}_2] \leq \negl(\secparam)$.
This follows from the witness extractability of the underlying adaptor signature scheme $\as$. 
In a little more detail, we can show that if an adversary \A successfully causes $G_2$ to set ${\sf Bad}_2$, then, we can use it to come up with a reduction \B that breaks the witness extractability of the underlying adaptor signature scheme. This is because ${\sf Bad}_2 = \true$ implies $(m^* \notin \mcal{Q}) \wedge ((\aux_\bfy^*, z) \notin R'_\ipfe) \wedge \Verify(\vk, m^*, \sigma^*)$, i.e., the winning condition for \B.
\item 
Lastly, we can show that $\Pr[G_2 (1^\secparam) = 1] \leq \negl(\secparam)$ assuming $\mcal{F}_{{\sf IP}, \ell}$-hardness of relation $R$ and $R'_\ipfe$-robustness of the $\ipfe$ scheme. 
Essentially, we show that if \A wins $G_2$, then, we can build a reduction \B that breaks the $\mcal{F}_{{\sf IP}, \ell}$-hardness of relation $R$.
\B outputs $(\bfy^*, v)$, where $v = \ipfe.\Dec(z, \ct)$.
To win, \B's output must satisfy $(\bfy^* \in \mcal{F}_{{\sf IP}, \ell}) \wedge (v \in \{ f_{\bfy^*}(\bfx)  : \exists \bfx \ s.t. \ (X, \bfx) \in R\})$.
If \A wins, then $\bfy^* \in \mcal{F}_{{\sf IP}, \ell}$ and $\pk_{\bfy^*} = \aux_\bfy^*$.
Next, ${\sf Bad}_1 = \false$ implies that $\stmt \in L_{NIZK}$, where $\stmt = (X^*, \pp', \mpk, \ct)$. Hence, it follows that $\ct$ encrypts some vector $\widetilde{\bfx} = (\bfx^T, 0)^T \in \mcal{M}' \subseteq \Z^{\ell+1}$ under $\mpk$ such that $(X^*, \bfx) \in R$. 
Next, ${\sf Bad}_2 = \false$ implies that $(\aux_\bfy^*, z) \in R'_\ipfe$. 
As $\pk_{\bfy^*} = \aux_\bfy^*$ and IPFE satisfies $R'_\ipfe$-robustness, it follows that $v = f_{\widetilde{\bfy^*}}(\widetilde{\bfx})$. As $\widetilde{\bfy^*} = ({\bfy^*}^T, \pi_\bfy^*)^T$ and $\widetilde{\bfx} = (\bfx^T, 0)^T$, we get that $f_{\widetilde{\bfy^*}}(\widetilde{\bfx}) = f_{\bfy^*}(\bfx)$. 
Hence, we conclude that $v \in \{ f_{\bfy^*}(\bfx) : \exists \bfx \ s.t. \ (X, \bfx) \in R\}$. Thus, \B wins its game. This completes the proof.

\end{itemize}
\ifhldiff
}
\fi

\ifhldiff
{\color{hldiffcolor}
\fi
\smallskip\noindent\textbf{Proof sketch of Zero-Knowledge.}
We first describe the stateful simulator $\Sim = (\Setup^*, \AdGen^*, \allowbreak \AuxGen^*, \Adapt^*)$.  
Let the NIZK simulator be $\nizk.\Sim = (\nizk.\Setup^*, \nizk.\Prove^*)$. 
% and let the IPFE simulator be $\ipfe.\Sim=(\ipfe.\Setup^*, \ipfe.\Enc^*, $ $\ipfe.\KGen^*)$. 
Then, the simulator $\Sim$ is as follows.
\begin{itemize}[leftmargin=*]
\item $\Setup^*(1^\secparam)$: same as $\Setup$ except $(\crs, \td) \gets \nizk.\Setup^*(1^\secparam)$ and trapdoor $\td$ is stored as internal state by $\Sim$.
\item $\AdGen^*(\pp, X)$: same as $\AdGen$ except $\widetilde{\bfx} := (-\bft^T, 1)^T \in \Z_p^{\ell+1}$ and $\pi \gets \nizk.\Prove^*(\crs, \td, (X, \pp', \allowbreak \mpk, \ct))$.
\item $\AuxGen^*(\advt, \bfy, f_{\bfy}(\bfx))$: same as $\AuxGen$ except $\pk_\bfy$ is computed for $\widetilde{\bfy} := (\bfy^T, f_\bfy(\bft) + f_\bfy(\bfx))^T \in \mcal{F}_{{\sf IP}, \ell+1}$ and thus $\pi_\bfy := f_\bfy(\bft) + f_\bfy(\bfx)$ as well.
\item $\Adapt^*(\advt, \vk, m, X, \bfy, \aux_\bfy, \widetilde{\sigma}, f_\bfy(\bfx))$: same as $\Adapt$ except $\sk_\bfy$ for $\widetilde{\bfy} := (\bfy^T, f_\bfy(\bft) + f_\bfy(\bfx))^T$ is used.
\end{itemize}

% The simulator induces the ideal-world experiment. 
To see that adversary's views in real and ideal world (~\Cref{fig:fas-zk-real}) are indistinguishable, consider a sequence of games $G_0, G_1, G_2, G_3$.
\begin{itemize}[leftmargin=*]
\item 
Game $G_0$ corresponds to the real-world experiment where $\Setup$, $\AdGen$, $\AuxGen$, $\Adapt$ are used. 
\item 
Game $G_1$ is same as $G_0$, except the NIZK is switched to simulation mode, i.e., we change $\Setup$ to perform $(\crs, \td) \gets \nizk.\Setup^*(\allowbreak 1^\secparam)$ and change $\AdGen$ to perform $\pi \gets \nizk.\Prove^*(\crs, \td, (X, \allowbreak \pp', \allowbreak \mpk, \ct))$. It follows from the zero-knowledge property of NIZK that $G_0 \approx_c G_1$. 
\item 
Game $G_2$ is same as $G_1$ except $\widetilde{\bfy}$ used in $\AuxGen$ and $\Adapt$ is switched from $\widetilde{\bfy} := (\bfy^T, f_\bfy(\bft))^T$ to $\widetilde{\bfy} := (\bfy^T, f_\bfy(\bft) + f_\bfy(\bfx))^T$. Since $f$ is a linear function, it follows that $f_\bfy(\bft) + f_\bfy(\bfx) = f_\bfy(\bft + \bfx)$, thus, one can view the transition to $G_2$ as a change of variables $\bft \to \bft+\bfx$. Since $\bft$ is uniform random, it follows that $G_1$ and $G_2$ are identically distributed. 
\item 
Game $G_3$ is same as $G_2$ except that in $\AdGen$, $\ct$ encrypts 
$\widetilde{\bfx} := (-\bft^T, 1)^T$ 
instead of 
$\widetilde{\bfx} := (\bfx^T, 0)^T$. 
Observe that this does not change the inner product value $\widetilde{\bfx}^T\widetilde{\bfy}$. Thus, $G_2 \approx_c G_3$ follows from IND-security of the IPFE scheme.
\item
Lastly observe that adversary's view in $G_3$ is same as that in the ideal-world experiment. This complete the proof.
\end{itemize}

We emphasize that in $G_3$, even the underlying IPFE only uses the information $f_\bfy(\bfx)$ about $\bfx$. Thus, going from $G_1$ to $G_3$ is essentially building an IND-security to simulation-security IPFE compiler. The random coins of this compiler are $\bft$. Crucially, $G_2$ and $G_3$ explicitly use $f_\bfy(\bft)$, which is a leakage on these random coins and for this reason the IPFE compiler is used in a non-blackbox way. We refer the reader back to~\Cref{remark:ipfe-simulation} on why $f_\bfy(\bft)$ is an acceptable leakage on the IPFE simulator's random coins $\bft$.
\ifhldiff
}
\fi

%% file: fas-from-prime-groups-eprint.tex
\ifhldiff
{\color{hldiffcolor}
\fi
\section{FAS from Groups of Prime-Order}
\label{sec:fas-from-prime-groups}\label{sec:detailed_instantiation}

% \ifacm
% We discuss the instantiations of the two building blocks namely $\as$ as Schnorr adaptor signature scheme~\cite{PKC:EFHMR21}, and the selective, IND-secure IPFE scheme by Abdalla et al.~\cite{ipe}. These two are the key components of our FAS instantiation.
% \else
We provide instantiations of the FAS construction in~\Cref{sec:construction} from prime order groups. 
For this, it suffices to instantiate the building blocks $\IPFE$ and $\AS$ from prime order groups, while ensuring the two are compatible with each other. We describe these next.
We can instantiate the $\nizk$ for the language $L_\nizk$ depending on the concrete relation $R$. 
That is, given $R$ and the above instantiation of $\IPFE$, we can pick the most efficient $\nizk$ for $L_\nizk$.
Since this is not the main contribution of this work, we will assume $R$ to be a general NP relation and rely on the existence of $\nizk$ for general NP~\cite{nizk-from-ddh,peikert2019noninteractive} and not delve into its details here.
%The asterisk $^*$ in the section title is to note that $\nizk$ for NP are not known from prime order groups, so we do not specify its details here. As the lattice-based $\nizk$ systems support all NP langauges \pnote{which papers to cite?}, we can use those even for the language $L_\nizk$ related to the prime-order groups that the $\ipfe$ instantiation induces here.

More concretely, we instantiate the $\as$ as Schnorr adaptor signature scheme~\cite{asig} and thus, we have $R_\ipfe = R'_\ipfe = R_\DL$ as defined 
\ifcameraready
in the full version. 
\else
in~\Cref{def:dl-relation}. 
\fi
Further, we show that the selective, IND-secure IPFE scheme by Abdalla et al.~\cite{ipe} satisfies $R_\DL$-compliance and $R_\DL$-robustness when appropriately augmented with a $\PubKGen$ algorithm. 
% As the $\as$ scheme is not part of our contributions, we defer its details 
% \ifcameraready
% to the full version
% \else
% to~\Cref{sec:schnorr-adaptor} 
% \fi
% and describe the IPFE scheme here.
% % \fi 

%\subsection{Preliminaries}
%\noindent\textbf{Decisional Diffie-Hellman (DDH) assumption.}
%Let $\GroupGen$ be a \ppt algorithm that takes as input a security parameter $1^\secparam$ and outputs a triplet $(\G, p, g)$ where $\G$ is a group of order $p$ that is generated by $g \in \G$, and $p$ is a $\secparam$-bit prime number. Then, the DDH assumption states that the tuples $(g, g^a, g^b, g^{ab})$ and $(g, g^a, g^b, g^c)$ are computationally indistinguishable, where $(\G, p, g)\gets \GroupGen(1^\secparam)$, and $a, b, c \in \Z_p$ are chosen independently and uniformly at random. \anote{Why do we need to recall this assumption here? Remove if space constraint}
%
%\noindent\textbf{Additional notation.}
%For $d \in \G$, we denote the algorithm for computing the discrete log of $d$ with respect to base $g \in \G$ by $\DLog_g(d)$. As long as we are promised that the discrete log is bounded by a polynomial $B$, the $\DLog$ algorithm simply enumerates all the possibilities to find the discrete log. The running time is polynomial because of the bound $B$.

\subsection{Schnorr Adaptor Signature}\label{sec:schnorr-adaptor}
First, we recall the Schnorr Signature Scheme $\Sch$~\cite{schnorr}.
Suppose $\pp := (\G, p, g) \gets \ipfe.\Gen(1^\secparam)$ as described in~\Cref{fig:ipfe-abdp}, where $\G$ is a cyclic group of prime order $p$ and $g$ is a generator of $\G$. Suppose $H: \{0,1\}^* \to \ZZ_p$ is a hash function modeled as a random oracle. Then $\Sch$ is described below:
\begin{itemize}[leftmargin=*]
	\item $(\vk,\sk)\gets \Sch.\KGen(\pp)$: choose $\delta \gets \ZZ_p$ and set $\sk:=\delta$ and $\vk:=g^\delta$.
	\item $\sigma \gets \Sch.\Sign(\sk, m)$: sample a randomness $r \gets \ZZ_p$ to  compute $ h := H(\vk||g^r||m), s:=r+h\delta$ and output $\sigma:=(h,s)$.
	\item $0/1 \gets \Sch.\Verify(\vk,m,\sigma)$: parse $\sigma:=(h,s)$ and then compute $R := g^s / \vk^h$ and if $h = H(\vk ||R||m)$ output $1$, otherwise output $0$.
\end{itemize}

%\begin{figure}[t]
%\centering
%\captionsetup{justification=centering}
%\begin{pchstack}[boxed, space=1em]
%\begin{pcvstack}[space=1em]
%
%\procedure[linenumbering, mode=text]{$\Sch.\KGen(\pp)$}{
%$\delta \getr \Z_p$
%\\
%\ret $\sk := \delta$, $\vk := g^\delta$
%}
%
%\procedure[linenumbering, mode=text]{$\Sch.\Verify(\vk, m, \sigma)$}{
%$g^R := g^s / \vk^h$
%\\ 
%\ret $1$ iff $h = H(\vk \ ||\ g^R \ ||\ m)$
%}
%\end{pcvstack}
%\begin{pcvstack}[space=1em]
%
%\procedure[linenumbering, mode=text]{$\Sch.\Sign(\sk, m \in \{0,1\}^*)$}{
%$r \getr \Z_p$
%\\ 
%$h := H(\vk \ ||\ g^r \ ||\ m)$
%\\ 
%$s := r + h \delta$
%\\
%\ret $\sigma := (h, s)$
%}
%\end{pcvstack}
%\end{pchstack}
%
%\caption{$\Sch$: Schnorr Signature Scheme\nikhil{add cite}}
%\label{fig:sch-sig}
%\end{figure}

Next, we describe the Adaptor Signature scheme $\AS$ with respect to the digital signature scheme $\Sch$ and hard relations $R=R'=R_\DL$, where $R_\DL$ is the discrete log relation. Let $L_\DL$ be the corresponding language. Then, $\AS$ is described below:

\begin{itemize}[leftmargin=*]
	\item $\widetilde{\sigma}\gets \AS.\PreSign(\sk, m, X \in L_\DL)$: choose $r \getr \Z_p$ and compute $h := H(\vk ||g^r \cdot X||m)$ and $ \widetilde{s} := r + h \delta$. Set $\widetilde{\sigma} := (h, \widetilde{s})$.
	\item $0/1 \gets \AS.\PreVerify(\vk, m, X, \widetilde{\sigma})$: parse $\widetilde\sigma := (h, \widetilde{s})$ and compute $R := g^s / \vk^h$. If $h = H(\vk ||R \cdot X ||m)$ \ret $1$, else \ret $0$. 
	\item $\sigma \gets \AS.\Adapt(\vk, m , X, x, \widetilde{\sigma})$: parse $\widetilde{\sigma} = (h, \widetilde{s})$ and compute $s := \widetilde{s} + x$. \ret $\sigma := (h, s)$
	\item $x' \gets \AS.\Ext(\widetilde{\sigma}, \sigma, X)$: parse $\widetilde{\sigma} = (h, \widetilde{s})$ and $\sigma = (h, s)$. Compute $x' := s - \widetilde{s}$ and if $(X, x')\notin R_\DL$, \ret $\bot$, else \ret $x'$.  
\end{itemize}

%
%\begin{figure}[t]
%\centering
%\captionsetup{justification=centering}
%\begin{pchstack}[boxed, space=1em]
%\begin{pcvstack}[space=1em]
%
%\procedure[linenumbering, mode=text]{$\AS.\PreSign(\sk, m, X \in L_\DL)$}{
%$r \getr \Z_p$
%\\ 
%$h := H(\vk \ ||\ g^r \cdot X \ ||\ m)$
%\\ 
%$\widetilde{s} := r + h \delta$
%\\
%\ret $\widetilde{\sigma} := (h, \widetilde{s})$
%}
%
%\procedure[linenumbering, mode=text]{$\AS.\PreVerify(\vk, m, X, \widetilde{\sigma})$}{
%$g^R := g^s / \vk^h$
%\\ 
%\ret $h = H(\vk \ ||\ g^R \cdot X \ ||\ m)$
%}
%\end{pcvstack}
%\begin{pcvstack}[space=1em]
%
%\procedure[linenumbering, mode=text]{$\AS.\Adapt(\vk, m , X, x, \widetilde{\sigma})$}{
%Parse $\widetilde{\sigma} = (h, \widetilde{s})$
%\\ 
%$s := \widetilde{s} + x$
%\\ 
%\ret $\sigma := (h, s)$
%}
%
%\procedure[linenumbering, mode=text]{$\AS.\Ext(\widetilde{\sigma}, \sigma, X)$}{
%Parse $\widetilde{\sigma} = (h, \widetilde{s})$, $\sigma = (h, s)$
%\\ 
%$x' := s - \widetilde{s}$
%\\ 
%If $(X, x')\notin R_\DL$: \ret $\bot$
%\\ 
%Else: \ret $x'$
%}
%\end{pcvstack}
%\end{pchstack}
%
%\caption{$\AS$ w.r.t.\ $\Sch$ signature scheme and hard relation $R_\DL$\nikhil{add cites}}
%\label{fig:sch-as}
%\end{figure}

\begin{lemma}[\cite{asig}]
\label{lemma:sch-as}
The adaptor signature scheme in~\Cref{sec:schnorr-adaptor} satisfies witness extractability (\Cref{def:as-wit-ext}) and weak pre-signature adaptability (\Cref{def:as-pre-sig-adaptability}).
% \nikhil{add assumptions}
\end{lemma}

\subsection{IPFE from Groups of Prime-Order}
We first recall the IPFE scheme by Abdalla et al.~\cite{ipe}. Then, we show that it satisfies the additional compliance and robustness properties needed by our FAS scheme.

Suppose $p$ is a $\secparam$-bit prime number and suppose we want to compute inner products of vectors of length $\ell$.
The Abdalla et al.~\cite{ipe} IPFE scheme computes inner products with output values polynomially bounded by $B \ll p$. 
Specifically, the message space is $\mcal{M} = \Z_p^\ell$ and the function class is $\mcal{F}_{{\sf IP}, \ell, p, B} = \{ f_\bfy : \bfy \in \Z_p^\ell\} \subseteq \mcal{F}_{{\sf IP}, \ell}$, where $f_\bfy: \Z_p^\ell \to \{0, \ldots, B\}$ is defined as $f_\bfy(\bfx) = \bfx^T \bfy \in \{0, \ldots, B\}$. 
The scheme by Abdalla et al.\cite{ipe} augmented with $\PubKGen$ algorithm is as in~\Cref{fig:ipfe-abdp}.
For $d \in \G$, we denote the algorithm for computing the discrete log of $d$ with respect to base $g \in \G$ by $\DLog_g(d)$. 
\ifcameraready\else
As long as we are promised that the discrete log is bounded by a polynomial $B$, the $\DLog$ algorithm simply enumerates all the possibilities to find the discrete log. The running time is polynomial because of the bound $B$.
\fi 

\begin{figure}[t]
\ifhldiff
{\color{hldiffcolor}
\fi
\centering
\captionsetup{justification=centering}
\begin{pchstack}[boxed, space=0.5em]
\begin{pcvstack}

\procedure[linenumbering, mode=text]{$\ipfe.\Gen(1^\secparam)$}{
$\pp := (\G, p, g) \gets \GroupGen(1^\secparam)$\\
\ret $\pp$
}

\procedure[linenumbering, mode=text]{$\ipfe.\Setup(\pp, 1^\ell)$}{
$\bfs \getr \Z_p^\ell$, 
\ret $\msk:=\bfs, \mpk:=g^\bfs$
}

\procedure[linenumbering, mode=text]{$\ipfe.\Enc(\mpk, \bfx \in \Z_p^\ell)$}{
$r \getr \Z_p$, $\ct_0 := g^r$, $\ct_1 := g^\bfx \cdot \mpk^r$
\\
\ret $\ct := (\ct_0, \ct_1)$
}

\procedure[linenumbering, mode=text]{$\ipfe.\KGen(\msk, \bfy \in \Z_p^\ell)$}{
\ret $\sk_\bfy := \bfs^T\bfy \in \Z_p$
}

\end{pcvstack}
\begin{pcvstack}

\procedure[linenumbering, mode=text]{$\ipfe.\PubKGen(\mpk, \bfy \in \Z_p^\ell)$}{
Parse $\mpk = (k_1, \ldots, k_\ell)$
\\ 
Parse $\bfy = (y_1, \ldots, y_\ell)$
\\
\ret $\pk_\bfy := \prod_{i \in [\ell]} k_i^{y_i} $
}

\procedure[linenumbering, mode=text]{$\ipfe.\Dec(\sk_\bfy, \ct)$}{
Parse $\ct_1 = (c_1, \ldots, c_\ell)$
\\ 
Parse $\bfy = (y_1, \ldots, y_\ell)$
\\
$d := \prod_{i \in [\ell]} c_i^{y_i} / \ct_0^{\sk_\bfy}$
\\
\ret $v := \DLog_g(d)$
}

\end{pcvstack}
\end{pchstack}

\caption{
\ifhldiff
{\color{hldiffcolor}
\fi
Abdalla et al.~\cite{ipe} IPFE scheme augmented with $\PubKGen$ algorithm
\ifhldiff
}
\fi
}
\label{fig:ipfe-abdp}
\ifhldiff
}
\fi
\end{figure}

\begin{lemma}[\cite{ipe}]\label{lemma:1}
Suppose the DDH assumption holds. 
Then, the IPFE scheme in~\Cref{fig:ipfe-abdp} is selective, 
\ifcameraready
IND-secure.
\else
IND-secure (\Cref{def:ipfe-sel-ind-sec}).
\fi
\end{lemma}

Next, we show the IPFE scheme augmented with the above $\PubKGen$ algorithm satisfies $R_\DL$-compliance and $R_\DL$-robustness. 

\begin{lemma}\label{lemma:2}
The IPFE scheme in~\Cref{fig:ipfe-abdp} is $R_\DL$-compliant.
\end{lemma}
\begin{proof}
For any $\secparam \in \N$,
for any $\pp \gets \ipe.\Gen(1^\secparam)$ as implemented in~\Cref{fig:ipfe-abdp},
for any $(\mpk, \msk) \gets \ipe.\Setup(\pp, 1^\ell)$ as implemented in~\Cref{fig:ipfe-abdp}, 
for any $\bfy \in \mcal{F}_{{\sf IP}, \ell, p, B}$,
let $\pk_\bfy :=  \PubKGen(\allowbreak \mpk, \bfy)$ as implemented in~\Cref{fig:ipfe-abdp}, and 
$\sk_\bfy := \KGen(\msk, \bfy)$ as implemented in~\Cref{fig:ipfe-abdp}. Suppose here $\msk=\bfs \in \Z_p^\ell$. Then, $\sk_\bfy = \bfs^T\bfy \in \Z_p$ and $\pk_\bfy = g^{\bfs^T\bfy} \in \G$. Thus, clearly, $(\pk_\bfy, \sk_\bfy) \in R_\DL$.
\end{proof}

\begin{lemma}\label{lemma:3}
The IPFE scheme in~\Cref{fig:ipfe-abdp} is $R_\DL$-robust.
\end{lemma}
\begin{proof}
Observe that in the IPFE scheme in~\Cref{fig:ipfe-abdp}, the base relation $R$ and the extended relation $R'$ are the same: $R = R'=R_\DL$. Further, $R_\DL$ is a unique witness relation. Thus, $R_\DL$-robustness follows trivially from the correctness of the IPFE scheme.  
\end{proof}

\subsection{$\fas$ Construction}

Instantiating $\fas$ construction in~\Cref{sec:fas-construction} with $\ipfe$ from~\Cref{fig:ipfe-abdp} and $\as$ from~\Cref{sec:schnorr-adaptor}, we obtain the following corollary.

\begin{corollary}
\label{corollary:fas-strongly-secure-prime-order-groups}
\label{thm:fas-strongly-secure-prime-order-groups}
Let $p$ be a $\secparam$-bit prime number and let $\ell$ be an integer.
Let $\mcal{M}=\Z_p^\ell$ be an additive group.
let $\mcal{F}_{\sf{IP}, \ell, p, B}$ be the function family for computing inner products of vectors in $\Z_p^\ell$ such that the output value is bounded by some polynomial $B \ll p$. 
Let $R$ be any NP relation 
with statement/witness pairs $(X, \bfx)$ such that $\bfx \in \Z_p^\ell$.
Suppose that 
\begin{itemize}[leftmargin=*]
\item
$R$ is $\mcal{F}_{{\sf IP}, \ell, p, B}$-hard (\Cref{def:f-hard-relation}),
\item
$\nizk$ is a secure NIZK argument system (\Cref{def:nizk}),
\item
$\as$ construction in~\Cref{sec:schnorr-adaptor} is an adaptor signature scheme w.r.t.\ digital signature scheme $\Sch$ and hard relation $R_\DL$ that satisfies weak pre-signature adaptability and witness extractability (\Cref{lemma:sch-as}),
\item
$\ipfe$ construction in~\Cref{fig:ipfe-abdp} is a selective, IND-secure IPFE scheme (\Cref{lemma:1}) for function family $\mcal{F}_{{\sf IP}, \ell+1}$ that is $R_\ipe$-compliant (\Cref{lemma:2}) and $R_\ipfe$-robust (\Cref{lemma:3}). 
\end{itemize} 
Then, the functional adaptor signature scheme 
w.r.t.\ Schnorr signature scheme $\Sch$, NP relation $R$, and family of inner product functions $\mcal{F}_{{\sf IP}, \ell, p, B}$  
constructed in~\Cref{sec:construction} is strongly-secure (\Cref{def:fas-strongly-secure}).
\end{corollary}

The proof of the above corollary is immediate from the proof of the~\cref{thm:fas-strongly-secure} concerning the generic construction, and~\cref{lemma:1,lemma:2,lemma:3} that show that the $\ipfe$ scheme in~\Cref{fig:ipfe-abdp} has the required properties.

% \nikhil{move this to DL based instantiation section later}
% The IPFE scheme by Abdalla et al.~\cite{ipe} based on the DDH assumption satisfies our requirements. Specifically, in their construction, $\msk := \bfs \getr \Z_p^\ell$, $\mpk := g^\bfs$, and $\sk_\bfy := \bfs^T \bfy$. Then, given $\mpk$ and $\bfy$, we can efficiently compute $\pk_\bfy := g^{ \bfs^T \bfy} := \Pi_{i \in [m]} (g^{s_i})^{y_i}$. This choice meets our requirements that $(\mpk, \msk) \in R_{DL}$ and for any $\bfy$, $(\pk_\bfy, \sk_\bfy) \in R_{DL}$. \anote{where do we define $R_{DL}$}

%% file: fas-from-lattices-eprint.tex
\ifhldiff
{\color{hldiffcolor}
\fi
\section{$\fas$ from lattices}
\label{sec:fas-from-lattices}

In this section, we provide instantiations of the FAS construction in~\Cref{sec:construction} from lattices.
For this, it suffices to instantiate the building blocks IPFE and AS from lattices, while ensuring the two are compatible with each other. We describe these next.

More specifically, we instantiate the $\as$ as the lattice-based adaptor signature scheme by Esgin et al.~\cite{pq-as}. This scheme was built using cyclotomic ring $\Ring = \Z[x]/(x^d+1)$ of degree $d=256$ under Module-SIS and Module-LWE assumptions. Here, we set the degree to be $d=1$, thus giving us the ring of integers $\mcal{R} = \Z$. Consequently, security follows from plain SIS and LWE assumptions. One can look at the resulting $\as$ to be w.r.t.\ a digital scheme $\Lyu$ that is somewhere in between Lyubashevsky's signature scheme~\cite{lyu12lattice-sig} and Dilithium~\cite{dilithium} instantiated with unstructured lattices. 
% \nikhil{is there a closer match? Lyu12 is over unstructured lattices but uses gaussian sampling. If there is some work that does Lyu12 but does not do gaussian sampling, that would be the closest to our digital signature.}
Further, $\as$ is w.r.t.\ hard relations $R_\isis$ and $R'_\isis$ such that $R_\isis \subset R'_\isis$.

Further, we show that the IND-secure IPFE scheme by Agrawal et al.~\cite[Section 4.2]{ipe01} satisfies $R_\isis$-compliance and $R'_\isis$-robustness when appropriately augmented with a $\PubKGen$ algorithm. 
We describe the instantiations of these two building blocks in the subsequent sections.

\subsection{Lattice-based Adaptor Signature}
\label{sec:as-lattices}
First, we recall the lattice-based signature scheme $\Lyu$, which can be seen as a variant of Lyubashevsky's signature scheme~\cite{lyu12lattice-sig} where the vector $\bfy$ sampled during signing comes from a uniform distribution over a small range instead of discrete gaussian distribution. 
It can also be seen as a variant of Dilithium signature scheme~\cite{dilithium} instantiated over unstructured lattices. 
% \nikhil{add cites}

Suppose $\pp := (n, m, p, q, k, \alpha) \gets \ipfe.\Gen(1^\secparam, p)$ as described in~\Cref{fig:ipfe-als}, where $p$ is a prime number, $n=\secparam$, $q=p^k$ and $k \in \Z$, $m \in \Z$, real $\alpha \in (0,1)$ are chosen as described in the ``Parameter choices'' in~\Cref{sec:ipfe-lwe}. 
Let $\mcal{H}: \{0,1\}^* \to \C$ be a family of hash functions (modelled as a random oracle), where $\C = \{\bfc \in \Z^h: ||\bfc||_1 = \kappa \wedge ||\bfc||_\infty = 1\}$. 
\nikhil{define $\kappa$}
Let $\gamma$ be the maximum absolute coefficient of the masking randomness $\bfy$ used in the $\Sign$ algorithm below. 
Then, the digital signature scheme $\Lyu$ is as in~\Cref{fig:lyu-sig}.
\ifacm
	\begin{figure}[t]
	\centering
	\captionsetup{justification=centering}
	\begin{pcvstack}[boxed, space=1em]

	\procedure[linenumbering, mode=text]{$\Lyu.\KGen(1^\secparam, \pp)$}{ 
	Sample $\bfA' \getr \Z_q^{n \times m}$, let $\bfA := (I_n || \bfA')^T \in \Z_q^{n \times (n + m)}$
	\\ 
	Sample $\bfR \getr \S_1^{(n + m) \times h}$, let $\bfT := \bfA \bfR \bmod{q} \in \Z_q^{n \times h}$
	\\ 
	\ret $\vk := (\bfA, \bfT), \sk := \bfR$
	}

	\procedure[linenumbering, mode=text]{$\Lyu.\Sign(\sk, \mu \in \{0,1\}^*)$}{ 
	Sample $\bfy \getr \S_\gamma^{n + m}$, let $\bfw := \bfA \bfy$
	\\ 
	Compute $\bfc := \H(\vk \ ||\  \bfw \ ||\  \mu) \in \Z^h$
	\\ 
	Compute $\bfz := \bfy + \bfR \bfc$
	\\ 
	If $||\bfz||_\infty > \gamma - \kappa$: restart 
	\\ 
	\ret $\sigma := (\bfc, \bfz)$
	}

	\procedure[linenumbering, mode=text]{$\Lyu.\Verify(\vk, \mu, \sigma)$}{ 
	Parse $\sigma = (\bfc, \bfz)$
	\\ 
	If $||\bfz||_\infty > \gamma - \kappa$: \ret $0$
	\\ 
	Compute $\bfw' := \bfA \bfz - \bfT \bfc $
	\\ 
	If $\bfc \neq \H(\vk \ ||\  \bfw' \ ||\ \mu)$: \ret $0$
	\\ 
	\ret $1$
	}

	\end{pcvstack}
	\caption{Digital signature scheme $\Lyu$ from \sis and \lwe}
	\label{fig:lyu-sig}
	\end{figure}
\else
	\begin{figure}[H]
	\centering
	\captionsetup{justification=centering}
	\begin{pchstack}[boxed, space=0.1em]

	\procedure[linenumbering, mode=text]{$\Lyu.\KGen(1^\secparam, \pp)$}{ 
	Sample $\bfA' \getr \Z_q^{n \times m}$
	\\ 
	Let $\bfA := (I_n || \bfA') \in \Z_q^{n \times (n + m)}$
	\\ 
	Sample $\bfR \getr \S_1^{(n + m) \times h}$ 
	\\ 
	Let $\bfT := \bfA \bfR \bmod{q} \in \Z_q^{n \times h}$
	\\ 
	\ret $\vk := (\bfA, \bfT), \sk := \bfR$
	}

	\procedure[linenumbering, mode=text]{$\Lyu.\Sign(\sk, \mu \in \{0,1\}^*)$}{ 
	Sample $\bfy \getr \S_\gamma^{n + m}$
	\\ 
	Let $\bfw := \bfA \bfy$
	\\ 
	Let $\bfc := \H(\vk \ ||\  \bfw \ ||\  \mu) \in \Z^h$
	\\ 
	Let $\bfz := \bfy + \bfR \bfc$
	\\ 
	If $||\bfz||_\infty > \gamma - \kappa$: restart 
	\\ 
	\ret $\sigma := (\bfc, \bfz)$
	}

	\procedure[linenumbering, mode=text]{$\Lyu.\Verify(\vk, \mu, \sigma)$}{ 
	Parse $\sigma = (\bfc, \bfz)$
	\\ 
	If $||\bfz||_\infty > \gamma - \kappa$: \ret $0$
	\\ 
	Let $\bfw' := \bfA \bfz - \bfT \bfc $
	\\ 
	If $\bfc \neq \H(\vk \ ||\  \bfw' \ ||\ \mu)$: \ret $0$
	\\ 
	\ret $1$
	}

	\end{pchstack}
	\caption{Digital signature scheme $\Lyu$ from \sis and \lwe}
	\label{fig:lyu-sig}
	\end{figure}
\fi
Next, we describe the Adaptor Signature scheme $\AS$ with respect to the digital signature scheme $\Lyu$ and hard relations $R_\isis := \isis_{n, n + m, q, \beta_0}$ and $R'_\isis:= \isis_{n, n + m, q, \beta_1}$ such that $R_\isis \subseteq R'_\isis$. 
Here, $\beta_0 := \ell p B_\tau$ and $\beta_1 := 2(\gamma - \kappa) - \beta_0$ such that $\beta_1 > \beta_0$ and $\gamma -\kappa - \beta_0 > 0$.
Let $L_\isis$ be the  language corresponding to $R_\isis$ and $R'_\isis$.
Notice that $L_\isis \subseteq \Z_q^n$. 
Hence, for $R_\isis$, statements are of the form $(\bfA, X) \in \Z_q^{n \times (n + m)} \times \Z_q^n$ (same for $'_\isis$) and witnesses are of the form $\bfx \in \Z^{n + m}$ such that $||\bfx||_\infty \leq \beta_0$ ($||\bfx||_\infty \leq \beta_1$ for $R'_\isis$). 
As $\bfA$ is also part of the verification key, henceforth we drop it from the statement. 
We next describe the $\AS$ construction in ~\Cref{fig:lyu-as}.
\ifacm
	\begin{figure}[t]
	\centering
	\captionsetup{justification=centering}
	\begin{pcvstack}[boxed, space=1em]

	\procedure[linenumbering, mode=text]{$\AS.\PreSign(\sk, \mu \in \{0,1\}^*, X \in L_\isis)$}{ 
	Sample $\bfy \getr \S_\gamma^{n + m}$, let $\bfw := \bfA \bfy$
	\\ 
	Compute $\bfc := \H(\vk \ ||\  \bfw + X \ ||\  \mu) \in \Z^h$
	\\ 
	Compute $\widetilde{\bfz} := \bfy + \bfR \bfc$
	\\ 
	If $||\widetilde{\bfz}||_\infty > \gamma - \kappa - \beta_0$: restart 
	\\ 
	\ret $\widetilde{\sigma} := (\bfc, \widetilde{\bfz})$
	}

	\procedure[linenumbering, mode=text]{$\AS.\PreVerify(\vk, \mu, X, \widetilde{\sigma})$}{ 
	Parse $\widetilde{\sigma} = (\bfc, \widetilde{\bfz})$
	\\ 
	If $||\widetilde{\bfz}||_\infty > \gamma - \kappa - \beta_0$: \ret $0$
	\\ 
	Compute $\bfw' := \bfA \widetilde{\bfz} - \bfT \bfc $
	\\ 
	If $\bfc \neq \H(\vk \ ||\  \bfw' \ ||\ \mu)$: \ret $0$
	\\ 
	\ret $1$
	}

	\procedure[linenumbering, mode=text]{$\AS.\Adapt(\vk, \mu, X, x, \widetilde{\sigma})$}{ 
	If $\PreVerify(\vk, \mu, X, \widetilde{\sigma}) = 0$: \ret $\bot$ 
	\\ 
	Parse $\widetilde{\sigma} = (\bfc, \widetilde{\bfz})$
	\\ 
	$\bfz := \widetilde{\bfz} + x$
	\\ 
	\ret $\sigma := (\bfc, \bfz)$
	}

	\procedure[linenumbering, mode=text]{$\AS.\Ext(\widetilde{\sigma}, \sigma, X)$}{ 
	Parse $\widetilde{\sigma} = (\bfc, \widetilde{\bfz}), \sigma = (\bfc, \bfz)$
	\\ 
	$x' := \bfz - \widetilde{\bfz} $
	\\ 
	If $ X \neq \bfA x' \bmod{q}$: \ret $\bot$
	\\
	\ret $x'$
	}

	\end{pcvstack}
	\caption{$\as$ w.r.t.\ $\Lyu$ signature scheme and hard relations $R_\isis$ and $R'_\isis$}
	\label{fig:lyu-as}
	\end{figure}
\else
	\begin{figure}[H]
	\centering
	\captionsetup{justification=centering}
	\begin{pchstack}[boxed, space=1em]
	\begin{pcvstack}
	\procedure[linenumbering, mode=text]{$\AS.\PreSign(\sk, \mu \in \{0,1\}^*, X \in L_\isis)$}{ 
	Sample $\bfy \getr \S_\gamma^{n + m}$, let $\bfw := \bfA \bfy$
	\\ 
	Let $\bfc := \H(\vk \ ||\  \bfw + X \ ||\  \mu) \in \Z^h$
	\\ 
	Let $\widetilde{\bfz} := \bfy + \bfR \bfc$
	\\ 
	If $||\widetilde{\bfz}||_\infty > \gamma - \kappa - \beta_0$: restart 
	\\ 
	\ret $\widetilde{\sigma} := (\bfc, \widetilde{\bfz})$
	}

	\procedure[linenumbering, mode=text]{$\AS.\Adapt(\vk, \mu, X, x, \widetilde{\sigma})$}{ 
	If $\PreVerify(\vk, \mu, X, \widetilde{\sigma}) = 0$: \ret $\bot$ 
	\\ 
	Parse $\widetilde{\sigma} = (\bfc, \widetilde{\bfz})$
	\\ 
	$\bfz := \widetilde{\bfz} + x$
	\\ 
	\ret $\sigma := (\bfc, \bfz)$
	}

	\end{pcvstack}
	\begin{pcvstack}

	\procedure[linenumbering, mode=text]{$\AS.\PreVerify(\vk, \mu, X, \widetilde{\sigma})$}{ 
	Parse $\widetilde{\sigma} = (\bfc, \widetilde{\bfz})$
	\\ 
	If $||\widetilde{\bfz}||_\infty > \gamma - \kappa - \beta_0$: \ret $0$
	\\ 
	Let $\bfw' := \bfA \widetilde{\bfz} - \bfT \bfc $
	\\ 
	If $\bfc \neq \H(\vk \ ||\  \bfw' \ ||\ \mu)$: \ret $0$
	\\ 
	\ret $1$
	}

	\procedure[linenumbering, mode=text]{$\AS.\Ext(\widetilde{\sigma}, \sigma, X)$}{ 
	Parse $\widetilde{\sigma} = (\bfc, \widetilde{\bfz}), \sigma = (\bfc, \bfz)$
	\\ 
	$x' := \bfz - \widetilde{\bfz} $
	\\ 
	If $ X \neq \bfA x' \bmod{q}$: \ret $\bot$
	\\
	\ret $x'$
	}

	\end{pcvstack}
	\end{pchstack}
	\caption{$\as$ w.r.t.\ $\Lyu$ signature scheme and hard relations $R_\isis$ and $R'_\isis$}
	\label{fig:lyu-as}
	\end{figure}
\fi
\noindent\textbf{Parameter choices and comparison.}
Our construction can be seen as similar to~\cite{pq-as} but with degree $1$ instead of $256$.
\nikhil{how does this affect other parameters? CHECK!!} 
Our construction is different from~\cite{pq-as} in the following aspects though:~\cite{pq-as} chooses $\beta_0 = 1$, whereas we choose $\beta_0 = \ell p B_\tau$, where parameters $\ell, p, B_\tau$ are chosen based on the specification in the IPFE construction in~\Cref{fig:ipfe-als}. Further,~\cite{pq-as} sets $\beta_1 = 2(\gamma-\kappa)$ whereas we set it a stricter bound on it $\beta_1 = 2(\gamma-\kappa)-\beta_0$ as it is sufficient for correctness and security. Further, note that for the choice of $\beta_1$, the two constraints 
$\beta_1 > \beta_0$ and $\gamma - \kappa - \beta_0 > 0$ are equivalent. Lastly, we note that just like in~\cite{pq-as}, we can set $\gamma = 2 \kappa$ to ensure that the average number of restarts in $\Sign$ and $\PreSign$ is about $e < 3$, where $\kappa$ is the $\ell_1$ norm of values in the range set $\C$ of the hash family $\mcal{H}$.
Lastly, in~\cite{pq-as}, the hash outputs belong to some polynomial ring $\Ring_q$, whereas we choose the range of hash outputs to be $\Z_q^h$ for some integer $h$.\nikhil{how does $h$ affect any other parameters if at all. CHECK!!}

\begin{lemma}[\cite{pq-as}]
\label{lemma:pq-as}
The adaptor signature scheme in~\Cref{fig:lyu-as} satisfies witness extractability (\Cref{def:as-wit-ext}) and weak pre-signature adaptability (\Cref{def:as-pre-sig-adaptability}).
\nikhil{add assumptions}
\end{lemma}

\subsection{IPFE from Lattices}
\label{sec:ipfe-lwe}
We first recall the IPFE scheme by Agrawal et al.~\cite{ipe01}. Then, we show that it satisfies the additional compliance and robustness properties needed by our FAS construction.

Suppose $p$ is a 
% $\secparam$-bit 
prime number and suppose we want to compute inner products of vectors of length $\ell$.
Let the set of messages be $\mcal{M} = \Z_p^\ell$ and 
the function class be $\mcal{F}_{{\sf IP}, \ell, p} = \{ f_\bfy : \bfy \in \Z_p^\ell\}$, where $f_\bfy: \Z_p^\ell \to \Z_p$ is defined as $f_\bfy(\bfx) = \bfx^T \bfy \bmod{p}$.
Our IPFE scheme augmented with $\PubKGen$ algorithm is as in~\Cref{fig:ipfe-als}.

\begin{remark}
\label{remark:linear-independence}
\ifcameraready
We note that the original construction requires $\KGen$ to be stateful to support key generation queries that are linearly dependent modulo $p$, which is problematic since $\PubKGen$ would also need to be stateful, undermining the algorithm's public nature. To address this, we assume that all $\KGen$ queries to $\ipfe$ are linearly independent. While this restriction is acceptable for a single buyer, it limits support for multiple buyers. 
\ifcameraready
In the full version,
\else
In~\Cref{sec:fas-from-lattices-modified},
\fi 
we show how to assign unique IDs to each buyer and make the underlying $\ipfe$ key generation queries linearly independent.
\else
We note that the original construction needs $\KGen$ to be stateful in order to support key generation queries that are linearly dependent modulo $p$. This is problematic for us as the corresponding $\PubKGen$ will also have to be stateful then, but that defeats the purpose of this algorithm being public. So, in order to remedy the situation, we assume that all $\KGen$ queries to $\ipfe$ are going to be linearly independent. 
Naively instantiating our $\fas$ construction (\Cref{fig:fas-construction}) will result in $\fas$ only supporting linearly independent functions, which is acceptable if it involves a single buyer but in the most general case we would want to support multiple buyers who are agnostic of each other. Specifically, a buyer may engage in $\fas$ protocol to learn a function evaluation of the secret that is linearly dependent on something learnt by other some other buyers. To handle such a scenario, we show 
\ifcameraready
in the full version
\else 
in~\Cref{sec:fas-from-lattices-modified}
\fi
how to assign unique IDs to each buyer and make the underlying $\ipfe$ key generation queries linearly independent.
\fi
\end{remark}
\ifacm
	\begin{figure}[t]
	\ifhldiff
	{\color{hldiffcolor}
	\fi
	\centering
	\captionsetup{justification=centering}
	\begin{pcvstack}[boxed]

	\procedure[linenumbering, mode=text]{$\ipfe.\Gen(1^\secparam, p)$}{
	Let $n=\secparam$, 
	integers $m, k$, real $\alpha \in (0,1)$ as explained below
	\\ 
	Let $q = p^k$,
	\ret $\pp = (n, m, p, q, k, \alpha)$
	}

	\procedure[linenumbering, mode=text]{$\ipfe.\Setup(\pp, 1^\ell)$}{
	Sample $\bfA \getr \Z_q^{m \times n}$, $\bfZ \gets \tau$ where
	\pcskipln \\ 
	distribution $\tau$ over $\Z^{\ell \times m}$ is as explained below
	\\ 
	Let $\bfU := \bfZ \bfA  \in \Z_q^{\ell \times n}$,
	\ret $\mpk := (\bfA, \bfU), \msk := \bfZ$
	}

	\procedure[linenumbering, mode=text]{$\ipfe.\Enc(\mpk, \bfx \in \Z_p^\ell)$}{
	Sample $\bfs \getr \Z_q^n$, $\bfe_0 \gets D_{\Z, \alpha q}^m$, $\bfe_1 \gets D_{\Z, \alpha q}^\ell$ 
	\\ 
	Let $\ct_0 := \bfA \bfs + \bfe_0 \in \Z_q^m$,
	$\ct_1 := \bfU \bfs + \bfe_1 + p^{k-1} \bfx \in \Z_q^\ell$
	\\ 
	\ret $\ct := (\ct_0, \ct_1)$
	}

	\begin{pchstack}
	\procedure[linenumbering, mode=text]{$\ipfe.\KGen(\msk, \bfy \in \Z_p^\ell)$}{
	\ret $\sk_\bfy := \bfZ^T \bfy \in \Z^m$
	% Let internal state be $\state = \{ (\bfy_i, \overline{\bfy}_i, \sk_{\bfy_i}, pk_{\bfy_i})\}$
	% \\ 
	% If $\bfy$ is linearly independent of $\bfy_i$'s modulo $p$, 
	% \pcskipln \\ 
	% set $\overline{\bfy} := \bfy \in \Z^\ell$, $\sk_\bfy := \bfZ^T \overline{\bfy} \in \Z^m$, add $(\bfy, \overline{\bfy}, \sk_\bfy)$ to $\state$
	% \\ 
	% Else let $\bfy = \sum_i k_i \bfy_i \bmod{p}$ for $k_i$'s in $[0, p)$, 
	% \pcskipln \\ 
	% set $\overline{\bfy} := \sum_i k_i \overline{\bfy}_i \in \Z^\ell$, 
	% $\sk_\bfy := \sum_i k_i \bfz_i \in \Z^m$
	% \\ 
	% \ret $(\overline{\bfy}, \sk_\bfy)$
	}

	\procedure[linenumbering, mode=text]{$\ipfe.\PubKGen(\mpk, \bfy \in \Z_p^\ell)$}{
	\ret $\pk_\bfy := \bfU^T \bfy \in \Z_q^n$
	}
	\end{pchstack}

	\procedure[linenumbering, mode=text]{$\ipfe.\Dec(\sk_\bfy, \ct)$}{
	Let $d := \ct_1^T\bfy - \ct_0^T \sk_\bfy \bmod{q}$
	\\
	\ret $v \in \Z_p$ that minimizes $\mid p^{k-1} v - d \mid$
	% \pcskipln \\ \nikhil{specify rounding how}
	}

	\end{pcvstack}
	\caption{
	\ifhldiff
	{\color{hldiffcolor}
	\fi
	Agrawal et al.~\cite{ipe01} IPFE scheme augmented with $\PubKGen$ algorithm. Note that $\KGen$ here is stateless.
	\ifhldiff
	}
	\fi
	}
	\label{fig:ipfe-als}
	\ifhldiff
	}
	\fi
	\end{figure}
\else 
	\begin{figure}[H]
	\ifhldiff
	{\color{hldiffcolor}
	\fi
	\centering
	\captionsetup{justification=centering}
	\begin{pchstack}[boxed, space=0.1em]
	\begin{pcvstack}

	\procedure[linenumbering, mode=text]{$\ipfe.\Gen(1^\secparam, p)$}{
	Let $n=\secparam$, 
	integers $m, k$, real $\alpha \in (0,1)$ as defined below
	\\ 
	Let $q = p^k$,
	\ret $\pp = (n, m, p, q, k, \alpha)$
	}

	\procedure[linenumbering, mode=text]{$\ipfe.\Setup(\pp, 1^\ell)$}{
	Sample $\bfA \getr \Z_q^{m \times n}$, $\bfZ \gets \tau$ where
	\pcskipln \\ 
	distribution $\tau$ over $\Z^{\ell \times m}$ is as defined below
	\\ 
	Let $\bfU := \bfZ \bfA  \in \Z_q^{\ell \times n}$,
	\ret $\mpk := (\bfA, \bfU), \msk := \bfZ$
	}

	\procedure[linenumbering, mode=text]{$\ipfe.\Enc(\mpk, \bfx \in \Z_p^\ell)$}{
	Sample $\bfs \getr \Z_q^n$, $\bfe_0 \gets D_{\Z, \alpha q}^m$, $\bfe_1 \gets D_{\Z, \alpha q}^\ell$ 
	\\ 
	Let $\ct_0 := \bfA \bfs + \bfe_0 \in \Z_q^m$,
	$\ct_1 := \bfU \bfs + \bfe_1 + p^{k-1} \bfx \in \Z_q^\ell$
	\\ 
	\ret $\ct := (\ct_0, \ct_1)$
	}

	\end{pcvstack}
	\begin{pcvstack}

	\procedure[linenumbering, mode=text]{$\ipfe.\KGen(\msk, \bfy \in \Z_p^\ell)$}{
	\ret $\sk_\bfy := \bfZ^T \bfy \in \Z^m$
	% Let internal state be $\state = \{ (\bfy_i, \overline{\bfy}_i, \sk_{\bfy_i}, pk_{\bfy_i})\}$
	% \\ 
	% If $\bfy$ is linearly independent of $\bfy_i$'s modulo $p$, 
	% \pcskipln \\ 
	% set $\overline{\bfy} := \bfy \in \Z^\ell$, $\sk_\bfy := \bfZ^T \overline{\bfy} \in \Z^m$, add $(\bfy, \overline{\bfy}, \sk_\bfy)$ to $\state$
	% \\ 
	% Else let $\bfy = \sum_i k_i \bfy_i \bmod{p}$ for $k_i$'s in $[0, p)$, 
	% \pcskipln \\ 
	% set $\overline{\bfy} := \sum_i k_i \overline{\bfy}_i \in \Z^\ell$, 
	% $\sk_\bfy := \sum_i k_i \bfz_i \in \Z^m$
	% \\ 
	% \ret $(\overline{\bfy}, \sk_\bfy)$
	}

	\procedure[linenumbering, mode=text]{$\ipfe.\PubKGen(\mpk, \bfy \in \Z_p^\ell)$}{
	\ret $\pk_\bfy := \bfU^T \bfy \in \Z_q^n$
	}

	\procedure[linenumbering, mode=text]{$\ipfe.\Dec(\sk_\bfy, \ct)$}{
	Let $d := \ct_1^T\bfy - \ct_0^T \sk_\bfy \bmod{q}$
	\\
	\ret $v \in \Z_p$ that minimizes $\mid p^{k-1} v - d \mid$
	% \pcskipln \\ \nikhil{specify rounding how}
	}

	\end{pcvstack}
	\end{pchstack}
	\caption{
	\ifhldiff
	{\color{hldiffcolor}
	\fi
	Agrawal et al.~\cite{ipe01} IPFE scheme augmented with $\PubKGen$ algorithm. Note that $\KGen$ here is stateless.
	\ifhldiff
	}
	\fi
	}
	\label{fig:ipfe-als}
	\ifhldiff
	}
	\fi
	\end{figure}
\fi

\noindent\textbf{Parameter choices.}
Let $B_\tau$ be such that with probability at most $n^{-\omega(1)}$, each row of sample from $\tau$ has $\ell_2$-norm at least $B_\tau$. Then, the parameter constraints for correctness and security are as follows.
\begin{itemize}[leftmargin=*]
\item 
$\alpha^{-1} \geq \ell^2 p^3 B_\tau \omega(\sqrt{\log n})$, 
% \item 
$q \geq \alpha^{-1} \omega(\sqrt{\log n})$,
\item 
$ \tau = D^{\ell \times m/2}_{\Z,\sigma_1} \times (D_{\Z^{m/2},\sigma_2, \delta_1} \times \ldots \times D_{\Z^{m/2},\sigma_2, \delta_\ell})$, where $\delta_i \in Z^\ell$ denotes the $i$-th canonical vector, and the
standard deviation parameters satisfy $\sigma_1 = \Theta(\sqrt{n \log m} \max(\sqrt{m}, K'))$ and $\sigma_2 = \Theta(n^{7/2} m{1/2} \max(m, K'^2) \log^{5/2} m)$, with $K' = (\sqrt{\ell}p)^\ell$.
% \item 
% To ensure security based on LWEq,α0,m in dimension ≥ c · n for some c ∈ (0, 1) via Theorems 2 and 4 below, one may further impose that ` ≤ (1 c) · n and m = Θ(n log q), to obtain α0 = Ω(α/(n6K0 log2 q log5/2 n)).
\end{itemize}
Further, $R'_\isis$-robustness (~\Cref{lemma:ipfe-als-robust}) will require the constraint $\alpha^{-1} \geq 4p \omega(\sqrt{\log n}) (\ell p + m \beta_1)$, where $\beta_1$ is as chosen below.

\begin{lemma}[~\cite{ipe01}]
\label{lemma:ipfe-als-secure}
Suppose $\ell \leq n^{O(1)}$, $m \geq 4n \log_2(q)$ and $q, \alpha, \tau$ are as described above.
Suppose $\mheLWE_{q,\alpha,m,\ell,\tau}$ 
\ifcameraready
assumption 
\else
assumption (\Cref{def:mhe-lwe}) 
\fi
holds.
Then, the IPFE scheme in~\Cref{fig:ipfe-als} is selective, 
\ifcameraready
IND-secure.
\else
IND-secure (\Cref{def:ipfe-sel-ind-sec}).
\fi
\end{lemma}
% \textcolor{orange}{Aravind: Proof or citation for this lemma?}

Next, we show the IPFE scheme augmented with the above $\PubKGen$ algorithm satisfies $R_\isis$-compliance and $R'_\isis$-robustness, 
where $R_\isis = \isis_{n, m, q, \beta_0}$ and $R'_\isis = \isis_{n, m, q, \beta_1}$, where $\beta_0 = \ell p B_\tau$ and $\beta_1 = 2(\gamma - \kappa) - \beta_0$ s.t. $\beta_1 > \beta_0$ and $\gamma - \kappa - \beta_0 > 0$. 
% \nikhil{review $\beta_0$ and $\beta_1$ bounds after stabilizing IPFE and AS}

\begin{lemma}
\label{lemma:ipfe-als-compliant}
The IPFE scheme in~\Cref{fig:ipfe-als} is $R_\isis$-compliant, 
where $R_\isis = \isis_{n, m, q, \beta_0}$ and $\beta_0 = \ell p B_\tau$.
\end{lemma}
\begin{proof}
For any $\secparam \in \N$,
for any $\pp \gets \ipe.\Gen(1^\secparam)$ as implemented in~\Cref{fig:ipfe-als},
for any $(\mpk, \msk) \gets \ipe.\Setup(\pp, 1^\ell)$ as implemented in~\Cref{fig:ipfe-als}, 
for any $\bfy \in \mcal{F}_{{\sf IP}, \ell, p}$,
let $\pk_\bfy := \PubKGen(\allowbreak \mpk, \bfy)$ as implemented in~\Cref{fig:ipfe-als}, and 
$\sk_\bfy := \KGen(\msk, \bfy)$ as implemented in~\Cref{fig:ipfe-als}. 
Suppose here $\mpk = (\bfA, \bfU)$ and $\msk = \bfZ$. 
Then, $\sk_\bfy = \bfZ^T\bfy \in \Z^m$ and $\pk_\bfy = \bfU^T \bfy \bmod{q} = \bfA^T \sk_\bfy \bmod{q}$.
Further, note that $||\sk_\bfy||_\infty \leq \ell \cdot ||\bfZ||_\infty \cdot ||\bfy||_\infty  \leq \ell \cdot ||\bfZ||_2 \cdot ||\bfy||_\infty \leq \ell \cdot B_\tau \cdot p = \beta_0$. Hence, $(\pk_\bfy, \sk_\bfy) \in R_\isis$.
\end{proof}

\begin{lemma}
\label{lemma:ipfe-als-robust}
If $\alpha^{-1} \geq 4p \omega(\sqrt{\log n}) (\ell p + m \beta_1)$, then, 
the IPFE scheme in~\Cref{fig:ipfe-als} is $R'_\isis$-robust,
where $R'_\isis = \isis_{n, m, q, \beta_1}$ and $\beta_1 = 2 (\gamma - \kappa)-\beta_0$ such that $\beta_1>\beta_0$.
% \nikhil{fix $\beta_1$.}
\end{lemma}

\begin{proof}
For any $\secparam \in \N$ let $\pp \gets \ipfe.\Gen(1^\secparam)$ 
as implemented in~\Cref{fig:ipfe-als},
for any $\ell \in \N$, 
$\bfx \in \Z_p^\ell, \bfy \in \mcal{F}_{{\sf IP}, \ell, p}$, 
let $(\mpk, \msk) \leftarrow  \Setup(\pp, \ell)$
as implemented in~\Cref{fig:ipfe-als},
let $\pk_\bfy = \PubKGen(\msk, \bfy)$
as implemented in~\Cref{fig:ipfe-als},
let $\ct \leftarrow \Enc(\mpk, \bfx)$ 
as implemented in~\Cref{fig:ipfe-als}.
Then, for any 
$\sk'_\bfy$ such that $(\pk_\bfy, \sk'_\bfy) \in R'_\isis$, 
we want that $\Dec(\sk'_\bfy, \ct)$ outputs $\bfx^T\bfy \bmod{q}$.
Observe that the value $d$ computed by $\Dec$ 
implemented in~\Cref{fig:ipfe-als} simplifies to 
$d = p^{k-1} \bfx^T \bfy + (\bfe_1^T \bfy - \bfe_0^T \sk'_\bfy) \bmod{q}$.
% \begin{align*}
% d & := \ct_1^T\bfy - \ct_0^T \sk'_\bfy \bmod{q} \\ 
% & = (\bfU \bfs + \bfe_1 + p^{k-1} \bfx)^T \bfy - (\bfA \bfs + \bfe_0)^T \sk'_\bfy \bmod{q}\\
% % & = (\bfU \bfs + \bfe_1 + p^{k-1} \bfx)^T \bfy - \bfs^T \bfA^T \sk'_\bfy - \bfe_0^T \sk'_\bfy \bmod{q}\\
% % & = (\bfU \bfs + \bfe_1 + p^{k-1} \bfx)^T \bfy - \bfs^T \pk_\bfy - \bfe_0^T \sk'_\bfy \bmod{q}\\
% % & = (\bfU \bfs + \bfe_1 + p^{k-1} \bfx)^T \bfy - \bfs^T \bfU^T \bfy - \bfe_0^T \sk'_\bfy \bmod{q}\\
% & = p^{k-1} \bfx^T \bfy + (\bfe_1^T \bfy - \bfe_0^T \sk'_\bfy) \bmod{q}
% \end{align*}
Observe that 
\begin{align*}
|\bfe_1^T \bfy - \bfe_0^T \sk_\bfy|
% & \leq |\bfe_1^T \bfy| + |\bfe_0^T \sk_\bfy| \\
& \leq \ell p \alpha q \omega(\sqrt{\log n}) + m \alpha q \omega(\sqrt{\log n}) ||\sk'_\bfy||_\infty \\ 
& = \alpha q \omega(\sqrt{\log n}) (\ell p + m \beta_1).
\end{align*}
For decryption correctness to hold, we need that $|\bfe_1^T \bfy - \bfe_0^T \sk_\bfy| \leq q/4p$. Hence it suffices to have $\alpha^{-1} \geq 4p \omega(\sqrt{\log n}) (\ell p + m \beta_1)$.
\end{proof}

\subsection{$\fas$ Construction}

Instantiating $\fas$ construction in~\Cref{sec:fas-construction} with $\ipfe$ from~\Cref{fig:ipfe-als} and $\as$ from~\Cref{fig:lyu-as}, we obtain the following corollary.

\begin{corollary}
\label{corollary:fas-strongly-secure-lattices}
\label{thm:fas-strongly-secure-lattices}
Let $p$ be a $\secparam$-bit prime number and let $\ell$ be an integer.
Let $\mcal{M}=\Z_p^\ell$ be an additive group.
\nikhil{remove additive group from paper everywhere and keep it to $\Z_p^\ell$ everywhere.}
Let $\mcal{F}_{{\sf IP}, \ell, p}$ be the function family for computing inner products of vectors in $\Z_p^\ell$. 
Let $R$ be any NP relation 
with statement/witness pairs $(X, \bfx)$ such that $\bfx \in \Z_p^\ell$.
Let $R_\isis = \isis_{n, m, q, \beta_0}$ and $R'_\isis = \isis_{n, m, q, \beta_1}$, where $\beta_0 = \ell p B_\tau$ and $\beta_1 = 2(\gamma - \kappa) - \beta_0$ s.t. $\beta_1 > \beta_0$ and $\gamma - \kappa - \beta_0 > 0$. 
Suppose that 
\begin{itemize}
\item 
All function queries $\bfy$ to $\AuxGen$ are linearly independent modulo $p$,
\item
$R$ is $\mcal{F}_{{\sf IP}, \ell, p}$-hard (\Cref{def:f-hard-relation}),
\item
$\nizk$ is a secure NIZK argument system (\Cref{def:nizk}),
\item
$\as$ construction in~\Cref{fig:lyu-as} is an adaptor signature scheme w.r.t.\ digital signature scheme $\Lyu$ and hard relations  $R_\isis$ and $R'_\isis$ that satisfies weak pre-signature adaptability and witness extractability (\Cref{lemma:pq-as}),
\item
$\ipfe$ construction in~\Cref{fig:ipfe-abdp} is a selective, IND-secure IPFE scheme (\Cref{lemma:ipfe-als-secure}) for function family $\mcal{F}_{{\sf IP}, \ell+1}$ that is $R_\isis$-compliant (\Cref{lemma:ipfe-als-compliant}) and $R'_\isis$-robust (\Cref{lemma:ipfe-als-robust}). 
\end{itemize} 
Then, the functional adaptor signature scheme 
w.r.t.\ Lyubashevsky signature scheme $\Lyu$, NP relation $R$, and family of inner product functions $\mcal{F}_{{\sf IP}, \ell, p}$  
constructed in~\Cref{sec:construction} is strongly-secure (\Cref{def:fas-strongly-secure}).
\end{corollary}

The proof of the above corollary is immediate from the proof of the~\cref{thm:fas-strongly-secure} concerning the generic construction, and~\cref{lemma:ipfe-als-secure,lemma:ipfe-als-compliant,lemma:ipfe-als-robust} that show that the $\ipfe$ scheme in~\Cref{fig:ipfe-als} has the required properties.

% \textcolor{orange}{Informally describe why the theorem holds.}
\begin{remark}
Note that the linear independence modulo $p$ restriction for all $\AuxGen$ queries in this lattice-based instantiation comes due to the same restriction on the underlying IPFE scheme as discussed in~\Cref{remark:linear-independence}. In practice,  $\AuxGen$ queries can be sent by any buyer to the seller. On one hand it's reasonable to assume that all the queries by a single buyer are going to be linearly independent modulo $p$ because by linearity of the functionality, the buyer can locally compute the function evaluation for some function $\bfy$ that is linear combination of the functions $\bfy_i$'s it previously queried as follows: if $\bfy = \sum_i k_i \bfy_i \bmod{q}$, then, $f_\bfy(\bfx)$ can be computed as $f_\bfy(\bfx) := \sum_i k_i f_{\bfy_i}(\bfx)$. On the other hand, it's unreasonable to assume that all buyers are working in coordination and make sure that their queries are linearly independent. In~\Cref{sec:fas-from-lattices-modified}, we show how to remove this restriction on FAS.
\end{remark}

\subsection{Modifying $\fas$ Construction in~\Cref{fig:fas-construction}.}
\label{sec:fas-from-lattices-modified}
There is a simple way to ensure that even if $\AuxGen$ queries across multiple buyers are linearly dependent modulo $p$, the requests to the underlying $\ipfe.\PubKGen$ are always linearly independent. Assume a polynomial bound $k$ on the number of buyers. Then, we add $k$ additional slots to the underlying IPFE resulting in a total of $\ell_k+1$ slots and modify the $\fas$ construction as follows: 
\begin{itemize}
\item 
In $\AdGen$ on input $\bfx \in \Z_p^\ell$, the vector $\widetilde{\bfx} := (\bfx^T, \bfnum{0}^T)^T \in \Z_p^{\ell+k+1}$ is encrypted using $\ipfe.\Enc$. 
\item 
Further, the seller maintains an internal map of size $k$ from buyer verification keys $vk$'s to a unique unit vector assigned to them. In other words, $i$-th buyer with verification key denoted by $\vk_i$ is assigned $i$-th unit vector $\bfe_i \in \{0,1\}^k$. 
\item 
At time of $\AuxGen$ query for vector $\bfy \in \Z_p^\ell$ by $i$-th buyer, the seller uses $\widetilde{\bfy} := (\bfy^T, f_\bfy(\bft), \bfe_i^T)^T \in \Z_p^{\ell+k+1}$ to compute $\aux_\bfy$ and sets $\pi_\bfy := (f_\bfy(\bft), i)$.
\item 
In $\AuxVerify$, the buyer can recompute $\widetilde{\bfy}$ on its own using $\pi_\bfy = (f_\bfy(\bft), i)$ and return $1$ iff $\aux_\bfy = \ipfe.\PubKGen(\mpk, \widetilde{\bfy})$.
\item 
At the time of running $\Adapt$, the seller looks up $\vk$ in its internal map to find the index $i$ assigned to it. Then, it can recompute the same $\widetilde{\bfy}$ that it used during $\AuxGen$ query $\bfy$ by buyer $\vk$. This will ensure that the adapted signature passes the signature verification. 
\end{itemize}
Note that, the $k$ extra slots as used above do not affect correctness as the function evaluation for a function $\bfy$ is always $\widetilde{\bfx}^T\widetilde{\bfy} = \bfx^T\bfy \bmod{q}$. Further note that these slots remain the same in all experiments for all security properties.

\begin{remark}[Linear Independence]
Observe that with the above modification, even if two buyers $i$ and $j$ requested for the same function $\bfy \in \Z_p^\ell$, $\AuxGen$ will use linearly independent queries $\widetilde{\bfy}_i$ and $\widetilde{\bfy}_j$ to the underlying $\ipfe.\PubKGen$, where 
$\widetilde{\bfy}_i := (\bfy^T, f_\bfy(\bft), \bfe_i^T)^T \in \Z_p^{\ell+k+1}$ and 
$\widetilde{\bfy}_j := (\bfy^T, f_\bfy(\bft), \bfe_j^T)^T \in \Z_p^{\ell+k+1}$. These two are linearly independent because for $i\neq j$, the unit vectors $\bfe_i$ and $\bfe_j$ are linearly independent of each other.
\end{remark}

\begin{remark}[Communication overhead]
Note that in the IPFE construction in~\Cref{fig:ipfe-als}, the size of $\pk_\bfy$ and $\sk_\bfy$ are independent of the vector length parameter $\ell$. Hence, change from $\ell+1$ to $\ell+k+1$ does not increase the size $\aux_\bfy$. The size of $\pi_\bfy$ increases by $\log{k}$ bits, but that can also optimized to be sent only once per buyer and not for every $\AuxGen$ query. The main communication overhead is in the ciphertext $\ct$ communicated as part of the advertisement $\advt$. Now $\ct \in \Z_q^{m+\ell_k+1}$instead of $\ct \in \Z_q^{m+\ell+1}$. But fortunately, $\advt$ needs to be communicated only once.
\end{remark}

%% file: performance-evaluation.tex
\section{Performance Evaluation}\label{sec:implementation}
% \anote{Describe system setup: what computer did we use, what is the RAM, CPU/GPU, SSD, etc. Then what language, parameter choices, github repo that we used. Link for our open source implementation (anonymized). Then report numbers in a table/ graph. Each table must have a proper caption and referred to in the text. The columns must be annotated (what unit of measurement are we using).}

% \anote{Please put this table in the performance evaluation section at the end of the main body. The last column is unnecessary. The time columns should be marked in msec or sec.}

\ifhldiff
{\color{hldiffcolor}
\fi
Our implementation aims to (i) show that FAS can replace smart contracts for efficient functional sales, and (ii) benchmark the computational costs for each functional sale.
\ifhldiff
}
\fi
We provide
an open-source implementation~\cite{fas-impl} of our prime-order group-based strongly-secure $\fas$ in Python. 
% \nikhil{add link to github}
We also perform a series of benchmarks on our implementation for a wide range of parameters.
Our results show that our scheme is practical for variety of real-world scenarios.

We measured the costs on a Apple MacBook Pro with M2 chip, 16GB memory, 8 cores (4 cores @3.49 GHz and 4 cores @2.42 GHz).
\ifhldiff
{\color{hldiffcolor}
\fi
In typical applications, a seller wants to sell multiple functions of the same witness. The (application dependent) advertisement – containing NIZK proof – is just a one-time cost, while other processes are done once per sale. So, we benchmark computation costs of all algorithms except $\AdGen$ and $\AdVerify$.
\ifhldiff
}
\fi
% For example, selling function evaluations of vectors of length 10,000 completes in about 5 seconds.  

We use the ${\sf Secp256k1}$ elliptic curve that is used by Bitcoin~\cite{secp256k1} and other cryptocurrencies. 
We use the curve's python implementation from ${\sf hanabi1224}$~\cite{hanabi-python-secp256k1}.
For Schnorr signatures, we use a modified version 
of BIP-340 reference implementation~\cite{bip340}. 
We implement the Schnorr adaptor signatures~\cite{asig} on top of it. 
For compliance purpose, we implement the IPFE scheme in~\Cref{fig:ipfe-abdp}
% by Abdalla et al.~\cite{ipe} 
on the ${\sf Secp256k1}$ curve. 
Finally, using these adaptor signatures and IPFE schemes, we implement our functional adaptor signature construction. 
We do not implement the NIZKs as they are needed only for $\AdGen$ and $\AdVerify$ which we do not implement here. 

\smallskip\noindent\textbf{Optimizations.}
We make following implementation optimizations.
\begin{itemize}[leftmargin=*]
\item 
For vectors of length $\ell$, we parallelize $\ipfe.\Setup$, $\ipfe.\Enc$. 
\item 
Recall that $\ipfe.\Dec$ computes $\ct_1^T\bfy - \ct_0^T\sk_\bfy$ (See~\Cref{fig:ipfe-abdp}). Here, $\ct_1$ and $\bfy$ are vectors of length $\ell$ and $\ct_1^T \bfy$ is computed as $\prod_{i \in \ell} c_i^{y_i}$, where $\ct_1 = (c_1, \ldots, c_\ell)$ is a vector of group elements and $\bfy = (y_1, \ldots, y_\ell)$ is a vector of scalars. Thus, computing $\ct_1^T\bfy$ is expensive as it involves $\ell$ group exponentiation operations. But, this computation can be done in an offline stage by the buyer as it does not require knowledge of $\sk_\bfy$ from the seller. This helps us make $\ipfe.\Dec$ and thus, $\fas.\FExt$ very efficient in practice. 
\item 
\ifhldiff
{\color{hldiffcolor}
\fi
Note that $\ipfe.\PubKGen$ and $\ipfe.\Dec$ involve computing a product of powers such as $\prod_{i \in \ell} k_i^{y_i}$ and $\prod_{i \in \ell} c_i^{y_i}$ respectively. We compute these via ${\sf FastMult}$ algorithm of~\cite[Section 3.2]{bellare1998fast}.
\ifhldiff
}
\fi
\ifcameraready\else

\item 
Note that $\ipfe.\Dec$ involves a discrete log computation. For this task, we implement the baby-step giant-step algorithm~\cite{bsgs}.
\ifhldiff
{\color{hldiffcolor}
\fi
\item
Note that $\AuxGen$ involves computing $\pk_\bfy := \ipfe.\PubKGen ( \allowbreak \mpk, \widetilde{\bfy})$ which requires $\ell+1$ group exponentiation operations. As $\AuxGen$ is run by the seller who knows $\state = (\msk, \bft)$, hence, we optimize this computation as $\sk_\bfy := \ipfe.\KGen(\msk, \widetilde{\bfy})$, $\pk_\bfy := g^{\sk_\bfy}$. This resulting way only involves $1$ group exponentiation operation and hence is very efficient.
\ifhldiff
}
\fi

\fi
\end{itemize}

\ifhldiff
{\color{hldiffcolor}
\fi
\smallskip\noindent\textbf{Communication efficiency.}
In our implementation, each group element is 64 bytes elliptic curve point and each $Z_p$ element is 32 bytes. 
In~\Cref{fig:fas-comm}, we give bounds on communication cost for a witness x of dimension $\ell=10^6$ (total size 32 MB) and also specify if the communication is done on- or off-chain.
\begin{table}[t]
\ifhldiff
{\color{hldiffcolor}
\fi
\centering
\captionsetup{justification=centering}
\caption{
\ifhldiff
{\color{hldiffcolor}
\fi
Communication efficiency of $\fas$ for $\ell=10^6$. 
Last column denotes if a secure channel for buyer/seller is used. 
\ifhldiff
}
\fi
}
\ifacm
	\begin{tabular}{ |m{3.5cm}|m{1.1cm}|m{1.0cm}|m{1.5cm}|  }
\else 
	\begin{tabular}{ |m{5.0cm}|m{2.0cm}|m{2.0cm}|m{3.0cm}|  }
\fi
	\hline
	{\bf Variable} & {\bf Size} & {\bf On-chain?} & {\bf Secure channel?} \\
	\hline
	$\advt = (\ipfe.\mpk, \ipfe.\ct)$ & $128$ MB & no\tablefootnote{for example, $\advt$ can be published on a website.
	} & no \\
	Function $\bfy$ & $32$ MB & no & yes \\
	$(\aux_\bfy, \pi_\bfy)$ & $96$ bytes & no & yes \\
	Pre-signature & $64$ bytes & no & yes \\
	Adapted-signature & $64$ bytes & yes & no \\
	\hline
	\end{tabular}
\label{fig:fas-comm}
\ifhldiff
}
\fi
\end{table}
% \newcommand{\mc}[2]{\multicolumn{#1}{c}{#2}}
% \ifhldiff
% \newcolumntype{a}{>{\columncolor{hldiffcolor}}c}
% \else 
% \newcolumntype{a}{>{\columncolor{white}}c}
% \fi
% \newcolumntype{b}{>{\columncolor{white}}c}
\ifacm
	\begin{table}[t]
	\centering
	\captionsetup{justification=centering}
	\caption{Computational efficiency of $\fas$. }
	% \ifhldiff
	% \begin{tabu}{ |m{0.7cm}|m{0.8cm}||m{0.6cm}|m{0.6cm}|m{0.7cm}|m{0.7cm}|m{0.7cm}|m{0.7cm}|  }
	% \begin{tabu}{ |m{0.7cm}|m{0.8cm}||m{0.6cm}|>{\columncolor[hldiffcolor]{0.8}}m{0.6cm}|m{0.7cm}|m{0.7cm}|m{0.7cm}|m{0.7cm}|  }
	% \else
	\begin{tabu}{ |m{0.7cm}|m{0.8cm}||m{0.6cm}|m{0.6cm}|m{0.7cm}|m{0.7cm}|m{0.7cm}|m{0.7cm}|  }
	% \fi
	\hline
	\multicolumn{2}{|c||}{\bf Params} & \multicolumn{6}{c|}{\bf Running Times (seconds)
	\tablefootnote{${\sf AuxG}$, ${\sf AuxV}$, ${\sf FPS}$, ${\sf FPV}$  denote $\AuxGen$, $\AuxVerify$, $\FPreSign$, $\FPreVerify$.}} \\
	\hline
	$\ell$ & $B$ & ${\sf AuxG}$ & ${\sf AuxV}$ & ${\sf FPS}$ & ${\sf FPV}$ & $\Adapt$ & $\FExt$ \\
	\hline
	$1$ & $10^{6}$ & $0.003$ & 
	\ifhldiff
	{\color{hldiffcolor}$0.003$}
	\else
	$0.003$ 
	\fi
	& $0.008$ & 
	\ifhldiff
	{\color{hldiffcolor}$0.008$}
	\else
	$0.008$ 
	\fi
	& $0.013$ & $0.082$\\
	$10^2$ & $10^{6}$ & $0.003$ & 
	\ifhldiff
	{\color{hldiffcolor}$0.300$}
	\else
	$0.300$ 
	\fi
	& $0.013$ & 
	\ifhldiff
	{\color{hldiffcolor}$0.418$}
	\else
	$0.418$ 
	\fi
	& $0.021$ & $0.142$\\
	$10^2$ & $10^{8}$ & $0.003$ & 
	\ifhldiff
	{\color{hldiffcolor}$0.332$}
	\else
	$0.332$ 
	\fi
	& $0.014$ & 
	\ifhldiff
	{\color{hldiffcolor}$0.477$}
	\else
	$0.477$ 
	\fi 
	& $0.023$ & $0.722$\\
	$10^4$ & $10^{8}$ & $0.011$ & 
	\ifhldiff
	{\color{hldiffcolor}$0.323$}
	\else
	$0.323$ 
	\fi 
	& $0.010$ & 
	\ifhldiff
	{\color{hldiffcolor}$0.424$}
	\else
	$0.424$ 
	\fi 
	& $0.035$ & $1.025$\\
	$10^4$ & $10^{10}$ & $0.011$ & 
	\ifhldiff
	{\color{hldiffcolor}$0.369$}
	\else
	$0.369$ 
	\fi 
	& $0.009$ & 
	\ifhldiff
	{\color{hldiffcolor}$0.352$}
	\else
	$0.352$ 
	\fi 
	& $0.022$ & $9.952$\\
	$10^5$ & $10^{11}$ & $0.092$ & 
	\ifhldiff
	{\color{hldiffcolor}$2.067$}
	\else
	$2.067$ 
	\fi 
	& $0.009$ & 
	\ifhldiff
	{\color{hldiffcolor}$2.101$}
	\else
	$2.101$
	\fi 
	& $0.111$ & $30.38$\\
	$10^6$ & $10^{12}$ & $0.879$ & 
	\ifhldiff
	{\color{hldiffcolor}$19.08$}
	\else
	$19.08$
	\fi 
	& $0.008$ & 
	\ifhldiff
	{\color{hldiffcolor}$18.41$}
	\else
	$18.41$
	\fi 
	& $0.871$ & $95.07$\\
	\hline 
	\hline
	\ifhldiff
	\rowfont{\color{hldiffcolor}}
	\fi
	$10^7$ & $10^{13}$ & $9.191$ & $223.4$ & $0.011$ & $217.9$ & $10.22$ & $317.8$\\
	\ifhldiff
	\rowfont{\color{hldiffcolor}}
	\fi
	$3 \cdot 10^7$ & $3 \cdot 10^{13}$ & $32.03$ & $766.6$ & $0.011$ & $740.4$ & $40.27$ & $452.2$\\
	\hline
	\end{tabu}
	\label{fig:fas-perf}
	\end{table}
\else 
	\begin{table}[t]
	\centering
	\captionsetup{justification=centering}
	\caption{Computational efficiency of $\fas$. }
	% \ifhldiff
	% \begin{tabu}{ |m{0.7cm}|m{0.8cm}||m{0.6cm}|m{0.6cm}|m{0.7cm}|m{0.7cm}|m{0.7cm}|m{0.7cm}|  }
	% \begin{tabu}{ |m{0.7cm}|m{0.8cm}||m{0.6cm}|>{\columncolor[hldiffcolor]{0.8}}m{0.6cm}|m{0.7cm}|m{0.7cm}|m{0.7cm}|m{0.7cm}|  }
	% \else
	\begin{tabu}{ |m{1.2cm}|m{1.2cm}||m{1.2cm}|m{1.4cm}|m{1.4cm}|m{1.6cm}|m{0.8cm}|m{0.8cm}|  }
	% \fi
	\hline
	\multicolumn{2}{|c||}{\bf Params} & \multicolumn{6}{c|}{\bf Running Times (seconds)}\\
	% \tablefootnote{${\sf AuxG}$, ${\sf AuxV}$, ${\sf FPS}$, ${\sf FPV}$  denote $\AuxGen$, $\AuxVerify$, $\FPreSign$, $\FPreVerify$.}} \\
	\hline
	$\ell$ & $B$ & $\AuxGen$ & $\AuxVerify$ & $\FPreSign$ & $\FPreVerify$ & $\Adapt$ & $\FExt$ \\
	\hline
	$1$ & $10^{6}$ & $0.003$ & 
	\ifhldiff
	{\color{hldiffcolor}$0.003$}
	\else
	$0.003$ 
	\fi
	& $0.008$ & 
	\ifhldiff
	{\color{hldiffcolor}$0.008$}
	\else
	$0.008$ 
	\fi
	& $0.013$ & $0.082$\\
	$10^2$ & $10^{6}$ & $0.003$ & 
	\ifhldiff
	{\color{hldiffcolor}$0.300$}
	\else
	$0.300$ 
	\fi
	& $0.013$ & 
	\ifhldiff
	{\color{hldiffcolor}$0.418$}
	\else
	$0.418$ 
	\fi
	& $0.021$ & $0.142$\\
	$10^2$ & $10^{8}$ & $0.003$ & 
	\ifhldiff
	{\color{hldiffcolor}$0.332$}
	\else
	$0.332$ 
	\fi
	& $0.014$ & 
	\ifhldiff
	{\color{hldiffcolor}$0.477$}
	\else
	$0.477$ 
	\fi 
	& $0.023$ & $0.722$\\
	$10^4$ & $10^{8}$ & $0.011$ & 
	\ifhldiff
	{\color{hldiffcolor}$0.323$}
	\else
	$0.323$ 
	\fi 
	& $0.010$ & 
	\ifhldiff
	{\color{hldiffcolor}$0.424$}
	\else
	$0.424$ 
	\fi 
	& $0.035$ & $1.025$\\
	$10^4$ & $10^{10}$ & $0.011$ & 
	\ifhldiff
	{\color{hldiffcolor}$0.369$}
	\else
	$0.369$ 
	\fi 
	& $0.009$ & 
	\ifhldiff
	{\color{hldiffcolor}$0.352$}
	\else
	$0.352$ 
	\fi 
	& $0.022$ & $9.952$\\
	$10^5$ & $10^{11}$ & $0.092$ & 
	\ifhldiff
	{\color{hldiffcolor}$2.067$}
	\else
	$2.067$ 
	\fi 
	& $0.009$ & 
	\ifhldiff
	{\color{hldiffcolor}$2.101$}
	\else
	$2.101$
	\fi 
	& $0.111$ & $30.38$\\
	$10^6$ & $10^{12}$ & $0.879$ & 
	\ifhldiff
	{\color{hldiffcolor}$19.08$}
	\else
	$19.08$
	\fi 
	& $0.008$ & 
	\ifhldiff
	{\color{hldiffcolor}$18.41$}
	\else
	$18.41$
	\fi 
	& $0.871$ & $95.07$\\
	\hline 
	\hline
	\ifhldiff
	\rowfont{\color{hldiffcolor}}
	\fi
	$10^7$ & $10^{13}$ & $9.191$ & $223.4$ & $0.011$ & $217.9$ & $10.22$ & $317.8$\\
	\ifhldiff
	\rowfont{\color{hldiffcolor}}
	\fi
	$3 \cdot 10^7$ & $3 \cdot 10^{13}$ & $32.03$ & $766.6$ & $0.011$ & $740.4$ & $40.27$ & $452.2$\\
	\hline
	\end{tabu}
	\label{fig:fas-perf}
	\end{table}
\fi
Using standard compression techniques outlined in BIP-340~\cite{bip340}, advertisement size can be reduced to 66 MB by encoding 64 bytes elliptic curve points using 33 bytes. In practice, the size of $\advt = (\IPFE.\mpk, \IPFE.\ct)$ can be amortized to 33 MB by reusing the same $\IPFE.\mpk$ for all advertisements by a seller. The multi-message security of IPFE scheme should preserve FAS security and result in optimal amortized advertisement size of ~1.03x the witness size.
\ifhldiff
}
\fi

\smallskip\noindent\textbf{Computational efficiency.}
We share the performance numbers in~\Cref{fig:fas-perf}. 
Here, $\ell$ denotes the length of the witness vector.
Each entry of the vector is 32 Bytes integer. Thus, $\ell=10^6$ implies that the witness is of size 32MB.
We note that the pre-dominant cost in the implementation is the group exponentiation operation which hasn't been optimized here. 
We note that for parameter $\ell$, the number of group exponentiation operations in each algorithm are as follows: 
\ifhldiff
{\color{hldiffcolor}
\sout{$\ell+1$}
}
{\color{blue}
$1$ 
}
\fi
in $\AuxGen$,
$\ell+1$ in $\AuxVerify$,
$1$ in $\FPreSign$,
$\ell+2$ in $\FPreVerify$,
$0$ in $\Adapt$,
$\ell+2$ in $\FExt$.
Asymptotically, $\AuxVerify$ and $\FPreVerify$ run in time $O(\ell)$ and $\FExt$ runs in time $O(\sqrt{\ell})$.
But, concretely $\FExt$ is the slowest for practical scenarios as highlighted in~\Cref{fig:fas-perf}. 
The last two rows of~\Cref{fig:fas-perf} do not highlight practical scenarios as the costs of these algorithms are higher than $100$ seconds each, but we benchmark them to understand at what point do the concrete running times of $\AuxVerify$ and $\FPreVerify$ start becoming a bottleneck.

%% file: fas-proof-appendix.tex
\section{Proofs for Generic Construction of FAS}
\label{sec:fas-proof-appendix}

Here we present the formal proofs for the generic FAS construction from~\cref{sec:fas-construction}.
\subsection{Correctness}
First, we give the proof of correctness.

\begin{proof}[Proof of~\cref{lemma:correctness-fas}]
For every 
$\secparam \in \N$, every message $m \in \{0,1\}^*$, 
every statement/witness pair $(X, \bfx) \in R$, 
and every function $\bfy \in \mcal{F}_{{\sf IP}, \ell}$, suppose 
$\pp \gets \Setup(1^\secparam),
(\advt, \state) \gets \AdvertisementGen(\pp, X, \bfx),\allowbreak
(\sk, \vk) \gets \KGen(1^\secparam),
(\aux_\bfy, \pi_\bfy) \gets \AuxGen(\advt, \state, \bfy),
\widetilde{\sigma} \gets \allowbreak \FPreSign(\advt, \sk, $ $m, X, \bfy, \aux_\bfy),
\sigma := \Adapt(\advt, \state, \vk, m, X,\allowbreak \bfx, \bfy, \aux_\bfy \widetilde{\sigma}),
z := \FExt$ $(\advt, \widetilde{\sigma}, \sigma,  X, \bfy, \aux_\bfy)
$. Then,
\begin{itemize}[leftmargin=*]
\item 
From NIZK correctness, it follows that $\AdvertisementVerify(\pp, X, \advt) = 1$. 
\item 
For determinism of $\ipfe.\PubKGen$, it follows that $\AuxVerify(\advt, $ $ \bfy, \aux_\bfy, \pi_\bfy) = 1$.
\item
From correctness of $\AS$, it follows that $\ \FPreVerify(\advt, \vk, m, X, $ $\bfy, \aux_\bfy, \pi_\bfy, \widetilde{\sigma}) = 1$.
\item
Recall that $\widetilde{\sigma}$ is an $\AS$ pre-signature w.r.t.\ the statement $\pk_\bfy$.
Further, recall that $\sk_\bfy$ is the witness computed by the $\Adapt$ algorithm and used by $\AS.\Adapt$ to compute $\sigma$. 
From $R_\ipfe$-compliance of IPFE, it follows that $(\pk_\bfy, \sk_\bfy) \in R_\ipfe$. Hence, by the correctness of $\AS$, it follows that
$\Verify(\vk, m, \sigma) = 1$. 

\item
From $\AS$ correctness, it follows that $(\pk_\bfy, z) \in R'_\ipfe$, where $z$ is the extracted witness in the $\FExt$ algorithm. Then, from $R'_\ipfe$-robustness of IPFE, it follows that $v = f_\bfy(\bfx)$, where $v$ is output of $\FExt$.  
\end{itemize}

\end{proof}

In the subsequent sub-sections, we prove the lemmas regarding advertisement soundness, unforgeability, pre-signature adaptability, witness extractability, zero-knowledge.
% zero-knowledge (for ~\Cref{thm:fas-strongly-secure}), and witness indistinguishability (~\Cref{thm:fas-weakly-secure}).

\input{fas-proof-advertisement-soundness}
\input{fas-proof-pre-sig-validity}

\input{fas-proof-unforgeability-newer}

\input{fas-proof-func-wit-ext}
\input{fas-proof-pre-sig-adaptability}
\input{fas-proof-zk}

%% file: fas-proof-advertisement-soundness.tex
\subsection{Advertisement soundness}
% \nikhil{change to be compatible with the computational definition as in~\Cref{def:fas-ad-sound}}
\begin{lemma}
Suppose NIZK satisfies adaptive soundness. Then, the functional adaptor signature construction in~\Cref{sec:construction} is advertisement sound.
\label{lemma:fas-ad-sound}
\end{lemma}

\begin{proof}
We show that if a \ppt adversary \A breaks the advertisement soundness of our functional adaptor signature construction with non-negligible advantage, 
then, there exists a \ppt reduction \B that can break the adaptive soundness of the underlying NIZK scheme with the same non-negligible advantage.
Suppose the challenger for adaptive soundness of NIZK is \C.
The reduction \B is as follows.
\begin{itemize}
\item 
The reduction \B obtains $\crs$ from the challenger \C, computes $\pp' \gets \ipfe.\Gen(1^\secparam)$, and sends $\pp=(\crs, \pp')$  to \A.
\item 
The adversary \A sends public advertisement $\advt$ and a statement $X$ to the reduction \B. 
% If $\AdvertisementVerify(\pp, X, \advt) = 0$, \B aborts.
$\B$ parses $\advt=(\mpk, \ct, \pi)$ and sends the statement $(X, \pp, \mpk, \ct)$ and proof $\pi$ to the challenger \C.
\end{itemize}

Observe that if \A breaks advertisement soundness, then, 
$X \notin L_R$ and $\AdvertisementVerify(\pp, X, \advt)=1$.
From the definition of $L_\nizk$ it follows that if $X \notin L_R$, then, $(X, \pp, \mpk, \ct) \notin L_\nizk$. 
Further, recall that the implementation of $\AdvertisementVerify$ simply involves a $\NIZK$ proof verification. 
Therefore, if $\AdvertisementVerify(\pp, X, \advt) = 1$, then, $\NIZK.\Verify(\crs, (X, \pp, \mpk, \ct), \pi) = 1$. 
Hence, whenever \A breaks advertisement soundness of \fas, \B breaks adaptive soundness of \nizk.

\end{proof}

%% file: fas-proof-pre-sig-validity.tex
\subsection{Pre-Signature Validity}

\begin{lemma}
\label{lemma:fas-pre-sig-validity}
Suppose the adaptor signature scheme satisfies correctness (\Cref{def:as-correctness}).
Then, the functional adaptor signature construction in~\Cref{fig:fas-construction} satisfies pre-signature validity.
\end{lemma}

\begin{proof}
For 
any $\secparam \in \N$, 
any $\pp \gets \Setup(1^\secparam)$ as computed in~\Cref{fig:fas-construction}, 
any $X, \advt$ such that $\AdvertisementVerify(\pp, \advt, X)=1$ as computed in~\Cref{fig:fas-construction},
any message $m\in \{0,1\}^*$,
any function $\bfy \in \mcal{F_{\sf IP}, \ell}$,
any $(\aux_\bfy, \pi_\bfy)$, 
any key pair $(\sk, \vk) \gets \KGen(1^\secparam)$,
any pre-signature $\widetilde{\sigma} \gets \FPreSign(\advt, \sk, m, X, \bfy, \aux_\bfy)$ as computed in~\Cref{fig:fas-construction}, 
suppose $\AuxVerify(\advt, \bfy, \aux_\bfy, \pi_\bfy)=1$. 
Then, to prove that $\Pr[\FPreVerify(\allowbreak \advt, \vk, m, X, f, \aux_f, \pi_f, \widetilde{\sigma})]=1$ (as computed in~\Cref{fig:fas-construction}), it suffices to show that $\Pr[\AS.\PreVerify(\vk, \allowbreak m, \aux_\bfy, \widetilde{\sigma})]=1$. This follows from the correctness of the adaptor signature scheme $\AS$.
\end{proof}

%% file: fas-proof-unforgeability-newer.tex
\subsection{Unforgeability}
% \nikhil{update experiments and proofs to reflect changes related to $\aux$ info. See ~\Cref{fig:fasig-unf-exp}.}
% \nikhil{put $\AuxGen$, $\AuxVerify$ and update $\FPreSign$ and $\FPreVerify$ to take $\aux_f$ as inputs.}

\begin{lemma}
\label{lemma:fas-unf}
% Let $\mcal{F}_{{\sf IP}, \ell} = \{\bfy \in \Z_p^\ell \ s.t. \ \bfy \neq {\bf 0}\}$.
% \nikhil{this function class needs to be a little more complex than this. For eg, if $X$ trivially reveals the first bit $x_1$ of the witness, then, the function class should not include vectors $\bfy$ in the span of $(1, 0, \ldots, 0)$.} \anote{Why can't $\bfy$ be ${\bf 0}$?}
Suppose the NIZK argument system $\nizk$ satisfies adaptive soundness ((\Cref{def:nizk})), 
the Adaptor Signature scheme $\as$ satisfies witness extractability (\Cref{def:as-wit-ext}), 
the relation $R$ is a $\mcal{F}_{{\sf IP}, \ell}$-hard (\Cref{def:f-hard-relation}), and 
% \anote{Is this notion defined somewhere?}
$\ipe$ scheme satisfies correctness and $R'_\ipfe$-robustness (\Cref{def:ipfe-robust}).
Then, the functional adaptor signature instantiation above is $\faSigForge$-secure.
\end{lemma}

\begin{proof}
We prove the lemma by defining a sequence of games.

\paragraph{Game $G_0$:} It is the original game $\faSigForge_{\A, \FAS}$, where the adversary \A has to come up with a valid forgery on a message $m^*$ of his choice, while having access to functional pre-sign oracle $\mcal{O}_{\fpS}$ and sign oracle $\mcal{O}_S$. 
% Since we are in the random oracle model, the adversary additionally has access to a random oracle $H$.
Game $G_0$ is formally defined in~\Cref{fig:fasig-unf-games}.

\paragraph{Game $G_1$:}
same as game $G_0$, except that 
before computing the pre-signature, 
the game checks if the NIZK statement $(X^*, \pp', \mpk, \ct)$ is in the language $L_{NIZK}$. 
If no, the game sets the flag ${\sf Bad}_1 = \true$.
% \nikhil{note that this check is inefficient. Is that a problem?}
Game $G_1$ is formally defined in~\Cref{fig:fasig-unf-games}.

\paragraph{Game $G_2$:}
same as game $G_1$, except that 
when \A outputs the forgery $\sigma^*$, 
the game extracts the witness $z$ of the underlying adaptor signature scheme for the statement $\pk_{\bfy^*}$ and checks if the NIZK statement $(\aux_\bfy^*, z)$ satisfies $R'_\ipfe$. 
If no, the game sets the flag ${\sf Bad}_2 = \true$.
% \nikhil{note that this check is inefficient. Is that a problem?}
Game $G_2$ is formally defined in~\Cref{fig:fasig-unf-games}.

\ifacm
	\begin{figure}[t]
	\centering
	\captionsetup{justification=centering}
	\begin{gameproof}
	\begin{pcvstack}[boxed, space=1em]
	\begin{pcvstack}[space=1em]
	\procedure[linenumbering, mode=text]{Games $G_0$, \pcbox{G_1}, {\color{blue} \pcbox{ \color{black} G_2}}}{
	$\mcal{Q} := \emptyset$ 
	\\ 
	$\crs \gets \nizk.\Setup(1^\secparam)$
	\\ 
	$\pp' \gets \ipfe.\Gen(1^\secparam)$ 
	% \\ 
	% \nikhil{remove GroupGen specification.}
	% \pcskipln \\ \nikhil{keep it generic for the generic construction.}
	\\ 
	$\pp := (\crs, \pp')$
	\\ 
	$(\sk, \vk) \gets \KGen(\pp')$ 
	\\ 
	$(X^*, \bfx^*) \gets \GenR(1^\secparam)$ 
	\\ 
	$(\advt, m^*, \bfy^*, \aux_\bfy^*, \pi^*_\bfy) \gets \A^{\mcal{O}_S(\cdot)}(\pp, \vk, X^*)$ 
	\\ 
	Parse $\advt = (\mpk, \ct, \pi)$, let $\stmt := (X^*, \pp', \mpk, \ct)$
	\\ 
	$\widetilde{\bfy^*} := ({\bfy^*}^T, \pi^*_\bfy)^T$
	\\
	$\pk_{\bfy^*} := \ipe.\PubKGen(\mpk, \widetilde{\bfy^*})$ 
	\\ 
	$ {\sf Bad}_1 := \false, {\sf Bad}_2 := \false$
	\\ 
	If $\nizk.\Verify(\crs, \stmt, \pi) =0 \ \vee\ \bfy^* \notin \mcal{F}_{{\sf IP}, \ell} \ \vee \ \pk_{\bfy^*} \neq \aux_\bfy^*$: 
	\pcskipln \\ 
	$\pcind$ \ret $0$
	\\
	{\color{blue} \pcbox{
	\color{black} \pcbox{\text{If $\stmt \notin L_{NIZK}$: ${\sf Bad}_1 := \true$, \ret $0$}}
	}}
	\\ 
	$\widetilde{\sigma}^* \gets \as.\PreSign(\sk, m^*, \aux_\bfy^*)$ 
	\\ 
	$\sigma^* \gets \A^{\mcal{O}_S(\cdot), \mcal{O}_{\fpS}(\cdot, \cdot, \cdot, \cdot, \cdot)}(\widetilde{\sigma}^*)$ 
	\\
	{ \color{blue} \pcbox{ \color{black}
	z = \as.\Ext(\widetilde{\sigma}^*, \sigma^*, \aux_\bfy^*)
	} 
	} \\
	{ \color{blue} \pcbox{ \color{black}
	\text{if $ ((m^* \notin \mcal{Q}) \wedge \Verify(\vk, m^*, \sigma^*) \wedge ((\aux_\bfy^*, z) \notin R'_\ipfe))$: }	
	} 
	} \pcskipln \\
	$\pcind$
	{ \color{blue} \pcbox{ \color{black}
	\text{${\sf Bad}_2 := \true$, \ret $0$}	
	} 
	} \\
	\ret $((m^* \notin \mcal{Q}) \wedge \Verify(\vk, m^*, \sigma^*))$
	}
	\end{pcvstack}
	\begin{pcvstack}[space=1em]
	\procedure[linenumbering, mode=text]{Oracle $\mcal{O}_S(m)$}{
	$\sigma \gets \Sign(\sk, m)$
	\\
	$\mcal{Q} := \mcal{Q} \vee \{m\}$ 
	\\ 
	\ret $\sigma$
	}

	\procedure[linenumbering, mode=text]{Oracle $\mcal{O}_{\fpS}(m, X, \bfy, \aux_\bfy, \pi_\bfy)$}{
	If $\AuxVerify(\advt, \bfy, \aux_\bfy, \pi_\bfy) = 0$: \ret $\bot$
	\\
	$\widetilde{\sigma} \gets \FPreSign(\advt, \sk, m, X, \bfy, \aux_\bfy)$ 
	\\ 
	$\mcal{Q} := \mcal{Q} \vee \{m\}$ 
	\\ 
	\ret $\widetilde{\sigma}$
	}
	\end{pcvstack}
	\end{pcvstack}
	\end{gameproof}
	\caption{Unforgeability Proof: Games $G_0, G_1, G_2$
	}
	\label{fig:fasig-unf-games}
	\end{figure}
\else 
	\begin{figure}[h]
	\centering
	\captionsetup{justification=centering}
	\begin{gameproof}
	\begin{pchstack}[boxed, space=1em]
	\begin{pcvstack}
	\procedure[linenumbering, mode=text]{Games $G_0$, \pcbox{G_1}, {\color{blue} \pcbox{ \color{black} G_2}}}{
	$\mcal{Q} := \emptyset$ 
	\\ 
	$\crs \gets \nizk.\Setup(1^\secparam)$
	\\ 
	$\pp' \gets \ipfe.\Gen(1^\secparam)$ 
	% \\ 
	% \nikhil{remove GroupGen specification.}
	% \pcskipln \\ \nikhil{keep it generic for the generic construction.}
	\\ 
	$\pp := (\crs, \pp')$
	\\ 
	$(\sk, \vk) \gets \KGen(\pp')$ 
	\\ 
	$(X^*, \bfx^*) \gets \GenR(1^\secparam)$ 
	\\ 
	$(\advt, m^*, \bfy^*, \aux_\bfy^*, \pi^*_\bfy) \gets \A^{\mcal{O}_S(\cdot)}(\pp, \vk, X^*)$ 
	\\ 
	Parse $\advt = (\mpk, \ct, \pi)$ 
	\\ 
	Let $\stmt := (X^*, \pp', \mpk, \ct)$
	\\ 
	$\widetilde{\bfy^*} := ({\bfy^*}^T, \pi^*_\bfy)^T$
	\\
	$\pk_{\bfy^*} := \ipe.\PubKGen(\mpk, \widetilde{\bfy^*})$ 
	\\ 
	$ {\sf Bad}_1 := \false, {\sf Bad}_2 := \false$
	\\ 
	If $\nizk.\Verify(\crs, \stmt, \pi) =0 \ \vee\ \bfy^* \notin \mcal{F}_{{\sf IP}, \ell}$ 
	\pcskipln \\ 
	$\vee \ \pk_{\bfy^*} \neq \aux_\bfy^*$: \ret $0$
	\\
	{\color{blue} \pcbox{
	\color{black} \pcbox{\text{If $\stmt \notin L_{NIZK}$: ${\sf Bad}_1 := \true$, \ret $0$}}
	}}
	\\ 
	$\widetilde{\sigma}^* \gets \as.\PreSign(\sk, m^*, \aux_\bfy^*)$ 
	\\ 
	$\sigma^* \gets \A^{\mcal{O}_S(\cdot), \mcal{O}_{\fpS}(\cdot, \cdot, \cdot, \cdot, \cdot)}(\widetilde{\sigma}^*)$ 
	\\
	{ \color{blue} \pcbox{ \color{black}
	z = \as.\Ext(\widetilde{\sigma}^*, \sigma^*, \aux_\bfy^*)
	} 
	} \\
	{ \color{blue} \pcbox{ \color{black}
	\text{if $ ((m^* \notin \mcal{Q}) \wedge \Verify(\vk, m^*, \sigma^*)$ }	
	} 
	} \pcskipln \\
	% $\pcind$
	{ \color{blue} \pcbox{ \color{black}
	\text{$ \wedge ((\aux_\bfy^*, z) \notin R'_\ipfe))$: ${\sf Bad}_2 := \true$, \ret $0$}	
	} 
	} \\
	\ret $((m^* \notin \mcal{Q}) \wedge \Verify(\vk, m^*, \sigma^*))$
	}
	\end{pcvstack}
	\begin{pcvstack}[space=1em]
	\procedure[linenumbering, mode=text]{Oracle $\mcal{O}_S(m)$}{
	$\sigma \gets \Sign(\sk, m)$
	\\
	$\mcal{Q} := \mcal{Q} \vee \{m\}$ 
	\\ 
	\ret $\sigma$
	}

	\procedure[linenumbering, mode=text]{Oracle $\mcal{O}_{\fpS}(m, X, \bfy, \aux_\bfy, \pi_\bfy)$}{
	If $\AuxVerify(\advt, \bfy, \aux_\bfy, \pi_\bfy) = 0$: \ret $\bot$
	\\
	$\widetilde{\sigma} \gets \FPreSign(\advt, \sk, m, X, \bfy, \aux_\bfy)$ 
	\\ 
	$\mcal{Q} := \mcal{Q} \vee \{m\}$ 
	\\ 
	\ret $\widetilde{\sigma}$
	}
	\end{pcvstack}
	\end{pchstack}
	\end{gameproof}
	\caption{Unforgeability Proof: Games $G_0, G_1, G_2$
	}
	\label{fig:fasig-unf-games}
	\end{figure}
\fi
To prove the lemma, we need to show that 
\[
\Pr[ G_0(1^\secparam) = 1] \leq \negl(\secparam).
\]
Note that by triangle inequality, it follows that 
\begin{align*}
\Pr[ G_0(1^\secparam) = 1] 
& \leq |\Pr[ G_0(1^\secparam) = 1] - \Pr[ G_1(1^\secparam) = 1]| \\ 
& \quad \quad + |\Pr[ G_1(1^\secparam) = 1] - \Pr[ G_2(1^\secparam) = 1]| \\
& \quad \quad + \Pr[ G_2(1^\secparam) = 1].
\end{align*}
To complete the proof,  
we show in~\Cref{claim:fasig-unf-g1,claim:fasig-unf-g2,claim:fasig-unf-g2-final}
that each of the three terms on the right-hand-side are at most $\negl(\secparam)$. 
% \nikhil{add discussion on next claims and how that completes the proof.}
\end{proof}

\begin{claim}
\label{claim:fasig-unf-g1}
If the NIZK argument system $\nizk$ satisfies adaptive soundness (\Cref{def:nizk}), then, 
$| \Pr[G_0 (1^\secparam) = 1] - \Pr[G_1 (1^\secparam) = 1] | \leq \negl(\secparam)$.
\end{claim}

\begin{proof}
Observe that in game $G_0$, ${\sf Bad}_1 = \false$ always and game $G_1$ differs from it when the flag ${\sf Bad}_1 = \true$ is set. Both games are identical until $G_1$ checks the condition for setting ${\sf Bad}_1$.
Hence, 
\[
| \Pr[G_0 (1^\secparam) = 1] - \Pr[G_1 (1^\secparam) = 1] | \leq \Pr[{\sf Bad}_1 \text{ in $G_1$}]. 
\]
Hence, it suffices to show that in game $G_1$, $\Pr[{\sf Bad}_1] \leq \negl(\secparam)$.

We prove the claim using a reduction to the adaptive soundness of the underlying $\nizk$ scheme. More specifically, we show that if the adversary \A causes game $G_1$ to set ${\sf Bad}_1 = \true$, then, we can use it to construct a reduction \B that breaks the adaptive soundness of $\nizk$. Let \C be the challenger for the NIZK adaptive soundness. Then, the reduction is as in~\Cref{fig:fasig-unf-g1-proof}.

\ifacm
	\begin{figure}[t]
	\centering
	\captionsetup{justification=centering}
	\begin{gameproof}
	\begin{pcvstack}[boxed, space=1em]
	\begin{pcvstack}[space=1em]
	\procedure[linenumbering, mode=text]{Reduction \B for proof of~\Cref{claim:fasig-unf-g1}}{
	$\mcal{Q} := \emptyset$ 
	\\ 
	{\color{blue}
	$\crs \gets \C(1^\secparam)$
	}
	\\ 
	$\pp' \gets \ipfe.\Gen(1^\secparam)$ 
	\\ 
	$\pp = (\crs, \pp')$
	\\ 
	$(\sk, \vk) \gets \KGen(\pp')$ 
	\\ 
	$(X^*, \bfx^*) \gets \GenR(1^\secparam)$ 
	\\ 
	$(\advt, m^*, \bfy^*, \aux_\bfy^*, \pi^*_\bfy) \gets \A^{\mcal{O}_S(\cdot)}(\pp, \vk, X^*)$ 
	\\ 
	Parse $\advt = (\mpk, \ct, \pi)$, let $\stmt := (X^*, \pp', \mpk, \ct)$
	\\ 
	$\widetilde{\bfy^*} := ({\bfy^*}^T, \pi^*_\bfy)^T$
	\\
	$\pk_{\bfy^*} := \ipe.\PubKGen(\mpk, \widetilde{\bfy^*})$ 
	\\ 
	${\sf Bad}_1 = \false$
	\\ 
	% {\color{blue}
	\label{step:bad1-prev}
	If $\nizk.\Verify(\crs, \stmt, \pi) =0 \ \vee\ \bfy^* \notin \mcal{F}_{{\sf IP}, \ell} \ \vee \ \pk_{\bfy^*} \neq \aux_\bfy^*$: 
	% }
	\pcskipln \\ 
	% {\color{blue}
	$\pcind$ {\bf abort}
	% }
	\\ 
	{\color{blue}
	\label{step:bad1}
	If $\stmt \notin L_{NIZK}$: ${\sf Bad}_1 = \true$
	}
	\\ 
	{\color{blue}
	\ret $(\stmt, \pi)$ to \C
	}
	% \\ 
	% {\color{blue}
	% Else: abort game with \C
	% }
	}
	\end{pcvstack}
	\begin{pcvstack}[space=1em]
	\procedure[linenumbering, mode=text]{Oracle $\mcal{O}_S(m)$}{
	$\sigma \gets \Sign(\sk, m)$
	\\
	$\mcal{Q} := \mcal{Q} \vee \{m\}$ 
	\\ 
	\ret $\sigma$
	}

	% \procedure[linenumbering, mode=text]{Oracle $\mcal{O}_{\fpS}(m, X, f)$}{
	% $\widetilde{\sigma} \gets \FPreSign(\advt, \sk, m, X, f)$ 
	% \\ 
	% $\mcal{Q} := \mcal{Q} \vee \{m\}$ 
	% \\ 
	% \ret $\widetilde{\sigma}$
	% }
	\end{pcvstack}
	\end{pcvstack}
	\end{gameproof}
	\caption{Reduction \B for proof of~\Cref{claim:fasig-unf-g1}
	}
	\label{fig:fasig-unf-g1-proof}
	\end{figure}
\else
	\begin{figure}[h]
	\centering
	\captionsetup{justification=centering}
	\begin{gameproof}
	\begin{pchstack}[boxed, space=1em]
	\procedure[linenumbering, mode=text]{Reduction \B for proof of~\Cref{claim:fasig-unf-g1}}{
	$\mcal{Q} := \emptyset$ 
	\\ 
	{\color{blue}
	$\crs \gets \C(1^\secparam)$
	}
	\\ 
	$\pp' \gets \ipfe.\Gen(1^\secparam)$ 
	\\ 
	$\pp = (\crs, \pp')$
	\\ 
	$(\sk, \vk) \gets \KGen(\pp')$ 
	\\ 
	$(X^*, \bfx^*) \gets \GenR(1^\secparam)$ 
	\\ 
	$(\advt, m^*, \bfy^*, \aux_\bfy^*, \pi^*_\bfy) \gets \A^{\mcal{O}_S(\cdot)}(\pp, \vk, X^*)$ 
	\\ 
	Parse $\advt = (\mpk, \ct, \pi)$, let $\stmt := (X^*, \pp', \mpk, \ct)$
	\\ 
	$\widetilde{\bfy^*} := ({\bfy^*}^T, \pi^*_\bfy)^T$
	\\
	$\pk_{\bfy^*} := \ipe.\PubKGen(\mpk, \widetilde{\bfy^*})$ 
	\\ 
	${\sf Bad}_1 = \false$
	\\ 
	% {\color{blue}
	\label{step:bad1-prev}
	If $\nizk.\Verify(\crs, \stmt, \pi) =0 \ \vee\ \bfy^* \notin \mcal{F}_{{\sf IP}, \ell} \ \vee \ \pk_{\bfy^*} \neq \aux_\bfy^*$: 
	% }
	\pcskipln \\ 
	% {\color{blue}
	$\pcind$ {\bf abort}
	% }
	\\ 
	{\color{blue}
	\label{step:bad1}
	If $\stmt \notin L_{NIZK}$: ${\sf Bad}_1 = \true$
	}
	\\ 
	{\color{blue}
	\ret $(\stmt, \pi)$ to \C
	}
	% \\ 
	% {\color{blue}
	% Else: abort game with \C
	% }
	}

	\procedure[linenumbering, mode=text]{Oracle $\mcal{O}_S(m)$}{
	$\sigma \gets \Sign(\sk, m)$
	\\
	$\mcal{Q} := \mcal{Q} \vee \{m\}$ 
	\\ 
	\ret $\sigma$
	}

	% \procedure[linenumbering, mode=text]{Oracle $\mcal{O}_{\fpS}(m, X, f)$}{
	% $\widetilde{\sigma} \gets \FPreSign(\advt, \sk, m, X, f)$ 
	% \\ 
	% $\mcal{Q} := \mcal{Q} \vee \{m\}$ 
	% \\ 
	% \ret $\widetilde{\sigma}$
	% }
	\end{pchstack}
	\end{gameproof}
	\caption{Reduction \B for proof of~\Cref{claim:fasig-unf-g1}
	}
	\label{fig:fasig-unf-g1-proof}
	\end{figure}
\fi
Observe that if \A causes game $G_1$ to not abort, then, it must be the case that 
$\pi$ is a valid proof for $\stmt = (X^*, \pp', \mpk, \ct)$.
Further, if \A causes game $G_1$ to set ${\sf Bad}_1 = \true$, then, it must be the case that 
$\stmt \notin L_{NIZK}$. 
This means that in step~\ref{step:bad1}, the reduction \B successfully returns a valid proof $\pi$ to $\C$ for a statement $\stmt$ not in the language $L_{NIZK}$ and thus breaks the adaptive soundness of the $\nizk$. Thus, $\Pr[{\sf Bad}_1]$ is same as the probability that $\B$ breaks adaptive soundness of the $\nizk$. Hence, we can conclude that $\Pr[{\sf Bad}_1] \leq \negl(\secparam)$.
\end{proof}

\begin{claim}
\label{claim:fasig-unf-g2}
If the Adaptor Signature scheme $\as$ satisfies witness extractability (\Cref{def:as-wit-ext}), then, 
$| \Pr[G_1 (1^\secparam) = 1] - \Pr[G_2 (1^\secparam) = 1] | \leq \negl(\secparam)$.
\end{claim}

\begin{proof}
Observe that in game $G_1$, ${\sf Bad}_2 = \false$ always and game $G_2$ differs from it when the flag ${\sf Bad}_2 = \true$ is set. Both games are identical until $G_2$ checks the condition for setting ${\sf Bad}_2$.
Hence, 
\[
| \Pr[G_1 (1^\secparam) = 1] - \Pr[G_2 (1^\secparam) = 1] | \leq \Pr[{\sf Bad_2} \text{ in $G_2$}].
\]

% \begin{align*}
% & | \Pr[G_1 (1^\secparam) = 1] - \Pr[G_2 (1^\secparam) = 1] | \\
% & \quad = | \Pr[(m^* \notin \mcal{Q}) \wedge \Verify(\vk, m^*, \sigma^*) \wedge \neg {\sf Bad}_0 \wedge \neg {\sf Bad}_1] \\  
% & \quad \quad - \Pr[(m^* \notin \mcal{Q}) \wedge \Verify(\vk, m^*, \sigma^*) \wedge \neg {\sf Bad}_0 \wedge \neg {\sf Bad}_1 \wedge \neg {\sf Bad}_2] | \\
% & \quad = \Pr[(m^* \notin \mcal{Q}) \wedge \Verify(\vk, m^*, \sigma^*) \wedge \neg {\sf Bad}_0 \wedge \neg {\sf Bad}_1 \wedge {\sf Bad}_2] 
% \end{align*}
% Let ${\sf Event}_2 = (m^* \notin \mcal{Q}) \wedge \Verify(\vk, m^*, \sigma^*) \wedge \neg {\sf Bad}_0 \wedge \neg {\sf Bad}_1 \wedge {\sf Bad}_2$.
Then, it suffices to show that in game $G_2$, $\Pr[{\sf Bad}_2] \leq \negl(\secparam)$.

We prove the claim using a reduction to the witness extractability of the underlying $\as$ scheme. More specifically, 
we show that if the adversary \A causes ${\sf Bad}_2$ in game $G_2$, then, we can use it to construct a reduction \B that breaks the witness extractability of $\as$.
For the $\as$ witness extractability game, 
let \C be the challenger and let $\as.\mcal{O}_S, \as.\mcal{O}_{pS}$ be the signing and pre-signing oracles that the reduction \B has access to.
Then, the reduction is as in~\Cref{fig:fasig-unf-g2-proof}.

\ifacm
	\begin{figure}[t]
	\centering
	\captionsetup{justification=centering}
	\begin{gameproof}
	\begin{pcvstack}[boxed, space=1em]
	\begin{pcvstack}[space=1em]
	\procedure[linenumbering, mode=text]{Reduction \B for proof of~\Cref{claim:fasig-unf-g2}}{
	$\mcal{Q} := \emptyset$ 
	\\ 
	$\crs \gets \nizk.\Setup(1^\secparam)$
	\\ 
	{\color{blue}
	$(\pp', \vk) \gets \C(1^\secparam)$ 
	}
	\\ 
	$\pp = (\crs, \pp')$
	\\ 
	$(X^*, \bfx^*) \gets \GenR(1^\secparam)$ 
	\\ 
	$(\advt, m^*, \bfy^*, \aux_\bfy^*, \pi^*_\bfy) \gets \A^{\mcal{O}_S(\cdot)}(\pp, \vk, X^*)$ 
	\\ 
	Parse $\advt = (\mpk, \ct, \pi)$, let $\stmt := (X^*, \pp', \mpk, \ct)$
	\\ 
	$\widetilde{\bfy^*} := ({\bfy^*}^T, \pi^*_\bfy)^T$
	\\
	$\pk_{\bfy^*} := \ipe.\PubKGen(\mpk, \widetilde{\bfy^*})$ 
	\\ 
	${\sf Bad}_1 := \false, {\sf Bad}_2 := \false$
	\\ 
	If $\nizk.\Verify(\crs, \stmt, \pi) =0 \ \vee\ \bfy^* \notin \mcal{F}_{{\sf IP}, \ell} \ \vee \ \pk_{\bfy^*} \neq \aux_\bfy^*$: 
	\pcskipln \\ 
	$\pcind$ ${\bf abort}$
	\\
	If $\stmt \notin L_{NIZK}$: ${\sf Bad}_1 := \true$, {\bf abort}
	\\ 
	{\color{blue}
	$\widetilde{\sigma}^* \gets \C(m^*, \aux_\bfy^*)$ 
	}
	\\ 
	$\sigma^* \gets \A^{\mcal{O}_S(\cdot), \mcal{O}_{\fpS}(\cdot, \cdot, \cdot, \cdot, \cdot)}(\widetilde{\sigma}^*)$ 
	\\
	$z = \as.\Ext(\widetilde{\sigma}^*, \sigma^*, \aux_\bfy^*)$
	\\ 
	{\color{blue}
	\label{step:bad2-after}
	If $(m^* \notin \mcal{Q}) \wedge \Verify(\vk, m^*, \sigma^*) \wedge ((\aux_\bfy^*, z) \notin R'_\ipfe)$: 
	}
	\pcskipln \\
	$\pcind$
	{\color{blue}
	${\sf Bad}_2 = \true$, \ret $\sigma^*$ to \C
	}
	\\
	% if $((m^* \notin \mcal{Q}) \wedge \Verify(\vk, m^*, \sigma^*) \wedge \neg {\sf Bad}_1 \wedge {\sf Bad}_2)$: \ret $\sigma^*$ to \C
	% \\ 
	{\color{blue}
	Else: {\bf abort}
	}
	}
	\end{pcvstack}
	\begin{pcvstack}[space=1em]
	\procedure[linenumbering, mode=text]{Oracle $\mcal{O}_S(m)$}{
	{\color{blue}
	$\sigma \gets \as.\mcal{O}_S(m)$
	}
	\\
	$\mcal{Q} := \mcal{Q} \vee \{m\}$ 
	\\ 
	\ret $\sigma$
	}

	\procedure[linenumbering, mode=text]{Oracle $\mcal{O}_{\fpS}(m, X, \bfy, \aux_\bfy, \pi_\bfy)$}{
	If $\AuxVerify(\advt, \bfy, \aux_\bfy, \pi_\bfy) = 0$: \ret $\bot$
	\\
	{\color{blue}
	$\widetilde{\sigma} \gets \as.\mcal{O}_{pS}(m, \aux_\bfy)$ 
	}
	\\ 
	$\mcal{Q} := \mcal{Q} \vee \{m\}$ 
	\\ 
	\ret $\widetilde{\sigma}$
	}
	\end{pcvstack}
	\end{pcvstack}
	\end{gameproof}
	\caption{Reduction \B for proof of~\Cref{claim:fasig-unf-g2}
	}
	\label{fig:fasig-unf-g2-proof}
	\end{figure}
\else
	\begin{figure}[h]
	\centering
	\captionsetup{justification=centering}
	\begin{gameproof}
	\begin{pchstack}[boxed, space=1em]
	\begin{pcvstack}
	\procedure[linenumbering, mode=text]{Reduction \B for proof of~\Cref{claim:fasig-unf-g2}}{
	$\mcal{Q} := \emptyset$ 
	\\ 
	$\crs \gets \nizk.\Setup(1^\secparam)$
	\\ 
	{\color{blue}
	$(\pp', \vk) \gets \C(1^\secparam)$ 
	}
	\\ 
	$\pp = (\crs, \pp')$
	\\ 
	$(X^*, \bfx^*) \gets \GenR(1^\secparam)$ 
	\\ 
	$(\advt, m^*, \bfy^*, \aux_\bfy^*, \pi^*_\bfy) \gets \A^{\mcal{O}_S(\cdot)}(\pp, \vk, X^*)$ 
	\\ 
	Parse $\advt = (\mpk, \ct, \pi)$
	\\ 
	Let $\stmt := (X^*, \pp', \mpk, \ct)$
	\\ 
	$\widetilde{\bfy^*} := ({\bfy^*}^T, \pi^*_\bfy)^T$
	\\
	$\pk_{\bfy^*} := \ipe.\PubKGen(\mpk, \widetilde{\bfy^*})$ 
	\\ 
	${\sf Bad}_1 := \false, {\sf Bad}_2 := \false$
	\\ 
	If $\nizk.\Verify(\crs, \stmt, \pi) =0 \ \vee\ \bfy^* \notin \mcal{F}_{{\sf IP}, \ell}$ 
	\pcskipln \\ 
	$\vee \ \pk_{\bfy^*} \neq \aux_\bfy^*$: ${\bf abort}$
	\\
	If $\stmt \notin L_{NIZK}$: ${\sf Bad}_1 := \true$, {\bf abort}
	\\ 
	{\color{blue}
	$\widetilde{\sigma}^* \gets \C(m^*, \aux_\bfy^*)$ 
	}
	\\ 
	$\sigma^* \gets \A^{\mcal{O}_S(\cdot), \mcal{O}_{\fpS}(\cdot, \cdot, \cdot, \cdot, \cdot)}(\widetilde{\sigma}^*)$ 
	\\
	$z = \as.\Ext(\widetilde{\sigma}^*, \sigma^*, \aux_\bfy^*)$
	\\ 
	{\color{blue}
	\label{step:bad2-after}
	If $(m^* \notin \mcal{Q}) \wedge \Verify(\vk, m^*, \sigma^*) \wedge ((\aux_\bfy^*, z) \notin R'_\ipfe)$: 
	}
	\pcskipln \\
	$\pcind$
	{\color{blue}
	${\sf Bad}_2 = \true$, \ret $\sigma^*$ to \C
	}
	\\
	% if $((m^* \notin \mcal{Q}) \wedge \Verify(\vk, m^*, \sigma^*) \wedge \neg {\sf Bad}_1 \wedge {\sf Bad}_2)$: \ret $\sigma^*$ to \C
	% \\ 
	{\color{blue}
	Else: {\bf abort}
	}
	}
	\end{pcvstack}
	\begin{pcvstack}[space=1em]
	\procedure[linenumbering, mode=text]{Oracle $\mcal{O}_S(m)$}{
	{\color{blue}
	$\sigma \gets \as.\mcal{O}_S(m)$
	}
	\\
	$\mcal{Q} := \mcal{Q} \vee \{m\}$ 
	\\ 
	\ret $\sigma$
	}

	\procedure[linenumbering, mode=text]{Oracle $\mcal{O}_{\fpS}(m, X, \bfy, \aux_\bfy, \pi_\bfy)$}{
	If $\AuxVerify(\advt, \bfy, \aux_\bfy, \pi_\bfy) = 0$: \ret $\bot$
	\\
	{\color{blue}
	$\widetilde{\sigma} \gets \as.\mcal{O}_{pS}(m, \aux_\bfy)$ 
	}
	\\ 
	$\mcal{Q} := \mcal{Q} \vee \{m\}$ 
	\\ 
	\ret $\widetilde{\sigma}$
	}
	\end{pcvstack}
	\end{pchstack}
	\end{gameproof}
	\caption{Reduction \B for proof of~\Cref{claim:fasig-unf-g2}
	}
	\label{fig:fasig-unf-g2-proof}
	\end{figure}
\fi
Observe that the reduction \B perfectly simulates the game $G_2$ to the adversary \A as it appropriately forwards the signing and functional pre-signing queries to its the signing and pre-signing oracles of $\as$. 
Thus, to complete the proof it suffices to argue that if \B returns $\sigma^*$ to \C in step~\ref{step:bad2-after}, then, it breaks the witness extractability of $\as$. Suppose that the query set maintained by \C is denoted by $\as.\mcal{Q}$.
Recall then that to break the witness extractability of $\as$, the signature $\sigma^*$ on message $m^*$ must satisfy $(m^* \notin \as.\mcal{Q}) \wedge ((\aux_\bfy^*, z) \notin R'_\ipfe) \wedge \Verify(\vk, m^*, \sigma^*)$, where $z = \as.\Ext(\widetilde{\sigma}^*, \sigma^*, \aux_\bfy^*)$.
The last of the three conditions holds because the same is present in the pre-condition for returning $\sigma^*$ to \C. Hence, it remains to argue that the first two conditions also hold true.
We note that $m^* \notin \mcal{Q}$ implies $m^* \notin \as.\mcal{Q}$ as from the description of the reduction \B, it follows that $\mcal{Q} = \as.\mcal{Q}$. Next, $(\aux_\bfy^*, z) \notin R'_\ipfe$ because ${\sf Bad}_2 = \true$. Therefore, if the adversary \A causes ${\sf Event}_2$ in game $G_2$, then, \B that breaks the witness extractability of $\as$. Hence, it follows that $\Pr[{\sf Event}_2] \leq \negl(\secparam)$.
\end{proof}

\begin{claim}
\label{claim:fasig-unf-g2-final}
% Let $\mcal{F}_{{\sf IP}, \ell} = \{\bfy \in \Z_p^\ell \ s.t. \ \bfy \neq {\bf 0}\}$.
Suppose relation $R$ is a $\mcal{F}_{{\sf IP}, \ell}$-hard (\Cref{def:f-hard-relation}), 
$\ipe$ scheme satisfies $R'_\ipfe$-robustness (\Cref{def:ipfe-robust}). 
Then, $\Pr[G_2(1^\secparam) = 1] \leq \negl(\secparam)$.
\end{claim}

\begin{proof}
\ifacm
	\begin{figure}[t]
	\centering
	\captionsetup{justification=centering}
	\begin{gameproof}
	\begin{pcvstack}[boxed, space=1em]
	\begin{pcvstack}[space=1em]
	\procedure[linenumbering, mode=text]{Reduction \B for proof of~\Cref{claim:fasig-unf-g2-final}}{
	$\mcal{Q} := \emptyset$ 
	\\ 
	$\crs \gets \nizk.\Setup(1^\secparam)$
	\\ 
	$\pp' \gets \ipfe.\Gen(1^\secparam)$ 
	\\ 
	$\pp := (\crs, \pp')$
	\\ 
	$(\sk, \vk) \gets \KGen(\pp')$ 
	\\ 
	{\color{blue}
	$X^* \gets \C(1^\secparam)$ 
	}
	\\ 
	$(\advt, m^*, \bfy^*, \aux_\bfy^*, \pi^*_\bfy) \gets \A^{\mcal{O}_S(\cdot)}(\pp, \vk, X^*)$ 
	\\ 
	Parse $\advt = (\mpk, \ct, \pi)$, let $\stmt := (X^*, \pp', \mpk, \ct)$
	\\ 
	$\widetilde{\bfy^*} := ({\bfy^*}^T, \pi^*_\bfy)^T$
	\\
	$\pk_{\bfy^*} := \ipe.\PubKGen(\mpk, \widetilde{\bfy^*})$ 
	\\ 
	${\sf Bad}_1 := \false, {\sf Bad}_2 := \false$
	\\ 
	If $\nizk.\Verify(\crs, \stmt, \pi) =0 \ \vee\ \bfy^* \notin \mcal{F}_{{\sf IP}, \ell} \ \vee \ \pk_{\bfy^*} \neq \aux_\bfy^*$: 
	\pcskipln \\ 
	$\pcind$ {\bf abort}
	\\
	If $\stmt \notin L_{NIZK}$: ${\sf Bad}_1 := \true$, {\bf abort}
	\\ 
	$\widetilde{\sigma}^* \gets \as.\PreSign(\sk, m^*, \aux_\bfy^*)$ 
	\\ 
	$\sigma^* \gets \A^{\mcal{O}_S(\cdot), \mcal{O}_{\fpS}(\cdot, \cdot, \cdot, \cdot, \cdot)}(\widetilde{\sigma}^*)$ 
	\\
	$z = \as.\Ext(\widetilde{\sigma}^*, \sigma^*, \aux_\bfy^*)$
	\\
	If $(m^* \notin \mcal{Q}) \wedge \Verify(\vk, m^*, \sigma^*) \wedge ((\aux_\bfy^*, z) \notin R'_\ipfe)$: 
	\pcskipln \\ 
	$\pcind$ ${\sf Bad}_2 = \true$, {\bf abort}
	\\
	{\color{blue}
	$v = \ipe.\Dec(z, \ct)$
	}
	\\
	{\color{blue}
	\label{step:good}
	If $((m^* \notin \mcal{Q}) \wedge \Verify(\vk, m^*, \sigma^*))$: 
	}
	\pcskipln \\ 
	{\color{blue}
	$\pcind$ \ret $(\bfy^*, v)$
	}
	\\ 
	{\color{blue}
	Else: abort game with \C
	}
	}
	\end{pcvstack}
	\begin{pcvstack}[space=1em]
	\procedure[linenumbering, mode=text]{Oracle $\mcal{O}_S(m)$}{
	$\sigma \gets \Sign(\sk, m)$
	\\
	$\mcal{Q} := \mcal{Q} \vee \{m\}$ 
	\\ 
	\ret $\sigma$
	}

	\procedure[linenumbering, mode=text]{Oracle $\mcal{O}_{\fpS}(m, X, \bfy, \aux_\bfy, \pi_\bfy)$}{
	If $\AuxVerify(\advt, \bfy, \aux_\bfy, \pi_\bfy) = 0$: \ret $\bot$
	\\
	$\widetilde{\sigma} \gets \FPreSign(\advt, \sk, m, X, \bfy, \aux_\bfy)$ 
	\\ 
	$\mcal{Q} := \mcal{Q} \vee \{m\}$ 
	\\ 
	\ret $\widetilde{\sigma}$
	}
	\end{pcvstack}
	\end{pcvstack}
	\end{gameproof}
	\caption{Reduction \B for proof of~\Cref{claim:fasig-unf-g2-final}
	}
	\label{fig:fasig-unf-g2-final-proof}
	\end{figure}
\else
	\begin{figure}[h]
	\centering
	\captionsetup{justification=centering}
	\begin{gameproof}
	\begin{pchstack}[boxed, space=1em]
	\begin{pcvstack}[space=1em]
	\procedure[linenumbering, mode=text]{Reduction \B for proof of~\Cref{claim:fasig-unf-g2-final}}{
	$\mcal{Q} := \emptyset$ 
	\\ 
	$\crs \gets \nizk.\Setup(1^\secparam)$
	\\ 
	$\pp' \gets \ipfe.\Gen(1^\secparam)$ 
	\\ 
	$\pp := (\crs, \pp')$
	\\ 
	$(\sk, \vk) \gets \KGen(\pp')$ 
	\\ 
	{\color{blue}
	$X^* \gets \C(1^\secparam)$ 
	}
	\\ 
	$(\advt, m^*, \bfy^*, \aux_\bfy^*, \pi^*_\bfy) \gets \A^{\mcal{O}_S(\cdot)}(\pp, \vk, X^*)$ 
	\\ 
	Parse $\advt = (\mpk, \ct, \pi)$, let $\stmt := (X^*, \pp', \mpk, \ct)$
	\\ 
	$\widetilde{\bfy^*} := ({\bfy^*}^T, \pi^*_\bfy)^T$
	\\
	$\pk_{\bfy^*} := \ipe.\PubKGen(\mpk, \widetilde{\bfy^*})$ 
	\\ 
	${\sf Bad}_1 := \false, {\sf Bad}_2 := \false$
	\\ 
	If $\nizk.\Verify(\crs, \stmt, \pi) =0 \ \vee\ \bfy^* \notin \mcal{F}_{{\sf IP}, \ell}$ 
	\pcskipln \\ 
	$\vee \ \pk_{\bfy^*} \neq \aux_\bfy^*$: {\bf abort}
	\\
	If $\stmt \notin L_{NIZK}$: ${\sf Bad}_1 := \true$, {\bf abort}
	\\ 
	$\widetilde{\sigma}^* \gets \as.\PreSign(\sk, m^*, \aux_\bfy^*)$ 
	\\ 
	$\sigma^* \gets \A^{\mcal{O}_S(\cdot), \mcal{O}_{\fpS}(\cdot, \cdot, \cdot, \cdot, \cdot)}(\widetilde{\sigma}^*)$ 
	\\
	$z = \as.\Ext(\widetilde{\sigma}^*, \sigma^*, \aux_\bfy^*)$
	\\
	If $(m^* \notin \mcal{Q}) \wedge \Verify(\vk, m^*, \sigma^*) \wedge ((\aux_\bfy^*, z) \notin R'_\ipfe)$: 
	\pcskipln \\ 
	$\pcind$ ${\sf Bad}_2 = \true$, {\bf abort}
	\\
	{\color{blue}
	$v = \ipe.\Dec(z, \ct)$
	}
	\\
	{\color{blue}
	\label{step:good}
	If $((m^* \notin \mcal{Q}) \wedge \Verify(\vk, m^*, \sigma^*))$: 
	}
	\pcskipln \\ 
	{\color{blue}
	$\pcind$ \ret $(\bfy^*, v)$
	}
	\\ 
	{\color{blue}
	Else: abort game with \C
	}
	}
	\end{pcvstack}
	\begin{pcvstack}[space=1em]
	\procedure[linenumbering, mode=text]{Oracle $\mcal{O}_S(m)$}{
	$\sigma \gets \Sign(\sk, m)$
	\\
	$\mcal{Q} := \mcal{Q} \vee \{m\}$ 
	\\ 
	\ret $\sigma$
	}

	\procedure[linenumbering, mode=text]{Oracle $\mcal{O}_{\fpS}(m, X, \bfy, \aux_\bfy, \pi_\bfy)$}{
	If $\AuxVerify(\advt, \bfy, \aux_\bfy, \pi_\bfy) = 0$: \ret $\bot$
	\\
	$\widetilde{\sigma} \gets \FPreSign(\advt, \sk, m, X, \bfy, \aux_\bfy)$ 
	\\ 
	$\mcal{Q} := \mcal{Q} \vee \{m\}$ 
	\\ 
	\ret $\widetilde{\sigma}$
	}
	\end{pcvstack}
	\end{pchstack}
	\end{gameproof}
	\caption{Reduction \B for proof of~\Cref{claim:fasig-unf-g2-final}
	}
	\label{fig:fasig-unf-g2-final-proof}
	\end{figure}
\fi
We prove the claim using a reduction to the $\mcal{F}_{{\sf IP}, \ell}$-hardness of the relation $R$. 
More specifically, 
we show that if an adversary \A causes game $G_2$ to return $1$, then, 
we can use it to construct a reduction \B that breaks 
the $\mcal{F}_{{\sf IP}, \ell}$-hardness of the relation $R$. 
Let \C be the challenger for 
the $\mcal{F}_{{\sf IP}, \ell}$-hardness of the relation $R$. 
The reduction \B is as in~\Cref{fig:fasig-unf-g2-final-proof}.

Observe that the reduction \B perfectly simulates game $G_2$ to the adversary \A as it obtains $X^*$ from the challenger \C who computes it as $(X^*, \cdot) \gets \GenR(1^\secparam)$. Then, if \A causes $G_2$ to return $1$, it must be the case that the if condition in step~\ref{step:good} of the reduction must be true. In such a case, the reduction returns $(\bfy^*, v)$ to the challenger \C. To complete the proof, it remains to show that this results in \B winning the game against \C. 

To win the game against \C, the reduction \B's output must satisfy $(\bfy^* \in \mcal{F}_{{\sf IP}, \ell}) \wedge (v \in \{ f_{\bfy^*}(\bfx)  : \exists \bfx \ s.t. \ (X, \bfx) \in R\})$.
Note that since the reduction did not abort, it implies $\bfy^* \in \mcal{F}_{{\sf IP}, \ell}$ and $\pk_{\bfy^*} = \aux_\bfy^*$.
Next, ${\sf Bad}_1 = \false$ implies that $\stmt \in L_{NIZK}$, where $\stmt = (X^*, \pp', \mpk, \ct)$. Hence, it follows that $\ct$ encrypts some vector $\widetilde{\bfx} = (\bfx^T, 0)^T \in \mcal{M}' \subseteq \Z^{\ell+1}$ under $\mpk$ such that $(X^*, \bfx) \in R$, where $(\mpk, \msk) \gets \ipe.\Setup(\pp', 1^{\ell+1})$. Next, ${\sf Bad}_2 = \false$ implies that $(\aux_\bfy^*, z) \in R'_\ipfe$. 
As $\pk_{\bfy^*} = \aux_\bfy^*$ and IPFE satisfies $R'_\ipfe$-robustness, it follows that $v = f_{\widetilde{\bfy^*}}(\widetilde{\bfx})$. As $\widetilde{\bfy^*} = ({\bfy^*}^T, \pi_\bfy^*)^T$ and $\widetilde{\bfx} = (\bfx^T, 0)^T$, we get that $f_{\widetilde{\bfy^*}}(\widetilde{\bfx}) = f_{\bfy^*}(\bfx)$. 
Hence, we conclude that $v \in \{ f_{\bfy^*}(\bfx) : \exists \bfx \ s.t. \ (X, \bfx) \in R\}$. This completes the proof.

\end{proof}

%% file: fas-proof-func-wit-ext.tex
\subsection{Witness Extractability}
% \nikhil{update experiments and proofs to reflect changes related to $\aux$ info. See ~\Cref{fig:fasig-wit-ext-exp}.}

\begin{lemma}
Suppose $\AS$ satisfies witness extractability, 
$\nizk$ satisfies adaptive soundness, and 
$\ipe$ satisfies $R'_\ipfe$-robustness. 
Then, the functional adaptor signature construction in~\Cref{sec:construction} is $\faWitExt$-secure.
\label{lemma:fas-wit-ext}
\end{lemma}

\begin{proof}
We prove the lemma by defining a sequence of games.

\paragraph{Game $G_0$:} It is the original game $\faWitExt_{\A, \FAS}$, where the adversary \A who is  
given access to a pre-signature on a message,
must come up with a full signature such that it does not reveal
the function evaluation on a witness.
The adversary also has access to functional pre-sign oracle $\mcal{O}_{\fpS}$ and sign oracle $\mcal{O}_S$. 
% Since we are in the random oracle model, the adversary additionally has access to a random oracle $H$.
Game $G_0$ is formally defined in~\Cref{fig:fasig-wit-ext-games}.

\paragraph{Game $G_1$:}
same as game $G_0$, except that 
when \A outputs a public advertisement $\advt$, 
the game checks if the NIZK statement $(X^*, \mpk, \ct)$ is in the language $L_\nizk$. 
If no, the game sets the flag ${\sf Bad}_1 = \true$.
% \nikhil{note that this check is inefficient. Is that a problem?}
Game $G_1$ is formally defined in~\Cref{fig:fasig-wit-ext-games}.

\ifacm
	\begin{figure}[t]
	\centering
	\captionsetup{justification=centering}
	\begin{pcvstack}[boxed, space=1em]
	\begin{pcvstack}[space=1em]
	\procedure[linenumbering, mode=text]{Games $G_0$, $\pcbox{G_1}$}{
	$\mcal{Q} := \emptyset$ 
	\\ 
	$\crs \gets \nizk.\Setup(1^\secparam)$
	\\ 
	$\pp' \gets \ipfe.\Gen(1^\secparam)$ 
	% \\ 
	% \nikhil{remove GroupGen specification.} 
	% \pcskipln \\ \nikhil{keep it generic for the generic construction.}
	\\ 
	$\pp := (\crs, \pp')$
	\\ 
	$(\sk, \vk) \gets \KGen(\pp')$ 
	\\ 
	% $\advt = \bot$ 
	% \\ 
	$(X^*, \advt, m^*, \bfy^*, \aux_\bfy^*, \pi_\bfy^*) \gets \A^{\mcal{O}_S(\cdot)}(\pp, \vk)$ 
	\\ 
	Parse $\advt = (\mpk, \ct, \pi)$, let $\stmt := (X^*, \pp', \mpk, \ct)$
	\\ 
	$\widetilde{\bfy^*} := ({\bfy^*}^T, \pi^*_\bfy)^T$
	\\
	$\pk_{\bfy^*} := \ipe.\PubKGen(\mpk, \widetilde{\bfy^*})$ 
	\\ 
	${\sf Bad}_0 := \false, {\sf Bad}_1 := \false, {\sf Bad}_2 := \false$
	\\ 
	If $\nizk.\Verify(\crs, \stmt, \pi) =0 \vee \bfy^* \notin \mcal{F}_{{\sf IP}, \ell} \vee \pk_{\bfy^*} \neq \aux_\bfy^*$: 
	\pcskipln \\ 
	$\pcind$ ${\sf Bad}_0 := \true$
	\\ 
	\pcbox{\text{If $\stmt \notin L_\nizk$: ${\sf Bad}_1 := \true$}}
	\\
	$\widetilde{\sigma}^* \gets \as.\PreSign(\sk, m^*, \aux_\bfy^*)$ 
	\\ 
	$\sigma^* \gets \A^{\mcal{O}_S(\cdot), \mcal{O}_{\fpS}(\cdot, \cdot, \cdot, \cdot, \cdot)}(\widetilde{\sigma}^*)$ 
	\\
	$z := \as.\Ext(\widetilde{\sigma}^*, \sigma^*, \aux_\bfy^*)$ 
	\\ 
	$v := \ipe.\Dec(z, \ct)$
	\\ 
	If $v \in \{ f_{\bfy^*}(\bfx)  : \exists \ \bfx\ s.t.\ (X^*, \bfx) \in R\}$: ${\sf Bad}_2 = \true$ 
	\\
	% \nikhil{computing set $W$ requires to find all witnesses $x$ for the statement $X$ which is going to be inefficient.}
	\ret $((m^* \notin \mcal{Q}) \wedge \Verify(\vk, m^*, \sigma^*) \wedge \neg {\sf Bad}_0 \wedge \neg {\sf Bad}_1 \wedge \neg {\sf Bad}_2 )$ 
	}
	\end{pcvstack}
	\begin{pcvstack}[space=1em]
	\procedure[linenumbering, mode=text]{Oracle $\mcal{O}_S(m)$}{
	$\sigma \gets \Sign(\sk, m)$
	\\
	$\mcal{Q} := \mcal{Q} \vee \{m\}$ 
	\\ 
	\ret $\sigma$
	}

	\procedure[linenumbering, mode=text]{Oracle $\mcal{O}_{\fpS}(m, X, \bfy, \aux_\bfy, \pi_\bfy)$}{
	If $\AuxVerify(\advt, \bfy, \aux_\bfy, \pi_\bfy) = 0$: \ret $\bot$
	\\
	$\widetilde{\sigma} \gets \FPreSign(\advt, \sk, m, X, \bfy, \aux_\bfy)$ 
	\\ 
	$\mcal{Q} := \mcal{Q} \vee \{m\}$ 
	\\ 
	\ret $\widetilde{\sigma}$
	}
	\end{pcvstack}
	\end{pcvstack}
	\caption{Witness Extractability proof: Games $G_0$ and $G_1$}
	\label{fig:fasig-wit-ext-games}
	\end{figure}
\else
	\begin{figure}[h]
	\centering
	\captionsetup{justification=centering}
	\begin{pchstack}[boxed, space=0.1em]
	\begin{pcvstack}[space=1em]
	\procedure[linenumbering, mode=text]{Games $G_0$, $\pcbox{G_1}$}{
	$\mcal{Q} := \emptyset$ 
	\\ 
	$\crs \gets \nizk.\Setup(1^\secparam)$
	\\ 
	$\pp' \gets \ipfe.\Gen(1^\secparam)$ 
	% \\ 
	% \nikhil{remove GroupGen specification.} 
	% \pcskipln \\ \nikhil{keep it generic for the generic construction.}
	\\ 
	$\pp := (\crs, \pp')$
	\\ 
	$(\sk, \vk) \gets \KGen(\pp')$ 
	\\ 
	% $\advt = \bot$ 
	% \\ 
	$(X^*, \advt, m^*, \bfy^*, \aux_\bfy^*, \pi_\bfy^*) \gets \A^{\mcal{O}_S(\cdot)}(\pp, \vk)$ 
	\\ 
	Parse $\advt = (\mpk, \ct, \pi)$, let $\stmt := (X^*, \pp', \mpk, \ct)$
	\\ 
	$\widetilde{\bfy^*} := ({\bfy^*}^T, \pi^*_\bfy)^T$
	\\
	$\pk_{\bfy^*} := \ipe.\PubKGen(\mpk, \widetilde{\bfy^*})$ 
	\\ 
	${\sf Bad}_0 := \false, {\sf Bad}_1 := \false, {\sf Bad}_2 := \false$
	\\ 
	If $\nizk.\Verify(\crs, \stmt, \pi) =0 \vee \bfy^* \notin \mcal{F}_{{\sf IP}, \ell} \vee \pk_{\bfy^*} \neq \aux_\bfy^*$: 
	\pcskipln \\ 
	$\pcind$ ${\sf Bad}_0 := \true$
	\\ 
	\pcbox{\text{If $\stmt \notin L_\nizk$: ${\sf Bad}_1 := \true$}}
	\\
	$\widetilde{\sigma}^* \gets \as.\PreSign(\sk, m^*, \aux_\bfy^*)$ 
	\\ 
	$\sigma^* \gets \A^{\mcal{O}_S(\cdot), \mcal{O}_{\fpS}(\cdot, \cdot, \cdot, \cdot, \cdot)}(\widetilde{\sigma}^*)$ 
	\\
	$z := \as.\Ext(\widetilde{\sigma}^*, \sigma^*, \aux_\bfy^*)$ 
	\\ 
	$v := \ipe.\Dec(z, \ct)$
	\\ 
	If $v \in \{ f_{\bfy^*}(\bfx)  : \exists \ \bfx\ s.t.\ (X^*, \bfx) \in R\}$: ${\sf Bad}_2 = \true$ 
	\\
	% \nikhil{computing set $W$ requires to find all witnesses $x$ for the statement $X$ which is going to be inefficient.}
	\ret $((m^* \notin \mcal{Q}) \wedge \Verify(\vk, m^*, \sigma^*) \wedge \neg {\sf Bad}_0 \wedge \neg {\sf Bad}_1 \wedge \neg {\sf Bad}_2 )$ 
	}
	\end{pcvstack}
	\begin{pcvstack}[space=1em]
	\procedure[linenumbering, mode=text]{Oracle $\mcal{O}_S(m)$}{
	$\sigma \gets \Sign(\sk, m)$
	\\
	$\mcal{Q} := \mcal{Q} \vee \{m\}$ 
	\\ 
	\ret $\sigma$
	}

	\procedure[linenumbering, mode=text]{Oracle $\mcal{O}_{\fpS}(m, X, \bfy, \aux_\bfy, \pi_\bfy)$}{
	If $\AuxVerify(\advt, \bfy, \aux_\bfy, \pi_\bfy) = 0$: \ret $\bot$
	\\
	$\widetilde{\sigma} \gets \FPreSign(\advt, \sk, m, X, \bfy, \aux_\bfy)$ 
	\\ 
	$\mcal{Q} := \mcal{Q} \vee \{m\}$ 
	\\ 
	\ret $\widetilde{\sigma}$
	}
	\end{pcvstack}
	\end{pchstack}
	\caption{Witness Extractability proof: Games $G_0$ and $G_1$}
	\label{fig:fasig-wit-ext-games}
	\end{figure}
\fi
To prove the lemma, we need to show that 
\[
\Pr[ G_0(1^\secparam) = 1] \leq \negl(\secparam).
\]
Note that by triangle inequality, it follows that 
\ifacm
	\begin{align*}
	\Pr[ G_0(1^\secparam) = 1] 
	& \leq |\Pr[ G_0(1^\secparam) = 1] - \Pr[ G_1(1^\secparam) = 1]| \\ 
	& \quad + \Pr[ G_1(1^\secparam) = 1].
	\end{align*}
\else 
\[
	\Pr[ G_0(1^\secparam) = 1] 
	\leq |\Pr[ G_0(1^\secparam) = 1] - \Pr[ G_1(1^\secparam) = 1]|
	+ \Pr[ G_1(1^\secparam) = 1].
\]
\fi
To complete the proof,  
we show in~\Cref{claim:fasig-wit-ext-g1,claim:fasig-wit-ext-g1-final}
that each of the two terms on the right-hand-side are at most $\negl(\secparam)$. 
\end{proof}

\begin{claim}
\label{claim:fasig-wit-ext-g1}
If the NIZK argument system $\nizk$ satisfies adaptive soundness, then, 
$| \Pr[G_0 (1^\secparam) = 1] - \Pr[G_1 (1^\secparam) = 1] | \leq \negl(\secparam)$.
\end{claim}
\begin{proof}
Similar to proof of~\Cref{claim:fasig-unf-g1}.
\end{proof}

\begin{claim}
\label{claim:fasig-wit-ext-g1-final}
Suppose the Adaptor Signature scheme $\as$ satisfies witness extractability and 
the IPFE scheme $\ipfe$ satisfies $R'_\ipfe$-robustness. 
Then, 
$\Pr[G_1 (1^\secparam) = 1] \leq \negl(\secparam)$.
\end{claim}

\begin{proof}

To prove $\Pr[G_1(1^\secparam) = 1] \leq \negl(\secparam)$,
we show that if there exists a \ppt adversary \A such that it wins game $G_1$ with non-negligible advantage, then,
we can construct a \ppt reduction \B that breaks the witness extractability of the underlying adaptor signature scheme $\AS$.
For the $\as$ witness extractability game, 
let \C be the challenger and let $\as.\mcal{O}_S, \as.\mcal{O}_{pS}$ be the signing and pre-signing oracles that the reduction \B has access to.
Then, the reduction is as in~\Cref{fig:fasig-wit-ext-g1-final-proof}.

\ifacm
	\begin{figure}[t]
	\centering
	\captionsetup{justification=centering}
	\begin{pcvstack}[boxed, space=1em]
	\begin{pcvstack}[space=1em]
	\procedure[linenumbering, mode=text]{Reduction $B$ for proof of~\Cref{claim:fasig-wit-ext-g1-final}}{
	$\mcal{Q} := \emptyset$ 
	\\ 
	$\crs \gets \nizk.\Setup(1^\secparam)$
	\\ 
	{\color{blue}
	$(\pp', \vk) \gets \C(1^\secparam)$ 
	}
	\\ 
	$\pp := (\crs, \pp')$
	\\ 
	% $\advt = \bot$ 
	% \\ 
	$(X^*, \advt, m^*, \bfy^*, \aux_\bfy^*, \pi_\bfy^*) \gets \A^{\mcal{O}_S(\cdot)}(\pp, \vk)$ 
	\\ 
	Parse $\advt = (\mpk, \ct, \pi)$, let $\stmt := (X^*, \pp', \mpk, \ct)$
	\\ 
	$\widetilde{\bfy^*} := ({\bfy^*}^T, \pi^*_\bfy)^T$
	\\
	$\pk_{\bfy^*} := \ipe.\PubKGen(\mpk, \widetilde{\bfy^*})$ 
	\\ 
	${\sf Bad}_0 := \false, {\sf Bad}_1 := \false, {\sf Bad}_2 := \false$
	\\ 
	If $\nizk.\Verify(\crs, \stmt, \pi) =0 \vee \bfy^* \notin \mcal{F}_{{\sf IP}, \ell} \vee \pk_{\bfy^*} \neq \aux_\bfy^*$: 
	\pcskipln \\ 
	$\pcind$ ${\sf Bad}_0 := \true$
	\\ 
	If $\stmt \notin L_\nizk$: ${\sf Bad}_1 = \true$
	\\
	{\color{blue}
	$\widetilde{\sigma}^* \gets \C(m^*, \aux_\bfy^*)$ 
	}
	\\ 
	$\sigma^* \gets \A^{\mcal{O}_S(\cdot), \mcal{O}_{\fpS}(\cdot, \cdot, \cdot, \cdot, \cdot)}(\widetilde{\sigma}^*)$ 
	\\
	$z := \as.\Ext(\widetilde{\sigma}^*, \sigma^*, \aux_\bfy^*)$ 
	\\ 
	$v := \ipe.\Dec(z, \ct)$
	\\ 
	If $v \in \{ f_{\bfy^*}(\bfx)  : \exists \ \bfx\ s.t.\ (X^*, \bfx) \in R\}$: ${\sf Bad}_2 = \true$ 
	\\
	{\color{blue}
	If $((m^* \notin \mcal{Q}) \wedge \Verify(\vk, m^*, \sigma^*) \wedge \neg {\sf Bad}_0 \wedge \neg {\sf Bad}_1 \wedge \neg {\sf Bad}_2 )$: 
	}
	\pcskipln \\
	{\color{blue}
	$\pcind$ \ret $\sigma^*$
	}
	\\ 
	{\color{blue}
	Else: abort game with \C 
	}
	}
	\end{pcvstack}
	\begin{pcvstack}[space=1em]
	\procedure[linenumbering, mode=text]{Oracle $\mcal{O}_S(m)$}{
	{\color{blue}
	$\sigma \gets \as.\mcal{O}_S(m)$
	}
	\\
	$\mcal{Q} := \mcal{Q} \vee \{m\}$ 
	\\ 
	\ret $\sigma$
	}

	\procedure[linenumbering, mode=text]{Oracle $\mcal{O}_{\fpS}(m, X, f)$}{
	If $\advt = \bot$: \ret $\bot$ 
	\\ 
	$\pk_\bfy = \ipe.\PubKGen(\mpk, \bfy)$ 
	\\ 
	{\color{blue}
	$\widetilde{\sigma} \gets \as.\mcal{O}_{pS}(m, \pk_\bfy)$ 
	}
	\\ 
	$\mcal{Q} := \mcal{Q} \vee \{m\}$ 
	\\ 
	\ret $\widetilde{\sigma}$
	}
	\end{pcvstack}
	\end{pcvstack}
	\caption{Reduction $B$ for proof of~\Cref{claim:fasig-wit-ext-g1-final}}
	\label{fig:fasig-wit-ext-g1-final-proof}
	\end{figure}
\else 
	\begin{figure}[h]
	\centering
	\captionsetup{justification=centering}
	\begin{pchstack}[boxed, space=1em]
	\begin{pcvstack}[space=1em]
	\procedure[linenumbering, mode=text]{Reduction $B$ for proof of~\Cref{claim:fasig-wit-ext-g1-final}}{
	$\mcal{Q} := \emptyset$ 
	\\ 
	$\crs \gets \nizk.\Setup(1^\secparam)$
	\\ 
	{\color{blue}
	$(\pp', \vk) \gets \C(1^\secparam)$ 
	}
	\\ 
	$\pp := (\crs, \pp')$
	\\ 
	% $\advt = \bot$ 
	% \\ 
	$(X^*, \advt, m^*, \bfy^*, \aux_\bfy^*, \pi_\bfy^*) \gets \A^{\mcal{O}_S(\cdot)}(\pp, \vk)$ 
	\\ 
	Parse $\advt = (\mpk, \ct, \pi)$, let $\stmt := (X^*, \pp', \mpk, \ct)$
	\\ 
	$\widetilde{\bfy^*} := ({\bfy^*}^T, \pi^*_\bfy)^T$
	\\
	$\pk_{\bfy^*} := \ipe.\PubKGen(\mpk, \widetilde{\bfy^*})$ 
	\\ 
	${\sf Bad}_0 := \false, {\sf Bad}_1 := \false, {\sf Bad}_2 := \false$
	\\ 
	If $\nizk.\Verify(\crs, \stmt, \pi) =0 \vee \bfy^* \notin \mcal{F}_{{\sf IP}, \ell} \vee \pk_{\bfy^*} \neq \aux_\bfy^*$: 
	\pcskipln \\ 
	$\pcind$ ${\sf Bad}_0 := \true$
	\\ 
	If $\stmt \notin L_\nizk$: ${\sf Bad}_1 = \true$
	\\
	{\color{blue}
	$\widetilde{\sigma}^* \gets \C(m^*, \aux_\bfy^*)$ 
	}
	\\ 
	$\sigma^* \gets \A^{\mcal{O}_S(\cdot), \mcal{O}_{\fpS}(\cdot, \cdot, \cdot, \cdot, \cdot)}(\widetilde{\sigma}^*)$ 
	\\
	$z := \as.\Ext(\widetilde{\sigma}^*, \sigma^*, \aux_\bfy^*)$ 
	\\ 
	$v := \ipe.\Dec(z, \ct)$
	\\ 
	If $v \in \{ f_{\bfy^*}(\bfx)  : \exists \ \bfx\ s.t.\ (X^*, \bfx) \in R\}$: ${\sf Bad}_2 = \true$ 
	\\
	{\color{blue}
	If $((m^* \notin \mcal{Q}) \wedge \Verify(\vk, m^*, \sigma^*) \wedge \neg {\sf Bad}_0 \wedge \neg {\sf Bad}_1 \wedge \neg {\sf Bad}_2 )$: 
	}
	\pcskipln \\
	{\color{blue}
	$\pcind$ \ret $\sigma^*$
	}
	\\ 
	{\color{blue}
	Else: abort game with \C 
	}
	}
	\end{pcvstack}
	\begin{pcvstack}[space=1em]
	\procedure[linenumbering, mode=text]{Oracle $\mcal{O}_S(m)$}{
	{\color{blue}
	$\sigma \gets \as.\mcal{O}_S(m)$
	}
	\\
	$\mcal{Q} := \mcal{Q} \vee \{m\}$ 
	\\ 
	\ret $\sigma$
	}

	\procedure[linenumbering, mode=text]{Oracle $\mcal{O}_{\fpS}(m, X, f)$}{
	If $\advt = \bot$: \ret $\bot$ 
	\\ 
	$\pk_\bfy = \ipe.\PubKGen(\mpk, \bfy)$ 
	\\ 
	{\color{blue}
	$\widetilde{\sigma} \gets \as.\mcal{O}_{pS}(m, \pk_\bfy)$ 
	}
	\\ 
	$\mcal{Q} := \mcal{Q} \vee \{m\}$ 
	\\ 
	\ret $\widetilde{\sigma}$
	}
	\end{pcvstack}
	\end{pchstack}
	\caption{Reduction $B$ for proof of~\Cref{claim:fasig-wit-ext-g1-final}}
	\label{fig:fasig-wit-ext-g1-final-proof}
	\end{figure}
\fi

From the description of reduction \B, 
we can observe that whenever the challenger \C obtains the signature $\sigma^*$ on message $m^*$ obtained, it is the case that $((m^* \notin \mcal{Q}) \wedge \Verify(\vk, m^*, \sigma^*) \wedge \neg {\sf Bad}_0 \wedge \neg {\sf Bad}_1 \wedge \neg {\sf Bad}_2)$. 
Suppose that the query set maintained by \C is denoted by $\as.\mcal{Q}$. For \B to succeed, 
the signature $\sigma^*$ must satisfy $((m^* \notin \as.\mcal{Q}) \wedge \Verify(\vk, m^*, \sigma^*) \wedge ((\pk_{\bfy^*}, z) \notin R'_\ipfe))$. 
We note that $m^* \notin \mcal{Q}$ implies $m^* \notin \as.\mcal{Q}$ as from the description of the reduction \B, it follows that $\mcal{Q} = \as.\mcal{Q}$.
Note that $\sigma^*$ already satisfies the second condition and it suffices to show that if $\neg {\sf Bad}_0 \wedge \neg {\sf Bad}_1 \wedge \neg {\sf Bad}_2 $, then, $(\pk_{\bfy^*}, z) \notin R'_\ipfe$.

Note that ${\sf Bad}_0 = \false$ implies $\nizk.\Verify(\crs, \stmt, \pi) = 1$ and $\bfy^* \in \mcal{F}_{{\sf IP}, \ell}$ and $\aux_\bfy^* = \pk_{\bfy^*}$.
Next, ${\sf Bad}_1 = \false$ implies that $\stmt \in L_\nizk$, where $\stmt = (X^*, \pp', \mpk, \ct)$. Hence, it follows that $\ct$ encrypts some vector $\widetilde{\bfx^*} = ({\bfx^*}^T, 0)^T \in \mcal{M}' \subseteq \Z^{\ell+1}$ under $\mpk$ such that $(X^*, \bfx^*) \in R$, where $(\mpk, \msk) \gets \ipe.\Setup(\pp', 1^\ell)$. Next, ${\sf Bad}_2 = \false$ implies that $v \notin \{ f_{\bfy^*}(\bfx) : \exists \ \bfx\ s.t.\ (X^*, \bfx) \in R\}$, where $v$ is computed as $v := \ipe.\Dec(z, \ct)$. This implies $v \neq f_{\bfy^*}(\bfx^*)$. 
Note that for $\widetilde{\bfy^*} = ({\bfy^*}^T, \pi_\bfy^*)^T$, we have $f_{\widetilde{\bfy^*}}(\widetilde{\bfx^*}) = f_{\bfy^*}(\bfx^*)$.
Hence, it follows that $v \neq f_{\widetilde{\bfy^*}}(\widetilde{\bfx^*})$.

To complete the proof, we argue that if $v \neq f_{\widetilde{\bfy^*}}(\widetilde{\bfx^*})$, then, $(\pk_{\bfy^*}, z) \notin R'_\ipfe$. This follows from $R'_\ipfe$-robustness of IPFE. In particular, the contra-positive form of $R'_\ipfe$-robustness says that if $v \neq f_{\widetilde{\bfy^*}}(\widetilde{\bfx^*})$, then either $\pp'$ is not computed honestly, or $\mpk$ is not computed honestly or $\ct$ is not computed honestly or $\pk_{\bfy^*}$ is not computed honestly or $v$ is not computed honestly or $(\pk_{\bfy^*}, z) \notin R'_\ipfe$. Observe that the challenger \C samples $\pp'$ honestly, the NIZK proof $\pi$ attests that $\mpk$ and $\ct$ are computed honestly, the reduction \B computes $\pk_{\bfy^*}$ and $v$ honestly. Thus, it must be the case that $(\pk_{\bfy^*}, z) \notin R'_\ipfe$.

\end{proof}

% \end{proof}

%% file: fas-proof-pre-sig-adaptability.tex
\subsection{Pre-Signature Adaptability}
% \nikhil{update experiments and proofs to reflect changes related to $\aux$ info. See ~\Cref{def:fas-pre-sig-adaptability}}

\begin{lemma}
% \anote{Lets formalise the properties}
Suppose \ipe satisfies $R_\ipfe$-compliance (\Cref{def:ipfe-compliant}) and suppose that $\AS$ satisfies weak pre-signature adaptability (\Cref{def:as-pre-sig-adaptability}). 
Then, the functional adaptor signature construction in~\Cref{sec:construction} is pre-signature adaptable (\Cref{def:fas-pre-sig-adaptability}).
\label{lemma:fas-pre-sig-adaptability}
\end{lemma}

\begin{proof}
For any $\secparam \in \N$, 
let $\crs \gets \Setup(1^\secparam)$ 
be as computed in~\Cref{sec:construction}.
For any statement/witness pair $(X, \bfx) \in R$, 
let $(\advt, $ $\state) \gets \AdvertisementGen(1^\ell,\crs, X, \bfx)$ 
be as computed in~\Cref{sec:construction}.
For any message $m \in \{0,1\}^*$, 
any function $\bfy \in \mcal{F}_{IP, \Z_p^\ell}$,
any $(\aux_\bfy, \pi_\bfy) \gets \AuxGen(\advt, \state, \bfy)$ as computed in~\Cref{sec:construction},
any key pair $(\sk, \vk) \gets \KGen(1^\secparam)$ 
as computed in~\Cref{sec:construction}, 
any pre-signature $\widetilde{\sigma} \in \{0,1\}^*$ 
such that $\FPreVerify(\advt, \vk, m, X, \bfy, \aux_\bfy, \pi_\bfy, \widetilde{\sigma}) = 1$ as computed in~\Cref{sec:construction}.
This implies $\AS.\PreVerify(\vk, m, \aux_\bfy, \widetilde{\sigma}) = 1$.
By $R_\ipfe$-compliance of the $\ipe$ scheme, 
we know that $(\aux_\bfy, \sk_\bfy) \in R_\ipfe$, where $\sk_\bfy$ is as computed in the $\Adapt$ algorithm in~\Cref{sec:construction}.
Then, it follows by the weak pre-signature adaptability of $\AS$ that 
$\Pr[\AS.\Verify(\vk, m, \sigma) = 1] =1 $, where 
$\sigma = \AS.\Adapt(\vk, m, \aux_\bfy, \sk_\bfy, \widetilde{\sigma})$.
Then, from the implementation of the $\Adapt$ algorithm of our functional adaptor signature scheme as in~\Cref{sec:construction}, it follows that 
\[
\Pr[\Verify(\vk, m, \Adapt(\advt, \state, \vk, m, X, \bfx, \bfy, \aux_\bfy, \widetilde{\sigma})) = 1] = 1.
\]

\end{proof}

% \anote{Did you guys discuss the point on the underlying adaptor signatures being forgeable? Seems we don't need them to be unforgeable secure. Does that mean our inner AS can be a weak object?}

%% file: fas-proof-zk.tex
\subsection{Zero-Knowledge}
% \nikhil{update experiments and proofs to reflect changes related to $\aux$ info. See ~\Cref{fig:fas-zk-real,fig:fas-zk-ideal}}

\begin{lemma}
Suppose $\mcal{M}$ is an additive group, $\nizk$ satisfies zero-knowledge (\Cref{def:nizk}) and $\ipe$ satisfies selective, IND-security (\Cref{def:ipfe-sel-ind-sec}). 
Then, the functional adaptor signature construction in~\Cref{sec:construction} is zero-knowledge (\Cref{def:fas-zk}).
\label{lemma:fas-zk}
\end{lemma}

\begin{proof}

To prove the lemma, we need to show that for every stateful \ppt adversary \A, 
there exists a stateful \ppt simulator $\Sim = (\Setup^*, \AdGen^*, \AuxGen^*, \Adapt^*)$ and 
there exists a negligible function $\negl$ such that
for all \ppt distinguishers \D, 
for all $\secparam \in \N$, 
for all $(X, \bfx) \in R$,
\ifacm
	\begin{align*}
	& | \Pr[\D(\faZKReal_{\A, \FAS}(1^\secparam, X, \bfx)) = 1] \\
	& \quad - \Pr[ \D(\faZKIdeal_{\A, \FAS}^\Sim(1^\secparam, X, \bfx)) = 1] | \leq  \negl(\secparam).
	\end{align*}
\else 
\[
	| \Pr[\D(\faZKReal_{\A, \FAS}(1^\secparam, X, \bfx)) = 1]
	- \Pr[ \D(\faZKIdeal_{\A, \FAS}^\Sim(1^\secparam, X, \bfx)) = 1] | \leq  \negl(\secparam).
\]
\fi
We first describe the stateful simulator $\Sim = (\Setup^*, \AdGen^*, \allowbreak \AuxGen^*, \Adapt^*)$.  
Let $\nizk.\Sim = (\nizk.\Setup^*, \nizk.\Prove^*)$ be the NIZK simulator.
% and let the IPFE simulator be $\ipfe.\Sim=(\ipfe.\Setup^*, \ipfe.\Enc^*, $ $\ipfe.\KGen^*)$. 
Then, the simulator $\Sim$ is as in~\Cref{fig:fas-zk-sim}.

\ifacm
	\begin{figure}[t]
	\centering
	\captionsetup{justification=centering}
	\begin{pchstack}[boxed, space=1em]
	\begin{pcvstack}[space=1em]
	\procedure[linenumbering, mode=text]{$\Setup^*(1^\secparam)$}{
	% \nikhil{does the nizk setup need to take $1^\ell$ as input too?}
	{\color{blue}
	Compute $(\crs, \td) \gets \nizk.\Setup^*(1^\secparam)$
	}
	\\ 
	Sample $\pp' \gets \ipfe.\Gen(1^\secparam)$
	\\ 
	Store internal state $\td$, \ret $\pp := (\crs, \pp')$
	}

	% \nikhil{remove GroupGen specification.} 
	% \pcskipln \\ \nikhil{keep it generic for the generic}
	% \\ 
	\procedure[linenumbering, mode=text]{$\AdGen^*(\pp, X)$}{
	Sample $(\mpk, \msk) \gets \ipe.\Setup(\pp', 1^{\ell+1})$
	% }
	\\ 
	Sample $\bft \getr \mcal{M}$
	\\
	{\color{blue}
	Let $\widetilde{\bfx} := (-\bft^T, 1)^T \in \mcal{M}' \subseteq \Z^{\ell+1}$ 
	}
	\\ 
	Compute $\ct \gets \ipe.\Enc(\mpk, \widetilde{\bfx})$
	\\
	{\color{blue}
	Compute $\pi \gets \nizk.\Prove^*(\crs, \td, (X, \pp', \mpk, \ct))$
	}
	\\ 
	Store internal state $\state := (\msk, \bft)$
	\\ 
	\ret $\advt := (\mpk, \ct, \pi)$
	}

	\procedure[linenumbering, mode=text]{$\AuxGen^*(\advt, \bfy, f_{\bfy}(\bfx))$}{
	Parse $\advt = (\mpk, \ct, \pi)$
	\\ 
	Compute $\widetilde{\bfy} := (\bfy^T, {\color{blue}f_\bfy(\bft) + f_\bfy(\bfx)})^T \in \mcal{F}_{{\sf IP}, \ell+1}$
	\\ 
	Compute $\pk_\bfy := \ipfe.\PubKGen(\mpk, \widetilde{\bfy})$ 
	\\ 
	\ret $\aux_\bfy := \pk_\bfy$, $\pi_\bfy := {\color{blue}f_\bfy(\bft) + f_\bfy(\bfx)}$
	}

	\procedure[linenumbering, mode=text]{$ \Adapt^*(\advt, \vk, m, X, \bfy, \aux_\bfy, \widetilde{\sigma}, f_\bfy(\bfx))$}{
	Parse $\advt = (\mpk, \ct, \pi)$, and $\state = (\msk, \bft)$
	\\ 
	Compute $\widetilde{\bfy} := (\bfy^T, {\color{blue}f_\bfy(\bft) + f_\bfy(\bfx)})^T \in \mcal{F}_{{\sf IP}, \ell+1}$
	% Compute $\pk_\bfy := \ipe.\PubKGen(\mpk, \bfy)$
	\\ 
	Compute $\sk_\bfy := \ipe.\KGen(\msk, \widetilde{\bfy})$
	\\ 
	\ret $\sigma := \as.\Adapt(\vk, m, \aux_\bfy, \sk_\bfy, \widetilde{\sigma})$
	}

	\end{pcvstack}
	\end{pchstack}
	\caption{Zero-Knowledge Simulator $\Sim$}
	\label{fig:fas-zk-sim}
	\end{figure}
\else 
	\begin{figure}[h]
	\centering
	\captionsetup{justification=centering}
	\begin{pchstack}[boxed, space=1em]
	\begin{pcvstack}[space=1em]
	\procedure[linenumbering, mode=text]{$\Setup^*(1^\secparam)$}{
	% \nikhil{does the nizk setup need to take $1^\ell$ as input too?}
	{\color{blue}
	Let $(\crs, \td) \gets \nizk.\Setup^*(1^\secparam)$
	}
	\\ 
	Sample $\pp' \gets \ipfe.\Gen(1^\secparam)$
	\\ 
	Store internal state $\td$, \ret $\pp := (\crs, \pp')$
	}

	% \nikhil{remove GroupGen specification.} 
	% \pcskipln \\ \nikhil{keep it generic for the generic}
	% \\ 
	\procedure[linenumbering, mode=text]{$\AdGen^*(\pp, X)$}{
	Sample $(\mpk, \msk) \gets \ipe.\Setup(\pp', 1^{\ell+1})$
	% }
	\\ 
	Sample $\bft \getr \mcal{M}$
	\\
	{\color{blue}
	Let $\widetilde{\bfx} := (-\bft^T, 1)^T \in \mcal{M}' \subseteq \Z^{\ell+1}$ 
	}
	\\ 
	Let $\ct \gets \ipe.\Enc(\mpk, \widetilde{\bfx})$
	\\
	{\color{blue}
	Let $\pi \gets \nizk.\Prove^*(\crs, \td, (X, \pp', \mpk, \ct))$
	}
	\\ 
	Store internal state $\state := (\msk, \bft)$
	\\ 
	\ret $\advt := (\mpk, \ct, \pi)$
	}

	\end{pcvstack}
	\begin{pcvstack}[space=1em]

	\procedure[linenumbering, mode=text]{$\AuxGen^*(\advt, \bfy, f_{\bfy}(\bfx))$}{
	Parse $\advt = (\mpk, \ct, \pi)$
	\\ 
	Let $\widetilde{\bfy} := (\bfy^T, {\color{blue}f_\bfy(\bft) + f_\bfy(\bfx)})^T \in \mcal{F}_{{\sf IP}, \ell+1}$
	\\ 
	Let $\pk_\bfy := \ipfe.\PubKGen(\mpk, \widetilde{\bfy})$ 
	\\ 
	\ret $\aux_\bfy := \pk_\bfy$, $\pi_\bfy := {\color{blue}f_\bfy(\bft) + f_\bfy(\bfx)}$
	}

	\procedure[linenumbering, mode=text]{$ \Adapt^*(\advt, \vk, m, X, \bfy, \aux_\bfy, \widetilde{\sigma}, f_\bfy(\bfx))$}{
	Parse $\advt = (\mpk, \ct, \pi)$, and $\state = (\msk, \bft)$
	\\ 
	Let $\widetilde{\bfy} := (\bfy^T, {\color{blue}f_\bfy(\bft) + f_\bfy(\bfx)})^T \in \mcal{F}_{{\sf IP}, \ell+1}$
	% Compute $\pk_\bfy := \ipe.\PubKGen(\mpk, \bfy)$
	\\ 
	Let $\sk_\bfy := \ipe.\KGen(\msk, \widetilde{\bfy})$
	\\ 
	\ret $\sigma := \as.\Adapt(\vk, m, \aux_\bfy, \sk_\bfy, \widetilde{\sigma})$
	}

	\end{pcvstack}
	\end{pchstack}
	\caption{Zero-Knowledge Simulator $\Sim$}
	\label{fig:fas-zk-sim}
	\end{figure}
\fi
Having described the Simulator $\Sim$, the experiments $\faZKReal_{\A, \FAS}$ and $\faZKIdeal^\Sim_{\A, \FAS}$ are as in~\Cref{fig:fasig-zk-games}. Observe that $\faZKIdeal^\Sim_{\A, \FAS}$ is same as $\faZKReal_{\A, \FAS}$ except that NIZK is switched to simulation mode and IPFE ciphertext $\ct$ encrypts $\widetilde{\bfx} := (-\bft^T, 1)^T \in \mcal{M}' \subseteq \Z^{\ell+1}$ and functional keys $\pk_\bfy$ and $\sk_\bfy$ correspond to $\widetilde{\bfy} := (\bfy^T, f_\bfy(\bft) + f_\bfy(\bfx))^T \in \mcal{F}_{{\sf IP}, \ell+1}$. 
Observe that $f_\bfy(\bfx) = f_{\widetilde{\bfy}}(\widetilde{\bfx})$.
To prove zero-knowledge, we introduce an intermediate hybrid experiment $\Hyb_{\A, \fas}^{\nizk.\Sim}$ which is same as $\faZKReal_{\A, \FAS}$ except that only NIZK is switched to simulation mode.
Note that by triangle inequality, it follows that 
\ifacm
	\begin{align*}
	& | \Pr[\D(\faZKReal_{\A, \FAS}(1^\secparam, X, \bfx)) = 1] \\ 
	& - \Pr[ \D(\faZKIdeal_{\A, \FAS}^\Sim(1^\secparam, X, \bfx)) = 1] | \\ 
	& \quad \leq 
	| \Pr[\D(\faZKReal_{\A, \FAS}(1^\secparam, X, \bfx)) = 1]\\ 
	& \quad \quad - \Pr[ \D(\Hyb_{\A, \FAS}^{\nizk.\Sim}(1^\secparam, X, \bfx)) = 1] | \\ 
	& \quad \quad + | \Pr[ \D(\Hyb_{\A, \FAS}^{\nizk.\Sim}(1^\secparam, X, \bfx)) = 1]\\ 
	& \quad \quad - \Pr[\D(\faZKIdeal_{\A, \FAS}^\Sim(1^\secparam, X, \bfx)) = 1] |.
	\end{align*}
\else
	\begin{align*}
	& | \Pr[\D(\faZKReal_{\A, \FAS}(1^\secparam, X, \bfx)) = 1]
	- \Pr[ \D(\faZKIdeal_{\A, \FAS}^\Sim(1^\secparam, X, \bfx)) = 1] | \\ 
	& \quad \leq 
	| \Pr[\D(\faZKReal_{\A, \FAS}(1^\secparam, X, \bfx)) = 1]
	- \Pr[ \D(\Hyb_{\A, \FAS}^{\nizk.\Sim}(1^\secparam, X, \bfx)) = 1] | \\ 
	& \quad \quad + | \Pr[ \D(\Hyb_{\A, \FAS}^{\nizk.\Sim}(1^\secparam, X, \bfx)) = 1]
	- \Pr[\D(\faZKIdeal_{\A, \FAS}^\Sim(1^\secparam, X, \bfx)) = 1] |.
	\end{align*}
\fi
To complete the proof, we show in~\Cref{clm:fas-zk-nizk,clm:fas-zk-ipfe} that each of the two terms on the right-hand-side are at most $\negl(\secparam)$.
\end{proof}

\ifacm
	\begin{figure*}[t]
	\centering
	\captionsetup{justification=centering}
	\begin{pcvstack}[boxed, space=1em]
	\procedure[linenumbering, mode=text]{Experiments 
	$\faZKReal_{\A,\FAS}(1^\secparam, X, \bfx)$, 
	\pcbox{\text{$\Hyb_{\A,\FAS}^{\nizk.\Sim}(1^\secparam, X, \bfx)$}},
	{\color{blue}\pcbox{\color{black}\text{$\faZKIdeal^\Sim_{\A,\FAS}(1^\secparam, X, \bfx)$}}}.}{
	% \nikhil{does the nizk setup need to take $1^\ell$ as input too?}
	$\crs \gets \nizk.\Setup(1^\secparam)$
	{\color{blue} \pcbox{
	\color{black} \pcbox{\text{
	$(\crs, \td) \gets \nizk.\Setup^*(1^\secparam)$
	}}
	}}
	\\
	$\pp' \gets \ipfe.\Gen(1^\secparam)$
	\\ 
	$\pp := (\crs, \pp')$
	\\ 
	Sample random coins $r_0$, $(\mpk, \msk) := \ipe.\Setup(\pp', 1^{\ell+1}; r_0)$
	{\color{blue} \pcbox{
	\color{black} \pcbox{\text{
	$(\mpk, \msk) \gets \ipe.\Setup(\pp', 1^{\ell+1})$
	}}
	}}
	\\ 
	Sample $\bft \getr \mcal{M}$
	\\
	Let $\widetilde{\bfx} := (\bfx^T, 0)^T \in \mcal{M}' \subseteq \Z^{\ell+1}$ 
	{ \color{blue} \pcbox{ \color{black} \text{
	Let $\widetilde{\bfx} := (-\bft^T, 1)^T \in \mcal{M}' \subseteq \Z^{\ell+1}$ 
	}}}
	\\ 
	Sample random coins $r_1$, $\ct := \ipe.\Enc(\mpk, \widetilde{\bfx}; r_1)$
	{\color{blue} \pcbox{
	\color{black} \pcbox{\text{
	$\ct \gets \ipe.\Enc(\mpk, \widetilde{\bfx})$
	}}
	}}
	\\ 
	$\pi \gets \nizk.\Prove(\crs, (X, \pp', \mpk, \ct), (r_0, r_1, \bfx))$
	{\color{blue} \pcbox{
	\color{black} \pcbox{\text{
	$\pi \gets \nizk.\Prove^*(\crs, \td, (X, \pp', \mpk, \ct))$
	}}
	}}
	\\ 
	$\advt := (\mpk, \ct, \pi)$, $\state := (\msk, \bft)$
	\\ 
	$\vk = \bot$
	\\ 
	$\vk \gets \A^{\mcal{O}_{\AuxGen}(\cdot),\mcal{O}_{\Adapt}(\cdot, \cdot, \cdot)}(\pp, \advt, X)$ 
	{ \color{blue} \pcbox{ \color{black} \text{
	$\vk \gets \A^{\mcal{O}_{\AuxGen^*}(\cdot, \cdot), \mcal{O}_{\Adapt^*}(\cdot, \cdot, \cdot, \cdot)}(\pp, \advt, X)$ 
	}}}
	\\ 
	\ret view of $\A$
	}

	\procedure[linenumbering, mode=text]{Oracles 
	$\mcal{O}_{\AuxGen}(\bfy)$, 
	{ \color{blue} \pcbox{ \color{black} \text{
	$\mcal{O}^*_{\AuxGen}(\bfy, f_\bfy(\bfx))$}}}.}{
	Parse $\advt = (\mpk, \ct, \pi)$
	\\ 
	Compute $\widetilde{\bfy} := (\bfy^T, f_\bfy(\bft))^T \in \mcal{F}_{{\sf IP}, \ell+1}$
	{ \color{blue} \pcbox{ \color{black} \text{
	Compute $\widetilde{\bfy} := (\bfy^T, f_\bfy(\bft)+f_\bfy(\bfx))^T \in \mcal{F}_{{\sf IP}, \ell+1}$
	}}}
	\\ 
	Compute $\pk_\bfy := \ipfe.\PubKGen(\mpk, \widetilde{\bfy})$ 
	\\ 
	\ret $\aux_\bfy := \pk_\bfy$, $\pi_\bfy := f_\bfy(\bft)$
	{ \color{blue} \pcbox{ \color{black} \text{
	\ret $\aux_\bfy := \pk_\bfy$, $\pi_\bfy := f_\bfy(\bft)+f_\bfy(\bfx)$
	}}}
	}

	\procedure[linenumbering, mode=text]{Oracles 
	$\mcal{O}_{\Adapt}(m, \bfy, \widetilde{\sigma})$, 
	{ \color{blue} \pcbox{ \color{black} \text{
	$\mcal{O}^*_{\Adapt}(m, \bfy, \widetilde{\sigma}, f_\bfy(\bfx))$}}}.}{
	Parse $\advt = (\mpk, \ct, \pi)$, and $\state = (\msk, \bft)$
	\\ 
	If $\vk = \bot$: \ret $\bot$
	\\ 
	$(\aux_\bfy, \pi_\bfy) \gets \mcal{O}_{\AuxGen}(\bfy)$
	{ \color{blue} \pcbox{ \color{black} \text{
	$(\aux_\bfy, \pi_\bfy) \gets \mcal{O}^*_{\AuxGen}(\bfy, f_\bfy(\bfx))$
	}}}
	\\ 
	If $\AuxVerify(\advt, \bfy, \aux_\bfy, \pi_\bfy) = 0 \ \vee\ \as.\PreVerify(\vk, m, \aux_\bfy, \widetilde{\sigma}) = 0\ $: \ret $\bot$
	\\ 
	Compute $\widetilde{\bfy} := (\bfy^T, f_\bfy(\bft))^T \in \mcal{F}_{{\sf IP}, \ell+1}$
	{ \color{blue} \pcbox{ \color{black} \text{
	Compute $\widetilde{\bfy} := (\bfy^T, f_\bfy(\bft) + f_\bfy(\bfx))^T \in \mcal{F}_{{\sf IP}, \ell+1}$
	}}}
	\\ 
	Compute $\sk_\bfy := \ipe.\KGen(\msk, \widetilde{\bfy})$
	\\ 
	\ret $\sigma := \as.\Adapt(\vk, m, \aux_\bfy, \sk_\bfy, \widetilde{\sigma})$
	}
	\end{pcvstack}
	\caption{Zero-knowledge security of functional adaptor signatures: $\faZKReal_{\A,\FAS}$ is the real world experiment, $\Hyb_{\A,\FAS}^{\nizk.\Sim}$ is an intermediate hybrid experiment, $\faZKIdeal^\Sim_{\A,\FAS}$ is the ideal world experiment.}
	\label{fig:fasig-zk-games}
	\end{figure*}
\else
	\begin{figure}[h]
	\centering
	\captionsetup{justification=centering}
	\begin{pcvstack}[boxed, space=1em]
	\procedure[linenumbering, mode=text]{Experiments 
	$\faZKReal_{\A,\FAS}(1^\secparam, X, \bfx)$, 
	\pcbox{\text{$\Hyb_{\A,\FAS}^{\nizk.\Sim}(1^\secparam, X, \bfx)$}},
	{\color{blue}\pcbox{\color{black}\text{$\faZKIdeal^\Sim_{\A,\FAS}(1^\secparam, X, \bfx)$}}}.}{
	% \nikhil{does the nizk setup need to take $1^\ell$ as input too?}
	$\crs \gets \nizk.\Setup(1^\secparam)$
	{\color{blue} \pcbox{
	\color{black} \pcbox{\text{
	$(\crs, \td) \gets \nizk.\Setup^*(1^\secparam)$
	}}
	}}
	\\
	$\pp' \gets \ipfe.\Gen(1^\secparam)$,
	% \\ 
	$\pp := (\crs, \pp')$,
	% \\ 
	sample random coins $r_0$
	\\ 
	Let $(\mpk, \msk) := \ipe.\Setup(\pp', 1^{\ell+1}; r_0)$
	{\color{blue} \pcbox{
	\color{black} \pcbox{\text{
	$(\mpk, \msk) \gets \ipe.\Setup(\pp', 1^{\ell+1})$
	}}
	}}
	\\ 
	Sample $\bft \getr \mcal{M}$
	\\
	Let $\widetilde{\bfx} := (\bfx^T, 0)^T \in \mcal{M}' \subseteq \Z^{\ell+1}$ 
	{ \color{blue} \pcbox{ \color{black} \text{
	$\widetilde{\bfx} := (-\bft^T, 1)^T \in \mcal{M}' \subseteq \Z^{\ell+1}$ 
	}}}
	\\ 
	Sample random coins $r_1$, $\ct := \ipe.\Enc(\mpk, \widetilde{\bfx}; r_1)$
	{\color{blue} \pcbox{
	\color{black} \pcbox{\text{
	$\ct \gets \ipe.\Enc(\mpk, \widetilde{\bfx})$
	}}
	}}
	\\ 
	$\pi \gets \nizk.\Prove(\crs, (X, \pp', \mpk, \ct), (r_0, r_1, \bfx))$
	{\color{blue} \pcbox{
	\color{black} \pcbox{\text{
	$\pi \gets \nizk.\Prove^*(\crs, \td, (X, \pp', \mpk, \ct))$
	}}
	}}
	\\ 
	$\advt := (\mpk, \ct, \pi)$, $\state := (\msk, \bft)$,
	% \\ 
	$\vk := \bot$
	\\ 
	$\vk \gets \A^{\mcal{O}_{\AuxGen}(\cdot),\mcal{O}_{\Adapt}(\cdot, \cdot, \cdot)}(\pp, \advt, X)$ 
	{ \color{blue} \pcbox{ \color{black} \text{
	$\vk \gets \A^{\mcal{O}_{\AuxGen^*}(\cdot, \cdot), \mcal{O}_{\Adapt^*}(\cdot, \cdot, \cdot, \cdot)}(\pp, \advt, X)$ 
	}}}
	\\ 
	\ret view of $\A$
	}

	\procedure[linenumbering, mode=text]{Oracles 
	$\mcal{O}_{\AuxGen}(\bfy)$, 
	{ \color{blue} \pcbox{ \color{black} \text{
	$\mcal{O}^*_{\AuxGen}(\bfy, f_\bfy(\bfx))$}}}.}{
	Described in~\Cref{fig:fasig-zk-games-oracles}
	}

	\procedure[linenumbering, mode=text]{Oracles 
	$\mcal{O}_{\Adapt}(m, \bfy, \widetilde{\sigma})$, 
	{ \color{blue} \pcbox{ \color{black} \text{
	$\mcal{O}^*_{\Adapt}(m, \bfy, \widetilde{\sigma}, f_\bfy(\bfx))$}}}.}{
	Described in~\Cref{fig:fasig-zk-games-oracles}
	}
	\end{pcvstack}
	\caption{Zero-knowledge security of FAS: $\faZKReal_{\A,\FAS}$ is the real world experiment, $\Hyb_{\A,\FAS}^{\nizk.\Sim}$ is an intermediate hybrid experiment, $\faZKIdeal^\Sim_{\A,\FAS}$ is the ideal world experiment.}
	\label{fig:fasig-zk-games}
	\end{figure}

	\begin{figure}[h]
	\centering
	\captionsetup{justification=centering}
	\begin{pcvstack}[boxed, space=1em]

	\procedure[linenumbering, mode=text]{Oracles 
	$\mcal{O}_{\AuxGen}(\bfy)$, 
	{ \color{blue} \pcbox{ \color{black} \text{
	$\mcal{O}^*_{\AuxGen}(\bfy, f_\bfy(\bfx))$}}}.}{
	Parse $\advt = (\mpk, \ct, \pi)$
	\\ 
	Compute $\widetilde{\bfy} := (\bfy^T, f_\bfy(\bft))^T \in \mcal{F}_{{\sf IP}, \ell+1}$
	{ \color{blue} \pcbox{ \color{black} \text{
	Compute $\widetilde{\bfy} := (\bfy^T, f_\bfy(\bft)+f_\bfy(\bfx))^T \in \mcal{F}_{{\sf IP}, \ell+1}$
	}}}
	\\ 
	Compute $\pk_\bfy := \ipfe.\PubKGen(\mpk, \widetilde{\bfy})$ 
	\\ 
	\ret $\aux_\bfy := \pk_\bfy$, $\pi_\bfy := f_\bfy(\bft)$
	{ \color{blue} \pcbox{ \color{black} \text{
	\ret $\aux_\bfy := \pk_\bfy$, $\pi_\bfy := f_\bfy(\bft)+f_\bfy(\bfx)$
	}}}
	}

	\procedure[linenumbering, mode=text]{Oracles 
	$\mcal{O}_{\Adapt}(m, \bfy, \widetilde{\sigma})$, 
	{ \color{blue} \pcbox{ \color{black} \text{
	$\mcal{O}^*_{\Adapt}(m, \bfy, \widetilde{\sigma}, f_\bfy(\bfx))$}}}.}{
	Parse $\advt = (\mpk, \ct, \pi)$, and $\state = (\msk, \bft)$.
	% \\ 
	If $\vk = \bot$: \ret $\bot$
	\\ 
	$(\aux_\bfy, \pi_\bfy) \gets \mcal{O}_{\AuxGen}(\bfy)$
	{ \color{blue} \pcbox{ \color{black} \text{
	$(\aux_\bfy, \pi_\bfy) \gets \mcal{O}^*_{\AuxGen}(\bfy, f_\bfy(\bfx))$
	}}}
	\\ 
	If $\AuxVerify(\advt, \bfy, \aux_\bfy, \pi_\bfy) = 0 \ \vee\ \as.\PreVerify(\vk, m, \aux_\bfy, \widetilde{\sigma}) = 0\ $: \ret $\bot$
	\\ 
	Compute $\widetilde{\bfy} := (\bfy^T, f_\bfy(\bft))^T \in \mcal{F}_{{\sf IP}, \ell+1}$
	{ \color{blue} \pcbox{ \color{black} \text{
	Compute $\widetilde{\bfy} := (\bfy^T, f_\bfy(\bft) + f_\bfy(\bfx))^T \in \mcal{F}_{{\sf IP}, \ell+1}$
	}}}
	\\ 
	Compute $\sk_\bfy := \ipe.\KGen(\msk, \widetilde{\bfy})$
	\\ 
	\ret $\sigma := \as.\Adapt(\vk, m, \aux_\bfy, \sk_\bfy, \widetilde{\sigma})$
	}
	\end{pcvstack}
	\caption{Oracle descriptions for experiments $\faZKReal_{\A,\FAS}$, $\Hyb_{\A,\FAS}^{\nizk.\Sim}$, $\faZKIdeal^\Sim_{\A,\FAS}$.}
	\label{fig:fasig-zk-games-oracles}
	\end{figure}	
\fi

% \begin{figure}[t]
% \centering
% \captionsetup{justification=centering}
% \begin{pcvstack}[boxed, space=1em]
% \procedure[linenumbering, mode=text]{Experiment $\Hyb_{\A,\FAS}^{\nizk.\Sim}(1^\secparam, X, \bfx)$.}{
% {\color{blue}
% $\crs \gets \nizk.\Setup^*(1^\secparam)$
% }
% \\ 
% % \nikhil{does the nizk setup need to take $1^\ell$ as input too?}
% $\pp' \gets \ipfe.\Gen(1^\secparam)$
% \\ 
% $\pp = (\crs, \pp')$
% \\ 
% $(\mpk, \msk) \gets \ipe.\Setup(\pp', 1^\ell)$
% \\ 
% $r \getr \Z_p$, $\ct = \ipe.\Enc(\mpk, \bfx; r)$
% \\ 
% {\color{blue}
% $\pi \gets \nizk.\Prove^*(\crs, (X, \pp', \mpk, \ct))$
% } 
% \\ 
% $\advt = (\mpk, \ct, \pi)$, $\state = \msk$
% \\ 
% $\vk = \bot$
% \\ 
% $\vk \gets \A^{\mcal{O}_{\Adapt}(\cdot, \cdot, \cdot)}(\pp, \advt, X)$ 
% \\ 
% \ret view of $\A$
% }

% \procedure[linenumbering, mode=text]{Oracle $\mcal{O}_{\Adapt}(m, \bfy, \widetilde{\sigma})$.}{
% If $\vk = \bot$: \ret $\bot$
% \\ 
% $\pk_\bfy = \ipe.\PubKGen(\mpk, \bfy)$
% \\ 
% If $\as.\PreVerify(\vk, m, \pk_\bfy, \widetilde{\sigma}) = 0$: \ret $\bot$
% \\ 
% $\sk_\bfy = \ipe.\KGen(\msk, \bfy)$
% \\ 
% \ret $\sigma := \as.\Adapt(\vk, m, \pk_\bfy, \sk_\bfy, \widetilde{\sigma})$  
% }
% \end{pcvstack}

% \caption{Zero-Knowledge: hybrid experiment $\Hyb$}
% \label{fig:fasig-zk-hyb}
% \end{figure}

\begin{claim}
Suppose $\nizk$ satisfies zero-knowledge. 
% \nikhil{non-adaptive zk suffices, update the claim and the proof}
Then, for every stateful \ppt adversary \A, 
for the \ppt simulator $\nizk.\Sim$, 
there exists a negligible function $\negl$ such that
for all \ppt distinguishers \D, 
for all $\secparam \in \N$, 
for all $(X, \bfx) \in R$,
\ifacm
	\begin{align*}
	& | \Pr[\D(\faZKReal_{\A, \FAS}(1^\secparam, X, \bfx)) = 1] \\ 
	& \quad - \Pr[ \D(\Hyb_{\A, \FAS}^{\nizk.\Sim}(1^\secparam, X, \bfx)) = 1] | \leq  \negl(\secparam).
	\end{align*}
\else 
	\[
	| \Pr[\D(\faZKReal_{\A, \FAS}(1^\secparam, X, \bfx)) = 1] 
	- \Pr[ \D(\Hyb_{\A, \FAS}^{\nizk.\Sim}(1^\secparam, X, \bfx)) = 1] | \leq  \negl(\secparam).
	\]

\fi
\label{clm:fas-zk-nizk}
\end{claim}
\begin{proof}
Suppose towards a contradiction that for the stateful \ppt adversary \A that makes no queries to the oracles $\mcal{O}_\AuxGen$ and $\mcal{O}_\Adapt$, there exists a \ppt distinguisher $\D$ and a non-negligible value $\epsilon$ such that 
\ifacm
	\begin{align*}
	& | \Pr[\D(\faZKReal_{\A, \FAS}(1^\secparam, X, \bfx)) = 1] \\ 
	& \quad - \Pr[ \D(\Hyb_{\A, \FAS}^{\nizk.\Sim}(1^\secparam, X, \bfx)) = 1] | = \epsilon.
	\end{align*}
\else 
	\[
	| \Pr[\D(\faZKReal_{\A, \FAS}(1^\secparam, X, \bfx)) = 1] 
	- \Pr[ \D(\Hyb_{\A, \FAS}^{\nizk.\Sim}(1^\secparam, X, \bfx)) = 1] | = \epsilon.
	\]
\fi
Then, we build a \ppt distinguisher $\D_\nizk$ using $\D$ which breaks the zero-knowledge of NIZK for the NP language $L_\nizk$. This should complete the proof. 

Suppose the inputs to the distinguisher $D_\nizk$ are the NIZK common reference string $\crs$ and a NIZK proof $\pi$ for some statement $(X, \pp', \mpk, \ct) \in L_\nizk$. Then, $\D_\nizk$ is as follows. It sets $\pp=(\crs, \pp')$ and $\advt = (\mpk, \ct, \pi)$. As \A makes no oracle queries, the view of the adversary \A against the the zero-knowledge security of functional adaptor signatures is simply $(\pp, \advt, X)$. 
% \anote{The view includes the transcript of the experiment, which includes the oracle interactions too.} 
Then, $\D_\nizk$ runs the distinguisher $\D$ on inputs $(\pp, \advt, X)$ and returns whatever $\D$ returns.
Then, observe that 
\begin{align*}
\Pr&\left[ 
\begin{matrix} 
\D_\nizk(\crs, \pi) = 1: \\
\crs \gets \nizk.\Setup(1^\secparam), \\ 
\pi \gets \nizk.\Prove(\pp, (X, \pp', \mpk, \ct), (r_0, r_1, \bfx)) \end{matrix}\right] \\
& = 
\Pr\left[
\begin{matrix}
\D(\pp=(\crs, \pp'), \advt = (\mpk, \ct, \pi), X) = 1: \\ 
\crs \gets \nizk.\Setup(1^\secparam), \\ 
\pi \gets \nizk.\Prove(\pp, (X, \pp', \mpk, \ct), (r_0, r_1, \bfx)) 
\end{matrix}\right] \\ 
& = 
\Pr\left[\D(\faZKReal_{\A, \fas}(1^\secparam, X, \bfx)) = 1\right].
\end{align*}
Similarly, we get that 
\begin{align*}
\Pr&\left[
\begin{matrix} 
\D_\nizk(\crs, \pi) = 1: \\
(\crs, \td) \gets \nizk.\Setup^*(1^\secparam) \\
\pi \gets \nizk.\Prove^*(\crs, \td, \stmt=(X, \pp', \mpk, \ct)) 
\end{matrix}\right] \\ 
& =
\Pr\left[
\begin{matrix} 
\D(\pp=(\crs, \pp'), \advt = (\mpk, \ct, \pi), X) = 1: \\ 
(\crs, \td) \gets \nizk.\Setup^*(1^\secparam) \\
\pi \gets \nizk.\Prove^*(\crs, \td, \stmt=(X, \pp', \mpk, \ct)) 
\end{matrix}\right] \\ 
& = 
\Pr\left[\D(\Hyb_{\A, \fas}^{\nizk.\Sim}(1^\secparam, X, \bfx)) = 1\right].
\end{align*}
This implies that 
\begin{align*}
& \left| 
\begin{matrix}
\Pr\left[ 
\begin{matrix} 
\D_\nizk(\crs, \pi) = 1: \\
\crs \gets \nizk.\Setup(1^\secparam), \\ 
\pi \gets \nizk.\Prove(\pp, (X, \pp', \mpk, \ct), (\bfx, r)) \end{matrix}\right] \\
- 
\Pr\left[
\begin{matrix} 
\D_\nizk(\crs, \pi) = 1: \\
(\crs, \pi) \gets \nizk.\Sim(1^\secparam, \stmt=(X, \pp', \mpk, \ct)) \end{matrix}\right]
\end{matrix} \right|\\ 
& = \left| 
\begin{matrix}
\Pr\left[\D(\faZKReal_{\A, \fas}(1^\secparam, X, \bfx)) = 1\right] \\ 
- \Pr\left[\D(\Hyb_{\A, \fas}^{\nizk.\Sim}(1^\secparam, X, \bfx)) = 1\right] 
\end{matrix} \right| \\ 
& = \epsilon.
\end{align*}

As $\epsilon$ is non-negligible, hence, $\D_\nizk$ breaks the zero-knowledge of the underlying NIZK scheme.

\end{proof}

\begin{claim}
Suppose $\mcal{M}$ is an additive group and $\ipfe$ satisfies selective, IND-security.
Then, for every stateful \ppt adversary \A, 
for the \ppt simulators $\nizk.\Sim$ and $\Sim$ (described in~\Cref{fig:fas-zk-sim}), 
there exists a negligible function $\negl$ such that
for all \ppt distinguishers \D, 
for all $\secparam \in \N$, 
for all $(X, \bfx) \in R$,
\ifacm
	\begin{align*}
	& | \Pr[ \D(\Hyb_{\A, \FAS}^{\nizk.\Sim}(1^\secparam, X, \bfx)) = 1]\\ 
	& \quad  - \Pr[\D(\faZKIdeal_{\A, \FAS}^\Sim(1^\secparam, X, \bfx)) = 1] | \leq  \negl(\secparam).
	\end{align*}
\else 
\[
	| \Pr[ \D(\Hyb_{\A, \FAS}^{\nizk.\Sim}(1^\secparam, X, \bfx)) = 1] 
	- \Pr[\D(\faZKIdeal_{\A, \FAS}^\Sim(1^\secparam, X, \bfx)) = 1] | \leq  \negl(\secparam).
\]
\fi
\label{clm:fas-zk-ipfe}
\end{claim}
\begin{proof}
Observe that experiments $\Hyb_{\A, \FAS}^{\nizk.\Sim}$ and $\faZKIdeal_{\A, \FAS}^\Sim$ differ in the following two aspects:

\begin{itemize}[leftmargin=*]
\item 
The vector $\widetilde{\bfx}$ that ciphertext $\ct$ encrypts: 
In $\Hyb_{\A, \FAS}^{\nizk.\Sim}$, we have 
$\widetilde{\bfx} := (\bfx^T, 0)^T \in \mcal{M}' \subseteq \Z^{\ell+1}$, 
and in $\faZKIdeal_{\A, \FAS}^\Sim$, we have 
$\widetilde{\bfx} := (-\bft^T, 1)^T \in \mcal{M}' \subseteq \Z^{\ell+1}$. 
\item 
The vector $\widetilde{\bfy}$ used for computing $\pk_\bfy$ and $\sk_\bfy$ for every $\bfy$ queried to $\AuxGen$ and $\Adapt$ oracles: 
In $\Hyb_{\A, \FAS}^{\nizk.\Sim}$, we have 
$\widetilde{\bfy} := (\bfy^T, f_\bfy(\bft))^T \in \mcal{F}_{{\sf IP}, \ell+1}$,
and in $\faZKIdeal_{\A, \FAS}^\Sim$, we have 
$\widetilde{\bfy} := (\bfy^T, f_\bfy(\bft) + f_\bfy(\bfx))^T \in \mcal{F}_{{\sf IP}, \ell+1}$.

\end{itemize}

Consider an intermediate hybrid $\overline{\Hyb}_{\A, \FAS}^{\nizk.\Sim}$ that uses $\widetilde{\bfx}$ as in $\Hyb_{\A, \FAS}^{\nizk.\Sim}$ and $\widetilde{\bfy}$ as in $\faZKIdeal_{\A, \FAS}^\Sim$. 

We argue that $\Hyb$ and $\overline{\Hyb}$ are identically distributed. The two differ in how $\widetilde{\bfy}$ is computed. 
Notice that by linearity, it follows that $f_\bfy(\bft) + f_\bfy(\bfx) = f_\bfy(\bft+\bfx)$. Hence, one can view the change from $\Hyb$ to $\overline{\Hyb}$ as simply a change of variables $\widetilde{\bft} = \bft \to \widetilde{\bft} = \bft+\bfx$. 
As $\mcal{M}$ is an additive group and $\bft, \bfx \in \mcal{M}$, hence, $\bft+\bfx \in \mcal{M}$.
As $\bft \getr \mcal{M}$, hence, it follows that $\bft$ and $\bft+\bfx$ are identically distributed. 
Therefore, $\Hyb$ and $\overline{\Hyb}$ are identically distributed. In other words,
\ifacm
	\begin{align*}
	& | \Pr[ \D(\Hyb_{\A, \FAS}^{\nizk.\Sim}(1^\secparam, X, \bfx)) = 1]\\ 
	& \quad  - \Pr[\D(\overline{\Hyb}_{\A, \FAS}^{\nizk.\Sim}(1^\secparam, X, \bfx)) = 1] | = 0.
	\end{align*}
\else 
\[
	| \Pr[ \D(\Hyb_{\A, \FAS}^{\nizk.\Sim}(1^\secparam, X, \bfx)) = 1] 
	- \Pr[\D(\overline{\Hyb}_{\A, \FAS}^{\nizk.\Sim}(1^\secparam, X, \bfx)) = 1] | = 0.
\]
\fi
\nikhil{add more details to this proof.}
\nikhil{do we need to assume that the IPFE should satisfy the property that functional keys for vectors that are linearly dependent on previously queried vectors may be computed as linear combination of previously returned keys.}

To complete the proof, we argue that if $\ipfe$ is selective, IND-secure, then, 
\ifacm
	\begin{align*}
	& | \Pr[ \D(\overline{\Hyb}_{\A, \FAS}^{\nizk.\Sim}(1^\secparam, X, \bfx)) = 1]\\ 
	& \quad  - \Pr[\D(\faZKIdeal_{\A, \FAS}^\Sim(1^\secparam, X, \bfx)) = 1] | \leq  \negl(\secparam).
	\end{align*}
\else 
\[ 
	| \Pr[ \D(\overline{\Hyb}_{\A, \FAS}^{\nizk.\Sim}(1^\secparam, X, \bfx)) = 1]
	- \Pr[\D(\faZKIdeal_{\A, \FAS}^\Sim(1^\secparam, X, \bfx)) = 1] | \leq  \negl(\secparam).
\]
\fi

Suppose towards a contradiction that for any stateful \ppt adversary \A, there exists a 
there exists a \ppt distinguisher $\D$ and a non-negligible value $\epsilon$ such that 
\ifacm 
	\begin{align*}
	& | \Pr[ \D(\overline{\Hyb}_{\A, \FAS}^{\nizk.\Sim}(1^\secparam, X, \bfx)) = 1] \\ 
	& \quad - \Pr[\D(\faZKIdeal_{\A, \FAS}^\Sim(1^\secparam, X, \bfx)) = 1] | = \epsilon.
	\end{align*}
\else 
\[
	| \Pr[ \D(\overline{\Hyb}_{\A, \FAS}^{\nizk.\Sim}(1^\secparam, X, \bfx)) = 1]  
	- \Pr[\D(\faZKIdeal_{\A, \FAS}^\Sim(1^\secparam, X, \bfx)) = 1] | = \epsilon.
\]
\fi
Then, we build a \ppt reduction \B using \D that breaks the selective, IND-security of IPFE. This should complete the proof.

\ifacm
	\begin{figure}[t]
	\centering
	\captionsetup{justification=centering}
	\begin{pcvstack}[boxed, space=1em]
	\procedure[linenumbering, mode=text]{Reduction $\B(1^\secparam, X, \bfx)$ for proof of~\Cref{clm:fas-zk-ipfe}.}{
	$\pp' \gets \C(1^\secparam)$
	\\ 
	Sample $\bft \getr \mcal{M}$
	\\ 
	Let $\widetilde{\bfx_0} := (\bfx^T, 0)^T \in \mcal{M}' \subseteq \Z^{\ell+1}$ 
	\\
	Let $\widetilde{\bfx_1} := (-\bft^T, 1)^T \in \mcal{M}' \subseteq \Z^{\ell+1}$ 
	\\
	$(\mpk, \ct) \gets \C(\widetilde{\bfx_0}, \widetilde{\bfx_1})$
	\\ 
	$(\crs, \td) \gets \nizk.\Setup^*(1^\secparam)$
	\\ 
	$\pp = (\crs, \pp')$
	\\ 
	$\pi \gets \nizk.\Prove^*(\crs, \td, (X, \pp', \mpk, \ct))$
	\\ 
	$\advt = (\mpk, \ct, \pi)$, $\state = \bft$
	\\ 
	$\vk = \bot$
	\\
	$\vk \gets \A^{\mcal{O}_{\AuxGen^*}(\cdot, \cdot) \mcal{O}_{\Adapt^*}(\cdot, \cdot, \cdot, \cdot)}(\pp, \advt, X)$ 
	\\ 
	$b \gets \D(\text{view of \A})$ 
	\\ 
	\ret bit $b$ to \C
	}

	\procedure[linenumbering, mode=text]{Oracle 
	$\mcal{O}^*_{\AuxGen}(\bfy, f_\bfy(\bfx))$.}{
	Parse $\advt = (\mpk, \ct, \pi)$ and $\state=\bft$
	\\ 
	Compute $\widetilde{\bfy} := (\bfy^T, f_\bfy(\bft)+f_\bfy(\bfx))^T \in \mcal{F}_{{\sf IP}, \ell+1}$
	\\ 
	Compute $\pk_\bfy := \ipfe.\PubKGen(\mpk, \widetilde{\bfy})$ 
	\\ 
	\ret $\aux_\bfy := \pk_\bfy$, $\pi_\bfy := f_\bfy(\bft)+f_\bfy(\bfx)$
	}

	\procedure[linenumbering, mode=text]{Oracle 
	$\mcal{O}^*_{\Adapt}(m, \bfy, \widetilde{\sigma}, f_\bfy(\bfx))$.}{
	Parse $\advt = (\mpk, \ct, \pi)$, and $\state = \bft$
	\\ 
	If $\vk = \bot$: \ret $\bot$
	\\ 
	$(\aux_\bfy, \pi_\bfy) \gets \mcal{O}^*_{\AuxGen}(\bfy, f_\bfy(\bfx))$
	\\ 
	If $\AuxVerify(\advt, \bfy, \aux_\bfy, \pi_\bfy)=0 \ \vee$  
	\pcskipln \\
	$\as.\PreVerify(\vk, m, \aux_\bfy, \widetilde{\sigma}) = 0\ $: 
	\pcskipln \\
	$\pcind$ \ret $\bot$
	\\ 
	Compute $\widetilde{\bfy} := (\bfy^T, f_\bfy(\bft) + f_\bfy(\bfx))^T \in \mcal{F}_{{\sf IP}, \ell+1}$
	\\ 
	Compute $\sk_\bfy \gets \ipe.\mcal{O}_\KGen(\widetilde{\bfy})$
	\\ 
	\ret $\sigma := \as.\Adapt(\vk, m, \aux_\bfy, \sk_\bfy, \widetilde{\sigma})$
	}

	% \procedure[linenumbering, mode=text]{Oracle $\mcal{O}_{\Adapt}(m, \bfy, \widetilde{\sigma})$.}{
	% If $\vk = \bot$: \ret $\bot$
	% \\ 
	% $\pk_\bfy = \ipe.\PubKGen(\mpk, \bfy)$
	% \\ 
	% If $\as.\PreVerify(\vk, m, \pk_\bfy, \widetilde{\sigma}) = 0$: \ret $\bot$
	% \\ 
	% $\sk_\bfy \gets \ipfe.\mcal{O}_\KGen(\bfy)$ 
	% \pcskipln \\ \nikhil{in case of ideal game, should syntax change?} 
	% \pcskipln \\ \anote{Yes, it is just some oracle access $\B$ does not know which.}
	% \\ 
	% \ret $\sigma := \as.\Adapt(\vk, m, \pk_\bfy, \sk_\bfy, \widetilde{\sigma})$  
	% }
	\end{pcvstack}

	\caption{Reduction \B for proof of~\Cref{clm:fas-zk-ipfe}}
	\label{fig:fasig-zk-reduction}
	\end{figure}
\else
	\begin{figure}[h]
	\centering
	\captionsetup{justification=centering}
	\begin{pchstack}[boxed, space=1em]
	\begin{pcvstack}[space=1em]
	\procedure[linenumbering, mode=text]{Reduction $\B(1^\secparam, X, \bfx)$ for proof of~\Cref{clm:fas-zk-ipfe}.}{
	$\pp' \gets \C(1^\secparam)$
	\\ 
	Sample $\bft \getr \mcal{M}$
	\\ 
	Let $\widetilde{\bfx_0} := (\bfx^T, 0)^T \in \mcal{M}' \subseteq \Z^{\ell+1}$ 
	\\
	Let $\widetilde{\bfx_1} := (-\bft^T, 1)^T \in \mcal{M}' \subseteq \Z^{\ell+1}$ 
	\\
	$(\mpk, \ct) \gets \C(\widetilde{\bfx_0}, \widetilde{\bfx_1})$
	\\ 
	$(\crs, \td) \gets \nizk.\Setup^*(1^\secparam)$
	\\ 
	$\pp = (\crs, \pp')$
	\\ 
	$\pi \gets \nizk.\Prove^*(\crs, \td, (X, \pp', \mpk, \ct))$
	\\ 
	$\advt = (\mpk, \ct, \pi)$, $\state = \bft$
	\\ 
	$\vk = \bot$
	\\
	$\vk \gets \A^{\mcal{O}_{\AuxGen^*}(\cdot, \cdot) \mcal{O}_{\Adapt^*}(\cdot, \cdot, \cdot, \cdot)}(\pp, \advt, X)$ 
	\\ 
	$b \gets \D(\text{view of \A})$ 
	\\ 
	\ret bit $b$ to \C
	}

	\end{pcvstack}
	\begin{pcvstack}[space=1em]

	\procedure[linenumbering, mode=text]{Oracle 
	$\mcal{O}^*_{\AuxGen}(\bfy, f_\bfy(\bfx))$.}{
	Parse $\advt = (\mpk, \ct, \pi)$ and $\state=\bft$
	\\ 
	Compute $\widetilde{\bfy} := (\bfy^T, f_\bfy(\bft)+f_\bfy(\bfx))^T \in \mcal{F}_{{\sf IP}, \ell+1}$
	\\ 
	Compute $\pk_\bfy := \ipfe.\PubKGen(\mpk, \widetilde{\bfy})$ 
	\\ 
	\ret $\aux_\bfy := \pk_\bfy$, $\pi_\bfy := f_\bfy(\bft)+f_\bfy(\bfx)$
	}

	\procedure[linenumbering, mode=text]{Oracle 
	$\mcal{O}^*_{\Adapt}(m, \bfy, \widetilde{\sigma}, f_\bfy(\bfx))$.}{
	Parse $\advt = (\mpk, \ct, \pi)$, and $\state = \bft$
	\\ 
	If $\vk = \bot$: \ret $\bot$
	\\ 
	$(\aux_\bfy, \pi_\bfy) \gets \mcal{O}^*_{\AuxGen}(\bfy, f_\bfy(\bfx))$
	\\ 
	If $\AuxVerify(\advt, \bfy, \aux_\bfy, \pi_\bfy)=0 \ \vee$  
	\pcskipln \\
	$\as.\PreVerify(\vk, m, \aux_\bfy, \widetilde{\sigma}) = 0\ $: 
	\pcskipln \\
	$\pcind$ \ret $\bot$
	\\ 
	Compute $\widetilde{\bfy} := (\bfy^T, f_\bfy(\bft) + f_\bfy(\bfx))^T \in \mcal{F}_{{\sf IP}, \ell+1}$
	\\ 
	Compute $\sk_\bfy \gets \ipe.\mcal{O}_\KGen(\widetilde{\bfy})$
	\\ 
	\ret $\sigma := \as.\Adapt(\vk, m, \aux_\bfy, \sk_\bfy, \widetilde{\sigma})$
	}

	% \procedure[linenumbering, mode=text]{Oracle $\mcal{O}_{\Adapt}(m, \bfy, \widetilde{\sigma})$.}{
	% If $\vk = \bot$: \ret $\bot$
	% \\ 
	% $\pk_\bfy = \ipe.\PubKGen(\mpk, \bfy)$
	% \\ 
	% If $\as.\PreVerify(\vk, m, \pk_\bfy, \widetilde{\sigma}) = 0$: \ret $\bot$
	% \\ 
	% $\sk_\bfy \gets \ipfe.\mcal{O}_\KGen(\bfy)$ 
	% \pcskipln \\ \nikhil{in case of ideal game, should syntax change?} 
	% \pcskipln \\ \anote{Yes, it is just some oracle access $\B$ does not know which.}
	% \\ 
	% \ret $\sigma := \as.\Adapt(\vk, m, \pk_\bfy, \sk_\bfy, \widetilde{\sigma})$  
	% }
	\end{pcvstack}
	\end{pchstack}

	\caption{Reduction \B for proof of~\Cref{clm:fas-zk-ipfe}}
	\label{fig:fasig-zk-reduction}
	\end{figure}
\fi
Let the IPFE challenger be \C and let its $\KGen$ oracle be denoted be $\ipfe.\mcal{O}_\KGen$. The reduction \B is as in~\Cref{fig:fasig-zk-reduction}.
Observe that if the challenger \C chooses to encrypt $\widetilde{\bfx_0}$, then, \A's view is same as in $\overline{\Hyb}$ and if it chooses to encrypt $\widetilde{\bfx_1}$, then \A's view is same as in $\faZKIdeal$. Thus, if \D can distinguish the two views of \A with a non-negligible distinguishing advantage $\epsilon$ against \B, then, \B has the same non-negligible distinguishing advantage $\epsilon$ against \C. Thus, \B breaks the selective, IND-security of IPFE.

\end{proof}

%% file: fas-construction-weak.tex
\section{Weakly-Secure $\fas$ Construction}
\label{sec:fas-construction-weak}

Recall that our main $\fas$ construction was obtained by 
starting with a simulation-secure IPFE scheme. 
The construction was presented using an IND-secure IPFE and applying the IND-secure IPFE to simulation-secure IPFE compiler of~\cite{ad-sim-ipfe} in a non-black-box way within our $\fas$ construction. Here, we show that if (i) we do not apply the compiler and just use the IND-secure IPFE, (ii) use NIZK with adaptive zero-knowledge, then, the resulting $\fas$ construction is weakly-secure. This means that it satisfies all security properties as before except that for witness privacy, it satisfies witness indistinguishability. We first define adaptive zero-knowledge of NIZK and then describe the FAS construction next.
% \nikhil{add $L'_\nizk$ description.}

\begin{definition}[Adaptively Secure NIZK Argument System]
\label{def:adaptive-nizk}
An adaptively secure NIZK argument system must satisfy completeness and adaptive soundness as before. Additionally, it must satisfy \emph{adaptive zero-knowledge} which requires that 
% for all \ppt verifiers $\mcal{V}^*$,
there exists a 
% \pnote{NIZK simulators are strict polytime} 
\ppt simulator $\Sim = (\Setup^*, \Prove^*)$ 
and there exists a negligible function $\negl$ such that 
for all \ppt distinguishers \D, 
for all $\secparam \in \N$,
for all $(\stmt, \wit) \in R$, 
% \pnote{verifier shouldn't be choosing the witness}
% \pnote{what is adaptive about this definition?}
% \pnote{better to define two games: REAL and IDEAL played with the verifier where games output views, and then ask views to be indistinguishbable}
% \pnote{It is better to avoid the notation $\approx_c$ and spell it out completely using a distinguisher.}
\[
| 
\Pr[\D(\crs, \pi) = 1] 
-
\Pr[\D(\crs^*, \pi^*) = 1]
| \leq \negl(\secparam),
\]
where
$\crs \gets \Setup(1^\secparam)$,
% $\crs^* \gets \Setup^*(1^\secparam)$,
$\pi \gets \Prove(\crs, \stmt,$ $ \wit)$, 
and 
$(\crs^*, {\sf td}) \gets \Setup^*(1^\secparam)$, 
$\pi \gets \Prove^*(\crs^*, {\sf td}, \stmt)$.
Here, ${\sf td}$ is the trapdoor for the simulated $\crs^*$ that is used by $\Prove^*$ to generate an accepting proof for $\stmt$ without knowing the witness.
\end{definition}

Our generic construction of weakly-secure functional adaptor signatures will use the following building blocks:
\begin{itemize}[leftmargin=*]
\item 
An inner product functional encryption scheme $\ipe$ that satisfies selective, IND-security (\Cref{def:ipfe-sel-ind-sec}), 
$R_\ipfe$-compliance (\Cref{def:ipfe-compliant}) and 
$R'_\ipfe$-robustness (\Cref{def:ipfe-robust})  
w.r.t.\ hard relations $R_\ipfe$, $R'_\ipfe$ such that $R_\ipfe \subseteq R'_\ipfe$. The message space of $\ipfe$ is $\mcal{M}' \subseteq \Z^{\ell}$ and the function family is $\mcal{F}_{{\sf IP}, \ell}$.

\item 
An adaptor signature scheme $\as$ w.r.t.\ a digital signature scheme $\DS$, and hard relations $R_\ipfe$ and $R'_\ipfe$ that satisfies witness extractability (\Cref{def:as-wit-ext}) and weak pre-signature adaptability (\Cref{def:as-pre-sig-adaptability}) security properties. 

\item 
An adaptively secure $\nizk$ (\Cref{def:adaptive-nizk}) for the NP language $L_\nizk$ defined as 
\[
L_\nizk := \left\{ 
(X, \pp', \mpk, \ct) : 
\begin{matrix}
\exists (r_0, r_1, \bfx)\ \text{such that}\\
\pp' \in [\ipe.\Gen(1^\secparam)],\\
(\mpk, \msk) = \ipe.\Setup(\pp', 1^{\ell+1}; r_0),\\
(X, \bfx) \in R,\\ 
\ct = \ipe.\Enc(\mpk, \bfx; r_1)
\end{matrix}
\right\}.
\]\
\end{itemize}

Our construction is as in~\Cref{fig:fas-construction-weak}. Observe that it uses $\ell$ slots instead of $\ell+1$ slots in the underlying $\ipfe$, thus making the construction slightly more efficient than before. Further, compared to before, $\AuxGen$ and $\AuxVerify$ are redundant, thus, can be skipped in the protocol resulting in a non-interactive pre-signing stage. The construction can be instantiated from prime-order groups and lattices as before.
\ifacm
	\begin{figure}[t]
	\centering
	\captionsetup{justification=centering}
	\begin{pcvstack}[boxed, space=1em]
	\procedure[linenumbering, mode=text]{$\Setup(1^\secparam)$}{
	Sample $\crs \gets \nizk.\Setup(1^\secparam)$
	\\ 
	% \nikhil{does the nizk setup need to take $1^\ell$ as input too?}
	Sample $\pp' \gets \ipfe.\Gen(1^\secparam)$
	\\ 
	\ret $\pp := (\crs, \pp')$
	}

	\procedure[linenumbering, mode=text]{$\AdvertisementGen(\pp, X, \bfx)$:}{
	Sample random coins $r_0$, $r_1$ 
	\\ 
	Compute $(\mpk, \msk) := \ipe.\Setup(\pp', 1^{\ell}; r_0)$
	\\ 
	Compute $\ct := \ipe.\Enc(\mpk, \bfx; r_1)$
	\\ 
	Compute $\pi \gets \nizk.\Prove(\crs, (X, \pp', \mpk, \ct), (r_0, r_1, \bfx))$
	\\ 
	\ret $\advt := (\mpk, \ct, \pi)$, and $\state := \msk$
	}

	\procedure[linenumbering, mode=text]{$\AdvertisementVerify(\pp, X, \advt)$}{
	\ret $\nizk.\Verify(\crs, (X, \pp', \mpk, \ct), \pi)$
	}

	\procedure[linenumbering, mode=text]{$\AuxGen(\advt, \state, \bfy)$}{
	\ret $\aux_\bfy = \bot, \pi_\bfy = \bot$
	}

	\procedure[linenumbering, mode=text]{$\AuxVerify(\advt, \bfy, \aux_\bfy, \pi_\bfy)$}{
	\ret $1$
	}

	\procedure[linenumbering, mode=text]{$\FPreSign(\advt, \sk, m, X, \bfy, \aux_\bfy)$}{
	Parse $\advt = (\mpk, \ct, \pi)$
	\\
	Compute $\pk_\bfy := \ipfe.\PubKGen(\mpk, \bfy)$
	\\
	\ret $\widetilde{\sigma} \gets \as.\PreSign(\sk, m, \pk_\bfy)$
	}

	\procedure[linenumbering, mode=text]{$\FPreVerify(\advt, \vk, m, X, \bfy, \aux_\bfy, \pi_\bfy, \widetilde{\sigma})$}{
	Parse $\advt = (\mpk, \ct, \pi)$
	\\
	Compute $\pk_\bfy := \ipfe.\PubKGen(\mpk, \bfy)$
	\\
	\ret $(\AuxVerify(\advt, \bfy, \aux_\bfy, \pi_\bfy)) \wedge$ 
	\pcskipln \\ 
	$\quad \qquad (\as.\PreVerify(\vk, m, \pk_\bfy, \widetilde{\sigma})) $
	}

	\procedure[linenumbering, mode=text]{$\Adapt(\advt, \state, \vk, m, X, \bfx, \bfy, \aux_\bfy, \widetilde{\sigma})$}{
	Parse $\advt = (\mpk, \ct, \pi)$, and $\state = (\msk, \bft)$
	\\ 
	Compute $\pk_\bfy := \ipfe.\PubKGen(\mpk, \bfy)$
	\\
	Compute $\sk_\bfy := \ipe.\KGen(\msk, \bfy)$
	\\ 
	\ret $\sigma := \as.\Adapt(\vk, m, \pk_\bfy, \sk_\bfy, \widetilde{\sigma})$
	}

	\procedure[linenumbering, mode=text]{$\FExt(\advt, \widetilde{\sigma}, \sigma, X, \bfy, \aux_\bfy)$}{ 
	Parse $\advt = (\mpk, \ct, \pi)$.
	\\ 
	Compute $\pk_\bfy := \ipe.\PubKGen(\mpk, \bfy)$
	\\ 
	Compute $z := \as.\Ext(\widetilde{\sigma}, \sigma, \pk_\bfy)$
	\\ 
	\ret $v := \ipe.\Dec(z, \ct)$
	}

	\end{pcvstack}
	\caption{Construction: Weakly-secure $\fas$}
	\label{fig:fas-construction-weak}
	\label{sec:construction-weak}
	\end{figure}
\else 
	\begin{figure}[h]
	\centering
	\captionsetup{justification=centering}
	\begin{pchstack}[boxed, space=1em]
	\begin{pcvstack}[space=1em]
	\procedure[linenumbering, mode=text]{$\Setup(1^\secparam)$}{
	Sample $\crs \gets \nizk.\Setup(1^\secparam)$
	\\ 
	% \nikhil{does the nizk setup need to take $1^\ell$ as input too?}
	Sample $\pp' \gets \ipfe.\Gen(1^\secparam)$
	\\ 
	\ret $\pp := (\crs, \pp')$
	}

	\procedure[linenumbering, mode=text]{$\AdvertisementGen(\pp, X, \bfx)$:}{
	Sample random coins $r_0$, $r_1$ 
	\\ 
	Let $(\mpk, \msk) := \ipe.\Setup(\pp', 1^{\ell}; r_0)$
	\\ 
	Let $\ct := \ipe.\Enc(\mpk, \bfx; r_1)$
	\\ 
	Let $\pi \gets \nizk.\Prove(\crs, (X, \pp', \mpk, \ct), (r_0, r_1, \bfx))$
	\\ 
	\ret $\advt := (\mpk, \ct, \pi)$, and $\state := \msk$
	}

	\procedure[linenumbering, mode=text]{$\AdvertisementVerify(\pp, X, \advt)$}{
	\ret $\nizk.\Verify(\crs, (X, \pp', \mpk, \ct), \pi)$
	}

	\procedure[linenumbering, mode=text]{$\AuxGen(\advt, \state, \bfy)$}{
	\ret $\aux_\bfy = \bot, \pi_\bfy = \bot$
	}

	\procedure[linenumbering, mode=text]{$\AuxVerify(\advt, \bfy, \aux_\bfy, \pi_\bfy)$}{
	\ret $1$
	}

	\end{pcvstack}
	\begin{pcvstack}[space=1em]

	\procedure[linenumbering, mode=text]{$\FPreSign(\advt, \sk, m, X, \bfy, \aux_\bfy)$}{
	Parse $\advt = (\mpk, \ct, \pi)$
	\\
	Let $\pk_\bfy := \ipfe.\PubKGen(\mpk, \bfy)$
	\\
	\ret $\widetilde{\sigma} \gets \as.\PreSign(\sk, m, \pk_\bfy)$
	}

	\procedure[linenumbering, mode=text]{$\FPreVerify(\advt, \vk, m, X, \bfy, \aux_\bfy, \pi_\bfy, \widetilde{\sigma})$}{
	Parse $\advt = (\mpk, \ct, \pi)$
	\\
	Let $\pk_\bfy := \ipfe.\PubKGen(\mpk, \bfy)$
	\\
	\ret $(\AuxVerify(\advt, \bfy, \aux_\bfy, \pi_\bfy)) \wedge$ 
	\pcskipln \\ 
	$\quad \qquad (\as.\PreVerify(\vk, m, \pk_\bfy, \widetilde{\sigma})) $
	}

	\procedure[linenumbering, mode=text]{$\Adapt(\advt, \state, \vk, m, X, \bfx, \bfy, \aux_\bfy, \widetilde{\sigma})$}{
	Parse $\advt = (\mpk, \ct, \pi)$, and $\state = (\msk, \bft)$
	\\ 
	Let $\pk_\bfy := \ipfe.\PubKGen(\mpk, \bfy)$
	\\
	Let $\sk_\bfy := \ipe.\KGen(\msk, \bfy)$
	\\ 
	\ret $\sigma := \as.\Adapt(\vk, m, \pk_\bfy, \sk_\bfy, \widetilde{\sigma})$
	}

	\procedure[linenumbering, mode=text]{$\FExt(\advt, \widetilde{\sigma}, \sigma, X, \bfy, \aux_\bfy)$}{ 
	Parse $\advt = (\mpk, \ct, \pi)$.
	\\ 
	Let $\pk_\bfy := \ipe.\PubKGen(\mpk, \bfy)$
	\\ 
	Let $z := \as.\Ext(\widetilde{\sigma}, \sigma, \pk_\bfy)$
	\\ 
	\ret $v := \ipe.\Dec(z, \ct)$
	}

	\end{pcvstack}
	\end{pchstack}
	\caption{Construction: Weakly-secure $\fas$}
	\label{fig:fas-construction-weak}
	\label{sec:construction-weak}
	\end{figure}
\fi
\begin{lemma}\label{lemma:correctness-fas-weak}
Suppose NIZK satisfies correctness, $\AS$ satisfies correctness, and $\ipfe$ satisfies $R_\ipfe$-compliance and  $R'_\ipfe$-robustness. Then, the functional adaptor signatures construction in~\Cref{sec:construction-weak} satisfies correctness.
\end{lemma}
\begin{proof}
Similar to the proof of~\Cref{lemma:correctness-fas}.
\end{proof}

\begin{theorem}
\label{thm:fas-weakly-secure}
Let $\mcal{F}_{{\sf IP}, \ell}$ be the function family for computing inner products of vectors of length $\ell$. 
Let $R$ be any NP relation 
with statement/witness pairs $(X, \bfx)$ such that $\bfx \in \mcal{M}$ for some set $\mcal{M} \subseteq \Z^\ell$.
Suppose that 
\begin{itemize}
\item 
\mcal{M} is an additive group,
\item
$R$ is $\mcal{F}_{{\sf IP}, \ell}$-hard (\Cref{def:f-hard-relation}),
\item
$\nizk$ is a secure NIZK argument system (\Cref{def:nizk}),
\item
$\as$ is an adaptor signature scheme w.r.t.\ digital signature scheme $\ds$ and hard relations $R_\ipfe, R'_\ipfe$ that satisfies weak pre-signature adaptability (\Cref{def:as-pre-sig-adaptability}) and witness extractability (\Cref{def:as-wit-ext}),
% \anote{What is this notion? where are we defining this?} 
% \anote{refer to the definitions here.}
\item
$\ipfe$ is a selective, IND-secure IPFE scheme (\Cref{def:ipfe-sel-ind-sec}) for function family $\mcal{F}_{{\sf IP}, \ell}$ that is $R_\ipe$-compliant (\Cref{def:ipfe-compliant}) and $R'_\ipfe$-robust (\Cref{def:ipfe-robust}). 
% \anote{We better formalize these properties into definitions}
\end{itemize}  
Then, the functional adaptor signature scheme 
w.r.t.\ digital signature scheme $\DS$, NP relation $R$, and family of inner product functions $\mcal{F}_{{\sf IP}, \ell}$  
constructed in~\Cref{sec:construction-weak} is weakly-secure (\Cref{def:fas-weakly-secure}).

\end{theorem}

\begin{proof}
Follows from~\Cref{lemma:fas-ad-sound,lemma:fas-pre-sig-validity,lemma:fas-unf,lemma:fas-wit-ext,lemma:fas-pre-sig-adaptability,lemma:fas-wit-ind}.
\end{proof}

\input{fas-proof-func-wit-ind}

%% file: fas-proof-func-wit-ind.tex
\subsection{Witness Indistinguishability}
% \nikhil{TODO: update proof}
\begin{lemma}
Suppose $\nizk$ satisfies adaptive zero-knowledge and $\ipe$ satisfies selective, IND-security. 
Then, the functional adaptor signature construction in~\Cref{sec:construction} is witness indistinguishable.
\label{lemma:fas-wit-ind}
\end{lemma}

\begin{proof}
We proof this through a sequence of games $G_0^b$ and $G_1^b$ for $b \in \{0,1\}$.

\paragraph{Game $G_0^b$:}
This game corresponds to the original game 
$\faWitInd^b_{\A, \FAS}$.
 % where \nikhil{explain intuitively}.
Formally, the game is defined in~\Cref{fig:fasig-wit-ind-games}.
For sake of simplicity, we drop the oracle $\mcal{O}_\AuxGen$ as it is redundant in context of the construction.

\paragraph{Game $G_1^b$:} 
This game is same as $G_0^b$ except that NIZK simulator is used instead of NIZK prover. 
Formally, the game is defined in~\Cref{fig:fasig-wit-ind-games}.

% \nikhil{remove GroupGen specification. keep it generic for the generic construction.}

\ifacm
	\begin{figure}[t]
	\centering
	\captionsetup{justification=centering}
	\begin{gameproof}
	\begin{pchstack}[boxed, space=1em]
	\procedure[linenumbering, mode=text]{Games $G^b_0$, \pcbox{\text{$G^b_1$}}.}{
	$\crs \gets \nizk.\Setup(1^\secparam)$
	\pcbox{\text{
	$\crs \gets \nizk.\Setup^*(1^\secparam)$	
	}}
	\\ 
	$\pp' \gets \ipfe.\Gen(1^\secparam)$
	\\ 
	$\pp = (\crs, \pp')$
	\\ 
	$(\vk, X, \bfx_0, \bfx_1) \gets \A(\pp)$
	\\ 
	$\text{if } (((X, \bfx_0) \notin R) \vee ((X, \bfx_1) \notin R)): \text{\ret } 0$  
	\\ 
	Sample random coins $r_0, r_1$
	\\
	$(\mpk, \msk) := \ipe.\Setup(\pp', 1^\ell; r_0)$
	\\ 
	$\ct := \ipe.\Enc(\mpk, \bfx_b; r_1)$
	\\ 
	$\pi \gets \nizk.\Prove(\crs, (X, \pp', \mpk, \ct), (r_0, r_1, \bfx_b))$
	\pcskipln \\
	\pcbox{\text{
	$\pi \gets \nizk.\Prove^*(\crs, (X, \pp', \mpk, \ct))$
	}}
	\\ 
	$\advt := (\mpk, \ct, \pi)$, $\state := \msk$
	\\ 
	$(m, \bfy, \aux_\bfy, \pi_\bfy, \widetilde{\sigma}) \gets \A(\advt)$ 
	\\ 
	$\pk_\bfy := \ipe.\PubKGen(\mpk, \bfy)$
	\\ 
	If $\as.\PreVerify(\vk, m, \pk_\bfy, \widetilde{\sigma}) = 0$: \ret $0$
	\\ 
	If $f_\bfy(\bfx_0) \neq f_\bfy(\bfx_1)$: \ret $0$
	\\ 
	$\sk_\bfy := \ipe.\KGen(\msk, \bfy)$
	\\ 
	$\sigma = \as.\Adapt(\vk, m, \pk_\bfy, \sk_\bfy, \widetilde{\sigma})$
	\\ 
	$b' \gets \A(\sigma)$ 
	\\ 
	\ret $b'$ 
	} 
	\end{pchstack}
	\end{gameproof}
	\caption{Witness Indistinguishability proof: Games $G^b_0$ and $G^b_1$ for $b \in \{0,1\}$
	}
	\label{fig:fasig-wit-ind-games}
	\end{figure}
\else 
	\begin{figure}[h]
	\centering
	\captionsetup{justification=centering}
	\begin{gameproof}
	\begin{pchstack}[boxed, space=1em]
	\procedure[linenumbering, mode=text]{Games $G^b_0$, \pcbox{\text{$G^b_1$}}.}{
	$\crs \gets \nizk.\Setup(1^\secparam)$
	\pcbox{\text{
	$\crs \gets \nizk.\Setup^*(1^\secparam)$	
	}}
	\\ 
	$\pp' \gets \ipfe.\Gen(1^\secparam)$
	\\ 
	$\pp = (\crs, \pp')$
	\\ 
	$(\vk, X, \bfx_0, \bfx_1) \gets \A(\pp)$
	\\ 
	$\text{if } (((X, \bfx_0) \notin R) \vee ((X, \bfx_1) \notin R)): \text{\ret } 0$  
	\\ 
	Sample random coins $r_0, r_1$
	\\
	$(\mpk, \msk) := \ipe.\Setup(\pp', 1^\ell; r_0)$
	\\ 
	$\ct := \ipe.\Enc(\mpk, \bfx_b; r_1)$
	\\ 
	$\pi \gets \nizk.\Prove(\crs, (X, \pp', \mpk, \ct), (r_0, r_1, \bfx_b))$
	\pcskipln \\
	\pcbox{\text{
	$\pi \gets \nizk.\Prove^*(\crs, (X, \pp', \mpk, \ct))$
	}}
	\\ 
	$\advt := (\mpk, \ct, \pi)$, $\state := \msk$
	\\ 
	$(m, \bfy, \aux_\bfy, \pi_\bfy, \widetilde{\sigma}) \gets \A(\advt)$ 
	\\ 
	$\pk_\bfy := \ipe.\PubKGen(\mpk, \bfy)$
	\\ 
	If $\as.\PreVerify(\vk, m, \pk_\bfy, \widetilde{\sigma}) = 0$: \ret $0$
	\\ 
	If $f_\bfy(\bfx_0) \neq f_\bfy(\bfx_1)$: \ret $0$
	\\ 
	$\sk_\bfy := \ipe.\KGen(\msk, \bfy)$
	\\ 
	$\sigma = \as.\Adapt(\vk, m, \pk_\bfy, \sk_\bfy, \widetilde{\sigma})$
	\\ 
	$b' \gets \A(\sigma)$ 
	\\ 
	\ret $b'$ 
	} 
	\end{pchstack}
	\end{gameproof}
	\caption{Witness Indistinguishability proof: Games $G^b_0$ and $G^b_1$ for $b \in \{0,1\}$
	}
	\label{fig:fasig-wit-ind-games}
	\end{figure}
\fi
To prove the lemma, we need to show that 
\[
|
\Pr[ G^0_0(1^\secparam) = 1] 
-
\Pr[ G^1_0(1^\secparam) = 1] 
| \leq \negl(\secparam),
\]
Note that by triangle inequality, it follows that 
\begin{align*}
& | \Pr[ G^0_0(1^\secparam) = 1] - \Pr[ G^1_0(1^\secparam) = 1] | \\
& \quad \leq 
| \Pr[ G^0_0(1^\secparam) = 1] - \Pr[ G^0_1(1^\secparam) = 1] | \\
& \quad \quad +  
| \Pr[ G^0_1(1^\secparam) = 1] - \Pr[ G^1_1(1^\secparam) = 1] | \\
& \qquad +  
| \Pr[ G^1_1(1^\secparam) = 1] - \Pr[ G^1_0(1^\secparam) = 1] |. 
\end{align*}
To complete the proof, 
we show in~\Cref{claim:fas-wit-ind-g1} 
that the first and third terms on the right-hand-side are 
$\leq \negl(\secparam)$ 
and in~\Cref{claim:fas-wit-ind-g1-final} 
that the second term on the right-hand-side is 
$\leq \negl(\secparam)$.

\end{proof}

\begin{claim}
If $\nizk$ satisfies adaptive zero-knowledge, then, 
there exists a negligible function $\negl$ such that for all $\secparam \in \N$,
for all $b \in \{0,1\}$, 
\[
| \Pr[ G^b_0(1^\secparam) = 1] - \Pr[ G^b_1(1^\secparam) = 1] | \leq \negl(\secparam).
\]
\label{claim:fas-wit-ind-g1}
\end{claim}

\begin{proof}

We show that if there exists a \ppt adversary \A can distinguish its view in games $G^b_0$ and $G^b_1$, then, we can create a distinguisher $\D^b$ that breaks the adaptive zero-knowledge of the underlying $\nizk$. Let \C be the adaptive zero-knowledge challenger for $\nizk$. $\C$ samples a random bit $\beta \getr \{0,1\}$ and if $\beta=0$, it samples $\crs$ and $\pi$ as in the real-world and if $\beta=1$, it samples $\crs$ and $\pi$ using the $\nizk$ simulator algorithms $\Setup^*$ and $\Prove^*$ respectively.
Note that as the adaptive zero-knowledge of the NIZK proof system holds for all valid statement-witness pairs, we will allow the distinguisher $\D^b$ to choose a valid statement-witness pair adaptively after seeing $\crs$.
The distinguisher $\D^b$ must output a bit at the end and it will use the adversary \A for this task. 
The distinguisher $\D^b$ interacts with \A and either presents it the view as in $G^b_0$ or as in $G^b_1$ and at the end \A outputs its guess bit $\beta' \in \{0,1\}$. The distinguisher $\D^b$ is as in~\Cref{fig:fasig-wit-ind-g1-proof}.

\ifacm
\begin{figure}[t]
\else
\begin{figure}[h]
\fi
\centering
\captionsetup{justification=centering}
\begin{pchstack}[boxed, space=1em]
\procedure[linenumbering, mode=text]{Distinguisher $\D^b$ for proof of~\Cref{claim:fas-wit-ind-g1}}{
{\color{blue}
$\crs \gets \C(1^\secparam)$
}
\\ 
$\pp' \gets \ipfe.\Gen(1^\secparam)$
\\ 
$\pp = (\crs, \pp')$
\\ 
$(\vk, X, \bfx_0, \bfx_1) \gets \A(\pp)$
\\ 
$\text{if } (((X, \bfx_0) \notin R) \vee ((X, \bfx_1) \notin R)): \text{\ret } 0$  
\\ 
Sample random coins $r_0, r_1$
\\
$(\mpk, \msk) := \ipe.\Setup(\pp', 1^\ell; r_0)$
\\ 
$\ct := \ipe.\Enc(\mpk, \bfx_b; r_1)$
\\ 
{\color{blue}
$\pi \gets \C((X, \pp', \mpk, \ct), (r_0, r_1, \bfx_b))$
}
\\ 
$\advt := (\mpk, \ct, \pi)$, $\state := \msk$
\\ 
$(m, \bfy, \aux_\bfy, \pi_\bfy, \widetilde{\sigma}) \gets \A(\advt)$ 
\\ 
$\pk_\bfy := \ipe.\PubKGen(\mpk, \bfy)$
\\ 
If $\as.\PreVerify(\vk, m, \pk_\bfy, \widetilde{\sigma}) = 0$: \ret $0$
\\ 
If $f_\bfy(\bfx_0) \neq f_\bfy(\bfx_1)$: \ret $0$
\\ 
$\sk_\bfy := \ipe.\KGen(\msk, \bfy)$
\\ 
$\sigma = \as.\Adapt(\vk, m, \pk_\bfy, \sk_\bfy, \widetilde{\sigma})$
\\  
$\beta' \gets \A(\sigma)$ 
\\ 
\ret $\beta'$ 
} 

\end{pchstack}
\caption{Distinguisher $\D^b$ for proof of~\Cref{claim:fas-wit-ind-g1}
}
\label{fig:fasig-wit-ind-g1-proof}
\end{figure}

Observe that if \C chose bit $\beta$, then, the view observed by \A is that of game $G^b_\beta$.
Hence, it follows that 
\begin{align*}
& |\Pr[\D^b(\crs, \pi) = 1| \beta = 0] - \Pr[\D^b(\crs, \pi) = 1| \beta = 1]| \\ 
& \quad =  
|\Pr[G^b_0(1^\secparam) = 1] - \Pr[G^b_0(1^\secparam) = 1]| . 
\end{align*}
Hence, if \A has a noticeable distinguishing advantage $\epsilon$, then, $\D^b$ also has a noticeable distinguishing advantage $\epsilon$ in its game with challenger \C. This completes the proof.
\end{proof}

\begin{claim}
If $\ipe$ satisfies selective, IND-security, then,
there exists a negligible function $\negl$ such that for all $\secparam \in \N$,
\[
| \Pr[ G^0_1(1^\secparam) = 1] - \Pr[ G^1_1(1^\secparam) = 1] | \leq \negl(\secparam).
\]
\label{claim:fas-wit-ind-g1-final}
\end{claim}

\begin{proof}

We show that if there exists a \ppt adversary \A can distinguish its view in games $G^0_1$ and $G^1_1$, then, we can create a reduction \B that breaks the selective, IND-security of the underlying $\ipe$. 
For the $\ipe$ selective, IND-security game, 
let \C be the  challenger and let $\ipe.\mcal{O}_\KGen$ be the key generation oracle that the reduction \B has access to. 
The reduction \B is as in~\Cref{fig:fasig-wit-ind-g1-final-proof}.

\ifacm
\begin{figure}[t]
\else
\begin{figure}[h]
\fi
\centering
\captionsetup{justification=centering}
\begin{pchstack}[boxed, space=1em]
\procedure[linenumbering, mode=text]{Reduction \B for proof of~\Cref{claim:fas-wit-ind-g1-final}}{
$\crs \gets \nizk.\Setup^*(1^\secparam)$
\\ 
{\color{blue}
$\pp' \gets \C(1^\secparam)$
}
\\ 
$\pp = (\crs, \pp')$
\\ 
$(\vk, X, \bfx_0, \bfx_1) \gets \A(\pp)$
\\ 
$\text{if } (((X, \bfx_0) \notin R) \vee ((X, \bfx_1) \notin R)): \text{\ret } 0$  
\\ 
{\color{blue}
$(\mpk, \ct) \gets \C(\bfx_0, \bfx_1)$
}
\\ 
$\pi \gets \nizk.\Prove^*(\crs, (X, \pp', \mpk, \ct))$
\\ 
$\advt := (\mpk, \ct, \pi)$, $\state := \msk$
\\ 
$(m, \bfy, \aux_\bfy, \pi_\bfy, \widetilde{\sigma}) \gets \A(\advt)$ 
\\ 
$\pk_\bfy := \ipe.\PubKGen(\mpk, \bfy)$
\\ 
If $\as.\PreVerify(\vk, m, \pk_\bfy, \widetilde{\sigma}) = 0$: \ret $0$
\\ 
If $f_\bfy(\bfx_0) \neq f_\bfy(\bfx_1)$: \ret $0$
\\ 
{\color{blue}
$\sk_\bfy = \ipe.\mcal{O}_\KGen(\bfy)$
}
\\ 
$\sigma = \as.\Adapt(\vk, m, \pk_\bfy, \sk_\bfy, \widetilde{\sigma})$
\\ 
$b' \gets \A(\sigma)$ 
\\ 
\ret $b'$ 
} 

\end{pchstack}
\caption{Reduction \B for proof of~\Cref{claim:fas-wit-ind-g1-final}
}
\label{fig:fasig-wit-ind-g1-final-proof}
\end{figure}

Observe that when \C plays game ${\sf IPFE\text{-}Expt\text{-}SEL}_\A^b(1^\secparam, n)$ with the reduction \B for some $b \in \{0,1\}$, then, $\ct$ encrypts $\bfx_b$ and thus, \A's view is same as $G^b_1$.
Further, observe that \B makes query to the key generation oracle only if $f_\bfy(\bfx_0)  = f_\bfy(\bfx_1)$. Thus, \B is an admissible adversary in the game with \C. Hence, it follows that for all $b \in \{0,1\}$,
\begin{align*}
\Pr[{\sf IPFE\text{-}Expt\text{-}SEL}_\B^b(1^\secparam, n) = 1]
& = \Pr[b' = 1] \\
& = \Pr[G^b_1(1^\secparam) = 1].
\end{align*}
Therefore, we get that 
\ifacm
	\begin{align*}
	& |
	\Pr[{\sf IPFE\text{-}Expt\text{-}SEL}_\A^0(1^\secparam, n) = 1] \\ 
	& \quad - \Pr[{\sf IPFE\text{-}Expt\text{-}SEL}_\A^1(1^\secparam, n) = 1]
	| \\
	& \qquad = 
	| \Pr[ G^0_1(1^\secparam) = 1] - \Pr[ G^1_1(1^\secparam) = 1] |.
	\end{align*}
\else 
	\begin{align*}
	& |
	\Pr[{\sf IPFE\text{-}Expt\text{-}SEL}_\A^0(1^\secparam, n) = 1]
	- \Pr[{\sf IPFE\text{-}Expt\text{-}SEL}_\A^1(1^\secparam, n) = 1]
	| \\
	& \qquad = 
	| \Pr[ G^0_1(1^\secparam) = 1] - \Pr[ G^1_1(1^\secparam) = 1] |.
	\end{align*}
\fi
Hence, if \A  has a noticeable distinguishing advantage $\epsilon$ in its game with reduction \B, then, \B also has a noticeable distinguishing advantage $\epsilon$ in its game with challenger \C. This completes the proof.
\end{proof}